\documentclass[acmsmall]{acmart}
\renewcommand\footnotetextcopyrightpermission[1]{} 
\usepackage[utf8]{inputenc}

\newcommand{\figref}[1]{Figure~\ref{#1}}
\newcommand{\secref}[1]{Section~\ref{#1}}
\newcommand{\appref}[1]{Appendix~\ref{#1}}
\newcommand{\thmref}[1]{Theorem~\ref{#1}}
\newcommand{\lemref}[1]{Lemma~\ref{#1}}

\newcommand{\corref}[1]{Corollary~\ref{#1}}

\DeclareMathOperator{\Prop}{\coqdockw{Prop}}
\DeclareMathOperator{\Type}{\coqdockw{Type}}

\newcommand{\ptranslate}[1]{\ensuremath{\llbracket}#1\ensuremath{\rrbracket}}
\newcommand{\ptranslateIso}[1]{\ensuremath{\llbracket}#1\ensuremath{\rrbracket_{iso}}}
\newcommand{\ptranslateStrongIso}[1]{\ensuremath{\llbracket}#1\ensuremath{\rrbracket_{sIso}}}
\newcommand{\ptranslateC}[2]{\ensuremath{\llbracket}#1\ensuremath{\rrbracket_{#2}}}
\newtheorem{thm}{Theorem}
\newcommand{\indexedVar}[2]{#1#2}
\usepackage{amsmath}
\usepackage{bussproofs}
\usepackage{color,soul}
\newenvironment{myquote}[1]%
  {\list{}{\leftmargin=#1\rightmargin=#1}\item[]}%
  {\endlist}
\usepackage{mathdots}
\newcommand{\tprime}[1]{{#1}{₂}}
\newcommand{\trel}[1]{{#1}{ᵣ}}
\newcommand{\coqdocvarP}[1]{\coqdocvar{\tprime{#1}}}
\newcommand{\coqdocvarR}[1]{\coqdocvar{\trel{#1}}}
\newcommand{\coqdocvarFour}[1]{\coqdocvar{{#1}\ensuremath{_4}}}
\newcommand{\coqdocvarFive}[1]{\coqdocvar{{#1}\ensuremath{_5}}}
\newcommand{\isorel}{IsoRel}
\newcommand{\cocminus}{CoC\ensuremath{^-}}
\newcommand{\anyrel}{AnyRel}
\newcommand{\projTyRel}[2]{\ensuremath{\pi_{#1}} #2}
\newcommand{\BCGamma}[2]{\coqdocdefinition{safe}\ensuremath{_{{#1}}}{#2}}
\newcommand{\BCGammaC}[1]{\coqdocdefinition{safeC} {#1}}
\newcommand{\bcsubst}[4]{{#2}\{{#4}/{#3}\}\ensuremath{_{#1}}}
\newcommand{\fvars}[1]{\coqdocdefinition{freeVars} #1}
\newcommand{\bvars}[1]{\coqdocdefinition{boundVars} #1}
\newcommand{\alphaeq}[2]{#1 =\ensuremath{_\alpha} #2}
\newcommand{\coqsubscript}{i}
\newcommand{\vdashq}{\ensuremath{\vdash_{\coqsubscript}}}

\newcommand{\coqRefDefnR}[2]
{\coqref{#1.#2 R}{\coqdocdefinition{\trel{#2}}}}
\newcommand{\coqRefConstrR}[2]
{\coqref{#1.#2 R}{\coqdocconstructor{\trel{#2}}}}

\usepackage{fixme}
\fxsetup{
    status=final,
    author=,
    layout=marginnote, 
    theme=color
}
\usepackage{cancel}
\newcommand{\coqRefInductive}[2]
{\coqref{#1.#2}{\coqdocinductive{#2}}}
\newcommand{\coqRefDefn}[2]
{\coqref{#1.#2}{\coqdocdefinition{#2}}}
\newcommand{\coqRefConstr}[2]
{\coqref{#1.#2}{\coqdocconstructor{#2}}}
\newcommand{\coqRefVar}[2]
{\coqref{#1.#2}{\coqdocvar{#2}}}

\newcommand{\CoqList}{\coqexternalref{list}{http://coq.inria.fr/stdlib/Coq.Init.Datatypes}{\coqdocinductive{list}}}
\newcommand{\CoqNat}{\coqexternalref{nat}{http://coq.inria.fr/stdlib/Coq.Init.Datatypes}{\coqdocinductive{nat}}}

\newcommand{\CoqMap}{\coqexternalref{map}{http://coq.inria.fr/stdlib/Coq.Lists.List}{\coqdocdefinition{map}}}

\newcommand{\CoqBool}{\coqexternalref{bool}{http://coq.inria.fr/stdlib/Coq.Init.Datatypes}{\coqdocinductive{bool}}}
\newcommand{\CoqBTrue}{\coqexternalref{true}{http://coq.inria.fr/stdlib/Coq.Init.Datatypes}{\coqdocconstructor{true}}}
\newcommand{\CoqBFalse}{\coqexternalref{false}{http://coq.inria.fr/stdlib/Coq.Init.Datatypes}{\coqdocconstructor{false}}}

\newcommand{\CoqTrue}{\coqexternalref{True}{http://coq.inria.fr/stdlib/Coq.Init.Logic}{\coqdocinductive{True}}}
\newcommand{\CoqFalse}{\coqexternalref{False}{http://coq.inria.fr/stdlib/Coq.Init.Logic}{\coqdocinductive{False}}}
\newcommand{\CoqTrueI}{\coqexternalref{I}{http://coq.inria.fr/stdlib/Coq.Init.Logic}{\coqdocconstructor{I}}}
\newcommand{\CoqFalseRect}{\coqexternalref{False}{http://coq.inria.fr/stdlib/Coq.Init.Logic}{\coqdocdefinition{False\_rect}}}
\newcommand{\CoqExistT}{\coqexternalref{existT}{http://coq.inria.fr/stdlib/Coq.Init.Specif}{\coqdocconstructor{existT}}}
\newcommand{\CoqSigTProj}{\coqexternalref{projT1}{http://coq.inria.fr/stdlib/Coq.Init.Specif}{\coqdocdefinition{\ensuremath{\pi_1}}}}
\newcommand{\CoqNone}{\coqexternalref{None}{http://coq.inria.fr/stdlib/Coq.Init.Datatypes}{\coqdocconstructor{None}}}
\newcommand{\CoqConj}{\coqexternalref{:type scope:x '/x5C' x}{http://coq.inria.fr/distrib/8.5pl3/stdlib/Coq.Init.Logic}{\coqdocnotation{\ensuremath{\land}}}}
\newcommand{\CoqAppend}{\coqexternalref{app}{https://coq.inria.fr/stdlib/Coq.Init.Datatypes}{++}}

\usepackage{microtype}
\DisableLigatures{encoding = *, family = * }
\usepackage{booktabs}   
\usepackage{subcaption} 
\usepackage[color]{coqdoc}
\colorlet{DEFGREEN}{defgreen}
   
\usepackage{common-unicode}

\input{tikZHighlighter.tex}

\makeatletter\if@ACM@journal\makeatother
\acmJournal{PACMPL}
\acmVolume{1}
\acmNumber{1}
\acmArticle{1}
\acmYear{2017}
\acmMonth{1}
\acmDOI{10.1145/nnnnnnn.nnnnnnn}
\startPage{1}
\else\makeatother
\acmConference[PL'17]{ACM SIGPLAN Conference on Programming Languages}{January 01--03, 2017}{New York, NY, USA}
\acmYear{2017}
\acmISBN{978-x-xxxx-xxxx-x/YY/MM}
\acmDOI{10.1145/nnnnnnn.nnnnnnn}
\startPage{1}
\fi

\setcopyright{none}             

\begin{document}

\title{Revisiting Parametricity: Inductives and Uniformity of Propositions}
%


\author{Abhishek Anand}
\authornote{abhishek.anand.iitg@gmail.com}          
\affiliation{
  \department{Department of Computer Science}              
  \institution{Cornell University}            
  \city{Ithaca}
  \state{NY}
  \country{USA}
}
\email{abhishek.anand.iitg@gmail.com}          

\author{Greg Morrisett}
\affiliation{
  \department{Department of Computer Science}              
  \institution{Cornell University}            
  \city{Ithaca}
  \state{NY}
  \country{USA}
}
\email{greg.morrisett@cornell.edu}         


\begin{abstract}
Reynold's parametricity theory captures the property that parametrically polymorphic functions behave uniformly: they produce related results on related instantiations.
In dependently typed programming languages, such relations and uniformity proofs can be expressed internally, and generated as a program translation. 

We present a new parametricity translation for a significant fragment of Coq. 
Previous translations of parametrically polymorphic propositions allowed non-uniformity. 
For example, on related instantiations, a function may return propositions that are logically inequivalent (e.g. True and False).
We show that uniformity of polymorphic propositions is not achievable in general.
Nevertheless, our translation produces proofs that the two propositions are 
\emph{logically equivalent} and also that \emph{any} two proofs of those propositions are related.
This is achieved at the cost of potentially requiring more assumptions on the
instantiations, requiring them to be isomorphic in the worst case.

Our translation augments the previous one for Coq by carrying and
compositionally building extra proofs about parametricity
relations. It is made easier by a new method for translating inductive types and pattern matching.
The new method builds upon and generalizes previous such translations for dependently typed programming languages. 
 
Using reification and reflection, we have implemented our translation as Coq
programs\footnote{\url{https://github.com/aa755/paramcoq-iff}}.
We obtain several stronger free theorems applicable to an ongoing
compiler-correctness project. Previously, proofs of some of these theorems took
several hours to finish.
\end{abstract}

\maketitle

\section{Introduction}
\label{sec:intro}
\citet{Krishnaswami.Dreyer2013} summarize Reynold's work on parametricity in the
following perfect way:
\begin{myquote}{0.1cm}
\emph{
\citet{Reynolds1983} famously introduced the concept of
relational parametricity with a fable about data abstraction.
Professors Bessel and Descartes, each teaching a class on complex numbers, defined them differently in the first lecture, the former using polar coordinates and the latter using (of course) cartesian coordinates. 
But despite accidentally trading sections after the first lecture, 
they never taught their students anything false, since after the first class, both professors proved all their theorems in terms of the defined operations on complex numbers, 
and never in terms of their underlying coordinate representation.}

\emph{
Reynolds formalized this idea by giving a semantics for System F in which each type denoted not just a set of well-formed terms, but a logical relation between them, defined recursively on the type structure of the language. 
Then, the fact that well-typed client programs were insensitive to a specific choice of implementation could be formalized in terms of their taking logically related inputs to logically related results. 
Since the two constructions of the complex numbers share the same interface, and it is easy to show they are logically related at that interface, 
any client of the interface must return equivalent results regardless of which
implementation of the interface is used.}
\end{myquote}


In Reynold's work 
 and subsequent work for 
other modern languages (e.g. OCaml~\citep{Crary2017}), 
the logical relations for types are meta-theoretic
(not defined in the programming language being studied).
In contrast, in dependently typed programming languages such as Coq, one can
express within the language such logical relations and the proofs that programs are related.
Thus, recent works~\citep{Bernardy.Jansson.ea2010, Bernardy.Jansson.ea2012, Keller.Lasson2012} have defined program translations that 
translate types to their logical relations. 
Because terms can appear in types in dependently typed languages, these
translations translate both terms and types. An amazing aspect of the translation of terms is that it produces
proofs of the corresponding abstraction theorems:
Let \ptranslate{$T$} denote the parametricity translation of the type $T$.
\coqdocnoindent
For closed terms $t$ and $T$, if $t$:$T$ ($t$ has type $T$) in System F, 
Reynold's abstraction theorem says that
($t$,$t$) is in the relation \ptranslate{$T$}.
The proof of this theorem is in the meta-theory. 
In contrast, in Coq, amazingly, the proof is precisely \ptranslate{$t$}, the
translation of $t$.

In Coq, parametricity is a powerful tool to obtain not only statements of
free theorems~\cite{Wadler1989}, but also free Coq proofs of those theorems.
In our recent compiler correctness project, we have used the
implementation\footnote{\url{https://github.com/mlasson/paramcoq}} by \citet{Keller.Lasson2012}
 to automatically obtain for free several Coq proofs that
otherwise took many hours to manually write.
 For example, by polymorphically defining the big-step operational semantics of some intermediate languages, we were able to obtain for free~(\secref{sec:app}) that the semantics are preserved when we change the
 representation from de Bruijn (the representation used in the compiler's
 source language) to named-variable bindings (the representation used in the
 backend).
However, as we explain below, ~\citeauthor{Keller.Lasson2012}'s translation
produces useless abstraction theorems for polymorphic
propositions or relations. 
\emph{In contrast, our translation lets us obtain the uniformity of even
the polymorphically defined, \emph{undecidable} relations (e.g. observational
equivalence).}

Undecidable relations are particularly problematic
because, as we explain next, they cannot be equivalently redefined in a
way that allows reaping the benefits of the existing
translation~\citep{Keller.Lasson2012}.
In proof assistants such as Coq, some amount of logic can be done using the boolean datatype.
A predicate over a type $X$ can be represented as a function of type 
$X$ $\rightarrow$ {\CoqBool}.
Given a polymorphic function, say \coqdocvar{f}, returning a {\CoqBool},
Coq's parametricity translation produces a proof that on different, parametrically related instantiations, 
\coqdocvar{f} will produce the \emph{same} boolean value.
However, undecidable predicates (or n-ary relations in general) cannot be defined this way,
because Coq functions are \emph{computable} : a term of type {\CoqBool}
\emph{must} eventually \emph{compute} to one of the two boolean values: {\CoqBTrue}, 
{\CoqBFalse}.
One can cheat and use a strong version of the axiom of excluded middle to make such definitions.
However, the axiom is provably 
\emph{non-parametric}~\cite[Sec.~5.4.2]{Keller.Lasson2012}.
Hence parametricity translations cannot generate abstraction theorems for 
definitions using the axiom.

Proof Assistants based on dependent types (e.g. Agda~\cite{Norell2009}, Coq,
F*~\cite{Swamy.Hritcu.ea2016},
Idris~\citep{Brady2013}, LEAN~\cite{deMoura.Kong.ea2015},
Nuprl~\citep{Constable.Allen.ea1986}) have another, perhaps more idiomatic
mechanism for defining propositions/relations.
For
example, dependent function types can be used to express universal quantification.
Using such quantification, one can easily define \emph{undecidable} relations.
An $n$-ary relation is just a function that takes $n$ arguments and returns a
proposition.
In Coq, \coqdockw{Prop} is a special universe
whose inhabitants are intended to be types denoting logical propositions.
In the ``propositions as types, proofs as programs" tradition, by ``\coqdocvar{P} is a proposition'',
we mean 
\coqdocvar{P}:\coqdockw{Prop}, 
 and by ``\coqdocvar{p} is a proof of \coqdocvar{P}'', we mean
\coqdocvar{p}:\coqdocvar{P}.

Propositions enjoy a special status in Coq. 
For example, by restricting pattern matching on proofs (\secref{sec:anyrel:compare}),  Coq ensures
that one can consistently assume the proof irrelevance axiom which says that any two proofs of 
a proposition are equal. Also, as a result, Coq's compiler can erase all proofs~\cite{Letouzey2004} to a dummy term.

The existing parametricity translation~\citep{Keller.Lasson2012} translates propositions and
proofs as well.
However, propositions are treated just like other types, and proofs are treated just like members of other types.
As a result, \coqdockw{Prop}, which is a universe and whose inhabitants are propositions (types), is treated differently
than {\CoqBool} which is not a universe, and whose members are not types: they are 
mere data constructors: {\CoqBTrue} and {\CoqBFalse}.
\ptranslate{{\CoqBool}}, the parametricity relation for the type {\CoqBool}  relates
{\CoqBTrue} with {\CoqBTrue}
and {\CoqBFalse} with {\CoqBFalse},
and relates nothing else.
%
In contrast, propositions (types) $P1$ and $P2$ are related by 
\ptranslate{\coqdockw{Prop}}
if there is \emph{any} relation, say $R$, between the proofs of
$P1$ and $P2$.
Note that there exist relations even between logically inequivalent types.
For example, 
$\lambda$ (\coqdocvar{t} : {\CoqTrue}) (\coqdocvar{f} : {\CoqFalse}),
 {\CoqTrue} is a relation between the propositions
 {\CoqTrue}\footnote{{\CoqTrue} and {\CoqBTrue}
 are not the same. {\CoqTrue}:\coqdockw{Prop} and 
 {\CoqBTrue}:{\CoqBool}. {\CoqTrue} is an inductively defined
 proposition with only one constructor. {\CoqBTrue} is already a
 data constructor. {\CoqFalse} is an inductively defined proposition with 
 \emph{no} constructor.\\ 
 For convenience, mentions of Coq constants
 are usually hyperlinked to their definition, if defined 
 in this paper or in Coq's standard library.
 Also, to take advantage of syntax
 highlighting, we recommend reading this paper in color.
 } and
 {\CoqFalse}.
This means that polymorphically defined propositions may have 
logically inequivalent meanings in related instantiations.
Thus, abstraction theorems for polymorphic propositions, as
generated by the existing parametricity translation~\cite{Keller.Lasson2012}, are useless.

In the context of the previous paragraph, 
the main advantage of our translation is that it
additionally ensures/requires:
\begin{enumerate}
  \item logical equivalence of the related propositions: $P1$ $\leftrightarrow$ $P2$
  \item triviality of the relation
  : $\forall$ (\coqdocvar{p1}:$P1$) (\coqdocvar{p2}:$P2$), 
   $R$ \coqdocvar{p1} \coqdocvar{p2}
\end{enumerate}
Here, $R$ is the relation between the proofs of the
propositions $P1$ and $P2$.
The usefulness of the first property was already explained above.
The second is useful when instantiating an interface that includes proofs.
For example, an interface describing a semigroup (in abstract algebra) 
in Coq may also contain fields representing the proofs of associativity equations.
To use parametricity to obtain free proofs that polymorphic functions over
semigroups behave uniformly, one needs to provide two instantiations of the semigroup interface,
and prove that all the fields, including the proof fields, are related.
The triviality property makes it trivial to prove that the proof fields
are related. 
Previously, it took one
of us several 
hours 
to do one of these proofs.
The Appendix~(\secref{appendix:intro:triv}) provides a Coq statement of
the proof, in case the reader wants to independently assess the difficulty.

There is a cost to achieving the above two properties for polymorphic
propositions: our abstraction theorem may make stronger assumptions in some cases.
For example, consider Coq's polymorphic equality proposition, which is defined using
indexed induction:

\coqdockw{Inductive} \coqdef{Coq.Init.Logic.eq}{eq}{\coqdocinductive{eq}} (\coqdocvar{T}:\coqdockw{Type}) (\coqdocvar{x}:\coqdocvariable{T}) : \coqdocvar{T} \coqref{Coq.Init.Logic.:type scope:x '->' x}{\coqdocnotation{\ensuremath{\rightarrow}}} \coqdockw{Prop} :=
\coqdef{Coq.Init.Logic.eq refl}{eq\_refl}{\coqdocconstructor{eq\_refl}} : \coqdocinductive{eq} \coqdocvar{T} \coqdocvar{x} \coqdocvar{x}
\coqdoceol
\coqdocnoindent 
This syntax says that \coqRefInductive{Coq.Init.Logic}{eq} is a family of propositions (types)
and for \emph{any} type \coqdocvariable{T} and \coqdocvariable{x} of type \coqdocvariable{T},
\coqRefConstr{Coq.Init.Logic}{eq\_refl} is a proof that \coqdocvariable{x} is equal to itself.
Because Coq's typehood judgements are preserved under computation, for closed \coqdocvariable{x} and \coqdocvariable{y},
the proposition 
\coqRefInductive{Coq.Init.Logic}{eq} \coqdocvariable{T} \coqdocvariable{x} \coqdocvariable{y}
asserts that the normal forms of \coqdocvariable{x} and \coqdocvariable{y} are
the same.
Thus Coq lets us define propositions that make logical
observations that no computation can make: by parametricity, all functions of the type
$\forall$ \coqdocvar{T}:\coqdockw{Type},
\coqdocvar{T} $\rightarrow$ \coqdocvar{T} $\rightarrow$ {\CoqBool} are \emph{constant} functions.
In~\secref{sec:uniformProp}, we see that for indexed-inductive propositions to behave uniformly, the
parametricity relation between the two instantiations of the index type may need to be one-to-one.
Also, for universal quantification ($\forall$), the relation for the quantified type may need to be total.

After analysing the uniformity requirements for Coq's mechanisms for defining new 
propositions~(Section \ref{sec:uniformProp}, \ref{sec:uniformProp:type}), we explain our new parametricity 
translation
that ensures these requirements~(\secref{sec:isorel}).
We call our new translation the {\isorel} translation because in the worst case, the two instantiations
of type variables need to be isomorphic.
In contrast, we call the old translation~\citep{Keller.Lasson2012} 
the {\anyrel} translation, because one can pick
\emph{any} relation between the two instantiations, as long as each item in the interface respects
the relation.
In this sense, Reynold's original parametricity translation of types can be
considered an {\anyrel} translation. 

Our {\isorel} translation excludes propositions that mention types
of higher universes 
(\coqdockw{Type}$_i$ for $i>0$) at certain places~(\secref{sec:isorel:limitations}). 
For example, in universal quantification, the quantified type
must be in \coqdockw{Prop} or \coqdockw{Type}$_0$, which is also denoted by \coqdockw{Set} in
Coq. Also, in inductively defined propositions, the types of indices
and the types of arguments of constructors (except the parameters of the type)
must be in \coqdockw{Set} or
\coqdockw{Prop}.
%
%
\coqdockw{Set} and \coqdockw{Prop}
suffice for many concrete applications, such as
correctness of computer systems (e.g. compilers, operating systems) and cyber-physical systems.
For example, in \coqdockw{Set}, one can define natural, rational, and real numbers, functions and infinitely branching
trees of real numbers, and abstract syntax trees used by compilers.
Our restrictions may be problematic for some applications such 
as proving the consistency of powerful logics~\cite{Anand.Rahli2014a}.

The {\anyrel} translation serves as a core of our {\isorel} translation.
The {\isorel} translation adds extra proofs about the {\anyrel} translations of types
and propositions.
The main challenge is to compositionally
build the extra  proofs of new type and proposition constructions from
the corresponding proofs of their subcomponents.
Because understanding the {\anyrel} translation is crucial for understanding our
{\isorel} translation, we first present our version of the {\anyrel} translation
in \secref{sec:anyrel}.
Our {\anyrel} translation is similar to the one 
by~\citet{Keller.Lasson2012}, except for the translation of inductive types and pattern matching.
Our {\anyrel} translation of inductive types~(\secref{sec:anyrel:ind}) and 
pattern matching~(\secref{sec:anyrel:match})
simplifies our {\isorel} translation because it allows us 
to use the \coqdockw{Prop} universe for defining the parametricity relations of those types.
As explained above, the \coqdockw{Prop} universe is well-suited for defining logical relations.
Our {\anyrel} translation of inductive types and pattern matching 
is inspired by a translation by
~\citet[Sec 5.4]{Bernardy.Jansson.ea2012}. However, we uncover and fix
a subtle
flaw in how they translate indexed-inductive types and pattern matching on inhabitants of those types.

\paragraph{Summary of Contributions:}
\begin{itemize}
\item For a significant fragment of Coq, a new parametricity translation
({\isorel}) that augments our version of the {\anyrel} translation to enforce the
uniformity of polymorphically defined propositions (Section~\ref{sec:uniformProp}-\ref{sec:isorel}).
The {\isorel} translation uses the proof irrelevance and function extensionality axioms.

\item For indexed-inductive types and pattern matching, a new {\anyrel}
translation~(Section \ref{sec:anyrel:ind}, \ref{sec:anyrel:match}) which
has proof-irrelevance properties
that simplify the {\isorel} translation and are also independently useful.
The {\anyrel} translation does not use any axiom.

\item An application of parametricity translations ({\anyrel}, {\isorel}) to obtain for free many tedious Coq proofs about compiler
correctness (\secref{sec:app}). We show a theorem (observational equivalence
respects $\alpha$ equality) that the {\isorel} translation can prove but the {\anyrel} cannot.
\end{itemize}
\EnableBpAbbreviations

\section{{\anyrel} Translation}
\label{sec:anyrel}
In this section, we present the {\anyrel} translation that forms the core of the 
{\isorel} translation described in the next sections.
As mentioned above, unlike the {\isorel} translation, 
the {\anyrel} translation does not ensure the uniformity of
propositions, and treats propositions (types) just like other types,
and treats proofs just like members of other types. 
First, we describe the translation of a core calculus of Coq
that excludes inductive constructions.
This core is exactly the Calculus of Constructions
(CoC)~\citep{Coquand.Huet1988}.
Although our presentation is very similar to the one
by~\citet{Keller.Lasson2012}, it highlights why we will later need a new translation for inductive types.
Then we add inductive types to the calculus and compare,
in the setting of Coq, the
existing {\anyrel} translations of inductive constructions and associated
constructs such as pattern-matching
(\secref{sec:anyrel:compare}).
Finally, we describe our new translation (Section~\ref{sec:anyrel:ind}, \ref{sec:anyrel:match}), which is
inspired by the compared translations. Our translation has proof irrelevance properties that
simplify the {\isorel} translation.
Also, we uncover and fix a subtle flaw in one of the compared translations.

\subsection{Core Calculus}
\label{sec:anyrel:core}
The following grammar describes the language of CoC (both terms and types):
\newcommand{\pouf}{\hspace{0.8em}}
$$ s \pouf := \text{\coqdockw{Prop}} \pouf|\pouf \text{\coqdockw{Type}}_{i}$$
\vspace{-0.6cm}
$$ A, B \pouf := \pouf  \coqdocvar{x}  \pouf|\pouf  s \pouf|\pouf \forall
\coqdocvar{x}:A, B \pouf|\pouf \lambda \coqdocvar{x}:A, B  \pouf|\pouf  (A\,B)
$$
where \coqdocvar{$x$} ranges over variables and $i$ ranges over natural
numbers.
$s$ denotes universes (also known as sorts in the literature).
The translation often needs four extra variables for each variable in the input.
Just to avoid capture, without loss of generality, we assume that there are five disjoint classes of variables and the input
only has variables from the first class, and has no repeated bound variables.
We assume that \coqdocvarP{}, \coqdocvarR{}, \coqdocvarFour{}, and \coqdocvarFive{} are injective functions that respectively map
variables of the first class to variables of the next four classes.
Semantic concepts such as $\alpha$-equality, reduction, typehood
are totally agnostic to this distinction between classes of variables.
Finally, for any term $A$, \tprime{$A$} denotes the term obtained by replacing
every variable \coqdocvar{v} by \coqdocvarP{v}.

For now, we define $\hat{s}$ := $s$. Let \coqdocvar{c} be some variable of the first class.

\ptranslate{}, the {\anyrel} 
parametricity translation is defined by structural recursion. 
To understand it, it may be helpful to first recall
its main correctness property:
For closed terms $t$ and $T$, if
$t$ : $T$, then
\ptranslate{$t$} must be the proof that
$t$ is related to itself in the relation 
\ptranslate{$T$}.
Relations are represented as functions that take two arguments
and return a proposition or a type.
Thus, more formally, if $t$ : $T$,
then we must have \ptranslate{$t$}: (\ptranslate{$T$} $t$ $t$).
\citet{Keller.Lasson2012} prove a more general version, for open terms in typing
contexts:
\begin{thm}[\label{Abstraction}Abstraction Theorem]
If $Γ ⊢ A : B$, then $\ptranslate{Γ} ⊢ A : B$, $\ptranslate{Γ} ⊢ \tprime{A} :
\tprime{B}$, and $\ptranslate{Γ} ⊢ \ptranslate{A} : \ptranslate{B}\, A\,
\tprime{A}$
\end{thm}
\vspace{-0.5cm}
%
\begin{align*}
  \ptranslate{s} := &\lambda(\coqdocvar{c}:s)(\coqdocvarP{c}:s),\coqdocvar{c} →
  \coqdocvarP{c} → \hat{s}
  \\
  \nonumber
  \ptranslate{\coqdocvar{x}} := &\,\coqdocvarR{x} \\
  \ptranslate{∀\coqdocvar{x}\!:\! A.B} :=
  &\,λ(\coqdocvarFour{x}:∀\coqdocvar{x}:A.B)(\coqdocvarFive{x}:∀\coqdocvarP{x}:\tprime{A}.\tprime{B}),
  \quad
  ∀(\coqdocvar{x}:A)(\coqdocvarP{x}:\tprime{A})(\coqdocvarR{x}:\ptranslate{A}
  \coqdocvar{x}\,\coqdocvarP{x}),\\
  & \qquad
  \ptranslate{B}{(\coqdocvarFour{x}\;\coqdocvar{x})}{(\coqdocvarFive{x}\,\coqdocvarP{x})}
  \label{DefProd} \\
  \ptranslate{λ\coqdocvar{x}:A, B} :=
  &\,λ(\coqdocvar{x}:A)(\coqdocvarP{x}:\tprime{A})(\coqdocvarR{x}:\ptranslate{A}
  \coqdocvar{x}\, \coqdocvarP{x}), \ptranslate{B} \\
  \ptranslate{(A\,B)} := &\,(\ptranslate{A}\,B\,\tprime{B}\,\ptranslate{B})
\end{align*}

The translation of contexts is obvious from the translation of the $\lambda$
case:
\begin{align*}
  \ptranslate{\langle\rangle} := &\,\langle \rangle \\
  \ptranslate{Γ, \coqdocvar{x}:A} := &\,\ptranslate{Γ}, \coqdocvar{x}:A,
  \coqdocvarP{x}:\tprime{A}, \coqdocvarR{x}
  :\ptranslate{A}\,\coqdocvar{x}\,\coqdocvarP{x}
\end{align*}


\coqdocnoindent
As examples, \ptranslate{
{\coqdocnotation{∀}} \coqdocvar{A} : \coqdockw{Type}
{\coqdocnotation{,}} \coqdocvariable{A} 
{\coqdocnotation{→}} \coqdocvariable{A}}
$\beta$ reduces to the relation
\coqdocnoindent
{\coqdocnotation{\ensuremath{\lambda}}}
{\coqdocnotation{(}}\coqdocvarFour{A}:
  \coqdockw{\ensuremath{\forall}} \coqdocvar{A} : \coqdockw{Type}, \coqdocvariable{A} {\coqdocnotation{→}} \coqdocvariable{A}{\coqdocnotation{)}}
{\coqdocnotation{(}}\coqdocvarFive{A}:
  \coqdockw{\ensuremath{\forall}} \coqdocvarP{A} : \coqdockw{Type}, \coqdocvarP{A} {\coqdocnotation{→}} \coqdocvarP{A}{\coqdocnotation{)}}
  , \coqdockw{\ensuremath{\forall}} (\coqdocvar{A} \coqdocvar{\tprime{A}}: \coqdockw{Type}) (\coqdocvarR{A}: 
  \coqdocvariable{A} {\coqdocnotation{→}} 
  \coqdocvariable{\tprime{A}} {\coqdocnotation{→}} \coqdockw{Type}) (\coqdocvar{a}: \coqdocvariable{A})
  (\coqdocvar{\tprime{a}}: \coqdocvariable{\tprime{A}}), \coqdocvarR{A} \coqdocvariable{a} \coqdocvariable{\tprime{a}} {\coqdocnotation{→}}
  \coqdocvarR{A} 
  (\coqdocvarFour{x} \coqdocvariable{A} \coqdocvariable{a}) (\coqdocvarFive{x} 
  \coqdocvarP{A} \coqdocvariable{\tprime{a}})
and
\ptranslate{\coqexternalref{::'xCExBB' x '..' x ','
x}{http://coq.inria.fr/distrib/8.5pl3/stdlib/Coq.Unicode.Utf8\_core}{\coqdocnotation{\ensuremath{\lambda}}}
{\coqdocnotation{(}}\coqdocvar{A} : \coqdockw{Type}) (\coqdocvar{a} :
 \coqdocvariable{A}{\coqdocnotation{),}} \coqdocvariable{a}}
is
{\coqdocnotation{\ensuremath{\lambda}}}{\coqdocnotation{(}}\coqdocvar{A}
 \coqdocvar{\tprime{A}} : \coqdockw{Type}) (\coqdocvarR{A} : \coqdocvariable{A} {\coqdocnotation{→}} \coqdocvariable{\tprime{A}} {\coqdocnotation{→}} \coqdockw{Type}) (\coqdocvar{a}: \coqdocvariable{A}) (\coqdocvar{\tprime{a}}: \coqdocvariable{\tprime{A}}) 
 (\coqdocvarR{a} : \coqdocvarR{A} \coqdocvariable{a} \coqdocvariable{\tprime{a}}{\coqdocnotation{),}} \coqdocvarR{a}.\coqdoceol

\newcommand{\CoqTypeSet}[0]{{\coqdockw{Type}\ensuremath{_0}}}
\newcommand{\falseForall}[0]{{\ensuremath{\forall} \coqdocvar{$A$}:{\CoqTypeSet} \coqdocvar{$A$}}}

A problem with the above definition of $\hat{s}$ is that for a closed 
$T$:\coqdockw{Type}$_i$ , \ptranslate{$T$}
is a relation of type 
$T$ $\rightarrow$ \tprime{$T$} $\rightarrow$ \coqdockw{Type}$_i$.
In Coq, logical relations typically return propositions. Thus one may instead
desire the following type: 
$T$ $\rightarrow$ \tprime{$T$} $\rightarrow$ \coqdockw{Prop}, which is what 
we get by defining $\hat{s}$ := \coqdockw{Prop}.
As explained in the previous section, inhabitants of the \coqdockw{Prop} universe
enjoy a special status in Coq's logic and compiler.
%
Unfortunately, \citet[Sec. 4.2]{Keller.Lasson2012} observed that having $\hat{s}$ := \coqdockw{Prop}
breaks the abstraction theorem above 
for the typehood judgement \coqdockw{Type}$_i$:\coqdockw{Type}$_{i+1}$

\citet[Sec. 4.2]{Keller.Lasson2012} consider a different calculus (CIC$_r$), which has
two chains of universes \coqdockw{Type}$_{i}$ and \coqdockw{Set}$_{i}$.
The latter chain does not have the rule \coqdockw{Set}$_i$:\coqdockw{Set}$_{i+1}$
and thus they are able to have $\hat{\text{\coqdockw{Set}}_i}$ := \coqdockw{Prop}.
However, without that rule, the higher universes in the latter chain may have limited utility.
Also, although they defined an embedding from CIC$_r$ to Coq, they didn't define
any embedding of any fragment of Coq into
CIC$_r$. Thus, it is not clear how their theory applies to Coq. Indeed, their
implementation for Coq always picks $\hat{s}$ := $s$. 

Instead of switching to a different calculus, we consider Coq.
Note that the relations for the lowermost universe \emph{can} live in
\coqdockw{Prop}, i.e., we \emph{can} define
$\hat{{\CoqTypeSet}}$ := \coqdockw{Prop} and $\hat{s}$ := $s$ otherwise.
For $i>0$, the abstraction theorem for 
\coqdockw{Type}$_0$:\coqdockw{Type}$_{i}$ $\beta$-reduces to the following,
which typechecks in Coq:
(λ (\coqdocvar{$A$} \coqdocvarP{$A$}: {\CoqTypeSet}),
\coqdocvar{$A$} → \coqdocvarP{$A$} → \coqdockw{Prop}) 
        : ({\CoqTypeSet} → {\CoqTypeSet} → \coqdockw{Type}$_{i}$)).\\
To follow Coq's convention, we will henceforth write
\coqdockw{Set} instead of {{\CoqTypeSet}}.

In the next subsection, we will see that the desire to have
$\hat{\coqdockw{Set}}$ := \coqdockw{Prop} has major implications
on how the inductive types are translated.

\subsection{Previous Translations of Inductive Types and Propositions: Comparison}
In the above core calculus, the only way to form new types was to form dependent function types.
One can also \emph{inductively} define new types and propositions in Coq.
For example, below we have a  Peano-style inductive definition of natural
numbers:

\vspace{0.15cm}
\noindent
%
\begin{minipage}[t]{0.32\textwidth}
\coqdocnoindent
\coqdockw{Inductive} \coqdef{Top.paper.nat}{nat}{\coqdocinductive{nat}} : \coqdockw{Set} :=\coqdoceol
\coqdocnoindent
\ensuremath{|} \coqdef{Top.paper.O}{O}{\coqdocconstructor{O}} : \coqref{Top.paper.nat}{\coqdocinductive{nat}}\coqdoceol
\coqdocnoindent
\ensuremath{|} \coqdef{Top.paper.S}{S}{\coqdocconstructor{S}} : \coqref{Top.paper.nat}{\coqdocinductive{nat}} \coqexternalref{:type scope:x '->' x}{http://coq.inria.fr/distrib/8.5pl3/stdlib/Coq.Init.Logic}{\coqdocnotation{\ensuremath{\rightarrow}}} \coqref{Top.paper.nat}{\coqdocinductive{nat}}.\coqdoceol
\coqdocemptyline
\coqdocemptyline
\end{minipage}
\begin{minipage}[t]{0.68\textwidth}
One can write functions by pattern matching on inductive data/proofs.
For example, below are the definitions of the predecessor function (left) and a
logical predicate (right) asserting that the input is zero.

\end{minipage}
%
%
%
%

\vspace{0.15cm}
\noindent
\begin{minipage}[t]{0.5\textwidth}
\coqdocnoindent
\coqdockw{Definition} \coqdef{Top.paper.pred}{pred}{\coqdocdefinition{pred}} (\coqdocvar{n}:\coqref{Top.paper.nat}{\coqdocinductive{nat}}) : \coqref{Top.paper.nat}{\coqdocinductive{nat}} :=\coqdoceol
\coqdocindent{1.00em}
\coqdockw{match} \coqdocvariable{n} \coqdockw{with}\coqdoceol
\coqdocindent{1.00em}
\ensuremath{|} \coqref{Top.paper.O}{\coqdocconstructor{O}}  \ensuremath{\Rightarrow} \coqref{Top.paper.O}{\coqdocconstructor{O}}\coqdoceol
\coqdocindent{1.00em}
\ensuremath{|} \coqref{Top.paper.S}{\coqdocconstructor{S}} \coqdocvar{n} \ensuremath{\Rightarrow} \coqdocvariable{n}\coqdoceol
\coqdocindent{1.00em}
\coqdockw{end}.\coqdoceol
\coqdocemptyline
\coqdocemptyline
\end{minipage}
\begin{minipage}[t]{0.5\textwidth}
\coqdocnoindent
\coqdockw{Definition} \coqdef{Top.paper.isZero}{isZero}{\coqdocdefinition{isZero}} (\coqdocvar{n}:\coqref{Top.paper.nat}{\coqdocinductive{nat}}) : \coqdockw{Prop} :=\coqdoceol
\coqdocindent{1.00em}
\coqdockw{match} \coqdocvariable{n} \coqdockw{with}\coqdoceol
\coqdocindent{1.00em}
\ensuremath{|} \coqref{Top.paper.O}{\coqdocconstructor{O}}  \ensuremath{\Rightarrow} \coqexternalref{True}{http://coq.inria.fr/distrib/8.5pl3/stdlib/Coq.Init.Logic}{\coqdocinductive{True}}\coqdoceol
\coqdocindent{1.00em}
\ensuremath{|} \coqref{Top.paper.S}{\coqdocconstructor{S}} \coqdocvar{\_} \ensuremath{\Rightarrow} \coqexternalref{False}{http://coq.inria.fr/distrib/8.5pl3/stdlib/Coq.Init.Logic}{\coqdocinductive{False}}\coqdoceol
\coqdocindent{1.00em}
\coqdockw{end}.\coqdoceol
\coqdocemptyline
\coqdocemptyline
\end{minipage}

\label{sec:anyrel:compare}
\citet{Bernardy.Jansson.ea2012} presented two ways to translate inductive types and pattern matching:
the inductive style translation and the deductive style translation.
The two methods are, according to the authors, isomorphic in their Agda-like
setting where there is no universe analogous to \coqdockw{Prop}.
However, in Coq, as we explain next, the deductive style is more
suitable for translating inductive types, and the inductive style is the only
choice (among the two) for inductive propositions.
Also, in the next subsection,
we will uncover and fix a subtle flaw in the deductive-style translation.
%
Below, for \coqRefInductive{Top.paper}{nat}, we have the inductive-style translation (left) and the deductive-style
translation (right).\\
\noindent
\begin{minipage}[t]{0.6\textwidth}
\coqdocnoindent
\coqdockw{Inductive} \coqdef{Top.paper.nat R}{\trel{nat}}{\coqdocinductive{\trel{nat}}} : \coqref{Top.paper.nat}{\coqdocinductive{nat}} \coqexternalref{:type scope:x '->' x}{http://coq.inria.fr/distrib/8.5pl3/stdlib/Coq.Init.Logic}{\coqdocnotation{\ensuremath{\rightarrow}}} \coqref{Top.paper.nat}{\coqdocinductive{nat}} \coqexternalref{:type scope:x '->' x}{http://coq.inria.fr/distrib/8.5pl3/stdlib/Coq.Init.Logic}{\coqdocnotation{\ensuremath{\rightarrow}}} $\hat{\coqdockw{Set}}$ :=\coqdoceol
\coqdocnoindent
\ensuremath{|} \coqdef{Top.paper.O R}{\trel{O}}{\coqdocconstructor{\trel{O}}} : \coqref{Top.paper.nat R}{\coqdocinductive{\trel{nat}}} \coqref{Top.paper.O}{\coqdocconstructor{O}} \coqref{Top.paper.O}{\coqdocconstructor{O}}\coqdoceol
\coqdocnoindent
\ensuremath{|} \coqdef{Top.paper.S R}{\trel{S}}{\coqdocconstructor{\trel{S}}} : \coqdockw{\ensuremath{\forall}} \coqdocvar{n} \coqdocvar{\tprime{n}} : \coqref{Top.paper.nat}{\coqdocinductive{nat}}, \coqref{Top.paper.nat R}{\coqdocinductive{\trel{nat}}} \coqdocvariable{n} \coqdocvariable{\tprime{n}} \coqexternalref{:type scope:x '->' x}{http://coq.inria.fr/distrib/8.5pl3/stdlib/Coq.Init.Logic}{\coqdocnotation{\ensuremath{\rightarrow}}} \coqref{Top.paper.nat R}{\coqdocinductive{\trel{nat}}} (\coqref{Top.paper.S}{\coqdocconstructor{S}} \coqdocvariable{n}) (\coqref{Top.paper.S}{\coqdocconstructor{S}} \coqdocvariable{\tprime{n}}).\coqdoceol
\end{minipage}
\begin{minipage}[t]{0.6\textwidth}
\coqdocnoindent
\coqdockw{Fixpoint} \coqdef{Top.paper.Ded.nat
R}{\trel{nat}}{\coqdocdefinition{\trel{nat}}} (\coqdocvar{n} \coqdocvar{\tprime{n}} : \coqref{Top.paper.nat}{\coqdocinductive{nat}}) : $\hat{\coqdockw{Set}}$ :=\coqdoceol
\coqdocnoindent
\coqdockw{match} \coqdocvariable{n},\coqdocvariable{\tprime{n}} \coqdockw{with}\coqdoceol
\coqdocnoindent
\ensuremath{|} \coqref{Top.paper.O}{\coqdocconstructor{O}}, \coqref{Top.paper.O}{\coqdocconstructor{O}} \ensuremath{\Rightarrow} \coqexternalref{True}{http://coq.inria.fr/distrib/8.5pl3/stdlib/Coq.Init.Logic}{\coqdocinductive{True}}\coqdoceol
\coqdocnoindent
\ensuremath{|} \coqref{Top.paper.S}{\coqdocconstructor{S}} \coqdocvar{m}, \coqref{Top.paper.S}{\coqdocconstructor{S}} \coqdocvar{\tprime{m}} \ensuremath{\Rightarrow} \coqref{Top.paper.nat R}{\coqdocdefinition{\trel{nat}}} \coqdocvar{m} \coqdocvar{\tprime{m}}\coqdoceol
\coqdocnoindent
\ensuremath{|} \coqdocvar{\_},\coqdocvar{\_} \ensuremath{\Rightarrow} \coqexternalref{False}{http://coq.inria.fr/distrib/8.5pl3/stdlib/Coq.Init.Logic}{\coqdocinductive{False}}
\coqdoceol\coqdocnoindent
\coqdockw{end}.\coqdoceol
\coqdocnoindent
\coqdockw{Definition} \coqdef{Top.paper.Ded.O R}{\trel{O}}{\coqdocdefinition{\trel{O}}} : \coqref{Top.paper.Ded.nat R}{\coqdocdefinition{\trel{nat}}} \coqref{Top.paper.O}{\coqdocconstructor{O}} \coqref{Top.paper.O}{\coqdocconstructor{O}} := \coqexternalref{I}{http://coq.inria.fr/distrib/8.5pl3/stdlib/Coq.Init.Logic}{\coqdocconstructor{I}}.\coqdoceol
\coqdocnoindent
\coqdockw{Definition} \coqdef{Top.paper.Ded.S R}{\trel{S}}{\coqdocdefinition{\trel{S}}} (\coqdocvar{n} \coqdocvar{\tprime{n}} : \coqref{Top.paper.nat}{\coqdocinductive{nat}}) (\coqdocvar{\trel{n}} : \coqref{Top.paper.Ded.nat R}{\coqdocdefinition{\trel{nat}}} \coqdocvariable{n} \coqdocvariable{\tprime{n}}) \coqdoceol
\coqdocindent{0.50em}
: \coqref{Top.paper.Ded.nat R}{\coqdocdefinition{\trel{nat}}} (\coqref{Top.paper.S}{\coqdocconstructor{S}} \coqdocvariable{n}) (\coqref{Top.paper.S}{\coqdocconstructor{S}} \coqdocvariable{\tprime{n}}) := \coqdocvariable{\trel{n}}.\coqdoceol
\coqdocemptyline
\coqdocemptyline
\end{minipage}
The inductive-style translation is straightforward.
Roughly speaking, given an inductive
\coqdocinductive{$I$}:$T$, it introduces a new inductive
\coqdocinductive{\trel{$I$}}:\ptranslate{$T$} \coqdocinductive{$I$}
\coqdocinductive{$I$}.
For each constructor \coqdocconstructor{c}:$C$, 
\coqdocinductive{\trel{$I$}} has the constructor
\coqdocconstructor{\trel{c}}:\ptranslate{$C$} \coqdocconstructor{c}
\coqdocconstructor{c}.
In both the styles, \ptranslate{} is extended to define
\ptranslate{\coqdocinductive{$I$}}:= \coqdocinductive{\trel{$I$}}
and
\ptranslate{\coqdocconstructor{c}}:= \coqdocconstructor{\trel{c}}.

The deductive style translation defines the same relation by structural recursion.
The constructors are translated separately. {\CoqTrueI}:{\CoqTrue} is the
constructor of the inductively defined proposition {\CoqTrue}.
%
%

The translation of Coq's \coqdockw{match} construct 
depends on how the type of the discriminee, which must be inductive (or coinductive), is translated.
Below, we have the inductive-style (left) and the deductive-style (right) 
translation of the above-defined predecessor function.
Again, the inductive-style translation is straightforward. We just translate
each subterm of the \coqdockw{match} construct.
The deductive style translation of an inductive type is \emph{not} an inductively defined type.
Thus, in the deductive style, we cannot do a pattern match on the translation of the discriminee.
Instead, we pattern match on the original discriminee \coqdocvar{n} and \coqdocvar{\tprime{n}}. 
In the cases when the constructors are different, the type of the argument
\coqdocvarR{n} computes to {\CoqFalse} 
(see the last branch in the definition of \coqRefDefn{Top.Paper.Ded}{\trel{nat}}). 
For any type \coqdocvar{T},
and  \coqdocvar{p}:{\CoqFalse}, 
{\CoqFalseRect} \coqdocvar{T} \coqdocvar{p} has type \coqdocvar{T}.
(Readers who find it odd that we apply \coqdocvar{\trel{n}} to the 
\coqdockw{match} term, and lambda bind it with
refined types in each branch
may wish to read
``The One Rule of Dependent Pattern Matching in
Coq''~\citep[Sec 8.2]{Chlipala2011}.)

\noindent
\begin{minipage}[t]{0.45\textwidth}
\coqdockw{Definition} \coqdef{Top.paper.pred R}{\trel{pred}}{\coqdocdefinition{\trel{pred}}} 
(\coqdocvar{n} \coqdocvar{\tprime{n}} : \coqref{Top.paper.nat}{\coqdocinductive{nat}}) 
\coqdoceol\coqdocindent{0.50em}
(\coqdocvar{\trel{n}} : \coqref{Top.paper.nat R}{\coqdocinductive{\trel{nat}}} \coqdocvariable{n} \coqdocvariable{\tprime{n}}) 
: \coqref{Top.paper.nat R}{\coqdocinductive{\trel{nat}}} (\coqref{Top.paper.pred}{\coqdocdefinition{pred}} \coqdocvariable{n}) (\coqref{Top.paper.pred}{\coqdocdefinition{pred}} \coqdocvariable{\tprime{n}}) 
\coqdoceol\coqdocnoindent
:= 
\coqdockw{match} \coqdocvariable{\trel{n}} \coqdockw{with}\coqdoceol
\coqdocnoindent
\ensuremath{|} \coqref{Top.paper.O R}{\coqdocconstructor{\trel{O}}} \ensuremath{\Rightarrow} \coqref{Top.paper.O R}{\coqdocconstructor{\trel{O}}}\coqdoceol
\coqdocnoindent
\ensuremath{|} \coqref{Top.paper.S R}{\coqdocconstructor{\trel{S}}} \coqdocvar{m} \coqdocvar{\tprime{m}} \coqdocvar{\trel{m}} \ensuremath{\Rightarrow} \coqdocvariable{\trel{m}}\coqdoceol
\coqdocnoindent
\coqdockw{end}.\coqdoceol
\end{minipage}
\begin{minipage}[t]{0.55\textwidth}
\coqdocnoindent
\coqdockw{Definition} \coqdef{Top.paper.Ded.pred R}{\trel{pred}}{\coqdocdefinition{\trel{pred}}} (\coqdocvar{n} \coqdocvar{\tprime{n}} : \coqref{Top.paper.nat}{\coqdocinductive{nat}}) (\coqdocvar{\trel{n}}: \coqref{Top.paper.Ded.nat R}{\coqdocdefinition{\trel{nat}}} \coqdocvariable{n} \coqdocvariable{\tprime{n}})\coqdoceol
\coqdocindent{0.50em}
: \coqref{Top.paper.Ded.nat R}{\coqdocdefinition{\trel{nat}}} (\coqref{Top.paper.pred}{\coqdocdefinition{pred}} \coqdocvariable{n}) (\coqref{Top.paper.pred}{\coqdocdefinition{pred}} \coqdocvariable{\tprime{n}}) :=
(\coqdockw{match} \coqdocvariable{n}, \coqdocvariable{\tprime{n}} 
\coqdockw{return} 
\coqdoceol\coqdocnoindent
\coqexternalref{:type scope:x '->' x}{http://coq.inria.fr/distrib/8.5pl3/stdlib/Coq.Init.Logic}{\coqdocnotation{(}}\coqref{Top.paper.Ded.nat R}{\coqdocdefinition{\trel{nat}}} \coqdocvariable{n} \coqdocvariable{\tprime{n}}\coqexternalref{:type scope:x '->' x}{http://coq.inria.fr/distrib/8.5pl3/stdlib/Coq.Init.Logic}{\coqdocnotation{)}} \coqexternalref{:type scope:x '->' x}{http://coq.inria.fr/distrib/8.5pl3/stdlib/Coq.Init.Logic}{\coqdocnotation{\ensuremath{\rightarrow}}} \coqref{Top.paper.Ded.nat R}{\coqdocdefinition{\trel{nat}}} (\coqref{Top.paper.pred}{\coqdocdefinition{pred}} \coqdocvariable{n}) (\coqref{Top.paper.pred}{\coqdocdefinition{pred}} \coqdocvariable{\tprime{n}})
\coqdockw{with}
\coqdoceol\coqdocnoindent
\ensuremath{|} \coqref{Top.paper.O}{\coqdocconstructor{O}}, \coqref{Top.paper.O}{\coqdocconstructor{O}} \ensuremath{\Rightarrow} {\coqdocnotation{\ensuremath{\lambda}}} {\coqdocnotation{(}}\coqdocvar{\trel{n}}: \coqref{Top.paper.Ded.nat R}{\coqdocdefinition{\trel{nat}}} \coqref{Top.paper.O}{\coqdocconstructor{O}} \coqref{Top.paper.O}{\coqdocconstructor{O}}{\coqdocnotation{),}} \coqref{Top.paper.Ded.O R}{\coqdocdefinition{\trel{O}}}\coqdoceol
\coqdocnoindent
\ensuremath{|} \coqref{Top.paper.S}{\coqdocconstructor{S}} \coqdocvar{m}, \coqref{Top.paper.S}{\coqdocconstructor{S}} \coqdocvar{\tprime{m}} \ensuremath{\Rightarrow} {\coqdocnotation{\ensuremath{\lambda}}} {\coqdocnotation{(}}\coqdocvar{\trel{n}}: \coqref{Top.paper.Ded.nat R}{\coqdocdefinition{\trel{nat}}} (\coqref{Top.paper.S}{\coqdocconstructor{S}} \coqdocvar{m}) (\coqref{Top.paper.S}{\coqdocconstructor{S}} \coqdocvar{\tprime{m}}){\coqdocnotation{),}} \coqdocvariable{\trel{n}}\coqdoceol
\coqdocnoindent
\ensuremath{|} \coqdocvar{\_}, \coqdocvar{\_} \ensuremath{\Rightarrow} {\coqdocnotation{\ensuremath{\lambda}}} \coqdocvar{\trel{n}}{\coqdocnotation{,}} \coqexternalref{False}{http://coq.inria.fr/distrib/8.5pl3/stdlib/Coq.Init.Logic}{\coqdocdefinition{False\_rect}} \coqdocvar{\_} \coqdocvariable{\trel{n}}\coqdoceol
\coqdocnoindent
\coqdockw{end}) \coqdocvariable{\trel{n}}.\coqdoceol
\end{minipage}

Note that in the inductive style translation, we pattern match on 
the translated discriminee, whose type is the translated inductive (fully applied).
To ensure the consistency of the proof irrelevance axiom, Coq 
has a proof-elimination restriction that ensures that 
one can pattern 
match on proofs to \emph{only} create proofs. (There is an exception called
singleton elimination, which we describe in the next subsection.)
Recall that a term \coqdocvar{p} is a proof iff its type's type is \coqdockw{Prop} (\coqdocvar{p}'s type is a proposition).
If we define $\hat{\coqdockw{Set}}$ := \coqdockw{Prop}, the above inductive-style
translation of \coqRefDefn{Top.paper}{pred} is well-typed in Coq, 
because it matches on proofs to create proofs.
However, the inductive-style
translation of the above-defined \coqRefDefn{Top.paper}{isZero} predicate 
violates the proof-elimination restriction:
\begin{lemma}
When $\hat{\coqdockw{Set}}$ := \coqdockw{Prop}, the inductive style translation of large elimination
can be ill typed. 
\end{lemma}
\begin{proof}
Below is the inductive-style translation of the above-defined \coqRefDefn{Top.paper}{isZero} predicate:

\noindent
\coqdockw{Definition} \coqdef{Top.paper.isZero
R}{\trel{isZero}}{\coqdocdefinition{\trel{isZero}}} (\coqdocvar{n}
\coqdocvar{\tprime{n}} : \coqref{Top.paper.nat}{\coqdocinductive{nat}}) (\coqdocvar{\trel{n}} : \coqref{Top.paper.nat R}{\coqdocinductive{\trel{nat}}} \coqdocvariable{n} \coqdocvariable{\tprime{n}})
: \coqexternalref{:type scope:x '->' x}{http://coq.inria.fr/distrib/8.5pl3/stdlib/Coq.Init.Logic}{\coqdocnotation{(}}\coqref{Top.paper.isZero}{\coqdocdefinition{isZero}} \coqdocvariable{n}\coqexternalref{:type scope:x '->' x}{http://coq.inria.fr/distrib/8.5pl3/stdlib/Coq.Init.Logic}{\coqdocnotation{)}} \coqexternalref{:type scope:x '->' x}{http://coq.inria.fr/distrib/8.5pl3/stdlib/Coq.Init.Logic}{\coqdocnotation{\ensuremath{\rightarrow}}} \coqexternalref{:type scope:x '->' x}{http://coq.inria.fr/distrib/8.5pl3/stdlib/Coq.Init.Logic}{\coqdocnotation{(}}\coqref{Top.paper.isZero}{\coqdocdefinition{isZero}} \coqdocvariable{\tprime{n}}\coqexternalref{:type scope:x '->' x}{http://coq.inria.fr/distrib/8.5pl3/stdlib/Coq.Init.Logic}{\coqdocnotation{)}} \coqexternalref{:type scope:x '->' x}{http://coq.inria.fr/distrib/8.5pl3/stdlib/Coq.Init.Logic}{\coqdocnotation{\ensuremath{\rightarrow}}} \coqdockw{Prop} :=\coqdoceol
\coqdocnoindent
\coqdockw{match} \coqdocvariable{\trel{n}} \coqdockw{with}\coqdoceol
\coqdocnoindent
\ensuremath{|} \coqref{Top.paper.O R}{\coqdocconstructor{\trel{O}}} \ensuremath{\Rightarrow} \coqref{Top.paper.True R}{\coqdocinductive{\trel{True}}}\coqdoceol
\coqdocnoindent
\ensuremath{|} \coqref{Top.paper.S R}{\coqdocconstructor{\trel{S}}} \coqdocvar{\_} \coqdocvar{\_} \coqdocvar{\_} \ensuremath{\Rightarrow} \coqref{Top.paper.False R}{\coqdocinductive{\trel{False}}}\coqdoceol
\coqdocnoindent
\coqdockw{end}.\coqdoceol

\coqdocnoindent
It pattern-matches on a proof
(\coqdocvar{\trel{n}} has type \coqRefDefn{Top.paper}{\trel{nat}} \coqdocvar{n} \coqdocvar{n}, which has type \coqdockw{Prop})
to produce a relation, and not a proof.
Note that 
\coqexternalref{:type scope:x '->' x}{http://coq.inria.fr/distrib/8.5pl3/stdlib/Coq.Init.Logic}{\coqdocnotation{(}}\coqref{Top.paper.isZero}{\coqdocdefinition{isZero}} \coqdocvariable{n}\coqexternalref{:type scope:x '->' x}{http://coq.inria.fr/distrib/8.5pl3/stdlib/Coq.Init.Logic}{\coqdocnotation{)}} \coqexternalref{:type scope:x '->' x}{http://coq.inria.fr/distrib/8.5pl3/stdlib/Coq.Init.Logic}{\coqdocnotation{\ensuremath{\rightarrow}}} \coqexternalref{:type scope:x '->' x}{http://coq.inria.fr/distrib/8.5pl3/stdlib/Coq.Init.Logic}{\coqdocnotation{(}}\coqref{Top.paper.isZero}{\coqdocdefinition{isZero}} \coqdocvariable{\tprime{n}}\coqexternalref{:type scope:x '->' x}{http://coq.inria.fr/distrib/8.5pl3/stdlib/Coq.Init.Logic}{\coqdocnotation{)}} \coqexternalref{:type scope:x '->' x}{http://coq.inria.fr/distrib/8.5pl3/stdlib/Coq.Init.Logic}{\coqdocnotation{\ensuremath{\rightarrow}}} \coqdockw{Prop}
does \emph{not} have type \coqdockw{Prop}.
\end{proof}


%
Indeed, if Coq allowed one to match on proofs and produce the 
{\CoqTrue} proposition on one proof and the {\CoqFalse} proposition on another
(e.g., consider the definition \coqRefDefn{Top.paper}{isZero}
when \coqRefInductive{Top.Paper}{nat} is declared in the  
\coqdockw{Prop} universe), one can easily refute proof irrelevance, which says that
any two proofs of a proposition are equal (logically indistinguishable).

In contrast, the deductive-style translation doesn't suffer from this problem,
because the resultant pattern matches are on discriminees of the original
inductive type, and not the translated one (\coqRefInductive{Top.paper}{nat} has type \coqdockw{Set}, not \coqdockw{Prop}).
Thus, regarding the proof-elimination restriction,
the translatability of pattern matches in the deductive style is independent of
how we define $\hat{\coqdockw{Set}}$. Indeed,
the deductive-style translation of \coqRefDefn{Top.paper}{isZero} happily
typechecks when we define $\hat{\coqdockw{Set}}$ := \coqdockw{Prop}
or $\hat{\coqdockw{Set}}$ := \coqdockw{Set}.
\emph{Thus, the deductive-style translation of inductive types and corresponding pattern matching
allows more flexibility in
the choice of $\hat{\coqdockw{Set}}$}.
Although propositions (members of the universe \coqdockw{Prop}) are also called types in the literature,
in this paper, the word ``type'' usually refers to terms that are members of \coqdockw{Set} or \coqdockw{Type}$_i$, but \emph{not} \coqdockw{Prop}. 

The deductive style has other advantages over the inductive style:
In the inductive style, proofs that are by induction on variables of
the translated inductive type are often difficult.
We explain this at the beginning of \secref{sec:anyrel:ind}.
Also, the deductive-style translation 
enables proofs by computation, e.g., \coqRefDefnR{Top.paper}{nat} computes when
the two numbers are in normal form.

One can also inductively define logical propositions in Coq.
The story for translating inductive propositions is the opposite:
the deductive-style violates the proof-elimination restriction.

\begin{lemma}
The deductive-style translation of inductive propositions can be ill-typed.
\end{lemma}

\begin{proof}
Replace \coqdockw{Set} by \coqdockw{Prop} in
the above definition of \coqRefInductive{Top.paper}{nat}.
Then, \coqRefInductive{Top.paper}{nat} can be seen as the ``True'' proposition and
its members, e.g. \coqRefConstr{Top.paper}{O}, can be considered proofs of the proposition.
The deductive-style translation of \coqRefDefn{Top.paper}{nat},
as shown above (\coqRefDefnR{Top.paper}{nat}), would then be ill-typed because it would
then match on proofs (of \coqRefDefn{Top.paper}{nat}) to produce propositions, not
proofs.
\end{proof}

In summary, \emph{the deductive-style translation
is more suitable for translating inductive types,
and the inductive-style translation is the only choice (among the 2) for inductive propositions}.
Thus, unfortunately, one has to implement both styles to translate Coq in a way
that allows $\hat{\coqdockw{Set}}$ := \coqdockw{Prop}.
This is what we do in our {\anyrel} translation, because having
$\hat{\coqdockw{Set}}$ := \coqdockw{Prop} greatly simplifies our {\isorel}
translation.

\citet[Sec. 4.4]{Keller.Lasson2012} show a third approach, which is a hybrid approach, but only for a simple example which has no indices (indices are explained in the next subsection): 
they don't provide a general translation for inductive types.
Thus we exclude (just) that part of their paper from further consideration.
Also, their implementation always uses the 
inductive-style translation and 
always chooses $\hat{\coqdockw{Set}}$ := \coqdockw{Set}.

The inductive-style translation is quite simple and well explained and 
implemented
by \citet[Sec.~4.3]{Keller.Lasson2012}.
In the next subsection, we turn our attention to the deductive-style translation, 
which is more complex than the inductive-style translation.

For the rest of this paper, we define 
$\hat{\coqdockw{Set}}$ := \coqdockw{Prop}
\footnote{
Ensuring $\hat{\coqdockw{Set}}$:=\coqdockw{Prop}, which simplifies the
{\isorel} translation,
is problematic in the presence of universe-polymorphic
inductive types, regardless of whether we choose the deductive-style translation or the inductive-style translation.
The problem arises from limitations in the expressivity of Coq's universe
polymorphism.
We discuss the problem in the Appendix~(\secref{appendix:anyrel:univPoly}).
As mentioned before, our {\isorel} translation does not work for higher
universes anyway, for much more fundamental reasons (\secref{sec:isorel:limitations}). 
}.
%
\subsection{Deductive-style Translation of Indexed-Inductive Types}
\label{sec:anyrel:ind}
The above subsection established that to ensure that $\hat{\coqdockw{Set}}$ :=
\coqdockw{Prop}, inductive types (but not inductive propositions) should be
translated in the deductive style. This section takes a closer look at our
deductive-style translation, especially 
of indexed-inductive definitions.
Our translation is inspired by the translation by \citet{Bernardy.Jansson.ea2012}.
However, while implementing it for Coq,
we found that it is even more complex than the way it was presented
in the literature~\citep{Bernardy2011,Bernardy.Jansson.ea2012}.
The additional complexity is fundamental in nature and has nothing to do with Coq.

One can mutually inductively define an infinite family of
types/propositions using Coq's indexed-inductive definitions. Below is a typical indexed-inductive definition.
The type \coqRefInductive{Top.paper}{Vec} \coqdocvar{T} \coqdocvar{m} is just like the 
type {\CoqList} \coqdocvar{T}, except that its inhabitants must have length exactly \coqdocvar{m}.

\vspace{0.15cm}
\noindent
\begin{minipage}[t]{0.5\textwidth}
\coqdocnoindent
\coqdockw{Inductive} \coqdef{Top.paper.Vec}{Vec}{\coqdocinductive{Vec}} (\coqdocvar{T}: \coqdockw{Set}): \coqdockw{\ensuremath{\forall}} (\coqdocvar{m}: \coqref{Top.paper.nat}{\coqdocinductive{nat}}), \coqdockw{Set} :=\coqdoceol
\coqdocnoindent
\ensuremath{|} \coqdef{Top.paper.nilV}{nilV}{\coqdocconstructor{nilV}}: \coqref{Top.paper.Vec}{\coqdocinductive{Vec}} \coqdocvar{T} \coqref{Top.paper.O}{\coqdocconstructor{O}}\coqdoceol
\coqdocnoindent
\ensuremath{|} \coqdef{Top.paper.consV}{consV}{\coqdocconstructor{consV}}: \coqdockw{\ensuremath{\forall}} (\coqdocvar{n}: \coqref{Top.paper.nat}{\coqdocinductive{nat}}), \coqdocvar{T} \coqexternalref{:type scope:x '->' x}{http://coq.inria.fr/distrib/8.5pl3/stdlib/Coq.Init.Logic}{\coqdocnotation{\ensuremath{\rightarrow}}} \coqref{Top.paper.Vec}{\coqdocinductive{Vec}} \coqdocvar{T} \coqdocvariable{n} \coqexternalref{:type scope:x '->' x}{http://coq.inria.fr/distrib/8.5pl3/stdlib/Coq.Init.Logic}{\coqdocnotation{\ensuremath{\rightarrow}}} \coqref{Top.paper.Vec}{\coqdocinductive{Vec}} \coqdocvar{T} (\coqref{Top.paper.S}{\coqdocconstructor{S}} \coqdocvariable{n}).\coqdoceol
\end{minipage}
\begin{minipage}[t]{0.5\textwidth}
Note that the constructor \coqRefConstr{Top.paper}{consV} takes a 
\coqRefInductive{Top.paper}{Vec} \coqdocvar{T} \coqdocvar{n}
and constructs a \coqRefInductive{Top.paper}{Vec} \coqdocvar{T} (\coqRefConstr{Top.paper}{S} \coqdocvar{n}):
the input and output are \emph{different} members of the family.
\end{minipage}
\vspace{0.15cm}
\noindent

The arguments of the type that vary in the definition are called \emph{indices}.
The other arguments are called parameters. In the above type, \coqdocvar{T} is a
parameter and \coqdocvar{m} is an index. Coq requires that the parameters be listed before ``:'' and the indices be listed
after ``:''.
In general, the definition of a member of the family can depend on the
definition of other members of the family. Thus, even if we have a variable, say \coqdocvar{v} whose type is a specific
member of an inductive family (as determined by the indices), 
to do a proof by induction on \coqdocvar{v}, one has to consider
\emph{all} the members of the family. In particular, the property that is being
proved by induction must be \emph{well-defined} for all members of the family.
This often makes such proofs difficult, because 
one needs to generalize over indices (see \secref{sec:uniformProp:indt} for an
example).
We will see below that the inductive-style translation of an inductive type with
$n$ indices produces an inductive with $3n+2$ indices!

The deductive-style translation of the above type, as presented in previous 
literature
is flawed in a subtle way.
Below, we have first the (correct) inductive-style translation \citep[page 24, middle]{Bernardy.Jansson.ea2012}
and deductive-style translation from the literature (\citep[page 21, top]{Bernardy.Jansson.ea2012},
\citet[page 31]{Bernardy2011}). We have adapted these from Agda-like syntax to
Coq.
The authors claimed that the two styles are isomorphic.

\vspace{0.15cm}
\noindent
\begin{minipage}[t]{\textwidth}
\coqdocnoindent
\coqdockw{Inductive} \coqdef{Top.paper.Vec R}{\trel{Vec}}{\coqdocinductive{\trel{Vec}}} (\coqdocvar{T} \coqdocvar{\tprime{T}} : \coqdockw{Set}) (\coqdocvar{\trel{T}} : \coqdocvariable{T} \coqexternalref{:type scope:x '->' x}{http://coq.inria.fr/distrib/8.5pl3/stdlib/Coq.Init.Logic}{\coqdocnotation{\ensuremath{\rightarrow}}} \coqdocvariable{\tprime{T}} \coqexternalref{:type scope:x '->' x}{http://coq.inria.fr/distrib/8.5pl3/stdlib/Coq.Init.Logic}{\coqdocnotation{\ensuremath{\rightarrow}}} \coqdockw{Prop})\coqdoceol
\coqdocindent{1.00em}
: \coqdockw{\ensuremath{\forall}} (\coqdocvar{m} \coqdocvar{\tprime{m}} : \coqref{Top.paper.nat}{\coqdocinductive{nat}}) (\coqdocvar{\trel{m}} : \coqref{Top.paper.nat R}{\coqdocinductive{\trel{nat}}} \coqdocvariable{m} \coqdocvariable{\tprime{m}}) (\coqdocvar{v} : \coqref{Top.paper.Vec}{\coqdocinductive{Vec}} \coqdocvar{T} \coqdocvariable{m}) (\coqdocvar{\tprime{v}} : \coqref{Top.paper.Vec}{\coqdocinductive{Vec}} \coqdocvar{\tprime{T}} \coqdocvariable{\tprime{m}}), \coqdockw{Prop} :=\coqdoceol
\coqdocnoindent
\ensuremath{|} \coqdef{Top.paper.nilV R}{\trel{nilV}}{\coqdocconstructor{\trel{nilV}}} : \coqref{Top.paper.Vec R}{\coqdocinductive{\trel{Vec}}} \coqdocvar{T} \coqdocvar{\tprime{T}} \coqdocvar{\trel{T}} \coqref{Top.paper.O}{\coqdocconstructor{O}} \coqref{Top.paper.O}{\coqdocconstructor{O}} \coqref{Top.paper.O R}{\coqdocconstructor{\trel{O}}} (\coqref{Top.paper.nilV}{\coqdocconstructor{nilV}} \coqdocvar{T}) (\coqref{Top.paper.nilV}{\coqdocconstructor{nilV}} \coqdocvar{\tprime{T}})\coqdoceol
\coqdocnoindent
\ensuremath{|} \coqdef{Top.paper.consV R}{\trel{consV}}{\coqdocconstructor{\trel{consV}}} : \coqdockw{\ensuremath{\forall}} (\coqdocvar{n} \coqdocvar{\tprime{n}} : \coqref{Top.paper.nat}{\coqdocinductive{nat}}) (\coqdocvar{\trel{n}} : \coqref{Top.paper.nat R}{\coqdocinductive{\trel{nat}}} \coqdocvariable{n} \coqdocvariable{\tprime{n}}) (\coqdocvar{t} : \coqdocvar{T}) (\coqdocvar{\tprime{t}} : \coqdocvar{\tprime{T}}),\coqdoceol
\coqdocindent{1.00em}
\coqdocvar{\trel{T}} \coqdocvariable{t} \coqdocvariable{\tprime{t}} \coqexternalref{:type scope:x '->' x}{http://coq.inria.fr/distrib/8.5pl3/stdlib/Coq.Init.Logic}{\coqdocnotation{\ensuremath{\rightarrow}}} \coqdockw{\ensuremath{\forall}} (\coqdocvar{vn} : \coqref{Top.paper.Vec}{\coqdocinductive{Vec}} \coqdocvar{T} \coqdocvariable{n}) (\coqdocvar{\tprime{vn}} : \coqref{Top.paper.Vec}{\coqdocinductive{Vec}} \coqdocvar{\tprime{T}} \coqdocvariable{\tprime{n}}),
\coqdoceol\coqdocindent{1.00em}
\coqref{Top.paper.Vec R}{\coqdocinductive{\trel{Vec}}} \coqdocvar{T} \coqdocvar{\tprime{T}} \coqdocvar{\trel{T}} \coqdocvariable{n} \coqdocvariable{\tprime{n}} \coqdocvariable{\trel{n}} \coqdocvariable{vn} \coqdocvariable{\tprime{vn}} \coqexternalref{:type scope:x '->' x}{http://coq.inria.fr/distrib/8.5pl3/stdlib/Coq.Init.Logic}{\coqdocnotation{\ensuremath{\rightarrow}}}\coqdoceol
\coqdocindent{1.00em}
\coqref{Top.paper.Vec R}{\coqdocinductive{\trel{Vec}}} \coqdocvar{T} \coqdocvar{\tprime{T}} \coqdocvar{\trel{T}} (\coqref{Top.paper.S}{\coqdocconstructor{S}} \coqdocvariable{n}) (\coqref{Top.paper.S}{\coqdocconstructor{S}} \coqdocvariable{\tprime{n}}) (\coqref{Top.paper.S R}{\coqdocconstructor{\trel{S}}} \coqdocvariable{n} \coqdocvariable{\tprime{n}} \coqdocvariable{\trel{n}}) (\coqref{Top.paper.consV}{\coqdocconstructor{consV}} \coqdocvar{T} \coqdocvariable{n} \coqdocvariable{t} \coqdocvariable{vn}) (\coqref{Top.paper.consV}{\coqdocconstructor{consV}} \coqdocvar{\tprime{T}} \coqdocvariable{\tprime{n}} \coqdocvariable{\tprime{t}} \coqdocvariable{\tprime{vn}}).\coqdoceol
\end{minipage}

\vspace{0.15cm}
\noindent
\begin{minipage}[t]{\textwidth}
\coqdockw{Fixpoint} \coqdef{Top.paper.DedV.Vec R}{\trel{Vec}}{\coqdocdefinition{\trel{Vec}}} (\coqdocvar{T} \coqdocvar{\tprime{T}} : \coqdockw{Set}) (\coqdocvar{\trel{T}} : \coqdocvariable{T} \coqexternalref{:type scope:x '->' x}{http://coq.inria.fr/distrib/8.5pl3/stdlib/Coq.Init.Logic}{\coqdocnotation{\ensuremath{\rightarrow}}}  \coqdocvariable{\tprime{T}} \coqexternalref{:type scope:x '->' x}{http://coq.inria.fr/distrib/8.5pl3/stdlib/Coq.Init.Logic}{\coqdocnotation{\ensuremath{\rightarrow}}} \coqdockw{Prop})\coqdoceol
\coqdocindent{1.00em}
(\coqdocvar{m} \coqdocvar{\tprime{m}} : \coqref{Top.paper.nat}{\coqdocinductive{nat}}) (\coqdocvar{\trel{m}} : \coqref{Top.paper.nat R}{\coqdocabbreviation{\trel{nat}}} \coqdocvariable{m} \coqdocvariable{\tprime{m}}) (\coqdocvar{v} : \coqref{Top.paper.Vec}{\coqdocinductive{Vec}} \coqdocvariable{T} \coqdocvariable{m}) (\coqdocvar{\tprime{v}} : \coqref{Top.paper.Vec}{\coqdocinductive{Vec}} \coqdocvariable{\tprime{T}} \coqdocvariable{\tprime{m}}) : \coqdockw{Prop} :=\coqdoceol
\coqdocnoindent
\coqdockw{match} \coqdocvariable{v},\coqdocvariable{\tprime{v}} \coqdockw{with}\coqdoceol
\coqdocnoindent
\ensuremath{|} \coqref{Top.paper.nilV}{\coqdocconstructor{nilV}}, \coqref{Top.paper.nilV}{\coqdocconstructor{nilV}} \ensuremath{\Rightarrow} \coqexternalref{True}{http://coq.inria.fr/distrib/8.5pl3/stdlib/Coq.Init.Logic}{\coqdocinductive{True}}\coqdoceol
\coqdocnoindent
\ensuremath{|} \coqref{Top.paper.consV}{\coqdocconstructor{consV}} \coqdocvar{n} \coqdocvar{t} \coqdocvar{vn}, \coqref{Top.paper.consV}{\coqdocconstructor{consV}} \coqdocvar{\tprime{n}} \coqdocvar{\tprime{t}} \coqdocvar{\tprime{vn}} \ensuremath{\Rightarrow}
\coqexternalref{:type scope:'x7B' x ':' x '|' x
'x7D'}{http://coq.inria.fr/distrib/8.5pl3/stdlib/Coq.Init.Specif}{\coqdocnotation{\{}}\coqdocvar{\trel{n}}
\coqexternalref{:type scope:'x7B' x ':' x '|' x
'x7D'}{http://coq.inria.fr/distrib/8.5pl3/stdlib/Coq.Init.Specif}{\coqdocnotation{:}}
\coqref{Top.paper.nat R}{\coqdocabbreviation{\trel{nat}}} \coqdocvar{n}
\coqdocvar{\tprime{n}} \coqexternalref{:type scope:'x7B' x ':' x '|' x 'x7D'}{http://coq.inria.fr/distrib/8.5pl3/stdlib/Coq.Init.Specif}{\coqdocnotation{\ensuremath{|}}} \coqdocvariable{\trel{T}} \coqdocvar{t} \coqdocvar{\tprime{t}} \coqexternalref{:type scope:x '/x5C' x}{http://coq.inria.fr/distrib/8.5pl3/stdlib/Coq.Init.Logic}{\coqdocnotation{\ensuremath{\land}}} \coqref{Top.paper.DedV.Vec R}{\coqdocdefinition{\trel{Vec}}} \coqdocvariable{T} \coqdocvariable{\tprime{T}} \coqdocvariable{\trel{T}} \coqdocvar{n} \coqdocvar{\tprime{n}} \coqdocvar{\trel{n}} \coqdocvar{vn} \coqdocvar{\tprime{vn}}\coqexternalref{:type scope:'x7B' x ':' x '|' x 'x7D'}{http://coq.inria.fr/distrib/8.5pl3/stdlib/Coq.Init.Specif}{\coqdocnotation{\}}}\coqdoceol
\coqdocnoindent
\ensuremath{|} \coqdocvar{\_}, \coqdocvar{\_} \ensuremath{\Rightarrow} \coqexternalref{False}{http://coq.inria.fr/distrib/8.5pl3/stdlib/Coq.Init.Logic}{\coqdocinductive{False}}\coqdoceol
\coqdocnoindent
\coqdockw{end}.\coqdoceol
\end{minipage}

\vspace{0.15cm}
\noindent
The argument \coqdocvar{\trel{m}} is \emph{unused}
and \emph{irrelevant} in the deductive-style translation
(\coqRefDefnR{Top.paper.DedV}{Vec}), which is a recursive function (and not an
inductive).
Thus, one can prove by induction on \coqdocvar{v} that forall 
\coqdocvar{m}
\coqdocvar{\tprime{m}},
\coqdocvar{m\ensuremath{_{R1}}}
\coqdocvar{m\ensuremath{_{R2}}},
\coqdocvar{v},
\coqdocvar{\tprime{v}},
the proposition 
\coqRefDefnR{Top.paper.DedV}{Vec} 
\coqdocvar{m}
\coqdocvar{\tprime{m}}
\coqdocvar{m\ensuremath{_{R1}}}
\coqdocvar{v} \coqdocvar{\tprime{v}}
is equal to 
\coqRefDefnR{Top.paper.DedV}{Vec} 
\coqdocvar{m}
\coqdocvar{\tprime{m}}
\coqdocvar{m\ensuremath{_{R2}}}
\coqdocvar{v} \coqdocvar{\tprime{v}}.
This is not the case in the inductive-style translation. For example, 
the constructor \coqRefConstr{Top.paper}{\trel{nilV}} \emph{requires}
\coqdocvar{\trel{m}} to be (definitionally) equal to 
\coqRefConstr{Top.paper}{\trel{O}}.
Thus, to prove that the two styles are isomorphic in this example,
one needs to at least prove that $\forall$ (\coqdocvar{\trel{m}}:
\coqRefDefnR{Top.paper}{nat} \coqRefConstr{Top.paper}{O}
\coqRefConstr{Top.paper}{O}),  \coqdocvar{\trel{m}} =
\coqRefConstrR{Top.paper}{O}.
It just so happens that this is provable for this example of
\coqRefInductive{Top.paper}{Vec}.
However, in general, the index type may not be concrete: it may be a type
variable. Also, it may be in a higher universe, in which case, its relation need
not be in \coqdockw{Prop}.
In that case, \emph{we} get to pick \emph{any} relation
for the type, and we can easily pick a relation $R$ such that for some
\coqdocvar{x} and \coqdocvar{y}, there are multiple distinct inhabitants in the
type $R$ \coqdocvar{x} \coqdocvar{y}. For example, we can pick $R$ := $\lambda$ \coqdocvar{x} \coqdocvar{y}, {\CoqBool}.
Also, \emph{we will see in \secref{sec:anyrel:match} that the translation of \coqdockw{match} terms
requires proofs like the above},  that
$\forall$
(\coqdocvar{\trel{m}}:
\coqRefDefnR{Top.paper}{nat} \coqRefConstr{Top.paper}{O}
\coqRefConstr{Top.paper}{O}),  \coqdocvar{\trel{m}} =
\coqRefConstrR{Top.paper}{O}. Thus, even when
provable, 
\ptranslate{} will need to cook up these proofs: it is not clear
how to do that automatically.

Thus we strengthen the propositions returned in the
deductive-style translation to add the above-mentioned equality constraints.
Here is the corrected version:

\vspace{0.05cm}
\noindent
\coqdocnoindent
\coqdockw{Fixpoint} \coqdef{Top.paper.DedVC.Vec R}{\trel{Vec}}{\coqdocdefinition{\trel{Vec}}} (\coqdocvar{T} \coqdocvar{\tprime{T}} : \coqdockw{Set}) (\coqdocvar{\trel{T}} : \coqdocvariable{T} \coqexternalref{:type scope:x '->' x}{http://coq.inria.fr/distrib/8.5pl3/stdlib/Coq.Init.Logic}{\coqdocnotation{\ensuremath{\rightarrow}}}  \coqdocvariable{\tprime{T}} \coqexternalref{:type scope:x '->' x}{http://coq.inria.fr/distrib/8.5pl3/stdlib/Coq.Init.Logic}{\coqdocnotation{\ensuremath{\rightarrow}}} \coqdockw{Prop})\coqdoceol
\coqdocindent{1.00em}
(\coqdocvar{m} \coqdocvar{\tprime{m}} : \coqref{Top.paper.nat}{\coqdocinductive{nat}}) (\coqdocvar{\trel{m}} : \coqref{Top.paper.nat R}{\coqdocabbreviation{\trel{nat}}} \coqdocvariable{m} \coqdocvariable{\tprime{m}}) (\coqdocvar{v} : \coqref{Top.paper.Vec}{\coqdocinductive{Vec}} \coqdocvariable{T} \coqdocvariable{m}) (\coqdocvar{\tprime{v}} : \coqref{Top.paper.Vec}{\coqdocinductive{Vec}} \coqdocvariable{\tprime{T}} \coqdocvariable{\tprime{m}}) : \coqdockw{Prop} :=\coqdoceol
\coqdocnoindent
(\coqdockw{match} \coqdocvariable{v},\coqdocvariable{\tprime{v}} \coqdockw{with}\coqdoceol
\coqdocnoindent
\ensuremath{|} \coqref{Top.paper.nilV}{\coqdocconstructor{nilV}}, \coqref{Top.paper.nilV}{\coqdocconstructor{nilV}} \ensuremath{\Rightarrow} {\coqdocnotation{\ensuremath{\lambda}}} 
\coqdocvar{\trel{m}}{\coqdocnotation{,}}
\highlight{{\coqdocvariable{\trel{m}} \coqexternalref{:type scope:x '='
x}{http://coq.inria.fr/distrib/8.5pl3/stdlib/Coq.Init.Logic}{\coqdocnotation{=}}
\coqref{Top.paper.Ded.O R}{\coqdocabbreviation{\trel{O}}}}}\coqdoceol

\coqdocnoindent
\ensuremath{|} \coqref{Top.paper.consV}{\coqdocconstructor{consV}} \coqdocvar{n} \coqdocvar{t} \coqdocvar{vn}, \coqref{Top.paper.consV}{\coqdocconstructor{consV}} \coqdocvar{\tprime{n}} \coqdocvar{\tprime{t}} \coqdocvar{\tprime{vn}} \ensuremath{\Rightarrow} {\coqdocnotation{\ensuremath{\lambda}}} \coqdocvar{\trel{m}}{\coqdocnotation{,}}\coqdoceol
\coqdocindent{1.00em}
\coqexternalref{:type scope:'x7B' x ':' x 'x26' x
'x7D'}{http://coq.inria.fr/distrib/8.5pl3/stdlib/Coq.Init.Specif}{\coqdocnotation{\{}}\coqdocvar{\trel{n}}
\coqexternalref{:type scope:'x7B' x ':' x 'x26' x
'x7D'}{http://coq.inria.fr/distrib/8.5pl3/stdlib/Coq.Init.Specif}{\coqdocnotation{:}}
\coqref{Top.paper.nat R}{\coqdocabbreviation{\trel{nat}}} \coqdocvar{n}
\coqdocvar{\tprime{n}} \coqexternalref{:type scope:'x7B' x ':' x 'x26' x
'x7D'}{http://coq.inria.fr/distrib/8.5pl3/stdlib/Coq.Init.Specif}{\coqdocnotation{\&}}
\coqdocvariable{\trel{T}} \coqdocvar{t} \coqdocvar{\tprime{t}}
\coqexternalref{:type scope:x '/x5C'
x}{http://coq.inria.fr/distrib/8.5pl3/stdlib/Coq.Init.Logic}{\coqdocnotation{\ensuremath{\land}}}
\coqref{Top.paper.DedVC.Vec R}{\coqdocdefinition{\trel{Vec}}} \coqdocvariable{T} \coqdocvariable{\tprime{T}} \coqdocvariable{\trel{T}} \coqdocvar{n} \coqdocvar{\tprime{n}} \coqdocvar{\trel{n}} \coqdocvar{vn} \coqdocvar{\tprime{vn}} \coqexternalref{:type scope:x '/x5C' x}{http://coq.inria.fr/distrib/8.5pl3/stdlib/Coq.Init.Logic}{\coqdocnotation{\ensuremath{\land}}} 
\highlight{\coqdocvariable{\trel{m}} \coqexternalref{:type scope:x '='
x}{http://coq.inria.fr/distrib/8.5pl3/stdlib/Coq.Init.Logic}{\coqdocnotation{=}}
\coqexternalref{:type scope:x '=' x}{http://coq.inria.fr/distrib/8.5pl3/stdlib/Coq.Init.Logic}{\coqdocnotation{(}}\coqref{Top.paper.Ded.S R}{\coqdocabbreviation{\trel{S}}} \coqdocvar{n} \coqdocvar{\tprime{n}} \coqdocvar{\trel{n}}\coqexternalref{:type scope:x '=' x}{http://coq.inria.fr/distrib/8.5pl3/stdlib/Coq.Init.Logic}{\coqdocnotation{)}}}\coqdocnotation{\}}
\coqdoceol
\coqdocnoindent
\ensuremath{|} \coqdocvar{\_}, \coqdocvar{\_} \ensuremath{\Rightarrow} {\coqdocnotation{\ensuremath{\lambda}}} \coqdocvar{\_} {\coqdocnotation{,}} \coqexternalref{False}{http://coq.inria.fr/distrib/8.5pl3/stdlib/Coq.Init.Logic}{\coqdocinductive{False}}\coqdoceol
\coqdocnoindent
\coqdockw{end}) \coqdocvariable{\trel{m}}.\coqdoceol
\vspace{0.05cm}
\noindent
\emph{After} adding the equality constraints, the deductive-style
translation is isomorphic to the inductive-style translation.
%
%
If an inductive constructor has recursive arguments that are functions, our
proof of the isomorphism needs the function extensionality axiom.

\emph{The {\anyrel} translation does not use any axiom.}
Preservation of reduction is typically a step in proving the abstraction theorem~\cite[Lemma 2]{Keller.Lasson2012}. 
Thus, we need to be careful in using axioms or opaque definitions because at certain places%
, they may block reduction.
(The {\isorel} translation described after this section uses axioms, but nevertheless achieves preservation of reduction.)

The \emph{only reason we add the equality constraints is that, as mentioned above, the proofs of those constraints are needed in the translation of pattern matches}.
These constraints added significant complexity to our implementation of \ptranslate{}, even after we were able to simplify the constraints a bit,
as explained in the rest of this subsection.

In general, an indexed-inductive type may have several indices.
Also, the \emph{types of} the later indices may be \emph{dependent} on the previous indices
or parameters. Below is an example:

\noindent
\begin{minipage}[t]{\textwidth}
\coqdockw{Inductive} \coqdef{Top.paper.isNil}{isNil}{\coqdocinductive{isNil}} : \coqdockw{\ensuremath{\forall}} (\coqdocvar{n}:\coqref{Top.paper.nat}{\coqdocinductive{nat}}) (\coqdocvar{v}:\coqref{Top.paper.Vec}{\coqdocinductive{Vec}} \coqref{Top.paper.nat}{\coqdocinductive{nat}} \coqdocvariable{n}), \coqdockw{Set} :=\coqdoceol
\coqdocnoindent
\coqdef{Top.paper.isnil}{isnil}{\coqdocconstructor{isnil}} : \coqdockw{\ensuremath{\forall}} (\coqdocvar{vv} : \coqref{Top.paper.Vec}{\coqdocinductive{Vec}} \coqref{Top.paper.nat}{\coqdocinductive{nat}} \coqref{Top.paper.O}{\coqdocconstructor{O}}), \coqref{Top.paper.isNil}{\coqdocinductive{isNil}} \coqref{Top.paper.O}{\coqdocconstructor{O}} \coqdocvariable{vv}.
\coqdoceol
\end{minipage}

\vspace{0.05cm}
\noindent
It is tricky to even state the equality constraints of the dependent indices (e.g. \coqdocvar{v}
in the example above) because
the types of the two sides of the equality will not be definitionally equal.
We will illustrate this soon.
While implementing \ptranslate{}, the main source of complexity came even later,
when implementing the translation of pattern matches. There, we had to
not only ``rewrite'' with the proofs of these equality constraints one by one, but also show that the proofs 
are each equal
to the canonical equality proof (\coqdocconstructor{eq\_refl}).
Fortunately, we found a much simpler way: we define a generalized equality type
that can, in one step, assert the equality of \emph{all} the corresponding indices.
Here is such an equality type for translating \coqRefInductive{Top.paper}{isNil}: 

\coqdocnoindent
\begin{minipage}[t]{\textwidth}
\coqdocnoindent
\coqdockw{Inductive} \coqdef{Top.paper.isNil indicesEq}{isNil\_indicesEq}{\coqdocinductive{isNil\_indicesEq}} (\coqdocvar{n} \coqdocvar{\tprime{n}} : \coqref{Top.paper.nat}{\coqdocinductive{nat}})  (\coqdocvar{\trel{n}} : \coqref{Top.paper.Ded.nat R}{\coqdocabbreviation{\trel{nat}}} \coqdocvariable{n} \coqdocvariable{\tprime{n}}) (\coqdocvar{v} : \coqref{Top.paper.Vec}{\coqdocinductive{Vec}} \coqref{Top.paper.nat}{\coqdocinductive{nat}} \coqdocvariable{n}) (\coqdocvar{\tprime{v}} : \coqref{Top.paper.Vec}{\coqdocinductive{Vec}} \coqref{Top.paper.nat}{\coqdocinductive{nat}} \coqdocvariable{\tprime{n}})\coqdoceol
\coqdocnoindent
(\coqdocvar{\trel{v}} : \coqref{Top.paper.DedV.Vec R}{\coqdocabbreviation{\trel{Vec}}} \coqref{Top.paper.nat}{\coqdocinductive{nat}} \coqref{Top.paper.nat}{\coqdocinductive{nat}} \coqref{Top.paper.Ded.nat R}{\coqdocabbreviation{\trel{nat}}} \coqdocvariable{n} \coqdocvariable{\tprime{n}} \coqdocvariable{\trel{n}} \coqdocvariable{v} \coqdocvariable{\tprime{v}}): \coqdockw{\ensuremath{\forall}} (\coqdocvar{i\trel{n}} : \coqref{Top.paper.Ded.nat R}{\coqdocabbreviation{\trel{nat}}} \coqdocvar{n} \coqdocvar{\tprime{n}}) (\coqdocvar{i\trel{v}} : \coqref{Top.paper.DedV.Vec R}{\coqdocabbreviation{\trel{Vec}}} \coqref{Top.paper.nat}{\coqdocinductive{nat}} \coqref{Top.paper.nat}{\coqdocinductive{nat}} \coqref{Top.paper.Ded.nat R}{\coqdocabbreviation{\trel{nat}}} \coqdocvar{n} \coqdocvar{\tprime{n}} \coqdocvariable{i\trel{n}} \coqdocvar{v} \coqdocvar{\tprime{v}}), \coqdockw{Prop}
\coqdoceol\coqdocnoindent
:=
\coqdef{Top.paper.isNil refl}{isNil\_refl}{\coqdocconstructor{isNil\_refl}} :  \coqref{Top.paper.isNil indicesEq}{\coqdocinductive{isNil\_indicesEq}} \coqdocvar{n} \coqdocvar{\tprime{n}} \coqdocvar{\trel{n}} \coqdocvar{v} \coqdocvar{\tprime{v}} \coqdocvar{\trel{v}} \coqdocvar{\trel{n}} \coqdocvar{\trel{v}}.\coqdoceol
\coqdocemptyline
\coqdocemptyline
\end{minipage}
\coqdocnoindent
This generalized equality type asserts that the indices \coqdocvar{\trel{n}} and \coqdocvar{\trel{v}}
are equal to the indices \coqdocvar{i\trel{n}} and \coqdocvar{i\trel{v}}.
The types of \coqdocvar{\trel{v}} and \coqdocvar{i\trel{v}} are different (not convertible, for the purpose of typechecking). Thus, it is ill-typed to just write \coqdocvar{\trel{v}} = \coqdocvar{i\trel{v}}.
Unlike JMeq~\cite{McBride2002}, our generalized equality type \emph{simultaneously} asserts the equality of sequences of dependent indices.
This greatly simplifies our translation of pattern matching, where now just generating one match on the proof of this
generalized equality type changes \emph{all} the indices,
and changes the \emph{only one} proof to the canonical form, which is 
\coqRefConstr{Top.paper}{isNil\_refl}.

The relation \coqRefInductive{Top.paper.DedM}{isNil\_indicesEq} lives in the \coqdockw{Prop}
universe. Thus, the proof-elimination restrictions may prohibit matching on its proofs
for producing non-proofs. However, Coq has a ``singleton-elimination'' exception for inductive
propositions that have \emph{only one} constructor and all the arguments of the constructor
are proofs. The above generalized equality proposition and those for other inductive types have only one constructor which takes no
arguments.

The inductive \coqRefInductive{Top.paper}{isNil} 
can now be translated in deductive style as follows:

\coqdocnoindent
\coqdockw{Fixpoint} \coqdef{Top.paper.isNil R}{\trel{isNil}}{\coqdocdefinition{\trel{isNil}}} (\coqdocvar{n} \coqdocvar{\tprime{n}} : \coqref{Top.paper.nat}{\coqdocinductive{nat}}) (\coqdocvar{\trel{n}} : \coqref{Top.paper.Ded.nat R}{\coqdocabbreviation{\trel{nat}}} \coqdocvariable{n} \coqdocvariable{\tprime{n}}) (\coqdocvar{v} : \coqref{Top.paper.Vec}{\coqdocinductive{Vec}} \coqref{Top.paper.nat}{\coqdocinductive{nat}} \coqdocvariable{n}) (\coqdocvar{\tprime{v}} : \coqref{Top.paper.Vec}{\coqdocinductive{Vec}} \coqref{Top.paper.nat}{\coqdocinductive{nat}} \coqdocvariable{\tprime{n}}) \coqdoceol
\coqdocnoindent
(\coqdocvar{\trel{v}} : \coqref{Top.paper.DedV.Vec R}{\coqdocabbreviation{\trel{Vec}}} \coqref{Top.paper.nat}{\coqdocinductive{nat}} \coqref{Top.paper.nat}{\coqdocinductive{nat}} \coqref{Top.paper.Ded.nat R}{\coqdocabbreviation{\trel{nat}}} \coqdocvariable{n} \coqdocvariable{\tprime{n}} \coqdocvariable{\trel{n}} \coqdocvariable{v} \coqdocvariable{\tprime{v}}) (\coqdocvar{m} : \coqref{Top.paper.isNil}{\coqdocinductive{isNil}} \coqdocvariable{n} \coqdocvariable{v}) (\coqdocvar{\tprime{m}} : \coqref{Top.paper.isNil}{\coqdocinductive{isNil}} \coqdocvariable{\tprime{n}} \coqdocvariable{\tprime{v}}): \coqdockw{Prop} :=\coqdoceol
\coqdocnoindent
(\coqdockw{match} \coqdocvariable{m},\coqdocvariable{\tprime{m}} \coqdockw{with}\coqdoceol
\coqdocnoindent
\ensuremath{|} \coqref{Top.paper.isnil}{\coqdocconstructor{isnil}} \coqdocvar{vv}, \coqref{Top.paper.isnil}{\coqdocconstructor{isnil}} \coqdocvar{\tprime{vv}} \ensuremath{\Rightarrow} {\coqdocnotation{\ensuremath{\lambda}}} {\coqdocnotation{(}}\coqdocvar{\trel{n}} : \coqref{Top.paper.Ded.nat R}{\coqdocabbreviation{\trel{nat}}} \coqref{Top.paper.O}{\coqdocconstructor{O}} \coqref{Top.paper.O}{\coqdocconstructor{O}}) (\coqdocvar{\trel{v}} : \coqref{Top.paper.DedV.Vec R}{\coqdocabbreviation{\trel{Vec}}} \coqref{Top.paper.nat}{\coqdocinductive{nat}} \coqref{Top.paper.nat}{\coqdocinductive{nat}} \coqref{Top.paper.Ded.nat R}{\coqdocabbreviation{\trel{nat}}} \coqref{Top.paper.O}{\coqdocconstructor{O}} \coqref{Top.paper.O}{\coqdocconstructor{O}} \coqdocvariable{\trel{n}} \coqdocvar{vv} \coqdocvar{\tprime{vv}}{\coqdocnotation{),}}\coqdoceol
\coqdocindent{1.00em}
\coqexternalref{:type scope:'x7B' x ':' x 'x26' x 'x7D'}{http://coq.inria.fr/distrib/8.5pl3/stdlib/Coq.Init.Specif}{\coqdocnotation{\{}}\coqdocvar{\trel{vv}} \coqexternalref{:type scope:'x7B' x ':' x 'x26' x 'x7D'}{http://coq.inria.fr/distrib/8.5pl3/stdlib/Coq.Init.Specif}{\coqdocnotation{:}} \coqref{Top.paper.DedV.Vec R}{\coqdocabbreviation{\trel{Vec}}} \coqref{Top.paper.nat}{\coqdocinductive{nat}} \coqref{Top.paper.nat}{\coqdocinductive{nat}} \coqref{Top.paper.Ded.nat R}{\coqdocabbreviation{\trel{nat}}} \coqref{Top.paper.O}{\coqdocconstructor{O}} \coqref{Top.paper.O}{\coqdocconstructor{O}} \coqref{Top.paper.Ded.O R}{\coqdocabbreviation{\trel{O}}} \coqdocvar{vv} \coqdocvar{\tprime{vv}} \coqexternalref{:type scope:'x7B' x ':' x 'x26' x 'x7D'}{http://coq.inria.fr/distrib/8.5pl3/stdlib/Coq.Init.Specif}{\coqdocnotation{\&}}
\highlight{\coqref{Top.paper.isNil indicesEq}{\coqdocinductive{isNil\_indicesEq}} \coqref{Top.paper.O}{\coqdocconstructor{O}} \coqref{Top.paper.O}{\coqdocconstructor{O}} \coqref{Top.paper.Ded.O R}{\coqdocabbreviation{\trel{O}}} \coqdocvar{vv} \coqdocvar{\tprime{vv}} \coqdocvar{\trel{vv}} \coqdocvariable{\trel{n}} \coqdocvariable{\trel{v}}}\coqexternalref{:type scope:'x7B' x ':' x 'x26' x 'x7D'}{http://coq.inria.fr/distrib/8.5pl3/stdlib/Coq.Init.Specif}{\coqdocnotation{\}}}\coqdoceol
\coqdocnoindent
\coqdockw{end}) \coqdocvariable{\trel{n}} \coqdocvariable{\trel{v}}.\coqdoceol
While translating constructors of an inductive type, we have to furnish a proof of the equality constraint. Fortunately, the canonical proof
always works while translating constructors. Here is the translation of the \coqRefConstr{Top.paper}{isnil} constructor:

\coqdocnoindent
\coqdockw{Definition}  \coqdef{Top.paper.DedM.isnil R}{\trel{isnil}}{\coqdocdefinition{\trel{isnil}}}  (\coqdocvar{vv} \coqdocvar{vv₂} : \coqref{Top.paper.Vec}{\coqdocinductive{Vec}} \coqref{Top.paper.nat}{\coqdocinductive{nat}} \coqref{Top.paper.O}{\coqdocconstructor{O}}) (\coqdocvar{\trel{vv}} : \coqref{Top.paper.DedV.Vec R}{\coqdocabbreviation{\trel{Vec}}} \coqref{Top.paper.nat}{\coqdocinductive{nat}} \coqref{Top.paper.nat}{\coqdocinductive{nat}} \coqref{Top.paper.Ded.nat R}{\coqdocabbreviation{\trel{nat}}} \coqref{Top.paper.O}{\coqdocconstructor{O}} \coqref{Top.paper.O}{\coqdocconstructor{O}} \coqref{Top.paper.Ded.O R}{\coqdocabbreviation{\trel{O}}} \coqdocvariable{vv} \coqdocvariable{vv₂})\coqdoceol
\coqdocnoindent
: \coqref{Top.paper.isNil R}{\coqdocdefinition{\trel{isNil}}} \coqref{Top.paper.O}{\coqdocconstructor{O}} \coqref{Top.paper.O}{\coqdocconstructor{O}} \coqref{Top.paper.Ded.O R}{\coqdocabbreviation{\trel{O}}} \coqdocvariable{vv} \coqdocvariable{vv₂} \coqdocvariable{\trel{vv}} (\coqref{Top.paper.isnil}{\coqdocconstructor{isnil}} \coqdocvariable{vv}) (\coqref{Top.paper.isnil}{\coqdocconstructor{isnil}} \coqdocvariable{vv₂}) := \coqexternalref{existT}{http://coq.inria.fr/distrib/8.5pl3/stdlib/Coq.Init.Specif}{\coqdocconstructor{existT}} \coqdocvariable{\trel{vv}}
(\highlight{\coqref{Top.paper.isNil refl}{\coqdocconstructor{isNil\_refl}} \coqref{Top.paper.O}{\coqdocconstructor{O}} \coqref{Top.paper.O}{\coqdocconstructor{O}} \coqref{Top.paper.Ded.O R}{\coqdocabbreviation{\trel{O}}} \coqdocvariable{vv} \coqdocvariable{vv₂} \coqdocvariable{\trel{vv}}}).\coqdoceol

Except for adding the equality constraints and proofs 
as highlighted above, our translation is largely a straightforward implementation of the description of \citet[Sec. 5.4]{Bernardy.Jansson.ea2012}.
Nevertheless, we show the general construction 
in the Appendix~(\secref{appendix:anyrel:ded:ind}).

\subsection{Pattern Matching (deductive-style)}
\label{sec:anyrel:match}
We already saw some examples of deductive-style translations of 
pattern-matching (e.g. \coqRefDefn{Top.paper}{pred}) on non-indexed inductive types.
Implementing the translation of pattern-matches on indexed-inductives is more complex, as we will illustrate with an example.
The main goal of this subsection is to discharge the claim made in the previous subsection that 
the equality constraints described in the previous subsection are crucial
for translating pattern matches.

We consider the following pattern-matching function over the indexed-inductive \coqRefInductive{Top.paper}{isNil} defined
in the previous section:

\coqdocnoindent
\begin{minipage}[t]{\textwidth}
\coqdocnoindent
\coqdockw{Definition} \coqdef{Top.paper.isNilRec}{isNilRec}{\coqdocdefinition{isNilRec}} (\coqdocvar{P} : \coqexternalref{:type scope:'xE2x88x80' x '..' x ',' x}{http://coq.inria.fr/distrib/8.5pl3/stdlib/Coq.Unicode.Utf8\_core}{\coqdocnotation{∀}} \coqexternalref{:type scope:'xE2x88x80' x '..' x ',' x}{http://coq.inria.fr/distrib/8.5pl3/stdlib/Coq.Unicode.Utf8\_core}{\coqdocnotation{(}}\coqdocvar{n} : \coqref{Top.paper.nat}{\coqdocinductive{nat}}) (\coqdocvar{v} : \coqref{Top.paper.Vec}{\coqdocinductive{Vec}} \coqref{Top.paper.nat}{\coqdocinductive{nat}} \coqdocvariable{n}\coqexternalref{:type scope:'xE2x88x80' x '..' x ',' x}{http://coq.inria.fr/distrib/8.5pl3/stdlib/Coq.Unicode.Utf8\_core}{\coqdocnotation{),}} \coqref{Top.paper.isNil}{\coqdocinductive{isNil}} \coqdocvariable{n} \coqdocvariable{v} \coqexternalref{:type scope:x 'xE2x86x92' x}{http://coq.inria.fr/distrib/8.5pl3/stdlib/Coq.Unicode.Utf8\_core}{\coqdocnotation{→}} \coqdockw{Set}) \coqdoceol
\coqdocnoindent
(\coqdocvar{f} : \coqexternalref{:type scope:'xE2x88x80' x '..' x ',' x}{http://coq.inria.fr/distrib/8.5pl3/stdlib/Coq.Unicode.Utf8\_core}{\coqdocnotation{∀}} \coqdocvar{vv} : \coqref{Top.paper.Vec}{\coqdocinductive{Vec}} \coqref{Top.paper.nat}{\coqdocinductive{nat}} \coqref{Top.paper.O}{\coqdocconstructor{O}}\coqexternalref{:type scope:'xE2x88x80' x '..' x ',' x}{http://coq.inria.fr/distrib/8.5pl3/stdlib/Coq.Unicode.Utf8\_core}{\coqdocnotation{,}} \coqdocvariable{P} \coqref{Top.paper.O}{\coqdocconstructor{O}} \coqdocvariable{vv} (\coqref{Top.paper.isnil}{\coqdocconstructor{isnil}} \coqdocvariable{vv})) (\coqdocvar{n} : \coqref{Top.paper.nat}{\coqdocinductive{nat}}) (\coqdocvar{v} : \coqref{Top.paper.Vec}{\coqdocinductive{Vec}} \coqref{Top.paper.nat}{\coqdocinductive{nat}} \coqdocvariable{n}) (\coqdocvar{d} : \coqref{Top.paper.isNil}{\coqdocinductive{isNil}} \coqdocvariable{n} \coqdocvariable{v}) : \coqdocvariable{P} \coqdocvariable{n} \coqdocvariable{v} \coqdocvariable{d}:=\coqdoceol
\coqdocnoindent
\coqdockw{match} \coqdocvariable{d}  \coqdockw{with}\coqdoceol
\coqdocnoindent
\ensuremath{|} \coqref{Top.paper.isnil}{\coqdocconstructor{isnil}} \coqdocvar{x} \ensuremath{\Rightarrow} \coqdocvariable{f} \coqdocvar{x}\coqdoceol
\coqdocnoindent
\coqdockw{end}.\coqdoceol
\end{minipage}

\vspace{0.1cm}
\noindent
It can be considered an induction/recursion principle for
 \coqRefInductive{Top.paper}{isNil}. It takes an \coqdocvar{f}
that works for the canonical forms (there is only 1), and returns a
function that works for an arbitrary member of the inductive family.
The deductive-style translation of the above function is shown in
\figref{fig:match}.
Unfortunately, it is the most complex example presented in this paper.
In general, for every argument, the translation has three arguments: 
see the clause for $\lambda$ in the definition of \ptranslate{} (\secref{sec:anyrel:core}).
%
As we enter each pattern match, the return type gets refined: the discriminee
is replaced by the constructor applied to its arguments and the indices
are replaced with the indices returned by the constructor.
Inside the first two pattern matches, \coqdocvar{n} and \coqdocvar{\tprime{n}} each become
\coqRefConstr{Top.paper}{O}, \coqdocvar{d} becomes \coqRefConstr{Top.paper}{isnil} \coqdocvar{x}, $\hdots$
. However, \coqdocvar{\trel{nn}}, which is the proof that \coqdocvar{n} and
\coqdocvarP{n} are related, doesn't change to \coqRefDefn{Top.paper}{\trel{O}}

The translation of the original body of the match, \coqdocvar{f} \coqdocvar{x}, is
\coqdocvar{\trel{f}} \coqdocvar{x} \coqdocvar{\tprime{x}} \coqdocvar{\trel{x}}.
The two outermost pattern matches bring \coqdocvar{x} and \coqdocvar{\tprime{x}} in scope, but not \coqdocvarR{x}.
In general, they bring the original constructor arguments and the \coqdocvarP{}
versions in scope.
However, we also have to bring to scope the proofs that the corresponding constructor arguments are related, e.g. that \coqdocvar{x} and \coqdocvar{\tprime{x}} are related.
In the deductive style translation, these proofs are packed as dependent pairs in the proof
that the two discriminees (\coqdocvar{d} and \coqdocvar{\tprime{d}}) are related.
In this case, that proof is \coqdocvar{\trel{dd}}. See the definition of
\coqRefDefnR{Top.paper}{isNil} to understand how the type of \coqdocvar{\trel{dd}} computes to a dependent pair.
In general, if the constructor has $n$ arguments, this type would compute to the type
of nested dependent pairs containing a total of n+1 items.
The first $n$ pattern matches on \coqdocvar{\trel{dd}} will ensure that 
all the free variables of the  translation of the body are in scope.
For example, in \figref{fig:match}, the $3^{rd}$ innermost \coqdockw{match}
brings \coqdocvar{\trel{x}} in scope.
However, the type of the translation of the body needs rewriting.
In \figref{fig:match}, we have shown the type (as checked by Coq), of the innermost pattern
match.
This is the expected return type, which as described above, has
\coqdocvar{\trel{nn}} instead of \coqRefDefn{Top.paper}{\trel{O}}, etc. The type of the translation of the body, which 
is the innermost body in the translation, is also shown and aligned to the 
expected return type. 
Note that \coqdocvar{\trel{nn}} in the outer type needs to change to 
\coqRefDefn{Top.paper}{\trel{O}},
\coqdocvar{\trel{vv}} to \coqdocvar{\trel{x}},
and the dependent pair
(\coqexternalref{existT}{http://coq.inria.fr/distrib/8.5pl3/stdlib/Coq.Init.Specif}{\coqdocconstructor{existT}} \coqdocvar{\trel{x}} \coqdocvar{pdeq})
needs to change to
(\coqref{Top.paper.DedM.isnil R}{\coqdocdefinition{\trel{isnil}}} \coqdocvar{x} \coqdocvar{\tprime{x}} \coqdocvar{\trel{x}}).
The latter computes to the dependent pair 
(\coqexternalref{existT}{http://coq.inria.fr/distrib/8.5pl3/stdlib/Coq.Init.Specif}{\coqdocconstructor{existT}} \coqdocvar{\trel{x}}
(\coqref{Top.paper.isNil refl}{\coqdocconstructor{isNil\_refl}} 
$\hdots$
))
Thus, the last change is essentially to change \coqdocvar{pdeq} to the canonical proof
(\coqref{Top.paper.isNil refl}{\coqdocconstructor{isNil\_refl}} $\hdots$), as hinted
in the previous subsection. All these changes are achieved by \emph{just one}
pattern match on \coqdocvar{pdeq}, the proof of the generalized equality type described in the previous subsection.

\begin{figure}
\coqdocnoindent
\coqdockw{Definition}
\coqdef{Top.paper.isNilRec}{isNilRec}{\coqdocdefinition{\trel{isNilRec}}}
(\coqdocvar{P} : \coqexternalref{:type scope:'xE2x88x80' x '..' x ',' x}{http://coq.inria.fr/distrib/8.5pl3/stdlib/Coq.Unicode.Utf8\_core}{\coqdocnotation{∀}} \coqexternalref{:type scope:'xE2x88x80' x '..' x ',' x}{http://coq.inria.fr/distrib/8.5pl3/stdlib/Coq.Unicode.Utf8\_core}{\coqdocnotation{(}}\coqdocvar{n} : \coqref{Top.paper.nat}{\coqdocinductive{nat}}) (\coqdocvar{v} : \coqref{Top.paper.Vec}{\coqdocinductive{Vec}} \coqref{Top.paper.nat}{\coqdocinductive{nat}} \coqdocvariable{n}\coqexternalref{:type scope:'xE2x88x80' x '..' x ',' x}{http://coq.inria.fr/distrib/8.5pl3/stdlib/Coq.Unicode.Utf8\_core}{\coqdocnotation{),}} \coqref{Top.paper.isNil}{\coqdocinductive{isNil}} \coqdocvariable{n} \coqdocvariable{v} \coqexternalref{:type scope:x 'xE2x86x92' x}{http://coq.inria.fr/distrib/8.5pl3/stdlib/Coq.Unicode.Utf8\_core}{\coqdocnotation{→}} \coqdockw{Set})\coqdoceol
\coqdocindent{0.50em}
(\coqdocvar{\tprime{P}} : \coqexternalref{:type scope:'xE2x88x80' x '..' x ',' x}{http://coq.inria.fr/distrib/8.5pl3/stdlib/Coq.Unicode.Utf8\_core}{\coqdocnotation{∀}} \coqexternalref{:type scope:'xE2x88x80' x '..' x ',' x}{http://coq.inria.fr/distrib/8.5pl3/stdlib/Coq.Unicode.Utf8\_core}{\coqdocnotation{(}}\coqdocvar{\tprime{n}} : \coqref{Top.paper.nat}{\coqdocinductive{nat}}) (\coqdocvar{\tprime{v}} : \coqref{Top.paper.Vec}{\coqdocinductive{Vec}} \coqref{Top.paper.nat}{\coqdocinductive{nat}} \coqdocvariable{\tprime{n}}\coqexternalref{:type scope:'xE2x88x80' x '..' x ',' x}{http://coq.inria.fr/distrib/8.5pl3/stdlib/Coq.Unicode.Utf8\_core}{\coqdocnotation{),}} \coqref{Top.paper.isNil}{\coqdocinductive{isNil}} \coqdocvariable{\tprime{n}} \coqdocvariable{\tprime{v}} \coqexternalref{:type scope:x 'xE2x86x92' x}{http://coq.inria.fr/distrib/8.5pl3/stdlib/Coq.Unicode.Utf8\_core}{\coqdocnotation{→}} \coqdockw{Set})\coqdoceol
\coqdocindent{0.50em}
(\coqdocvar{\trel{P}} : \coqexternalref{:type scope:'xE2x88x80' x '..' x ',' x}{http://coq.inria.fr/distrib/8.5pl3/stdlib/Coq.Unicode.Utf8\_core}{\coqdocnotation{∀}} \coqexternalref{:type scope:'xE2x88x80' x '..' x ',' x}{http://coq.inria.fr/distrib/8.5pl3/stdlib/Coq.Unicode.Utf8\_core}{\coqdocnotation{(}}\coqdocvar{n} \coqdocvar{\tprime{n}} : \coqref{Top.paper.nat}{\coqdocinductive{nat}}) (\coqdocvar{\trel{n}} : \coqref{Top.paper.Ded.nat R}{\coqdocabbreviation{\trel{nat}}} \coqdocvariable{n} \coqdocvariable{\tprime{n}}) (\coqdocvar{v} : \coqref{Top.paper.Vec}{\coqdocinductive{Vec}} \coqref{Top.paper.nat}{\coqdocinductive{nat}} \coqdocvariable{n}) (\coqdocvar{\tprime{v}} : \coqref{Top.paper.Vec}{\coqdocinductive{Vec}} \coqref{Top.paper.nat}{\coqdocinductive{nat}} \coqdocvariable{\tprime{n}})\coqdoceol
\coqdocindent{5.50em}
(\coqdocvar{\trel{v}} : \coqref{Top.paper.DedV.Vec R}{\coqdocabbreviation{\trel{Vec}}} \coqref{Top.paper.nat}{\coqdocinductive{nat}} \coqref{Top.paper.nat}{\coqdocinductive{nat}} \coqref{Top.paper.Ded.nat R}{\coqdocabbreviation{\trel{nat}}} \coqdocvariable{n} \coqdocvariable{\tprime{n}} \coqdocvariable{\trel{n}} \coqdocvariable{v} \coqdocvariable{\tprime{v}}) (\coqdocvar{d} : \coqref{Top.paper.isNil}{\coqdocinductive{isNil}} \coqdocvariable{n} \coqdocvariable{v})\coqdoceol
\coqdocindent{6.00em}
(\coqdocvar{\tprime{d}} : \coqref{Top.paper.isNil}{\coqdocinductive{isNil}} \coqdocvariable{\tprime{n}} \coqdocvariable{\tprime{v}}) (\coqdocvar{\trel{d}} : \coqref{Top.paper.isNil R}{\coqdocdefinition{\trel{isNil}}} \coqdocvariable{n} \coqdocvariable{\tprime{n}} \coqdocvariable{\trel{n}} \coqdocvariable{v} \coqdocvariable{\tprime{v}} \coqdocvariable{\trel{v}} \coqdocvariable{d} \coqdocvariable{\tprime{d}}\coqexternalref{:type scope:'xE2x88x80' x '..' x ',' x}{http://coq.inria.fr/distrib/8.5pl3/stdlib/Coq.Unicode.Utf8\_core}{\coqdocnotation{),}} \coqdocvariable{P} \coqdocvariable{n} \coqdocvariable{v} \coqdocvariable{d} \coqexternalref{:type scope:x 'xE2x86x92' x}{http://coq.inria.fr/distrib/8.5pl3/stdlib/Coq.Unicode.Utf8\_core}{\coqdocnotation{→}} \coqdocvariable{\tprime{P}} \coqdocvariable{\tprime{n}} \coqdocvariable{\tprime{v}} \coqdocvariable{\tprime{d}} \coqexternalref{:type scope:x 'xE2x86x92' x}{http://coq.inria.fr/distrib/8.5pl3/stdlib/Coq.Unicode.Utf8\_core}{\coqdocnotation{→}} \coqdockw{Prop})\coqdoceol
\coqdocindent{0.50em}
(\coqdocvar{f} : \coqexternalref{:type scope:'xE2x88x80' x '..' x ',' x}{http://coq.inria.fr/distrib/8.5pl3/stdlib/Coq.Unicode.Utf8\_core}{\coqdocnotation{∀}} \coqdocvar{vv} : \coqref{Top.paper.Vec}{\coqdocinductive{Vec}} \coqref{Top.paper.nat}{\coqdocinductive{nat}} \coqref{Top.paper.O}{\coqdocconstructor{O}}\coqexternalref{:type scope:'xE2x88x80' x '..' x ',' x}{http://coq.inria.fr/distrib/8.5pl3/stdlib/Coq.Unicode.Utf8\_core}{\coqdocnotation{,}} \coqdocvariable{P} \coqref{Top.paper.O}{\coqdocconstructor{O}} \coqdocvariable{vv} (\coqref{Top.paper.isnil}{\coqdocconstructor{isnil}} \coqdocvariable{vv})) (\coqdocvar{\tprime{f}} : \coqexternalref{:type scope:'xE2x88x80' x '..' x ',' x}{http://coq.inria.fr/distrib/8.5pl3/stdlib/Coq.Unicode.Utf8\_core}{\coqdocnotation{∀}} \coqdocvar{\tprime{vv}} : \coqref{Top.paper.Vec}{\coqdocinductive{Vec}} \coqref{Top.paper.nat}{\coqdocinductive{nat}} \coqref{Top.paper.O}{\coqdocconstructor{O}}\coqexternalref{:type scope:'xE2x88x80' x '..' x ',' x}{http://coq.inria.fr/distrib/8.5pl3/stdlib/Coq.Unicode.Utf8\_core}{\coqdocnotation{,}} \coqdocvariable{\tprime{P}} \coqref{Top.paper.O}{\coqdocconstructor{O}} \coqdocvariable{\tprime{vv}} (\coqref{Top.paper.isnil}{\coqdocconstructor{isnil}} \coqdocvariable{\tprime{vv}}))\coqdoceol
\coqdocindent{0.50em}
(\coqdocvar{\trel{f}} : \coqexternalref{:type scope:'xE2x88x80' x '..' x ',' x}{http://coq.inria.fr/distrib/8.5pl3/stdlib/Coq.Unicode.Utf8\_core}{\coqdocnotation{∀}} \coqexternalref{:type scope:'xE2x88x80' x '..' x ',' x}{http://coq.inria.fr/distrib/8.5pl3/stdlib/Coq.Unicode.Utf8\_core}{\coqdocnotation{(}}\coqdocvar{vv} \coqdocvar{\tprime{vv}} : \coqref{Top.paper.Vec}{\coqdocinductive{Vec}} \coqref{Top.paper.nat}{\coqdocinductive{nat}} \coqref{Top.paper.O}{\coqdocconstructor{O}}) (\coqdocvar{\trel{vv}} : \coqref{Top.paper.DedV.Vec R}{\coqdocabbreviation{\trel{Vec}}} \coqref{Top.paper.nat}{\coqdocinductive{nat}} \coqref{Top.paper.nat}{\coqdocinductive{nat}} \coqref{Top.paper.Ded.nat R}{\coqdocabbreviation{\trel{nat}}} \coqref{Top.paper.O}{\coqdocconstructor{O}} \coqref{Top.paper.O}{\coqdocconstructor{O}} \coqref{Top.paper.Ded.O R}{\coqdocabbreviation{\trel{O}}} \coqdocvariable{vv} \coqdocvariable{\tprime{vv}}\coqexternalref{:type scope:'xE2x88x80' x '..' x ',' x}{http://coq.inria.fr/distrib/8.5pl3/stdlib/Coq.Unicode.Utf8\_core}{\coqdocnotation{),}}\coqdoceol
\coqdocindent{6.00em}
\coqdocvariable{\trel{P}} \coqref{Top.paper.O}{\coqdocconstructor{O}} \coqref{Top.paper.O}{\coqdocconstructor{O}} \coqref{Top.paper.Ded.O R}{\coqdocabbreviation{\trel{O}}} \coqdocvariable{vv} \coqdocvariable{\tprime{vv}} \coqdocvariable{\trel{vv}} (\coqref{Top.paper.isnil}{\coqdocconstructor{isnil}} \coqdocvariable{vv}) (\coqref{Top.paper.isnil}{\coqdocconstructor{isnil}} \coqdocvariable{\tprime{vv}}) (\coqref{Top.paper.DedM.isnil R}{\coqdocdefinition{\trel{isnil}}} \coqdocvariable{vv} \coqdocvariable{\tprime{vv}} \coqdocvariable{\trel{vv}}) (\coqdocvariable{f} \coqdocvariable{vv}) (\coqdocvariable{\tprime{f}} \coqdocvariable{\tprime{vv}})) \coqdoceol
\coqdocindent{0.50em}
(\coqdocvar{n} \coqdocvar{\tprime{n}} : \coqref{Top.paper.nat}{\coqdocinductive{nat}}) (\coqdocvar{\trel{n}} : \coqref{Top.paper.Ded.nat R}{\coqdocabbreviation{\trel{nat}}} \coqdocvariable{n} \coqdocvariable{\tprime{n}}) (\coqdocvar{v} : \coqref{Top.paper.Vec}{\coqdocinductive{Vec}} \coqref{Top.paper.nat}{\coqdocinductive{nat}} \coqdocvariable{n}) (\coqdocvar{\tprime{v}} : \coqref{Top.paper.Vec}{\coqdocinductive{Vec}} \coqref{Top.paper.nat}{\coqdocinductive{nat}} \coqdocvariable{\tprime{n}}) \coqdoceol
\coqdocindent{0.50em}
(\coqdocvar{\trel{v}} : \coqref{Top.paper.DedV.Vec R}{\coqdocabbreviation{\trel{Vec}}} \coqref{Top.paper.nat}{\coqdocinductive{nat}} \coqref{Top.paper.nat}{\coqdocinductive{nat}} \coqref{Top.paper.Ded.nat R}{\coqdocabbreviation{\trel{nat}}} \coqdocvariable{n} \coqdocvariable{\tprime{n}} \coqdocvariable{\trel{n}} \coqdocvariable{v} \coqdocvariable{\tprime{v}}) (\coqdocvar{d} : \coqref{Top.paper.isNil}{\coqdocinductive{isNil}} \coqdocvariable{n} \coqdocvariable{v})\coqdoceol
\coqdocindent{0.50em}
(\coqdocvar{\tprime{d}} : \coqref{Top.paper.isNil}{\coqdocinductive{isNil}} \coqdocvariable{\tprime{n}} \coqdocvariable{\tprime{v}}) (\coqdocvar{\trel{d}} : \coqref{Top.paper.isNil R}{\coqdocdefinition{\trel{isNil}}} \coqdocvariable{n} \coqdocvariable{\tprime{n}} \coqdocvariable{\trel{n}} \coqdocvariable{v} \coqdocvariable{\tprime{v}} \coqdocvariable{\trel{v}} \coqdocvariable{d} \coqdocvariable{\tprime{d}}) :\coqdoceol
\coqdocindent{0.50em}
\coqdocvariable{\trel{P}} \coqdocvar{\_} \coqdocvar{\_} \coqdocvariable{\trel{n}} \coqdocvar{\_} \coqdocvar{\_} \coqdocvariable{\trel{v}} \coqdocvar{\_} \coqdocvar{\_} \coqdocvariable{\trel{d}} (\coqref{Top.paper.isNilRec}{\coqdocdefinition{isNilRec}} \coqdocvar{\_} \coqdocvariable{f} \coqdocvar{\_} \coqdocvar{\_}  \coqdocvariable{d}) (\coqref{Top.paper.isNilRec}{\coqdocdefinition{isNilRec}} \coqdocvar{\_} \coqdocvariable{\tprime{f}} \coqdocvar{\_} \coqdocvar{\_}  \coqdocvariable{\tprime{d}}) :=\coqdoceol
\coqdocnoindent
\coqdockw{match} \coqdocvariable{d} 
\coqdockw{as} $\hdots$ \coqdockw{in} $\hdots$ \coqdockw{return} $\hdots$  \coqdockw{with}\coqdoceol
\coqdocnoindent
\ensuremath{|} \coqref{Top.paper.isnil}{\coqdocconstructor{isnil}} \coqdocvar{x} \ensuremath{\Rightarrow} 
\coqdockw{match} \coqdocvariable{\tprime{d}}
\coqdockw{as} $\hdots$ \coqdockw{in} $\hdots$ \coqdockw{return} $\hdots$  \coqdockw{with}\coqdoceol
\coqdocindent{1.00em}
\ensuremath{|} \coqref{Top.paper.isnil}{\coqdocconstructor{isnil}} \coqdocvar{\tprime{x}} \ensuremath{\Rightarrow}
{\coqdocnotation{\ensuremath{\lambda}}} {\coqdocnotation{(}}\coqdocvar{\trel{nn}} : \coqref{Top.paper.Ded.nat R}{\coqdocabbreviation{\trel{nat}}} \coqref{Top.paper.O}{\coqdocconstructor{O}} \coqref{Top.paper.O}{\coqdocconstructor{O}}) (\coqdocvar{\trel{vv}} : \coqref{Top.paper.DedV.Vec R}{\coqdocabbreviation{\trel{Vec}}} \coqref{Top.paper.nat}{\coqdocinductive{nat}} \coqref{Top.paper.nat}{\coqdocinductive{nat}} \coqref{Top.paper.Ded.nat R}{\coqdocabbreviation{\trel{nat}}} \coqref{Top.paper.O}{\coqdocconstructor{O}} \coqref{Top.paper.O}{\coqdocconstructor{O}} \coqdocvariable{\trel{nn}} \coqdocvar{x} \coqdocvar{\tprime{x}})\coqdoceol
\coqdocindent{2.00em}
(\coqdocvar{\trel{dd}} : \coqref{Top.paper.isNil R}{\coqdocdefinition{\trel{isNil}}} \coqref{Top.paper.O}{\coqdocconstructor{O}} \coqref{Top.paper.O}{\coqdocconstructor{O}} \coqdocvariable{\trel{nn}} \coqdocvar{x} \coqdocvar{\tprime{x}} \coqdocvariable{\trel{vv}} (\coqref{Top.paper.isnil}{\coqdocconstructor{isnil}} \coqdocvar{x}) (\coqref{Top.paper.isnil}{\coqdocconstructor{isnil}} \coqdocvar{\tprime{x}}){\coqdocnotation{),}}
\coqdoceol\coqdocindent{2.00em}
\coqdockw{match} \coqdocvariable{\trel{dd}} \coqdockw{with}\coqdoceol
\coqdocindent{2.00em}
\ensuremath{|} \coqexternalref{existT}{http://coq.inria.fr/distrib/8.5pl3/stdlib/Coq.Init.Specif}{\coqdocconstructor{existT}} \coqdocvar{\trel{x}} \coqdocvar{pdeq} \ensuremath{\Rightarrow} 
\coqdoceol
\coqdocindent{3.00em}
(\coqdockw{match} \coqdocvar{pdeq} 
\coqdockw{as} $\hdots$ \coqdockw{in} $\hdots$ \coqdockw{return} $\hdots$  \coqdockw{with}\coqdoceol
\coqdocindent{3.50em}
\ensuremath{|} \coqref{Top.paper.isNil refl}{\coqdocconstructor{isNil\_refl}} \coqdocvar{\_} \coqdocvar{\_} \coqdocvar{\_} \coqdocvar{\_} \coqdocvar{\_} \coqdocvar{\_} \ensuremath{\Rightarrow} (\coqdocvariable{\trel{f}} \coqdocvar{x} \coqdocvar{\tprime{x}} \coqdocvar{\trel{x}}):\coqdoceol
\coqdocindent{5.90em}
(\coqdocvariable{\trel{P}} \coqref{Top.paper.O}{\coqdocconstructor{O}} \coqref{Top.paper.O}{\coqdocconstructor{O}} 
\highlight{\coqref{Top.paper.Ded.O R}{\coqdocabbreviation{\trel{O}}}}
 $\,$ \coqdocvar{x} \coqdocvar{\tprime{x}} 
 \highlight{\coqdocvar{\trel{x}}} $\,$
(\coqref{Top.paper.isnil}{\coqdocconstructor{isnil}} \coqdocvar{x}) 
(\coqref{Top.paper.isnil}{\coqdocconstructor{isnil}} \coqdocvar{\tprime{x}}) 
\highlight{(\coqref{Top.paper.DedM.isnil R}{\coqdocdefinition{\trel{isnil}}} \coqdocvar{x} \coqdocvar{\tprime{x}} \coqdocvar{\trel{x}})}
(\coqdocvariable{f} \coqdocvar{x}) (\coqdocvariable{\tprime{f}} \coqdocvar{\tprime{x}}))\coqdoceol
\coqdocindent{3.50em}
\coqdockw{end}):
(\coqdocvariable{\trel{P}} \coqref{Top.paper.O}{\coqdocconstructor{O}} \coqref{Top.paper.O}{\coqdocconstructor{O}}
\highlight{\coqdocvariable{\trel{nn}}} 
\coqdocvar{x} \coqdocvar{\tprime{x}} 
\highlight{\coqdocvariable{\trel{vv}}}
(\coqref{Top.paper.isnil}{\coqdocconstructor{isnil}} \coqdocvar{x}) (\coqref{Top.paper.isnil}{\coqdocconstructor{isnil}} \coqdocvar{\tprime{x}}) 
\highlight{(\coqexternalref{existT}{http://coq.inria.fr/distrib/8.5pl3/stdlib/Coq.Init.Specif}{\coqdocconstructor{existT}} \coqdocvar{\trel{x}}
\coqdocvar{pdeq})} 
(\coqdocvariable{f} \coqdocvar{x}) (\coqdocvariable{\tprime{f}} \coqdocvar{\tprime{x}}))\coqdoceol
\coqdocindent{2.00em}
\coqdockw{end}\coqdoceol
\coqdocindent{1.00em}
\coqdockw{end}\coqdoceol
\coqdocnoindent
\coqdockw{end} \coqdocvariable{\trel{n}} \coqdocvariable{\trel{v}} \coqdocvariable{\trel{d}}.
\coqdoceol
\caption{Translation of pattern matching requires the equality constraints}
\label{fig:match}
\end{figure}

The general scheme for translating pattern matches can be found in the
Appendix~(\secref{appendix:anyrel:ded:match}).

\subsection{Fixpoints (recursive functions)}
Our translation of \coqdockw{fix} (or \coqdockw{Fixpoint}) terms is largely 
as described by \citet{Keller.Lasson2012}. 
A minor change was required because we translate inductive types and corresponding
pattern matches in the deductive style.
It is explained in \appref{appendix:anyrel:fix}.
 


\subsection{Summary}
In this section, we presented the {\anyrel} translation that will serve as the
core of the {\isorel} translation described in the rest of the paper.
The main advantage of the translation in this section over the {\anyrel}
translation implemented by \citet{Keller.Lasson2012} is that we have $\hat{\coqdockw{Set}}$:=\coqdockw{Prop}, which means that
relations for types in the \coqdockw{Set} universe enjoy the proof irrelevance
property, which is useful not only in the {\isorel} translation, but in other
applications as well.
Ensuring $\hat{\coqdockw{Set}}$:=\coqdockw{Prop} required a deductive-style
translation of inductive types. 
We found that the deductive-style translation of pattern matches on inhabitants
of indexed-inductive types 
requires strengthening the deductive-style translation of those types with
equality constraints that were erroneously missing in the literature
\citep{Bernardy.Jansson.ea2012,Bernardy2011}.
Stating and using those equality constraints becomes challenging for inductive
types with multiple, dependent indices. We showed how to simplify the
construction.
We also showed that the deductive-style
translation does not work for inductively defined propositions: those need to be translated in the inductive
style, as does pattern matching on proofs of those propositions.

We have implemented our {\anyrel} translation as functions in Coq
itself. Using reification and reflection~\cite{Malecha.Sozeau2014}, 
we have used those Coq functions to
translate several examples.
The translated program is delivered to the reflection mechanism which
ensures that the result is \emph{well-typed} before adding it to Coq's
environment of definitions and declarations. Also, 
our translation produces all the implicit arguments, and is thus immune to the incompleteness of Coq's type inference mechanism.

%


%

\section{Uniformity of Propositions}
\label{sec:uniformProp}
We begin this section by describing why \ptranslate{\coqdockw{Prop}} is too weak to ensure the uniformity of propositions. 
Then and in the next section, 
we develop the main technical lemmas needed to ensure the uniformity. 
In \secref{sec:isorel}, we use these lemmas in the {\isorel} translation \ptranslateIso{}, which ensures the uniformity of propositions.
We believe the lemmas in Section \ref{sec:uniformProp} and \ref{sec:uniformProp:type} are independently interesting and useful.

Recall (\secref{sec:anyrel:core}) that 
\ptranslate{\coqdockw{Prop}} :=
\ensuremath{\lambda(\coqdocvar{P}\;\coqdocvarP{P}:\coqdockw{Prop}),\coqdocvar{P}
→ \coqdocvarP{P} → \coqdockw{Prop}}.
If we have $\theta$:\coqdockw{Prop},
\thmref{Abstraction} says 
\ptranslate{$\theta$}: $\theta$ → \tprime{$\theta$} → \coqdockw{Prop}.
In applications of parametricity, $\theta$ would typically denote
a proposition in one instantiation and \tprime{$\theta$} would denote the corresponding
proposition in the other instantiation. In the example at the beginning of 
\secref{sec:intro}, one instantiation is the cartesian representation of complex numbers,
and the other instantiation is the polar representation of complex numbers.
$\theta$ and \tprime{$\theta$}, in the respective instantiations, could
be the proposition that addition is commutative. 

We want the two propositions to mean the same in both the instantiations.
%
However, the statement (type) of \ptranslate{$\theta$},
which is the proof that  $\theta$  and \tprime{$\theta$} are parametrically related,
is too weak. \ptranslate{$\theta$} is merely a relation between
$\theta$  and \tprime{$\theta$}. As explained in \secref{sec:intro}, there is a
relation even between logically inequivalent propositions, such as 
{\CoqTrue} and {\CoqFalse}. In contrast, if we instead had $\theta$:{\CoqBool},
\thmref{Abstraction} says:
\ptranslate{$\theta$}: \coqRefDefn{Top.paper}{\trel{bool}} $\theta$ \tprime{$\theta$}, 
where \coqRefDefn{Top.paper}{\trel{bool}} is the deductive-style translation of
the inductive type {\CoqBool}, which has only two constructors:
{\CoqBTrue} and {\CoqBFalse}.
We hope that from the previous section, it is clear that 
\coqRefDefn{Top.paper}{\trel{bool}} $\theta$ \tprime{$\theta$} implies that
either \emph{both} $\theta$ and \tprime{$\theta$} reduce to  
{\CoqBTrue}, or \emph{both} reduce to {\CoqBFalse}.

The main goal of this paper is to strengthen the translation of the universe
\coqdockw{Prop} to get uniformity properties similar to the type {\CoqBool}.
In \secref{sec:intro}, we identified and motivated two properties that we wish to have for the relations between (proofs of) propositions:

\coqdocnoindent
\coqdockw{Definition} \coqdef{Top.paper.IffProps}{IffProps}{\coqdocdefinition{IffProps}} \{\coqdocvar{A} \coqdocvar{B} : \coqdockw{Prop}\} (\coqdocvar{R} : \coqdocvariable{A} \coqexternalref{:type scope:x '->' x}{http://coq.inria.fr/distrib/8.5pl3/stdlib/Coq.Init.Logic}{\coqdocnotation{\ensuremath{\rightarrow}}} \coqdocvariable{B} \coqexternalref{:type scope:x '->' x}{http://coq.inria.fr/distrib/8.5pl3/stdlib/Coq.Init.Logic}{\coqdocnotation{\ensuremath{\rightarrow}}} \coqdockw{Prop}) : \coqdockw{Prop} := \coqdocvariable{A} \coqexternalref{:type scope:x '<->' x}{http://coq.inria.fr/distrib/8.5pl3/stdlib/Coq.Init.Logic}{\coqdocnotation{\ensuremath{\leftrightarrow}}} \coqdocvariable{B}.\coqdoceol\coqdocnoindent
\coqdockw{Definition} \coqdef{Top.paper.CompleteRel}{CompleteRel}{\coqdocdefinition{CompleteRel}}  \{\coqdocvar{A} \coqdocvar{B} : \coqdockw{Prop}\} (\coqdocvar{R} : \coqdocvariable{A} \coqexternalref{:type scope:x '->' x}{http://coq.inria.fr/distrib/8.5pl3/stdlib/Coq.Init.Logic}{\coqdocnotation{\ensuremath{\rightarrow}}} \coqdocvariable{B} \coqexternalref{:type scope:x '->' x}{http://coq.inria.fr/distrib/8.5pl3/stdlib/Coq.Init.Logic}{\coqdocnotation{\ensuremath{\rightarrow}}} \coqdockw{Prop}) : \coqdockw{Prop} := \coqdockw{\ensuremath{\forall}} (\coqdocvar{a} : \coqdocvariable{A}) (\coqdocvar{b} : \coqdocvariable{B}), \coqdocvariable{R} \coqdocvariable{a} \coqdocvariable{b}.\coqdoceol
\coqdocnoindent
To ensure these properties, in
the {\isorel} translation, we define the translation of \coqdockw{Prop} in a way that is equivalent to the following:\\
\ptranslateIso{\coqdockw{Prop}} :=
{\coqdocnotation{\ensuremath{\lambda}}} {\coqdocnotation{(}}\coqdocvar{A} \coqdocvar{\tprime{A}}: \coqdockw{Prop}{\coqdocnotation{),}} \coqexternalref{:type scope:'x7B' x ':' x 'x26' x 'x7D'}{http://coq.inria.fr/distrib/8.5pl3/stdlib/Coq.Init.Specif}{\coqdocnotation{\{}}\coqdocvar{R} \coqexternalref{:type scope:'x7B' x ':' x 'x26' x 'x7D'}{http://coq.inria.fr/distrib/8.5pl3/stdlib/Coq.Init.Specif}{\coqdocnotation{:}} \coqdocvariable{A} \coqexternalref{:type scope:x '->' x}{http://coq.inria.fr/distrib/8.5pl3/stdlib/Coq.Init.Logic}{\coqdocnotation{\ensuremath{\rightarrow}}} \coqdocvariable{\tprime{A}} \coqexternalref{:type scope:x '->' x}{http://coq.inria.fr/distrib/8.5pl3/stdlib/Coq.Init.Logic}{\coqdocnotation{\ensuremath{\rightarrow}}} \coqdockw{Prop} \coqexternalref{:type scope:'x7B' x ':' x 'x26' x 'x7D'}{http://coq.inria.fr/distrib/8.5pl3/stdlib/Coq.Init.Specif}{\coqdocnotation{\&}} \coqref{Top.paper.IffProps}{\coqdocdefinition{IffProps}} \coqdocvar{R} \coqexternalref{:type scope:x '/x5C' x}{http://coq.inria.fr/distrib/8.5pl3/stdlib/Coq.Init.Logic}{\coqdocnotation{\ensuremath{\land}}} \coqref{Top.paper.CompleteRel}{\coqdocdefinition{CompleteRel}} \coqdocvar{R}\coqexternalref{:type scope:'x7B' x ':' x 'x26' x 'x7D'}{http://coq.inria.fr/distrib/8.5pl3/stdlib/Coq.Init.Specif}{\coqdocnotation{\}}}.\coqdoceol

Instead of returning an arbitrary relation,
the {\isorel} translation requires the relation to come bundled (as a dependent
pair) with proofs of the above two properties (of the relation) of interest.
This paper is mainly about tackling the far-reaching consequences of the above change.
The contributions of the previous section, although independently interesting, were made to ensure
that we can have
$\hat{\coqdockw{Set}}$ := \coqdockw{Prop}
(instead of
$\hat{\coqdockw{Set}}$ := \coqdockw{Set}), which makes it easy to tackle some of the consequences.
In \ptranslateIso{}, other parts of \ptranslate{} also need to be updated to cope with the change in the translation of \coqdockw{Prop}, so that we get essentially the same
abstraction theorem as before (\thmref{Abstraction}).
For example, as we will see in this section, we also need 
to bundle relations for types with some properties.
We will see in the next subsections that 
the relations produced by translating the types mentioned in propositions may need to have one or both
of the following properties:

\coqdocnoindent
\coqdockw{Definition} \coqdef{Top.paper.OneToOne}{OneToOne}{\coqdocdefinition{OneToOne}}  \{\coqdocvar{A} \coqdocvar{B} : \coqdockw{Set}\} (\coqdocvar{R} : \coqdocvariable{A} \coqexternalref{:type scope:x '->' x}{http://coq.inria.fr/distrib/8.5pl3/stdlib/Coq.Init.Logic}{\coqdocnotation{\ensuremath{\rightarrow}}} \coqdocvariable{B} \coqexternalref{:type scope:x '->' x}{http://coq.inria.fr/distrib/8.5pl3/stdlib/Coq.Init.Logic}{\coqdocnotation{\ensuremath{\rightarrow}}} \coqdockw{Prop}) : \coqdockw{Prop} :=
\coqdoceol\coqdocindent{1em}
\coqexternalref{:type scope:x '/x5C' x}{http://coq.inria.fr/distrib/8.5pl3/stdlib/Coq.Init.Logic}{\coqdocnotation{(}}\coqdockw{\ensuremath{\forall}} (\coqdocvar{a}:\coqdocvariable{A}) (\coqdocvar{b1} \coqdocvar{b2}: \coqdocvariable{B}), \coqdocvariable{R} \coqdocvariable{a} \coqdocvariable{b1} \coqexternalref{:type scope:x '->' x}{http://coq.inria.fr/distrib/8.5pl3/stdlib/Coq.Init.Logic}{\coqdocnotation{\ensuremath{\rightarrow}}} \coqdocvariable{R} \coqdocvariable{a} \coqdocvariable{b2} \coqexternalref{:type scope:x '->' x}{http://coq.inria.fr/distrib/8.5pl3/stdlib/Coq.Init.Logic}{\coqdocnotation{\ensuremath{\rightarrow}}}  \coqdocvariable{b1}\coqexternalref{:type scope:x '=' x}{http://coq.inria.fr/distrib/8.5pl3/stdlib/Coq.Init.Logic}{\coqdocnotation{=}}\coqdocvariable{b2}\coqexternalref{:type scope:x '/x5C' x}{http://coq.inria.fr/distrib/8.5pl3/stdlib/Coq.Init.Logic}{\coqdocnotation{)}} \coqexternalref{:type scope:x '/x5C' x}{http://coq.inria.fr/distrib/8.5pl3/stdlib/Coq.Init.Logic}{\coqdocnotation{\ensuremath{\land}}} \coqexternalref{:type scope:x '/x5C' x}{http://coq.inria.fr/distrib/8.5pl3/stdlib/Coq.Init.Logic}{\coqdocnotation{(}}\coqdockw{\ensuremath{\forall}} (\coqdocvar{b}:\coqdocvariable{B}) (\coqdocvar{a1} \coqdocvar{a2}: \coqdocvariable{A}), \coqdocvariable{R} \coqdocvariable{a1} \coqdocvariable{b} \coqexternalref{:type scope:x '->' x}{http://coq.inria.fr/distrib/8.5pl3/stdlib/Coq.Init.Logic}{\coqdocnotation{\ensuremath{\rightarrow}}} \coqdocvariable{R} \coqdocvariable{a2} \coqdocvariable{b} \coqexternalref{:type scope:x '->' x}{http://coq.inria.fr/distrib/8.5pl3/stdlib/Coq.Init.Logic}{\coqdocnotation{\ensuremath{\rightarrow}}}  \coqdocvariable{a1}\coqexternalref{:type scope:x '=' x}{http://coq.inria.fr/distrib/8.5pl3/stdlib/Coq.Init.Logic}{\coqdocnotation{=}}\coqdocvariable{a2}\coqexternalref{:type scope:x '/x5C' x}{http://coq.inria.fr/distrib/8.5pl3/stdlib/Coq.Init.Logic}{\coqdocnotation{)}}.\coqdoceol
%
\coqdocnoindent
\coqdockw{Definition} \coqdef{Top.paper.Total}{Total}{\coqdocdefinition{Total}} \{\coqdocvar{A} \coqdocvar{B} : \coqdockw{Set}\} (\coqdocvar{R}: \coqdocvariable{A} \coqexternalref{:type scope:x '->' x}{http://coq.inria.fr/distrib/8.5pl3/stdlib/Coq.Init.Logic}{\coqdocnotation{\ensuremath{\rightarrow}}} \coqdocvariable{B} \coqexternalref{:type scope:x '->' x}{http://coq.inria.fr/distrib/8.5pl3/stdlib/Coq.Init.Logic}{\coqdocnotation{\ensuremath{\rightarrow}}} \coqdockw{Prop}) : \coqdockw{Type} :=
\coqdoceol\coqdocindent{1em}
\coqexternalref{:type scope:x '*' x}{http://coq.inria.fr/distrib/8.5pl3/stdlib/Coq.Init.Datatypes}{\coqdocnotation{(}}\coqdockw{\ensuremath{\forall}} (\coqdocvar{a}:\coqdocvariable{A}), \coqexternalref{:type scope:'x7B' x ':' x 'x26' x 'x7D'}{http://coq.inria.fr/distrib/8.5pl3/stdlib/Coq.Init.Specif}{\coqdocnotation{\{}}\coqdocvar{b}\coqexternalref{:type scope:'x7B' x ':' x 'x26' x 'x7D'}{http://coq.inria.fr/distrib/8.5pl3/stdlib/Coq.Init.Specif}{\coqdocnotation{:}}\coqdocvariable{B} \coqexternalref{:type scope:'x7B' x ':' x 'x26' x 'x7D'}{http://coq.inria.fr/distrib/8.5pl3/stdlib/Coq.Init.Specif}{\coqdocnotation{\&}} \coqexternalref{:type scope:'x7B' x ':' x 'x26' x 'x7D'}{http://coq.inria.fr/distrib/8.5pl3/stdlib/Coq.Init.Specif}{\coqdocnotation{(}}\coqdocvariable{R} \coqdocvariable{a} \coqdocvar{b}\coqexternalref{:type scope:'x7B' x ':' x 'x26' x 'x7D'}{http://coq.inria.fr/distrib/8.5pl3/stdlib/Coq.Init.Specif}{\coqdocnotation{)\}}}\coqexternalref{:type scope:x '*' x}{http://coq.inria.fr/distrib/8.5pl3/stdlib/Coq.Init.Datatypes}{\coqdocnotation{)}} \coqexternalref{:type scope:x '*' x}{http://coq.inria.fr/distrib/8.5pl3/stdlib/Coq.Init.Datatypes}{\coqdocnotation{\ensuremath{\times}}} \coqexternalref{:type scope:x '*' x}{http://coq.inria.fr/distrib/8.5pl3/stdlib/Coq.Init.Datatypes}{\coqdocnotation{(}}\coqdockw{\ensuremath{\forall}} (\coqdocvar{b}:\coqdocvariable{B}), \coqexternalref{:type scope:'x7B' x ':' x 'x26' x 'x7D'}{http://coq.inria.fr/distrib/8.5pl3/stdlib/Coq.Init.Specif}{\coqdocnotation{\{}}\coqdocvar{a}\coqexternalref{:type scope:'x7B' x ':' x 'x26' x 'x7D'}{http://coq.inria.fr/distrib/8.5pl3/stdlib/Coq.Init.Specif}{\coqdocnotation{:}}\coqdocvariable{A} \coqexternalref{:type scope:'x7B' x ':' x 'x26' x 'x7D'}{http://coq.inria.fr/distrib/8.5pl3/stdlib/Coq.Init.Specif}{\coqdocnotation{\&}} \coqexternalref{:type scope:'x7B' x ':' x 'x26' x 'x7D'}{http://coq.inria.fr/distrib/8.5pl3/stdlib/Coq.Init.Specif}{\coqdocnotation{(}}\coqdocvariable{R} \coqdocvar{a} \coqdocvariable{b}\coqexternalref{:type scope:'x7B' x ':' x 'x26' x 'x7D'}{http://coq.inria.fr/distrib/8.5pl3/stdlib/Coq.Init.Specif}{\coqdocnotation{)\}}}\coqexternalref{:type scope:x '*' x}{http://coq.inria.fr/distrib/8.5pl3/stdlib/Coq.Init.Datatypes}{\coqdocnotation{)}}.\coqdoceol

The \coqRefDefn{Top.paper}{Total} property says that for all \coqdocvar{a}:\coqdocvariable{A} there exists a related \coqdocvar{b}:\coqdocvariable{B} and vice versa.
A relation satisfying both of the above properties 
can be considered an \emph{isomorphism}. Thus, in the worst case, the {\isorel}
translation produces free Coq proofs justifying the commonly held belief that
isomorphic instantiations of interfaces have the same logical properties.
\emph{However, as we will see in this section,
many propositions need \emph{neither} of the above properties to behave uniformly.} 
Here is an example where \emph{any} relation works for the first argument \coqdocvar{T}:

\coqdocnoindent
\coqdockw{Definition} \coqdef{Top.paper.PNone}{PNone}{\coqdocdefinition{PNone}} := {\coqdocnotation{\ensuremath{\lambda}}} {\coqdocnotation{(}}\coqdocvar{T}:\coqdockw{Set}) (\coqdocvar{f}:\coqdocvariable{T}\coqexternalref{:type scope:x '->' x}{http://coq.inria.fr/distrib/8.5pl3/stdlib/Coq.Init.Logic}{\coqdocnotation{\ensuremath{\rightarrow}}}\coqref{Top.paper.nat}{\coqdocinductive{nat}}) (\coqdocvar{a} \coqdocvar{b} :\coqdocvariable{T}{\coqdocnotation{)}} {\coqdocnotation{,}} {\coqdocnotation{(}}\coqdocvariable{f} \coqdocvariable{a} \coqexternalref{:type scope:x '=' x}{http://coq.inria.fr/distrib/8.5pl3/stdlib/Coq.Init.Logic}{\coqdocnotation{=}} \coqdocvariable{f} \coqdocvariable{b}{\coqdocnotation{)}}.\coqdoceol

\coqdocnoindent
The {\anyrel} translation of the argument \coqdocvar{f} already implies that on
related inputs, \coqdocvar{f} produces equal numbers.
Some need \emph{only one}: the next two polymorphic propositions respectively
only need the \coqRefDefn{Top.paper}{Total} and \coqRefDefn{Top.paper}{OneToOne}
properties for the first argument \coqdocvar{T}.

\coqdocnoindent
\coqdockw{Definition} \coqdef{Top.paper.PTot}{PTot}{\coqdocdefinition{PTot}} := {\coqdocnotation{\ensuremath{\lambda}}} {\coqdocnotation{(}}\coqdocvar{T}:\coqdockw{Set}) (\coqdocvar{f}:\coqdocvariable{T}\coqexternalref{:type scope:x '->' x}{http://coq.inria.fr/distrib/8.5pl3/stdlib/Coq.Init.Logic}{\coqdocnotation{\ensuremath{\rightarrow}}}\coqref{Top.paper.nat}{\coqdocinductive{nat}}{\coqdocnotation{),}} \coqdockw{\ensuremath{\forall}} (\coqdocvar{t}:\coqdocvariable{T}), \coqdocvariable{f} \coqdocvariable{t} \coqexternalref{:type scope:x '=' x}{http://coq.inria.fr/distrib/8.5pl3/stdlib/Coq.Init.Logic}{\coqdocnotation{=}} \coqref{Top.paper.O}{\coqdocconstructor{O}}.\coqdoceol
\coqdocnoindent
\coqdockw{Definition} \coqdef{Top.paper.POne}{POne}{\coqdocdefinition{POne}} := 
{\coqdocnotation{\ensuremath{\lambda}}} {\coqdocnotation{(}}\coqdocvar{T}:\coqdockw{Set}) (\coqdocvar{f}:\coqref{Top.paper.nat}{\coqdocinductive{nat}}\coqexternalref{:type scope:x '->' x}{http://coq.inria.fr/distrib/8.5pl3/stdlib/Coq.Init.Logic}{\coqdocnotation{\ensuremath{\rightarrow}}}\coqdocvariable{T}{\coqdocnotation{),}} 
\coqdocnotation{$\forall$} (\coqdocvar{n}: {\CoqNat}), \coqdocvariable{f} \coqdocvar{n} \coqexternalref{:type scope:x '=' x}{http://coq.inria.fr/distrib/8.5pl3/stdlib/Coq.Init.Logic}{\coqdocnotation{=}} \coqdocvariable{f} (\coqref{Top.paper.S}{\coqdocconstructor{S}} (\coqref{Top.paper.S}{\coqdocconstructor{S}} \coqdocvar{n})).\coqdoceol

\coqdocnoindent
We will see that we need the \coqRefDefn{Top.paper}{Total} property for universally quantified types and types of arguments of inductive constructors.
Also, we need the \coqRefDefn{Top.paper}{OneToOne} property for index types of inductively defined propositions, such as the equality proposition.
To allow such fine-grained analysis,
for now, unlike for propositions, we don't globally assume the
\coqRefDefn{Top.paper}{Total} and \coqRefDefn{Top.paper}{OneToOne} properties for
relations produced by translating types. We could have done that by defining:\\
\ptranslateIso{\coqdockw{Set}}:=
{\coqdocnotation{\ensuremath{\lambda}}} {\coqdocnotation{(}}\coqdocvar{A} \coqdocvar{\tprime{A}}: \coqdockw{Set}{\coqdocnotation{),}} \coqexternalref{:type scope:'x7B' x ':' x 'x26' x 'x7D'}{http://coq.inria.fr/distrib/8.5pl3/stdlib/Coq.Init.Specif}{\coqdocnotation{\{}}\coqdocvar{R} \coqexternalref{:type scope:'x7B' x ':' x 'x26' x 'x7D'}{http://coq.inria.fr/distrib/8.5pl3/stdlib/Coq.Init.Specif}{\coqdocnotation{:}} \coqdocvariable{A} \coqexternalref{:type scope:x '->' x}{http://coq.inria.fr/distrib/8.5pl3/stdlib/Coq.Init.Logic}{\coqdocnotation{\ensuremath{\rightarrow}}} \coqdocvariable{\tprime{A}} \coqexternalref{:type scope:x '->' x}{http://coq.inria.fr/distrib/8.5pl3/stdlib/Coq.Init.Logic}{\coqdocnotation{\ensuremath{\rightarrow}}} \coqdockw{Prop} \coqexternalref{:type scope:'x7B' x ':' x 'x26' x 'x7D'}{http://coq.inria.fr/distrib/8.5pl3/stdlib/Coq.Init.Specif}{\coqdocnotation{\&}} \coqref{Top.paper.Total}{\coqdocdefinition{Total}} \coqdocvar{R} \coqref{Top.paper.::x 'xC3x97' x}{\coqdocnotation{×}} \coqref{Top.paper.OneToOne}{\coqdocdefinition{OneToOne}} \coqdocvar{R}\coqexternalref{:type scope:'x7B' x ':' x 'x26' x 'x7D'}{http://coq.inria.fr/distrib/8.5pl3/stdlib/Coq.Init.Specif}{\coqdocnotation{\}}}.\coqdoceol

In the {\anyrel} translation described in the previous section, had we not ensured that $\hat{\coqdockw{Set}}$ := \coqdockw{Prop},
and instead chosen $\hat{\coqdockw{Set}}$ := \coqdockw{Set}, we would also need
(\secref{sec:uniformProp:pnode}) to consider and compositionally build proofs of a third property about parametricity relations of types, which seems hard but doable:

\coqdocnoindent
\coqdockw{Definition} \coqdef{Top.paper.irrelevant}{irrelevant}{\coqdocdefinition{irrelevant}}  \{\coqdocvar{A} \coqdocvar{\tprime{A}} : \coqdockw{Set}\} (\coqdocvar{R} : \coqdocvariable{A} \coqexternalref{:type scope:x '->' x}{http://coq.inria.fr/distrib/8.5pl3/stdlib/Coq.Init.Logic}{\coqdocnotation{\ensuremath{\rightarrow}}} \coqdocvariable{\tprime{A}} \coqexternalref{:type scope:x '->' x}{http://coq.inria.fr/distrib/8.5pl3/stdlib/Coq.Init.Logic}{\coqdocnotation{\ensuremath{\rightarrow}}} \coqdockw{Set})  :=
\coqdockw{\ensuremath{\forall}} (\coqdocvar{a}:\coqdocvariable{A}) (\coqdocvar{\tprime{a}}:\coqdocvariable{\tprime{A}}) (\coqdocvar{p1} \coqdocvar{p2}: \coqdocvariable{R} \coqdocvariable{a} \coqdocvariable{\tprime{a}}), \coqdocvariable{p1} \coqexternalref{:type scope:x '=' x}{http://coq.inria.fr/distrib/8.5pl3/stdlib/Coq.Init.Logic}{\coqdocnotation{=}} \coqdocvariable{p2}.\coqdoceol

\coqdocnoindent
Because we have $\hat{\coqdockw{Set}}$ := \coqdockw{Prop}, the type of
\coqdocvar{R} in the above definition becomes
\coqdocvariable{A} {\coqdocnotation{\ensuremath{\rightarrow}}} \coqdocvariable{\tprime{A}} {\coqdocnotation{\ensuremath{\rightarrow}}} \coqdockw{Prop}
and thus the above property becomes a trivial consequence of the proof irrelevance axiom.
Also, in \secref{sec:uniformProp:indt}, we will see that our {\isorel}
translation needs that (or a similar) axiom anyway.

The abstraction theorem (\thmref{Abstraction}) for the {\anyrel} translation says that 
for closed terms $t$ and $T$, 
if $t$ : $T$, then \ptranslate{$t$}: (\ptranslate{$T$} $t$ $t$).
The main change
now is that in some cases, \ptranslateIso{$T$} may be a dependent pair: of a
relation and some proofs about the relation.
Thus, we may need to project out the relation from \ptranslateIso{$T$}.

The above change in the translation of \coqdockw{Prop} means that for 
relations of propositional variables, we get to assume the two extra properties.
However, for composite propositions, we must build the proofs of those two properties
while assuming the property for the subcomponents, if any. Fortunately, 
starting from the universes, 
there are only two ways
to construct new types or propositions in Coq: dependent function types and inductive types.
Although one can also construct propositions or types by pattern matching and
returning different types in each branch, recursively, those types always
originate from the two primitive mechanisms mentioned above.\fxnote{How can this be made formal? Can we say that it suffices to consider normal forms of types
because the translation preserves reduction?}
When viewed through the lens of logic, 
dependent function types correspond to universal quantification,
and one can construct inductive types that correspond to familiar logical
constructs such as existential quantification.

In the next two subsections, we see how to compositionally build the
proofs for the two ways to build new propositions, and the additional assumptions (\coqRefDefn{Top.paper}{Total} or  
\coqRefDefn{Top.paper}{OneToOne} property) needed about
relations of types mentioned in the propositions.
Then, in the next section, we will see how to compositionally build the
proofs of \coqRefDefn{Top.paper}{Total} and \coqRefDefn{Top.paper}{OneToOne} properties 
for {\anyrel} translations of types.
Finally, in \secref{sec:isorel}, we use these constructions of proofs in \ptranslateIso{}.
Propositions where types of higher universes (\coqdockw{Type}$_i$ for
$i>0$) occur at certain places are excluded for fundamental reasons~(\secref{sec:isorel:limitations}).

\emph{All the proofs in this and the next section (except \secref{sec:uniformProp:necessity}) were originally done in Coq.} 
The appendix~(\secref{appendix:suppl}) has pointers to Coq proofs submitted
as anonymous supplementary material.
\appref{appendix:table} summarizes the main lemmas of this and the next section as tables.
\subsection{Universal Quantification}
\label{sec:uniformProp:piProp}

Consider \coqdocvar{A}:\coqdockw{Set}, \coqdocvar{B}:\coqdocvar{A}
$\rightarrow$ \coqdockw{Prop} $\vdash$ ($\forall$ (\coqdocvar{a}:\coqdocvar{A}), \coqdocvar{B} \coqdocvar{a}):\coqdockw{Prop}.
By \thmref{Abstraction}, we have
\ptranslate{\coqdocvar{A}:\coqdockw{Set}, \coqdocvar{B}:\coqdocvar{A}
$\rightarrow$ \coqdockw{Prop}} $\vdash$ 
\ptranslate{$\forall$ (\coqdocvar{a}:\coqdocvar{A}), \coqdocvar{B} \coqdocvar{a})}
: \ptranslate{Prop}
($\forall$ (\coqdocvar{a}:\coqdocvar{A}), \coqdocvar{B} \coqdocvar{a})
($\forall$ (\coqdocvarP{a}:\coqdocvarP{A}), \coqdocvarP{B} \coqdocvarP{a}).
By unfolding definitions and $\beta$ reduction, we get:
\coqdocvar{A}:\coqdockw{Set}, 
\coqdocvarP{A}:\coqdockw{Set}, 
\coqdocvarR{A}:\coqdocvar{A}~$\rightarrow$~\coqdocvarP{A}~$\rightarrow$~\coqdockw{Set}, 
\coqdocvar{B}:\coqdocvar{A} $\rightarrow$ \coqdockw{Prop}, 
\coqdocvarP{B}:\coqdocvarP{A} $\rightarrow$ \coqdockw{Prop}, 
\coqdocvarR{B}:$\forall$~(\coqdocvar{a}:\coqdocvar{A})~(\coqdocvar{\tprime{a}}:\coqdocvar{\tprime{A}}),%
~\coqdocvar{\trel{A}}~\coqdocvar{a}~\coqdocvar{\tprime{a}}~$\rightarrow$%
(\coqdocvar{B}~\coqdocvar{a}) $\rightarrow$%
~(\coqdocvar{\tprime{B}}~\coqdocvar{\tprime{a}})~$\rightarrow$~\coqdockw{Prop}
$\vdash$
\ptranslate{$\forall$ (\coqdocvar{a}:\coqdocvar{A}), \coqdocvar{B} \coqdocvar{a})}
: 
($\forall$ (\coqdocvar{a}:\coqdocvar{A}), \coqdocvar{B} \coqdocvar{a})
$\rightarrow$
($\forall$ (\coqdocvarP{a}:\coqdocvarP{A}), \coqdocvarP{B} \coqdocvarP{a})
$\rightarrow$
\coqdockw{Prop}

We wish to prove that the relation
\ptranslate{$\forall$ (\coqdocvar{a}:\coqdocvar{A}), \coqdocvar{B} \coqdocvar{a})} has the
\coqRefDefn{Top.paper}{IffProps} and
\coqRefDefn{Top.paper}{CompleteRel} properties.
Because \ptranslateIso{} will be structurally recursive~(\secref{sec:isorel}),
we have 2 hypotheses asserting that \coqdocvarR{B} already has the 2 properties, i.e,
\coqdocvar{iHrec}: $\forall$ (\coqdocvar{a}:\coqdocvar{A}) (\coqdocvarP{a}:\coqdocvarP{A})
(\coqdocvarR{a}:\coqdocvarR{A} \coqdocvar{a} \coqdocvarP{a}),
\coqRefDefn{Top.paper}{IffProps} (\coqdocvarR{B} \coqdocvar{a} \coqdocvarP{a} \coqdocvarR{a})
and  
\coqdocvar{cHrec}: $\forall$ (\coqdocvar{a}:\coqdocvar{A}) (\coqdocvarP{a}:\coqdocvarP{A})
(\coqdocvarR{a}:\coqdocvarR{A} \coqdocvar{a} \coqdocvarP{a}),
\coqRefDefn{Top.paper}{CompleteRel} (\coqdocvarR{B} \coqdocvar{a} \coqdocvarP{a} \coqdocvarR{a}).
Using \coqdocvar{cHrec}, it is trivial to prove the following:
\begin{lemma}
\label{lemma:piComplete}
\coqRefDefn{Top.paper}{CompleteRel}
(\ptranslate{$\forall$ (\coqdocvar{a}:\coqdocvar{A}), \coqdocvar{B} \coqdocvar{a})})
\end{lemma}

In contrast, \coqRefDefn{Top.paper}{IffProps} \ptranslate{($\forall$ (\coqdocvar{a}:\coqdocvar{A}), \coqdocvar{B} \coqdocvar{a})},
which $\beta$-reduces to
($\forall$ (\coqdocvar{a}:\coqdocvar{A}), \coqdocvar{B} \coqdocvar{a})
$\leftrightarrow$
($\forall$ (\coqdocvar{\tprime{a}}:\coqdocvar{\tprime{A}}), \coqdocvar{\tprime{B}} \coqdocvar{\tprime{a}}), is impossible to prove without additional assumption(s).
As a counterexample, take \coqdocvar{A} to be a
non-empty type, \coqdocvar{\tprime{A}} to be an empty type (e.g. {\CoqFalse}),
and \coqdocvar{B} and \coqdocvar{\tprime{B}} to be {$\lambda$ \_,\CoqFalse}.
A simple and sufficient assumption is 
\coqRefDefn{Top.paper}{Total} \coqdocvar{\trel{A}}:
\begin{lemma}
\label{lemma:piIff}
\coqRefDefn{Top.paper}{Total} \coqdocvar{\trel{A}}
$\rightarrow$
\coqRefDefn{Top.paper}{IffProps}
(\ptranslate{$\forall$ (\coqdocvar{a}:\coqdocvar{A}), \coqdocvar{B} \coqdocvar{a})})
\end{lemma}%
\noindent
Using \coqdocvar{iHrec}, the proof is straighforward.  We defer the discussion 
of the necessity of the \coqRefDefn{Top.paper}{Total} assumption to \secref{sec:uniformProp:necessity}.
In summary, universal quantifications behave uniformly if the relation corresponding to the quantified type
is \coqRefDefn{Top.paper}{Total}.

\subsection{Inductively defined propositions}
\label{sec:uniformProp:pnode}
We already saw an indexed-inductive proposition (the polymorphic equality proposition) in \secref{sec:intro}. 
In Coq, relations and predicates are often defined using indexed-induction.
Here is the definition of $\le$ on natural numbers:

\coqdocnoindent
\coqdockw{Inductive} \coqdef{Top.paper.le}{le}{\coqdocinductive{le}} (\coqdocvar{n} : \coqref{Top.paper.nat}{\coqdocinductive{nat}}) : \coqref{Top.paper.nat}{\coqdocinductive{nat}} \coqexternalref{:type scope:x '->' x}{http://coq.inria.fr/distrib/8.5pl3/stdlib/Coq.Init.Logic}{\coqdocnotation{\ensuremath{\rightarrow}}} \coqdockw{Prop} :=\coqdoceol
\coqdocnoindent
\ensuremath{|}\coqdef{Top.paper.le n}{le\_n}{\coqdocconstructor{le\_n}} : \coqref{Top.paper.le}{\coqdocinductive{le}} \coqdocvar{n} \coqdocvar{n} \coqdoceol
\coqdocnoindent
\ensuremath{|}\coqdef{Top.paper.le S}{le\_S}{\coqdocconstructor{le\_S}} : \coqdockw{\ensuremath{\forall}} \coqdocvar{m} : \coqref{Top.paper.nat}{\coqdocinductive{nat}}, \coqref{Top.paper.le}{\coqdocinductive{le}} \coqdocvar{n} \coqdocvariable{m} \coqexternalref{:type scope:x '->' x}{http://coq.inria.fr/distrib/8.5pl3/stdlib/Coq.Init.Logic}{\coqdocnotation{\ensuremath{\rightarrow}}} \coqref{Top.paper.le}{\coqdocinductive{le}} \coqdocvar{n} (\coqref{Top.paper.S}{\coqdocconstructor{S}} \coqdocvariable{m}).\coqdoceol 

Unlike universal quantification, inductively defined propositions
come in infinitely many shapes. For example, there can be an arbitrary number of parameters, indices, constructors and arguments of constructors.
To explain the key ideas, we consider just one type, which is an indexed version
of the W type~\cite{Martin-Lof1984} and can be understood as trees with possibly infinite branching.
W types can be used to encode a large class of inductively defined
types~\cite{Dybjer1997, Abbott.Altenkirch.ea2004}.

\coqdocnoindent
\coqdockw{Inductive} \coqdef{Top.paper.IWP}{IWP}{\coqdocinductive{IWP}} (\coqdocvar{I} \coqdocvar{A} : \coqdockw{Set}) (\coqdocvar{B} : \coqdocvariable{A} \coqexternalref{:type scope:x '->' x}{http://coq.inria.fr/distrib/8.5pl3/stdlib/Coq.Init.Logic}{\coqdocnotation{\ensuremath{\rightarrow}}} \coqdockw{Set}) (\coqdocvar{AI} : \coqdocvariable{A} \coqexternalref{:type scope:x '->' x}{http://coq.inria.fr/distrib/8.5pl3/stdlib/Coq.Init.Logic}{\coqdocnotation{\ensuremath{\rightarrow}}} \coqdocvariable{I})  (\coqdocvar{BI} : \coqdockw{\ensuremath{\forall}} (\coqdocvar{a} : \coqdocvariable{A}), \coqdocvariable{B} \coqdocvariable{a} \coqexternalref{:type scope:x '->' x}{http://coq.inria.fr/distrib/8.5pl3/stdlib/Coq.Init.Logic}{\coqdocnotation{\ensuremath{\rightarrow}}} \coqdocvariable{I}) : \coqdockw{\ensuremath{\forall}} (\coqdocvar{i}:\coqdocvar{I}), \coqdockw{Prop} :=\coqdoceol
\coqdocnoindent
\coqdef{Top.paper.pnode}{pnode}{\coqdocconstructor{pnode}} : \coqdockw{\ensuremath{\forall}} (\coqdocvar{a} : \coqdocvar{A}) (\coqdocvar{branches} : \coqdockw{\ensuremath{\forall}} \coqdocvar{b} : \coqdocvar{B} \coqdocvariable{a}, \coqref{Top.paper.IWP}{\coqdocinductive{IWP}} \coqdocvar{I} \coqdocvar{A} \coqdocvar{B} \coqdocvar{AI} \coqdocvar{BI} (\coqdocvar{BI} \coqdocvariable{a} \coqdocvariable{b})), \coqref{Top.paper.IWP}{\coqdocinductive{IWP}} \coqdocvar{I} \coqdocvar{A} \coqdocvar{B} \coqdocvar{AI} \coqdocvar{BI} (\coqdocvar{AI} \coqdocvariable{a}).\coqdoceol

\coqdocvar{I} is the type of indices. There is only one index type. This may be a loss of convenience, but is not a loss of generality, because one can use (dependent) pairs to
encode multiple, dependent indices. The type \coqdocvar{A} encodes the non-recursive arguments. 
Given any \coqdocvar{a}:\coqdocvar{A}, \coqdocvar{B} \coqdocvar{a} denotes the branching factor 
of a node of the tree: see the \coqdocvar{branches} argument of the constructor \coqRefConstr{Top.paper}{pnode}.
For example, we can choose \coqdocvar{B} := $\lambda$ (\coqdocvar{a}: \coqdocvar{A}), {\CoqBool} for binary (proof) trees.
The function \coqdocvar{AI} determines the index of the return type of the constructor. 
Similarly, the function \coqdocvar{BI} determines the indices of the subtrees in
the \coqdocvar{branches} argument of the constructor \coqRefConstr{Top.paper}{pnode}.
\appref{appendix:suppl} shows how to encode the above-defined
relation \coqRefInductive{Top.paper}{le} as an instance of \coqRefInductive{Top.paper}{IWP}.

Using \coqRefInductive{Top.paper}{IWP}, we proved in Coq the uniformity properties for the large class of inductive propositions encodable using \coqRefInductive{Top.paper}{IWP}.
Otherwise, this proof may have needed reasoning about a deep embedding of Coq's inductives.
Our implementation, which although is inspired by the uniformity proofs
for \coqRefInductive{Top.paper}{IWP}, directly translates each inductive,
without using the encoding. 
%
This has several advantages.
Users don't have to use unnatural encodings of their inductive propositions.
Even if the encoding could be automated, 
users may prefer to directly understand how the translation works for their definitions, instead of understanding how it is obtained via an encoding.
Below, although we mainly focus on the uniformity proof for \coqRefInductive{Top.paper}{IWP}, we include hints for generalizing the construction to other 
inductive propositions.

As in the previous subsection, in the translated context, we need to prove the
\coqRefDefn{Top.paper}{IffProps} and the
\coqRefDefn{Top.paper}{CompleteRel} properties for
\ptranslate{\coqRefInductive{Top.paper}{IWP} \coqdocvar{I}
\coqdocvar{A} \coqdocvar{B} \coqdocvar{AI} \coqdocvar{BI} \coqdocvar{i}}.
Because \coqRefInductive{Top.paper}{IWP}
 returns a \coqdockw{Prop},
it is translated in the inductive style~(\secref{sec:anyrel:compare}).
Let {\coqdocinductive{\trel{IWP}}} denote the inductive-style translation of \coqRefInductive{Top.paper}{IWP}. 
%
We explain the proof of the \coqRefDefn{Top.paper}{CompleteRel} property and one direction of the \coqRefDefn{Top.paper}{IffProps} property.
We conveniently prove both the properties simultaneously.
We will use the following abbreviations in a translated context:
\coqdef{Top.paper.W}{W}{\coqdocdefinition{W}} := 
\coqRefInductive{Top.paper}{IWP} \coqdocvar{I}
\coqdocvar{A} \coqdocvar{B} \coqdocvar{AI} \coqdocvar{BI},
\tprime{\coqdocdefinition{W}} := 
\coqRefInductive{Top.paper}{IWP} \coqdocvarP{I}
\coqdocvarP{A} \coqdocvarP{B} \coqdocvarP{AI} \coqdocvarP{BI},
\trel{\coqdocdefinition{W}} :=
\coqref{Top.paper.IWP R}{\coqdocinductive{\trel{IWP}}} \coqdocvariable{I} \coqdocvariable{\tprime{I}} \coqdocvariable{\trel{I}} \coqdocvariable{A} \coqdocvariable{\tprime{A}} \coqdocvariable{\trel{A}} \coqdocvariable{B} \coqdocvariable{\tprime{B}} \coqdocvariable{\trel{B}} \coqdocvariable{AI} \coqdocvariable{\tprime{AI}} \coqdocvariable{\trel{AI}} \coqdocvariable{BI} \coqdocvariable{\tprime{BI}} \coqdocvariable{\trel{BI}},
\coqdef{Top.paper.pnodew}{pnodew}{\coqdocdefinition{pnodew}} := 
\coqRefConstr{Top.paper}{pnode} \coqdocvar{I}
\coqdocvar{A} \coqdocvar{B} \coqdocvar{AI} \coqdocvar{BI},
\tprime{\coqdef{Top.paper.pnodew}{pnodew}{\coqdocdefinition{pnodew}}} := 
\coqRefConstr{Top.paper}{pnode} \coqdocvarP{I}
\coqdocvarP{A} \coqdocvarP{B} \coqdocvarP{AI} \coqdocvarP{BI}

\begin{lemma}
\label{lemma:iwp}
\coqRefDefn{Top.paper}{Total}  \coqdocvarR{A}
$\rightarrow$
 ($\forall$ (\coqdocvar{a}:\coqdocvar{A}) (\coqdocvarP{a}:\coqdocvarP{A})
(\coqdocvarR{a}:\coqdocvarR{A} \coqdocvar{a} \coqdocvarP{a}),
\coqRefDefn{Top.paper}{Total} (\coqdocvarR{B} \coqdocvar{a} \coqdocvarP{a} \coqdocvarR{a}))
$\rightarrow$
\coqRefDefn{Top.paper}{OneToOne} \coqdocvarR{I}
\coqdoceol
$\rightarrow$
$\forall$ (\coqdocvar{p}:\coqRefDefn{Top.paper}{W} \coqdocvar{i}),
(%
\tprime{\coqRefDefn{Top.paper}{W}} \coqdocvariable{\tprime{i}}
{\CoqConj} 
\coqdockw{\ensuremath{\forall}} \coqdocvar{y} : \tprime{\coqRefDefn{Top.paper}{W}} \coqdocvariable{\tprime{i}}, 
 \trel{\coqRefDefn{Top.paper}{W}} \coqdocvariable{i} \coqdocvariable{\tprime{i}} \coqdocvariable{\trel{i}} \coqdocvariable{p} \coqdocvariable{y}).\coqdoceol
\end{lemma}

\coqdocnoindent
Note that \coqdocvarR{i} has type \coqdocvarR{I} \coqdocvar{i} \coqdocvarP{i}.  
We proceed by induction on \coqdocvar{p}. The corresponding proof term is a structurally recursive function which pattern matches on \coqdocvar{p}.
(Our translation directly produces fully elaborated Gallina proof terms, and
not LTac proof scripts which have less well-defined semantics.) In the inductive step, we have, for some \coqdocvar{a} and
\coqdocvar{branches},  \coqdocvar{p} := 
(\coqRefDefn{Top.paper}{pnodew} \coqdocvar{a} \coqdocvar{branches}):
\coqRefDefn{Top.paper}{W} (\coqdocvar{AI} \coqdocvar{a}).
Note that the pattern matching (induction) refines the index of the discriminee \coqdocvar{p} from \coqdocvar{i} to (\coqdocvar{AI} \coqdocvar{a}).
Also, \coqdocvarR{i} now has type \coqdocvarR{I} (\coqdocvar{AI} \coqdocvar{a}) \coqdocvarP{i}.
It is straightforward to use the induction hypothesis and the 
\coqRefDefn{Top.paper}{Total} property for \coqdocvarR{A} and \coqdocvarR{B} to
obtain \coqdocvar{\tprime{a}} and \coqdocvarP{branches} such that
(\tprime{\coqRefDefn{Top.paper}{pnodew}} \coqdocvarP{a} \coqdocvarP{branches}):
\tprime{\coqRefDefn{Top.paper}{W}} (\coqdocvarP{AI} \coqdocvarP{a}).
\coqRefDefn{Top.paper}{Total} \coqdocvarR{A} also provides an \coqdocvarR{a}:(\coqdocvarR{A} \coqdocvar{a} \coqdocvarP{a}).
We are not done yet even for the left conjunct because it needs something of
type \tprime{\coqRefDefn{Top.paper}{W}} \coqdocvar{\tprime{i}}.
Thus we need a proof of \coqdocvarP{AI} \coqdocvarP{a} = \coqdocvar{\tprime{i}}. 
This is where the \coqRefDefn{Top.paper}{OneToOne} property of
\coqdocvarR{I} comes to the rescue.
Recall that we have \coqdocvarR{i}:(\coqdocvarR{I} (\coqdocvar{AI} \coqdocvar{a}) \coqdocvarP{i}). Also
\ptranslate{\coqdocvar{AI} \coqdocvar{a}}:= 
\coqdocvarR{AI} \coqdocvar{a} \coqdocvarP{a} \coqdocvarR{a}, which has type 
\coqdocvarR{I} (\coqdocvar{AI} \coqdocvar{a}) (\coqdocvarP{AI} \coqdocvarP{a}).
Thus, we get the needed equality by invoking the
hypothesis \coqRefDefn{Top.paper}{OneToOne} \coqdocvarR{I}.
Now we can substitute \coqdocvar{\tprime{i}} with (\coqdocvarP{AI} \coqdocvarP{a}) \emph{everywhere} (all hypotheses and the conclusion).
In general, this rewriting step has to be done for each index of an inductive
proposition and rewriting everywhere becomes important, especially while implementing the translation, when the later indices are dependent. 
Now
\tprime{\coqRefDefn{Top.paper}{pnodew}} \coqdocvarP{a} \coqdocvarP{branches} is a proof of the left conjunct.

The right conjunct now has type:
%
\coqexternalref{:type scope:'xE2x88x80' x '..' x ',' x}{http://coq.inria.fr/distrib/8.5pl3/stdlib/Coq.Unicode.Utf8\_core}{\coqdocnotation{∀}} 
\coqdocvar{y} : \tprime{\coqRefDefn{Top.paper}{W}} (\coqdocvar{\tprime{AI}} \coqdocvar{\tprime{a}})\coqexternalref{:type scope:'xE2x88x80' x '..' x ',' x}{http://coq.inria.fr/distrib/8.5pl3/stdlib/Coq.Unicode.Utf8\_core}{\coqdocnotation{,}} 
\trel{\coqRefDefn{Top.paper}{W}}  (\coqdocvar{AI} \coqdocvar{a}) (\coqdocvar{\tprime{AI}} \coqdocvar{\tprime{a}}) \coqdocvar{\trel{i}} ({\coqRefDefn{Top.paper}{pnodew}} \coqdocvar{a} \coqdocvar{branches}) \coqdocvariable{y}.\coqdoceol
\coqdocnoindent
Now we pick an arbitrary \coqdocvar{y} and
use proof irrelevance for the proposition \tprime{\coqRefDefn{Top.paper}{W}} (\coqdocvarP{AI} \coqdocvarP{a})
to produce a proof
that \coqdocvar{y} = \tprime{\coqRefDefn{Top.paper}{pnodew}} \coqdocvarP{a} \coqdocvarP{branches} and then substitute the former with the latter.
This step is crucial: we don't have the \coqRefDefn{Top.paper}{CompleteRel} property for \coqdocvarR{A}: \coqdocvar{A} is not a proposition.
We are only assuming the \coqRefDefn{Top.paper}{Total} property for \coqdocvarR{A}. Thus, if we had analyzed the original \coqdocvar{y} by pattern
matching on it, we would have obtained an {\coqdocvarP{a}\ensuremath{'}} that
may be different from \coqdocvarP{a} and \emph{unrelated} to \coqdocvar{a}.
 
Next, we use proof irrelevance 
for the proposition
\coqdocvarR{I} (\coqdocvar{AI} \coqdocvar{a}) (\coqdocvarP{AI} \coqdocvarP{a})
to replace \coqdocvar{\trel{i}} with (\coqdocvarR{AI} \coqdocvar{a} \coqdocvarP{a} \coqdocvarR{a}).
Had we not ensured that $\hat{\coqdockw{Set}}$ := \coqdockw{Prop},
and instead chosen $\hat{\coqdockw{Set}}$ := \coqdockw{Set}, we would be
unable to invoke proof irrelevance here and may need to explicitly assume 
\coqRefDefn{Top.paper}{irrelevant} \coqdocvarR{I}.
In general, this rewriting has to be done for each index, from the leftmost index to the rightmost index, because that is the order of dependencies.
%
Then we can use the constructor \coqRefConstrR{Top.paper}{{pnode}} and
the induction hypothesis to finish the proof.
\begin{corollary}
\label{corr:iwp}
\coqRefDefn{Top.paper}{Total}  \coqdocvarR{A}
$\rightarrow$
 ($\forall$ (\coqdocvar{a}:\coqdocvar{A}) (\coqdocvarP{a}:\coqdocvarP{A})
(\coqdocvarR{a}:\coqdocvarR{A} \coqdocvar{a} \coqdocvarP{a}),
\coqRefDefn{Top.paper}{Total} (\coqdocvarR{B} \coqdocvar{a} \coqdocvarP{a} \coqdocvarR{a}))
$\rightarrow$
\coqRefDefn{Top.paper}{OneToOne} \coqdocvarR{I}
\coqdoceol
$\rightarrow$
(
\coqRefDefn{Top.paper}{IffProps}
(\trel{\coqRefDefn{Top.paper}{W}} \coqdocvariable{i} \coqdocvariable{\tprime{i}} \coqdocvariable{\trel{i}})
{\CoqConj}
\coqRefDefn{Top.paper}{CompleteRel}
(\trel{\coqRefDefn{Top.paper}{W}} \coqdocvariable{i} \coqdocvariable{\tprime{i}} \coqdocvariable{\trel{i}})
)
\end{corollary}

\subsection{Necessity of our assumptions}
\label{sec:uniformProp:necessity}
In the previous two subsections, to prove the uniformity of the two canonical constructions of propositions, 
we sometimes needed to assume the
\coqRefDefn{Top.paper}{Total} and/or \coqRefDefn{Top.paper}{OneToOne} property for the translations of types mentioned
in those propositions. Now we consider the necessity of the two assumptions.

\begin{lemma}
Suppose $U$:\coqdockw{Set}, $V$:\coqdockw{Set} are closed, 
and 
that there is a tool T than can, for \emph{any} closed $P$: 
\coqdockw{Set} $\rightarrow$ \coqdockw{Prop} whose body does not mention types of higher universes, 
produce a proof of $P$ $U$ $\leftrightarrow$ 
$P$ $V$.
Then there exists a \coqRefDefn{Top.paper}{Total} and 
\coqRefDefn{Top.paper}{OneToOne} relation between
 $U$ and $V$.
\end{lemma}
\begin{proof}
Define
\coqdocdefinition{isoTypes} :=
$\lambda$ \coqdocvar{A} \coqdocvarP{A} : \coqdockw{Set},
  $\exists$ (\coqdocvar{f} : \coqdocvar{A} $\rightarrow$ \coqdocvarP{A}) (\coqdocvar{g} : \coqdocvarP{A} $\rightarrow$ \coqdocvar{A}),
   $\forall$ (\coqdocvar{s} : \coqdocvar{A}, \coqdocvar{g} (\coqdocvar{f} \coqdocvar{s}) = \coqdocvar{s}) $\wedge$ ($\forall$ (\coqdocvar{s} : \coqdocvarP{A}, 
   \coqdocvar{f} (\coqdocvar{g} \coqdocvar{s}) = \coqdocvar{s}).
Now, invoke the tool T on \coqdocvar{P} := (\coqdocdefinition{isoTypes} $U$): (\coqdockw{Set} $\rightarrow$ \coqdockw{Prop}) to get a proof of
(\coqdocdefinition{isoTypes} $U$ $U$) $\leftrightarrow$ (\coqdocdefinition{isoTypes} $U$ $V$), which implies
\coqdocdefinition{isoTypes} $U$ $V$, which implies that there exists
a \coqRefDefn{Top.paper}{Total} and 
\coqRefDefn{Top.paper}{OneToOne} relation between $U$ and $V$.
\end{proof}

In contrast, there are examples where our translations will make unnecessary assumptions.
Suppose $f$:{\CoqNat}$\rightarrow${\CoqBool} is a closed function that always returns {\CoqBFalse}.
Now consider
($\lambda$ (\coqdocvar{T}:\coqdockw{Set}), $\forall$ (\coqdocvar{n}:{\CoqNat}) 
\coqdockw{if} $f$ \coqdocvar{n} \coqdockw{then} ($\forall$ (\coqdocvar{t}: \coqdocvar{T}), \coqdocvar{t} = \coqdocvar{t}) \coqdockw{else} {\CoqTrue}): 
(\coqdockw{Set} $\rightarrow$ \coqdockw{Prop}).
In this case, because the returned proposition has a quantification on \coqdocvar{T},
our translation would require \coqRefDefn{Top.paper}{Total} \coqdocvarR{T}.
However, a smarter translation could figure out in some cases that 
$f$ always returns {\CoqBFalse} and thus make the quantification disappear.
However, it is impossible to determine whether an arbitrary closed function of type  
{\CoqNat}$\rightarrow${\CoqBool} always returns {\CoqBFalse}.
Thus, there will be examples where every such tool makes unnecessary assumptions.
 



\section{Total and One-to-one  Properties of Relations of Types}
\label{sec:uniformProp:type}
In the above section, we saw that to ensure the uniformity of
propositions, the {\anyrel} translation of types appearing in propositions may need to have 
the \coqRefDefn{Top.paper}{Total} or \coqRefDefn{Top.paper}{OneToOne} properties.
In this section, we consider all the ways to construct new canonical types in the universe \coqdockw{Set}
and show how to build the compositional proofs of
the \coqRefDefn{Top.paper}{Total} and \coqRefDefn{Top.paper}{OneToOne} properties.
As mentioned before, we only consider the lowermost universe (\coqdockw{Set}) in the {\isorel} translation.

\subsection{Dependent Function Types}
\label{sec:uniformProp:funt}
We have \coqdocvar{A}:\coqdockw{Set}, \coqdocvar{B}:\coqdocvar{A} $\rightarrow$ \coqdockw{Set},
$\vdash$ ($\forall$ (\coqdocvar{a}:\coqdocvar{A}), \coqdocvar{B} \coqdocvar{a}):\coqdockw{Set}.
In the translated context, we need to prove
\coqRefDefn{Top.paper}{Total} 
\ptranslate{$\forall$ (\coqdocvar{a}:\coqdocvar{A}), \coqdocvar{B}
\coqdocvar{a}}
and \coqRefDefn{Top.paper}{OneToOne}
\ptranslate{$\forall$ (\coqdocvar{a}:\coqdocvar{A}), \coqdocvar{B}
\coqdocvar{a}}.
%
%
The assumptions 
\coqRefDefn{Top.paper}{Total} \coqdocvarR{A} and
($\forall$ (\coqdocvar{a}:\coqdocvar{A}) (\coqdocvarP{a}:\coqdocvarP{A})
(\coqdocvarR{a}:\coqdocvarR{A} \coqdocvar{a} \coqdocvarP{a}),
\coqRefDefn{Top.paper}{Total} (\coqdocvarR{B} \coqdocvar{a} \coqdocvarP{a} \coqdocvarR{a}))
are not sufficient to prove
\coqRefDefn{Top.paper}{Total} \ptranslate{$\forall$ (\coqdocvar{a}:\coqdocvar{A}), \coqdocvar{B}
\coqdocvar{a}}.
As a counterexample, consider
\coqdocvar{A}, \coqdocvarP{A} := {\CoqBool}; \coqdocvar{B}, \coqdocvarP{B}:= $\lambda$ \_, {\CoqBool};
\coqdocvarR{A} := $\lambda$ (\coqdocvar{a} \coqdocvarP{a} : {\CoqBool}), {\CoqTrue}; and
\coqdocvarR{B} := $\lambda$ \_ \_ \_ (\coqdocvar{b} \coqdocvarP{b} : {\CoqBool}), \coqdocvar{b} = \coqdocvarP{b}.
\ptranslate{$\forall$ (\coqdocvar{a}:\coqdocvar{A}), \coqdocvar{B}
\coqdocvar{a}} := 
λ(\coqdocvarFour{x}:∀\coqdocvar{x}:\coqdocvar{A}.\coqdocvar{B})(\coqdocvarFive{x}:∀\coqdocvarP{x}:\coqdocvarP{A}.\coqdocvarP{B}),
  ∀(\coqdocvar{x}:\coqdocvar{A})(\coqdocvarP{x}:\coqdocvarP{A})(\coqdocvarR{x}:\coqdocvarR{A}
  \coqdocvar{x}\,\coqdocvarP{x}), \coqdocvarR{B} \coqdocvar{x} \coqdocvarP{x} \coqdocvarR{x} (\coqdocvarFour{x} \coqdocvar{x}) (\coqdocvarFive{x} \coqdocvarP{x})
relates nothing to $\lambda$(\coqdocvar{x}:{\CoqBool}), \coqdocvar{x}.
Intuitively, because \coqdocvarR{A} is a complete relation,
\ptranslate{$\forall$ (\coqdocvar{a}:\coqdocvar{A}), \coqdocvar{B}
\coqdocvar{a}} only relates constant functions.
The above counterexample was mainly enabled by the coarseness of \coqdocvarR{A}:
\coqdocvarR{A} is \coqRefDefn{Top.paper}{Total} but \emph{not} \coqRefDefn{Top.paper}{OneToOne}.
Indeed, the proof is easy
after adding the assumption
\coqRefDefn{Top.paper}{OneToOne} \coqdocvarR{A}:
\begin{lemma}
\label{lemma:piTot}
\coqRefDefn{Top.paper}{Total} \coqdocvar{\trel{A}}
$\rightarrow$
($\forall$ (\coqdocvar{a}:\coqdocvar{A}) (\coqdocvarP{a}:\coqdocvarP{A})
(\coqdocvarR{a}:\coqdocvarR{A} \coqdocvar{a} \coqdocvarP{a}),
\coqRefDefn{Top.paper}{Total} (\coqdocvarR{B} \coqdocvar{a} \coqdocvarP{a} \coqdocvarR{a}))
$\rightarrow$ 
\coqRefDefn{Top.paper}{OneToOne} \coqdocvar{\trel{A}}
\coqdoceol
$\rightarrow$
\coqRefDefn{Top.paper}{Total}
(\ptranslate{$\forall$ (\coqdocvar{a}:\coqdocvar{A}), \coqdocvar{B} \coqdocvar{a}})
\end{lemma}%
\noindent
Consider the proof of one side. Given an arbitrary \coqdocvar{f}:({$\forall$ (\coqdocvar{a}:\coqdocvar{A}), \coqdocvar{B} \coqdocvar{a})},
using the totality of \coqdocvarR{A} and \coqdocvarR{B},
it is easy to cook up an \coqdocvarP{f}:({$\forall$ (\coqdocvarP{a}:\coqdocvarP{A}), \coqdocvarP{B} \coqdocvarP{a})}.
Then we need to prove
\ptranslate{$\forall$ (\coqdocvar{a}:\coqdocvar{A}), \coqdocvar{B} \coqdocvar{a}}
\coqdocvar{f} \coqdocvarP{f}. For this part, we needed 
the hypothesis \coqRefDefn{Top.paper}{OneToOne} \coqdocvar{\trel{A}} and proof irrelevance of the relation 
\coqdocvar{\trel{A}}.

\begin{lemma}
\label{lemma:piOne}
\coqRefDefn{Top.paper}{Total} \coqdocvar{\trel{A}}
$\rightarrow$
($\forall$ (\coqdocvar{a}:\coqdocvar{A}) (\coqdocvarP{a}:\coqdocvarP{A})
(\coqdocvarR{a}:\coqdocvarR{A} \coqdocvar{a} \coqdocvarP{a}),
\coqRefDefn{Top.paper}{OneToOne} (\coqdocvarR{B} \coqdocvar{a} \coqdocvarP{a} \coqdocvarR{a}))
\coqdoceol
$\rightarrow$ 
\coqRefDefn{Top.paper}{OneToOne}
(\ptranslate{$\forall$ (\coqdocvar{a}:\coqdocvar{A}), \coqdocvar{B} \coqdocvar{a})})
\end{lemma}%
The proof is straightforward.
To prove equality of functions, it uses the dependent function extensionality axiom, which is believed to be consistent with the proof irrelevance
axiom in Coq:

\coqdocnoindent
\coqdockw{\ensuremath{\forall}} \{\coqdocvar{A}:\coqdockw{Type}\} \{\coqdocvar{B} : \coqdocvariable{A} \coqexternalref{:type scope:x '->' x}{http://coq.inria.fr/distrib/8.5pl3/stdlib/Coq.Init.Logic}{\coqdocnotation{\ensuremath{\rightarrow}}} \coqdockw{Type}\}, \coqdockw{\ensuremath{\forall}} (\coqdocvar{f} \coqdocvar{g} : \coqdockw{\ensuremath{\forall}} \coqdocvar{x} : \coqdocvariable{A}, \coqdocvariable{B} \coqdocvariable{x}), \coqexternalref{:type scope:x '->' x}{http://coq.inria.fr/distrib/8.5pl3/stdlib/Coq.Init.Logic}{\coqdocnotation{(}}\coqdockw{\ensuremath{\forall}} \coqdocvar{x}, \coqdocvariable{f} \coqdocvariable{x} \coqexternalref{:type scope:x '=' x}{http://coq.inria.fr/distrib/8.5pl3/stdlib/Coq.Init.Logic}{\coqdocnotation{=}} \coqdocvariable{g} \coqdocvariable{x}\coqexternalref{:type scope:x '->' x}{http://coq.inria.fr/distrib/8.5pl3/stdlib/Coq.Init.Logic}{\coqdocnotation{)}} \coqexternalref{:type scope:x '->' x}{http://coq.inria.fr/distrib/8.5pl3/stdlib/Coq.Init.Logic}{\coqdocnotation{\ensuremath{\rightarrow}}} \coqdocvariable{f} \coqexternalref{:type scope:x '=' x}{http://coq.inria.fr/distrib/8.5pl3/stdlib/Coq.Init.Logic}{\coqdocnotation{=}} \coqdocvariable{g}.\coqdoceol

\subsection{Inductive Types}
\label{sec:uniformProp:indt}
The \coqRefDefn{Top.paper}{Total} and the \coqRefDefn{Top.paper}{OneToOne} properties of the {\anyrel} translations of inductive types boil down to
the same properties for the types of arguments of their constructors. 
Let \coqdocconstructor{c} be a constructor of an inductive type (family)
\coqdocinductive{$I$}.
It is useful to classify the arguments of \coqdocconstructor{c} into two categories: those that are recursive (whose types mention
\coqdocinductive{$I$}) and those that are not.
For example, in the constructor \coqRefConstr{Top.paper}{pnode} in
\secref{sec:uniformProp:pnode}, \coqdocvar{a} is a non-recursive argument and
\coqdocvar{branches} is a recursive argument.
The non-recursive arguments are easy to tackle. Because \ptranslateIso{} will
be~(\secref{sec:isorel}) structurally recursive, we can assume that we
already have the \coqRefDefn{Top.paper}{Total} and the \coqRefDefn{Top.paper}{OneToOne} properties for the types of those arguments.
The recursive arguments are harder to tackle. Their types mention members of the type family \coqdocinductive{$I$}, and we don't yet have 
their proofs of the \coqRefDefn{Top.paper}{Total} and
\coqRefDefn{Top.paper}{OneToOne} properties yet: we are in the process of building that.
Thus we need to carefully analyse the types of the recursive arguments and 
build the recursive proofs of the \coqRefDefn{Top.paper}{Total} and the
\coqRefDefn{Top.paper}{OneToOne} properties in a way that satisfies Coq's
termination (well-definedness) checker for recursive functions.

Fortunately, Coq has a strict-positivity restriction on the  shape of the types of recursive arguments of constructors.
These types must be of the form\footnote{
Coq's strict-positivity restriction is a bit more permissive. For example, the type {\CoqNat $\rightarrow$ {\CoqList} (\coqdocinductive{$I$} $\hdots$)}
is acceptable as a type of a constructor argument.
Inductives with such constructors are called nested inductives. Our theory and implementation don't support them yet.
However, nested inductives can be encoded as mutual-inductive definitions. We do support mutual inductive definitions.}
:
$\forall$ (\coqdocvar{t}\ensuremath{_1} : \coqdocvar{T}\ensuremath{_1}) (\coqdocvar{t}\ensuremath{_2} : \coqdocvar{T}\ensuremath{_2})
 $\hdots$ 
(\coqdocvar{t}\ensuremath{_m} : \coqdocvar{T}\ensuremath{_m}),
(\coqdocinductive{$I$} $\hdots$), where
\coqdocinductive{$I$}~$\hdots$ represents \coqdocinductive{$I$} applied to enough arguments so that it becomes a type. 
Also, the types \coqdocvar{T}\ensuremath{_i} must not mention \coqdocinductive{$I$}. 
(Thus, we can assume \coqRefDefn{Top.paper}{Total} \ptranslate{\coqdocvar{T}\ensuremath{_i}} and
\coqRefDefn{Top.paper}{OneToOne} \ptranslate{\coqdocvar{T}\ensuremath{_i}}.)
\fxnote{put the subscripts inside coqdocvar. Ts are not vars}
So, the types of recursive arguments are \emph{(dependent) function types} returning
the inductive to which the constructor belongs. 
$m$ can be 0, as in the definition of natural numbers or lists.

Fortunately, in the previous subsection, we already saw how to compositionally construct the 
\coqRefDefn{Top.paper}{Total} and \coqRefDefn{Top.paper}{OneToOne} properties
for (dependent) function types. Those proofs were non-trivial. 
Thus, we encapsulate those constructions as reusable lemmas and use them in the {\isorel} translation of inductives. 
For example, the lemma \coqRefDefn{Top.paper}{totalPiHalf} below is the
combinator for one direction of the  \coqRefDefn{Top.paper}{Total} property.

\coqdocnoindent
\coqdockw{Definition} \coqdef{Top.paper.IsoRel}{IsoRel}{\coqdocdefinition{IsoRel}} :=
{\coqdocnotation{\ensuremath{\lambda}}} {\coqdocnotation{(}}\coqdocvar{A} \coqdocvar{\tprime{A}}: \coqdockw{Set}{\coqdocnotation{),}} \coqexternalref{:type scope:'x7B' x ':' x 'x26' x 'x7D'}{http://coq.inria.fr/distrib/8.5pl3/stdlib/Coq.Init.Specif}{\coqdocnotation{\{}}\coqdocvar{\trel{A}} \coqexternalref{:type scope:'x7B' x ':' x 'x26' x 'x7D'}{http://coq.inria.fr/distrib/8.5pl3/stdlib/Coq.Init.Specif}{\coqdocnotation{:}} \coqdocvariable{A} \coqexternalref{:type scope:x '->' x}{http://coq.inria.fr/distrib/8.5pl3/stdlib/Coq.Init.Logic}{\coqdocnotation{\ensuremath{\rightarrow}}} \coqdocvariable{\tprime{A}} \coqexternalref{:type scope:x '->' x}{http://coq.inria.fr/distrib/8.5pl3/stdlib/Coq.Init.Logic}{\coqdocnotation{\ensuremath{\rightarrow}}} \coqdockw{Prop} \coqexternalref{:type scope:'x7B' x ':' x 'x26' x 'x7D'}{http://coq.inria.fr/distrib/8.5pl3/stdlib/Coq.Init.Specif}{\coqdocnotation{\&}} \coqref{Top.paper.::x 'xC3x97' x}{\coqdocnotation{(}}\coqref{Top.paper.Total}{\coqdocdefinition{Total}} \coqdocvar{\trel{A}}\coqref{Top.paper.::x 'xC3x97' x}{\coqdocnotation{)}} \coqref{Top.paper.::x 'xC3x97' x}{\coqdocnotation{×}} \coqref{Top.paper.::x 'xC3x97' x}{\coqdocnotation{(}}\coqref{Top.paper.OneToOne}{\coqdocdefinition{OneToOne}} \coqdocvar{\trel{A}}\coqref{Top.paper.::x 'xC3x97' x}{\coqdocnotation{)}}\coqexternalref{:type scope:'x7B' x ':' x 'x26' x 'x7D'}{http://coq.inria.fr/distrib/8.5pl3/stdlib/Coq.Init.Specif}{\coqdocnotation{\}}}.\coqdoceol
\coqdocemptyline
\coqdocnoindent
\coqdockw{Definition} \coqdef{Top.paper.TotalHalf}{TotalHalf}{\coqdocdefinition{TotalHalf}} \{\coqdocvar{A} \coqdocvar{\tprime{A}} : \coqdockw{Set}\} (\coqdocvar{\trel{A}}: \coqdocvariable{A} \coqexternalref{:type scope:x '->' x}{http://coq.inria.fr/distrib/8.5pl3/stdlib/Coq.Init.Logic}{\coqdocnotation{\ensuremath{\rightarrow}}} \coqdocvariable{\tprime{A}} \coqexternalref{:type scope:x '->' x}{http://coq.inria.fr/distrib/8.5pl3/stdlib/Coq.Init.Logic}{\coqdocnotation{\ensuremath{\rightarrow}}} \coqdockw{Prop}) : \coqdockw{Type} := \coqdockw{\ensuremath{\forall}} (\coqdocvar{a}:\coqdocvariable{A}), \coqexternalref{:type scope:'x7B' x ':' x 'x26' x 'x7D'}{http://coq.inria.fr/distrib/8.5pl3/stdlib/Coq.Init.Specif}{\coqdocnotation{\{}}\coqdocvar{\tprime{a}}\coqexternalref{:type scope:'x7B' x ':' x 'x26' x 'x7D'}{http://coq.inria.fr/distrib/8.5pl3/stdlib/Coq.Init.Specif}{\coqdocnotation{:}}\coqdocvariable{\tprime{A}} \coqexternalref{:type scope:'x7B' x ':' x 'x26' x 'x7D'}{http://coq.inria.fr/distrib/8.5pl3/stdlib/Coq.Init.Specif}{\coqdocnotation{\&}} \coqexternalref{:type scope:'x7B' x ':' x 'x26' x 'x7D'}{http://coq.inria.fr/distrib/8.5pl3/stdlib/Coq.Init.Specif}{\coqdocnotation{(}}\coqdocvariable{\trel{A}} \coqdocvariable{a} \coqdocvar{\tprime{a}}\coqexternalref{:type scope:'x7B' x ':' x 'x26' x 'x7D'}{http://coq.inria.fr/distrib/8.5pl3/stdlib/Coq.Init.Specif}{\coqdocnotation{)\}}}.\coqdoceol
\coqdocemptyline
\coqdocnoindent
\coqdockw{Definition} \coqdef{Top.paper.anyRelPi}{anyRelPi}{\coqdocdefinition{anyRelPi}} \{\coqdocvar{A} \coqdocvar{\tprime{A}} :\coqdockw{Set}\} (\coqdocvar{\trel{A}}: \coqdocvariable{A} \coqexternalref{:type scope:x '->' x}{http://coq.inria.fr/distrib/8.5pl3/stdlib/Coq.Init.Logic}{\coqdocnotation{\ensuremath{\rightarrow}}} \coqdocvariable{\tprime{A}} \coqexternalref{:type scope:x '->' x}{http://coq.inria.fr/distrib/8.5pl3/stdlib/Coq.Init.Logic}{\coqdocnotation{\ensuremath{\rightarrow}}} \coqdockw{Prop}) \{\coqdocvar{B}: \coqdocvariable{A} \coqexternalref{:type scope:x '->' x}{http://coq.inria.fr/distrib/8.5pl3/stdlib/Coq.Init.Logic}{\coqdocnotation{\ensuremath{\rightarrow}}} \coqdockw{Set}\} \{\coqdocvar{\tprime{B}}: \coqdocvariable{\tprime{A}} \coqexternalref{:type scope:x '->' x}{http://coq.inria.fr/distrib/8.5pl3/stdlib/Coq.Init.Logic}{\coqdocnotation{\ensuremath{\rightarrow}}} \coqdockw{Set}\} \coqdoceol
\coqdocindent{1.00em}
(\coqdocvar{\trel{B}}: \coqdockw{\ensuremath{\forall}} \coqdocvar{a} \coqdocvar{\tprime{a}}, \coqdocvariable{\trel{A}} \coqdocvariable{a} \coqdocvariable{\tprime{a}} \coqexternalref{:type scope:x '->' x}{http://coq.inria.fr/distrib/8.5pl3/stdlib/Coq.Init.Logic}{\coqdocnotation{\ensuremath{\rightarrow}}} \coqexternalref{:type scope:x '->' x}{http://coq.inria.fr/distrib/8.5pl3/stdlib/Coq.Init.Logic}{\coqdocnotation{(}}\coqdocvariable{B} \coqdocvariable{a}\coqexternalref{:type scope:x '->' x}{http://coq.inria.fr/distrib/8.5pl3/stdlib/Coq.Init.Logic}{\coqdocnotation{)}} \coqexternalref{:type scope:x '->' x}{http://coq.inria.fr/distrib/8.5pl3/stdlib/Coq.Init.Logic}{\coqdocnotation{\ensuremath{\rightarrow}}} \coqexternalref{:type scope:x '->' x}{http://coq.inria.fr/distrib/8.5pl3/stdlib/Coq.Init.Logic}{\coqdocnotation{(}}\coqdocvariable{\tprime{B}} \coqdocvariable{\tprime{a}}\coqexternalref{:type scope:x '->' x}{http://coq.inria.fr/distrib/8.5pl3/stdlib/Coq.Init.Logic}{\coqdocnotation{)}} \coqexternalref{:type scope:x '->' x}{http://coq.inria.fr/distrib/8.5pl3/stdlib/Coq.Init.Logic}{\coqdocnotation{\ensuremath{\rightarrow}}} \coqdockw{Prop}) (\coqdocvar{f}: \coqdockw{\ensuremath{\forall}} \coqdocvar{a}, \coqdocvariable{B} \coqdocvariable{a}) (\coqdocvar{\tprime{f}}: \coqdockw{\ensuremath{\forall}} \coqdocvar{\tprime{a}}, \coqdocvariable{\tprime{B}} \coqdocvariable{\tprime{a}}) \coqdoceol
\coqdocindent{1.00em}
: \coqdockw{Prop} := \coqdockw{\ensuremath{\forall}} \coqdocvar{a} \coqdocvar{\tprime{a}} (\coqdocvar{\trel{a}}: \coqdocvariable{\trel{A}} \coqdocvariable{a} \coqdocvariable{\tprime{a}}), \coqdocvariable{\trel{B}} \coqdocvar{\_} \coqdocvar{\_} \coqdocvariable{\trel{a}} (\coqdocvariable{f} \coqdocvariable{a}) (\coqdocvariable{\tprime{f}} \coqdocvariable{\tprime{a}}).\coqdoceol
\coqdocemptyline
\coqdocnoindent
\coqdockw{Lemma} \coqdef{Top.paper.totalPiHalf}{totalPiHalf}{\coqdoclemma{totalPiHalf}}: \coqdockw{\ensuremath{\forall}} \{\coqdocvar{A} \coqdocvar{\tprime{A}} :\coqdockw{Set}\} (\coqdocvar{\trel{A}}: \coqref{Top.paper.IsoRel}{\coqdocdefinition{IsoRel}} \coqdocvariable{A} \coqdocvariable{\tprime{A}}) \{\coqdocvar{B}: \coqdocvariable{A} \coqexternalref{:type scope:x '->' x}{http://coq.inria.fr/distrib/8.5pl3/stdlib/Coq.Init.Logic}{\coqdocnotation{\ensuremath{\rightarrow}}} \coqdockw{Set}\} \{\coqdocvar{\tprime{B}}: \coqdocvariable{\tprime{A}} \coqexternalref{:type scope:x '->' x}{http://coq.inria.fr/distrib/8.5pl3/stdlib/Coq.Init.Logic}{\coqdocnotation{\ensuremath{\rightarrow}}} \coqdockw{Set}\} \coqdoceol
\coqdocindent{1.00em}
(\coqdocvar{\trel{B}}: \coqdockw{\ensuremath{\forall}} \coqdocvar{a} \coqdocvar{\tprime{a}}, ({\CoqSigTProj} \coqdocvariable{\trel{A}}) \coqdocvariable{a} \coqdocvariable{\tprime{a}} \coqexternalref{:type scope:x '->' x}{http://coq.inria.fr/distrib/8.5pl3/stdlib/Coq.Init.Logic}{\coqdocnotation{\ensuremath{\rightarrow}}} \coqexternalref{:type scope:x '->' x}{http://coq.inria.fr/distrib/8.5pl3/stdlib/Coq.Init.Logic}{\coqdocnotation{(}}\coqdocvariable{B} \coqdocvariable{a}\coqexternalref{:type scope:x '->' x}{http://coq.inria.fr/distrib/8.5pl3/stdlib/Coq.Init.Logic}{\coqdocnotation{)}} \coqexternalref{:type scope:x '->' x}{http://coq.inria.fr/distrib/8.5pl3/stdlib/Coq.Init.Logic}{\coqdocnotation{\ensuremath{\rightarrow}}} \coqexternalref{:type scope:x '->' x}{http://coq.inria.fr/distrib/8.5pl3/stdlib/Coq.Init.Logic}{\coqdocnotation{(}}\coqdocvariable{\tprime{B}} \coqdocvariable{\tprime{a}}\coqexternalref{:type scope:x '->' x}{http://coq.inria.fr/distrib/8.5pl3/stdlib/Coq.Init.Logic}{\coqdocnotation{)}} \coqexternalref{:type scope:x '->' x}{http://coq.inria.fr/distrib/8.5pl3/stdlib/Coq.Init.Logic}{\coqdocnotation{\ensuremath{\rightarrow}}} \coqdockw{Prop})\coqdoceol
\coqdocindent{1.00em}
(\coqdocvar{BTot} : \coqdockw{\ensuremath{\forall}} \coqdocvar{a} \coqdocvar{\tprime{a}} (\coqdocvar{\trel{a}}:({\CoqSigTProj} \coqdocvariable{\trel{A}}) \coqdocvariable{a} \coqdocvariable{\tprime{a}}), \coqref{Top.paper.TotalHalf}{\coqdocdefinition{TotalHalf}} (\coqdocvariable{\trel{B}} \coqdocvar{\_} \coqdocvar{\_} \coqdocvariable{\trel{a}})), \coqref{Top.paper.TotalHalf}{\coqdocdefinition{TotalHalf}} (\coqref{Top.paper.anyRelPi}{\coqdocdefinition{anyRelPi}} ({\CoqSigTProj} \coqdocvariable{\trel{A}}) \coqdocvariable{\trel{B}}).\coqdoceol

\coqdocnoindent
We have a similar combinator for the other direction, and similar combinators, one for each direction of the \coqRefDefn{Top.paper}{OneToOne}
property.
If the type of the recursive constructor argument has nested function types, we nest the appropriate combinator to get the 
proof of one direction of the \coqRefDefn{Top.paper}{Total} or  \coqRefDefn{Top.paper}{OneToOne} property.
For example, in the type $\forall$ (\coqdocvar{t}\ensuremath{_1} : \coqdocvar{T}\ensuremath{_1}) (\coqdocvar{t}\ensuremath{_2} : \coqdocvar{T}\ensuremath{_2})
 $\hdots$ 
(\coqdocvar{t}\ensuremath{_m} : \coqdocvar{T}\ensuremath{_m}),
(\coqdocinductive{$I$} $\hdots$) mentioned above, there will be an $m$-level nesting.
In the base case, when the type is just (\coqdocinductive{$I$} $\hdots$), we recursively call the proof (of one half of 
the \coqRefDefn{Top.paper}{Total} or \coqRefDefn{Top.paper}{OneToOne} property) currently being recursively defined. 

In the above discussion, we saw how to construct the proofs of one direction of the 
\coqRefDefn{Top.paper}{Total} and the \coqRefDefn{Top.paper}{OneToOne}
properties of types of all arguments (both recursive and non-recursive) of all
constructors. Now we explain how we use these proofs to build the proofs of the
same properties of the {\anyrel} translations of inductive types.
As in \secref{sec:uniformProp:pnode}, we use a W type to illustrate the
construction. However, our implementation \emph{directly} translates inductive types.
The type below is the same as the proposition
\coqRefInductive{Top.paper}{IWP} in \secref{sec:uniformProp:pnode}, except that we change its universe \coqdockw{Prop} to \coqdockw{Set} and change names to avoid clashes.

\coqdocnoindent
\coqdockw{Inductive} \coqdef{Top.paper.IWT}{IWT}{\coqdocinductive{IWT}} (\coqdocvar{I} \coqdocvar{A} : \coqdockw{Set}) (\coqdocvar{B} : \coqdocvariable{A} \coqexternalref{:type scope:x '->' x}{http://coq.inria.fr/distrib/8.5pl3/stdlib/Coq.Init.Logic}{\coqdocnotation{\ensuremath{\rightarrow}}} \coqdockw{Set}) (\coqdocvar{AI} : \coqdocvariable{A} \coqexternalref{:type scope:x '->' x}{http://coq.inria.fr/distrib/8.5pl3/stdlib/Coq.Init.Logic}{\coqdocnotation{\ensuremath{\rightarrow}}} \coqdocvariable{I})  (\coqdocvar{BI} : \coqdockw{\ensuremath{\forall}} (\coqdocvar{a} : \coqdocvariable{A}), \coqdocvariable{B} \coqdocvariable{a} \coqexternalref{:type scope:x '->' x}{http://coq.inria.fr/distrib/8.5pl3/stdlib/Coq.Init.Logic}{\coqdocnotation{\ensuremath{\rightarrow}}} \coqdocvariable{I}) : \coqdockw{\ensuremath{\forall}} (\coqdocvar{i}:\coqdocvar{I}), \coqdockw{Set} :=\coqdoceol
\coqdocnoindent
\coqdef{Top.paper.tnode}{tnode}{\coqdocconstructor{tnode}} : \coqdockw{\ensuremath{\forall}} (\coqdocvar{a} : \coqdocvar{A}) (\coqdocvar{branches} : \coqdockw{\ensuremath{\forall}} \coqdocvar{b} : \coqdocvar{B} \coqdocvariable{a}, \coqref{Top.paper.IWT}{\coqdocinductive{IWT}} \coqdocvar{I} \coqdocvar{A} \coqdocvar{B} \coqdocvar{AI} \coqdocvar{BI} (\coqdocvar{BI} \coqdocvariable{a} \coqdocvariable{b})), \coqref{Top.paper.IWT}{\coqdocinductive{IWT}} \coqdocvar{I} \coqdocvar{A} \coqdocvar{B} \coqdocvar{AI} \coqdocvar{BI} (\coqdocvar{AI} \coqdocvariable{a}).\coqdoceol

Again, we use the following abbreviations in a translated context:
\coqdef{Top.paper.WT}{WT}{\coqdocdefinition{WT}} := 
\coqRefInductive{Top.paper}{IWT} \coqdocvar{I}
\coqdocvar{A} \coqdocvar{B} \coqdocvar{AI} \coqdocvar{BI},
\tprime{\coqdocdefinition{WT}} := 
\coqRefInductive{Top.paper}{IWT} \coqdocvarP{I}
\coqdocvarP{A} \coqdocvarP{B} \coqdocvarP{AI} \coqdocvarP{BI},
\trel{\coqdocdefinition{WT}} :=
\coqref{Top.paper.IWT R}{\coqdocdefinition{\trel{IWT}}} \coqdocvariable{I} \coqdocvariable{\tprime{I}} \coqdocvariable{\trel{I}} \coqdocvariable{A} \coqdocvariable{\tprime{A}} \coqdocvariable{\trel{A}} \coqdocvariable{B} \coqdocvariable{\tprime{B}} \coqdocvariable{\trel{B}} \coqdocvariable{AI} \coqdocvariable{\tprime{AI}} \coqdocvariable{\trel{AI}} \coqdocvariable{BI} \coqdocvariable{\tprime{BI}} \coqdocvariable{\trel{BI}},
\coqdef{Top.paper.tnodew}{tnodew}{\coqdocdefinition{tnodew}} := 
\coqRefConstr{Top.paper}{tnode} \coqdocvar{I}
\coqdocvar{A} \coqdocvar{B} \coqdocvar{AI} \coqdocvar{BI},
\tprime{{\coqdocdefinition{tnodew}}} := 
\coqRefConstr{Top.paper}{tnode} \coqdocvarP{I}
\coqdocvarP{A} \coqdocvarP{B} \coqdocvarP{AI} \coqdocvarP{BI}

\begin{lemma}
\label{lemma:iwtTot}
\coqRefDefn{Top.paper}{Total}  \coqdocvarR{A}
$\rightarrow$
($\forall$ (\coqdocvar{a}:\coqdocvar{A}) (\coqdocvarP{a}:\coqdocvarP{A})
(\coqdocvarR{a}:\coqdocvarR{A} \coqdocvar{a} \coqdocvarP{a}),
\coqRefDefn{Top.paper}{Total} (\coqdocvarR{B} \coqdocvar{a} \coqdocvarP{a} \coqdocvarR{a}))
\coqdoceol\coqdocindent{2em}
$\rightarrow$
($\forall$ (\coqdocvar{a}:\coqdocvar{A}) (\coqdocvarP{a}:\coqdocvarP{A})
(\coqdocvarR{a}:\coqdocvarR{A} \coqdocvar{a} \coqdocvarP{a}),
\coqRefDefn{Top.paper}{OneToOne} (\coqdocvarR{B} \coqdocvar{a} \coqdocvarP{a} \coqdocvarR{a}))
$\rightarrow$
\coqRefDefn{Top.paper}{OneToOne} \coqdocvarR{I}
$\rightarrow$
\coqRefDefn{Top.paper}{Total}
(\trel{\coqRefDefn{Top.paper}{WT}} \coqdocvariable{i} \coqdocvariable{\tprime{i}} \coqdocvariable{\trel{i}})
\end{lemma}
\noindent
For one direction of totality, given a 
\coqdocvar{t}:(\coqRefDefn{Top.paper}{WT} \coqdocvar{i}), we need to produce a
\coqdocvarP{t}:(\tprime{\coqRefDefn{Top.paper}{WT}} \coqdocvarP{i}),
and prove 
\trel{\coqRefDefn{Top.paper}{WT}} \coqdocvariable{i} \coqdocvariable{\tprime{i}} \coqdocvariable{\trel{i}} \coqdocvar{t} \coqdocvarP{t}.
This proof is by induction on \coqdocvar{t}.
Note that 
\coqdocvar{B} serves as a domain type in the type of \coqdocvar{branches} in \coqRefInductive{Top.paper}{IWT} 
and 
that in the combinator \coqRefDefn{Top.paper}{totalPiHalf} shown above,
both the \coqRefDefn{Top.paper}{Total} and \coqRefDefn{Top.paper}{OneToOne} properties are needed for the relation for the domain type. 
This is because we needed both properties for the domain type in \lemref{lemma:piTot}.
Therefore, here we needed both properties for 
\coqdocvarR{B} to produce the argument \coqdocvarP{branches} in \coqdocvarP{t}.
We also needed \coqRefDefn{Top.paper}{OneToOne} \coqdocvarR{I}, 
for the same reason we needed it in \lemref{lemma:iwp}: to do
rewriting in indices.
%
The construction generalizes to other inductives, subject to the limitations discussed in \secref{sec:isorel:limitations}.

The proof of the \coqRefDefn{Top.paper}{OneToOne} property is straightforward, except at one place:
\begin{lemma}
\label{lemma:iwtOne}
\coqRefDefn{Top.paper}{OneToOne}  \coqdocvarR{A}
$\rightarrow$
($\forall$ (\coqdocvar{a}:\coqdocvar{A}) (\coqdocvarP{a}:\coqdocvarP{A})
(\coqdocvarR{a}:\coqdocvarR{A} \coqdocvar{a} \coqdocvarP{a}),
\coqRefDefn{Top.paper}{Total} (\coqdocvarR{B} \coqdocvar{a} \coqdocvarP{a} \coqdocvarR{a}))
\coqdoceol\coqdocindent{2em}
$\rightarrow$
\coqRefDefn{Top.paper}{OneToOne}
(\trel{\coqRefDefn{Top.paper}{WT}} \coqdocvariable{i} \coqdocvariable{\tprime{i}} \coqdocvariable{\trel{i}})
\end{lemma}
\noindent
The difficulty unsurprisingly involves indices.
First, in the above lemma, note that we don't need any property about \coqdocvarR{I}.
Also, recall that in \lemref{lemma:piOne}, 
we only needed the \coqRefDefn{Top.paper}{Total} property for the domain type.
Therefore, here, we need only the \coqRefDefn{Top.paper}{Total} property for
\coqdocvarR{B}.

Given \coqdocvar{t}:(\coqRefDefn{Top.paper}{WT} \coqdocvar{i}),
\coqdocvarP{t}:(\tprime{\coqRefDefn{Top.paper}{WT}} \coqdocvarP{i}),
\coqdocvarP{t2}:(\tprime{\coqRefDefn{Top.paper}{WT}} \coqdocvarP{i}),
\coqdocvarR{t}:\trel{\coqRefDefn{Top.paper}{WT}} \coqdocvariable{i} \coqdocvariable{\tprime{i}} \coqdocvariable{\trel{i}} \coqdocvar{t} \coqdocvarP{t},
and 
\coqdocvarR{t2}:\trel{\coqRefDefn{Top.paper}{WT}} \coqdocvariable{i} \coqdocvariable{\tprime{i}} \coqdocvariable{\trel{i}} \coqdocvar{t} \coqdocvarP{t2},
we need to produce a proof of \coqdocvarP{t} = \coqdocvarP{t2}.
The proof begins by pattern matching (induction) on \coqdocvar{t} and then another (nested) pattern match on \coqdocvarP{t}.
In general, inductives may have several constructors. In cases where the
constructors from the two pattern matches are different, we're done because
\coqdocvarR{t} computes to {\CoqFalse} (see \secref{sec:anyrel:ind}). We are now left only with cases that have the same constructor.
Back to the concrete example, we now have 
for some 
\coqdocvar{a},
\coqdocvar{branches},
\coqdocvarP{a}, and
\coqdocvarP{branches},
\coqdocvar{t} := \coqRefDefn{Top.paper}{tnodew} \coqdocvar{a} \coqdocvar{branches}
and 
\coqdocvarP{t} := \tprime{\coqRefDefn{Top.paper}{tnodew}} \coqdocvarP{a} \coqdocvarP{branches}.
\coqdocvarP{t} and \coqdocvarP{t2} now have type
\tprime{\coqRefDefn{Top.paper}{WT}} (\coqdocvarP{AI} \coqdocvarP{a}),
and we need to prove \coqdocvarP{t} = \coqdocvarP{t2}.
The obvious step now is to do a (nested) pattern match on \coqdocvarP{t2}. However, this is illegal.
As explained in \secref{sec:anyrel:ind}, for indexed inductive types, the definition of the type for one index may depend on the definition for other indices.
Therefore, to do induction on an indexed inductive type, the property being
proved by induction must be well-typed for \emph{all} indices.
Also, an equality is only well-typed if both sides have the same type. 
Thus, when we do a pattern match on \coqdocvarP{t2}, the index  (\coqdocvarP{AI} \coqdocvarP{a}) of its type gets generalized to a fresh variable, say
\coqdocvarP{i2}. Then the type of \coqdocvarP{t2} becomes 
\tprime{\coqRefDefn{Top.paper}{WT}} \coqdocvarP{i2},
and thus the types of \coqdocvarP{t} and 
\coqdocvarP{t2} become non-convertible.

A common solution to such problems is to state the equality in a more general type.
We can generalize the statement \coqdocvarP{t} = \coqdocvarP{t2} to the statement that
the dependent pair of (\coqdocvarP{AI} \coqdocvarP{a}) and \coqdocvarP{t} and
the dependent pair of (\coqdocvarP{AI} \coqdocvarP{a}) and \coqdocvarP{t2}
are equal in the sigma type \{ \coqdocvarP{i2} : \coqdocvarP{I} \coqdocnotation{\&} \tprime{\coqRefDefn{Top.paper}{WT}} \coqdocvarP{i2} \}.
Now when we pattern match on \coqdocvarP{t2}, the type of the RHS of the equality remains unchanged.
The rest of the proof is straightforward.

Finally, we have to undo the generalization of the equality statement. 
For that, we use the following lemma from Coq's standard library, which although unprovable~\cite{Hofmann.Streicher1998},
is a consequence of proof irrelevance (or the UIP (Unicity of Identity Proofs) axiom).

\coqdocnoindent
\coqdockw{Lemma} \coqdef{Top.paper.inj pair2}{inj\_pair2}{\coqdoclemma{inj\_pair2}}: \coqdockw{\ensuremath{\forall}} (\coqdocvar{U} : \coqdockw{Type}) (\coqdocvar{P} : \coqdocvariable{U} \coqexternalref{:type scope:x '->' x}{http://coq.inria.fr/distrib/8.5pl3/stdlib/Coq.Init.Logic}{\coqdocnotation{\ensuremath{\rightarrow}}} \coqdockw{Type}) (\coqdocvar{p} : \coqdocvariable{U}) (\coqdocvar{x} \coqdocvar{y} : \coqdocvariable{P} \coqdocvariable{p}), \coqexternalref{existT}{http://coq.inria.fr/distrib/8.5pl3/stdlib/Coq.Init.Specif}{\coqdocconstructor{existT}} \coqdocvariable{p} \coqdocvariable{x} \coqexternalref{:type scope:x '=' x}{http://coq.inria.fr/distrib/8.5pl3/stdlib/Coq.Init.Logic}{\coqdocnotation{=}} \coqexternalref{existT}{http://coq.inria.fr/distrib/8.5pl3/stdlib/Coq.Init.Specif}{\coqdocconstructor{existT}} \coqdocvariable{p} \coqdocvariable{y} \coqexternalref{:type scope:x '->' x}{http://coq.inria.fr/distrib/8.5pl3/stdlib/Coq.Init.Logic}{\coqdocnotation{\ensuremath{\rightarrow}}} \coqdocvariable{x} \coqexternalref{:type scope:x '=' x}{http://coq.inria.fr/distrib/8.5pl3/stdlib/Coq.Init.Logic}{\coqdocnotation{=}} \coqdocvariable{y}.\coqdoceol 

\coqdocnoindent
In general, an inductive type may have several (say $n$) indices. Our
translation then uses $n$ nested dependent pairs. Also, the above lemma is then
invoked $n$ times.

\section{{\isorel} Translation}
\label{sec:isorel}
Now we use the lemmas developed in the previous two sections to define the {\isorel} translation.
Those lemmas are summarized in tables in \appref{appendix:table}.
 
In \secref{sec:uniformProp}, we saw how to systematically produce proofs of
the two desirable properties (\coqRefDefn{Top.paper}{IffProps} and \coqRefDefn{Top.paper}{CompleteRel}) for {\anyrel} translations of propositions.
In the {\isorel} translation, we augment the {\anyrel} translation to ensure
that parametricity relations of propositions always come bundled with those
two properties. We wish to define:

\coqdocnoindent
\ptranslateIso{\coqdockw{Prop}} :=
{\coqdocnotation{\ensuremath{\lambda}}} {\coqdocnotation{(}}\coqdocvar{A} \coqdocvar{\tprime{A}}: \coqdockw{Prop}{\coqdocnotation{),}} \coqexternalref{:type scope:'x7B' x ':' x 'x26' x 'x7D'}{http://coq.inria.fr/distrib/8.5pl3/stdlib/Coq.Init.Specif}{\coqdocnotation{\{}}\coqdocvar{R} \coqexternalref{:type scope:'x7B' x ':' x 'x26' x 'x7D'}{http://coq.inria.fr/distrib/8.5pl3/stdlib/Coq.Init.Specif}{\coqdocnotation{:}} \coqdocvariable{A} \coqexternalref{:type scope:x '->' x}{http://coq.inria.fr/distrib/8.5pl3/stdlib/Coq.Init.Logic}{\coqdocnotation{\ensuremath{\rightarrow}}} \coqdocvariable{\tprime{A}} \coqexternalref{:type scope:x '->' x}{http://coq.inria.fr/distrib/8.5pl3/stdlib/Coq.Init.Logic}{\coqdocnotation{\ensuremath{\rightarrow}}} \coqdockw{Prop} \coqexternalref{:type scope:'x7B' x ':' x 'x26' x 'x7D'}{http://coq.inria.fr/distrib/8.5pl3/stdlib/Coq.Init.Specif}{\coqdocnotation{\&}} \coqref{Top.paper.IffProps}{\coqdocdefinition{IffProps}} \coqdocvar{R} \coqexternalref{:type scope:x '/x5C' x}{http://coq.inria.fr/distrib/8.5pl3/stdlib/Coq.Init.Logic}{\coqdocnotation{\ensuremath{\land}}} \coqref{Top.paper.CompleteRel}{\coqdocdefinition{CompleteRel}} \coqdocvar{R}\coqexternalref{:type scope:'x7B' x ':' x 'x26' x 'x7D'}{http://coq.inria.fr/distrib/8.5pl3/stdlib/Coq.Init.Specif}{\coqdocnotation{\}}}.\coqdoceol

\begin{lemma}
\label{lemma:unify}
For any \coqdocvar{A}:\coqdockw{Prop}, \coqdocvar{B}:\coqdockw{Prop}, and \coqdocvar{R}:(\coqdocvar{A}~$\rightarrow$~\coqdocvar{B}~$\rightarrow$~\coqdockw{Prop}),
(\coqRefDefn{Top.paper}{IffProps} \coqdocvar{R} $\wedge$ \coqRefDefn{Top.paper}{CompleteRel} \coqdocvar{R}) $\leftrightarrow$
(\coqRefDefn{Top.paper}{Total} \coqdocvar{R} $\times$ \coqRefDefn{Top.paper}{OneToOne} \coqdocvar{R}).
\end{lemma}

\begin{proof}
\coqRefDefn{Top.paper}{OneToOne} \coqdocvar{R} is a trivial consequence of proof irrelevance.
Also, using proof irrelevance, it is straightforward to prove 
\coqRefDefn{Top.paper}{Total} \coqdocvar{R} $\leftrightarrow$ 
(\coqRefDefn{Top.paper}{IffProps} \coqdocvar{R} $\wedge$ \coqRefDefn{Top.paper}{CompleteRel} \coqdocvar{R}).
\end{proof}

Thus we instead choose the following equivalent definition:\coqdoceol
\coqdocnoindent
\coqdockw{Definition} \coqdef{Top.paper.IsoRel}{IsoRel}{\coqdocdefinition{IsoRel}} :=
{\coqdocnotation{\ensuremath{\lambda}}} {\coqdocnotation{(}}\coqdocvar{A} \coqdocvar{\tprime{A}}: \coqdockw{Set}{\coqdocnotation{),}} \coqexternalref{:type scope:'x7B' x ':' x 'x26' x 'x7D'}{http://coq.inria.fr/distrib/8.5pl3/stdlib/Coq.Init.Specif}{\coqdocnotation{\{}}\coqdocvar{\trel{A}} \coqexternalref{:type scope:'x7B' x ':' x 'x26' x 'x7D'}{http://coq.inria.fr/distrib/8.5pl3/stdlib/Coq.Init.Specif}{\coqdocnotation{:}} \coqdocvariable{A} \coqexternalref{:type scope:x '->' x}{http://coq.inria.fr/distrib/8.5pl3/stdlib/Coq.Init.Logic}{\coqdocnotation{\ensuremath{\rightarrow}}} \coqdocvariable{\tprime{A}} \coqexternalref{:type scope:x '->' x}{http://coq.inria.fr/distrib/8.5pl3/stdlib/Coq.Init.Logic}{\coqdocnotation{\ensuremath{\rightarrow}}} \coqdockw{Prop} \coqexternalref{:type scope:'x7B' x ':' x 'x26' x 'x7D'}{http://coq.inria.fr/distrib/8.5pl3/stdlib/Coq.Init.Specif}{\coqdocnotation{\&}} \coqref{Top.paper.::x 'xC3x97' x}{\coqdocnotation{(}}\coqref{Top.paper.Total}{\coqdocdefinition{Total}} \coqdocvar{\trel{A}}\coqref{Top.paper.::x 'xC3x97' x}{\coqdocnotation{)}} \coqref{Top.paper.::x 'xC3x97' x}{\coqdocnotation{×}} \coqref{Top.paper.::x 'xC3x97' x}{\coqdocnotation{(}}\coqref{Top.paper.OneToOne}{\coqdocdefinition{OneToOne}} \coqdocvar{\trel{A}}\coqref{Top.paper.::x 'xC3x97' x}{\coqdocnotation{)}}\coqexternalref{:type scope:'x7B' x ':' x 'x26' x 'x7D'}{http://coq.inria.fr/distrib/8.5pl3/stdlib/Coq.Init.Specif}{\coqdocnotation{\}}}.\coqdoceol
\coqdocnoindent 
\ptranslateIso{\coqdockw{Prop}} :=
{\coqdocnotation{\ensuremath{\lambda}}} 
{\coqdocnotation{(}}\coqdocvar{A} \coqdocvar{\tprime{A}}: \coqdockw{Prop}
{\coqdocnotation{),}} 
\coqRefDefn{Top.paper}{IsoRel} \coqdocvar{A} \coqdocvar{\tprime{A}}.\coqdoceol

\coqdocnoindent 
Also, when propositions mention types, we may need the {\anyrel} parametricity relations of 
those types to have the \coqRefDefn{Top.paper}{Total} or \coqRefDefn{Top.paper}{OneToOne} property.
In \secref{sec:uniformProp:type}, we saw how to systematically build these properties for types in \coqdockw{Set}.
Thus, we can choose to define:

\coqdocnoindent 
\ptranslateIso{\coqdockw{Set}}:=
{\coqdocnotation{\ensuremath{\lambda}}} {\coqdocnotation{(}}\coqdocvar{A} \coqdocvar{\tprime{A}}: \coqdockw{Set}{\coqdocnotation{),}} 
\coqRefDefn{Top.paper}{IsoRel} \coqdocvar{A} \coqdocvar{\tprime{A}}.\coqdoceol

This choice is not ideal because the proofs of the desirable properties of many propositions 
don't need one or both of the two bundled properties of the types mentioned in the
propositions. We saw three examples (\coqRefDefn{Top.paper}{PNone}, \coqRefDefn{Top.paper}{PTot}, \coqRefDefn{Top.paper}{POne}) in \secref{sec:uniformProp}.
We use a 2-stage process in our {\isorel}
translation. In the first stage, which we call the weak {\isorel} translation and denote by \ptranslateIso{}, we \emph{always}
bundle the relations for types with \emph{both} the two properties. \ptranslateIso{} is structurally recursive and implemented in Coq (Gallina).
In the 2$^{nd}$ stage (\secref{sec:isorel:unused}),
we attempt to remove unused assumptions from the generated abstraction theorems. For efficiency, this 2$^{nd}$ stage is implemented as an OCaml plugin for Coq.
We denote the composition of the two stages by \ptranslateStrongIso{}, and call it the (strong) {\isorel} translation. 

It is natural to consider a 1-phase approach where the main translation itself
determines the minimally needed assumptions on type variables.
We considered and rejected that approach because it seemed very complex to implement.
A discussion can nevertheless be found in \appref{appendix:onephase}.

\subsection{\ptranslateIso{} (weak {\isorel} translation)}
\label{sec:isorel:weak}
First we define the following functions to construct and destruct
\coqRefDefn{Top.paper}{IsoRel}s.

\coqdocnoindent
\coqdockw{Definition} \coqdef{Top.paper.mkIsoRel}{mkIsoRel}{\coqdocdefinition{mkIsoRel}} (\coqdocvar{A} \coqdocvar{\tprime{A}} : \coqdockw{Set}) (\coqdocvar{\trel{A}}: \coqdocvariable{A} \coqexternalref{:type scope:x '->' x}{http://coq.inria.fr/distrib/8.5pl3/stdlib/Coq.Init.Logic}{\coqdocnotation{\ensuremath{\rightarrow}}} \coqdocvariable{\tprime{A}} \coqexternalref{:type scope:x '->' x}{http://coq.inria.fr/distrib/8.5pl3/stdlib/Coq.Init.Logic}{\coqdocnotation{\ensuremath{\rightarrow}}} \coqdockw{Prop}) (\coqdocvar{\trel{A}tot}: \coqref{Top.paper.Total}{\coqdocdefinition{Total}} \coqdocvariable{\trel{A}}) 
\coqdoceol\coqdocindent{1em}
(\coqdocvar{\trel{A}one}: \coqref{Top.paper.OneToOne}{\coqdocdefinition{OneToOne}} \coqdocvariable{\trel{A}}) : \coqref{Top.paper.IsoRel}{\coqdocdefinition{IsoRel}} \coqdocvariable{A} \coqdocvariable{\tprime{A}}
:= {\CoqExistT} \coqdocvarR{A} (\coqdocvar{\trel{A}tot}, \coqdocvar{\trel{A}one}).
\coqdoceol
\coqdocnoindent
\coqdockw{Definition} \coqdef{Top.paper.projRel}{projRel}{\coqdocdefinition{projRel}} (\coqdocvar{A} \coqdocvar{\tprime{A}} : \coqdockw{Set}) (\coqdocvar{\trel{A}iso} : \coqref{Top.paper.IsoRel}{\coqdocdefinition{IsoRel}} \coqdocvariable{A} \coqdocvariable{\tprime{A}}) : \coqdocvariable{A} \coqexternalref{:type scope:x '->' x}{http://coq.inria.fr/distrib/8.5pl3/stdlib/Coq.Init.Logic}{\coqdocnotation{\ensuremath{\rightarrow}}} \coqdocvariable{\tprime{A}} \coqexternalref{:type scope:x '->' x}{http://coq.inria.fr/distrib/8.5pl3/stdlib/Coq.Init.Logic}{\coqdocnotation{\ensuremath{\rightarrow}}} \coqdockw{Prop}
:= {\CoqSigTProj} \coqdocvar{\trel{A}iso}.\coqdoceol
\fxnote{subscripts should come at the end of variable names}
\coqdocnoindent
W.r.t. \ptranslate{}, the main change in \ptranslateIso{}
is that the parametricity relations of types and propositions come bundled with proofs. As a result, we often have
to project out relations from bundles before applying them.
Let
\projTyRel{A}{$t$} denote \coqRefDefn{Top.paper}{projRel} $A$ \tprime{$A$} $t$
if $A$ has type \coqdockw{Prop} or \coqdockw{Set}, and just  $t$ otherwise.
In our implementation, wherever needed, our reifier invokes Coq's typechecker and
includes this information (a flag indicating that a term has type \coqdockw{Prop} or \coqdockw{Set}) in the reified terms.
\ptranslateIso{} needs this information for the domain and codomain types of $\Pi$
types, the argument types of $\lambda$ terms, and the return types of \coqdockw{match} and
\coqdockw{fix} terms.
\coqdocnoindent
The desired correctness property of \ptranslateIso{} is: for closed $t$ and
$T$, if $t$ : $T$,
then we must have \ptranslateIso{$t$}:
((\projTyRel{T}{\ptranslateIso{$T$}}) $t$ $t$).
\fxnote{add let bindings to the translation. in template-coq, need add sort
annotations for its type as well}

\coqdocemptyline

\ptranslateIso{\coqdockw{Prop}} :=
{\coqdocnotation{\ensuremath{\lambda}}} 
{\coqdocnotation{(}}\coqdocvar{c} \coqdocvar{\tprime{c}}: \coqdockw{Prop}
{\coqdocnotation{),}} 
\coqRefDefn{Top.paper}{IsoRel} \coqdocvar{c} \coqdocvar{\tprime{c}}.\coqdoceol

\ptranslateIso{\coqdockw{Set}}:=
{\coqdocnotation{\ensuremath{\lambda}}} {\coqdocnotation{(}}\coqdocvar{c} \coqdocvar{\tprime{c}}: \coqdockw{Set}
{\coqdocnotation{),}} 
\coqRefDefn{Top.paper}{IsoRel} \coqdocvar{c} \coqdocvar{\tprime{c}}.\coqdoceol

\coqdocnoindent
For $i>0$, we have:
 
\ptranslateIso{\coqdockw{Type}$_i$} :=
{\coqdocnotation{\ensuremath{\lambda}}} 
{\coqdocnotation{(}}\coqdocvar{c} \coqdocvar{\tprime{c}}: \coqdockw{Type}$_i$
{\coqdocnotation{),}} 
\coqdocvar{c} $\rightarrow$ \coqdocvar{\tprime{c}} $\rightarrow$ \coqdockw{Type}$_i$

\coqdocemptyline
\ptranslateIso{\coqdocvar{x}}:= \coqdocvarR{x}

\ptranslateIso{λ\coqdocvar{x}:A.B} := λ
(\coqdocvar{x}: A)
(\coqdocvarP{x}: \tprime{A}) 
(\coqdocvarR{x}: (\projTyRel{A}{\ptranslateIso{A}}) \coqdocvar{x}
\coqdocvarP{x}), \ptranslateIso{B}

\ptranslateIso{(A\,B)} := (\ptranslateIso{A}\,B\,\tprime{B}\,\ptranslateIso{B})

\coqdocnoindent
The translation of dependent function types/propositions has two cases. First,
we define the following relation, which is the same as the {\anyrel}
translation, except that if necessary, it projects out the relations of the domain 
and the codomain type. 

$\ptranslateC{∀\coqdocvar{x}\!:\! A.B}{\Pi} :=
  λ(\coqdocvarFour{x}:∀\coqdocvar{x}:A.B)(\coqdocvarFive{x}:∀\coqdocvarP{x}:\tprime{A}.\tprime{B}),$
\coqdoceol
\coqdocindent{8em}
$∀(\coqdocvar{x}:A)(\coqdocvarP{x}:\tprime{A})(\coqdocvarR{x}:
  (\projTyRel{A}{\ptranslateIso{A}})\,\coqdocvar{x}\,{\coqdocvarP{x}}), \;
  (\projTyRel{B}\ptranslateIso{B})
  {(\coqdocvarFour{x}\;\coqdocvar{x})}{(\coqdocvarFive{x}\,\coqdocvarP{x})}$

\coqdocnoindent
If $∀x\!:\! A.B$ has type \coqdockw{Type}$_i$ where $i>0$, then we have

\ptranslateIso{$∀ \coqdocvar{x}\!:\! A.B$} := 
\ptranslateC{$∀
\coqdocvar{x}\!:\! A.B$}{\Pi}

\coqdocnoindent
If $∀x\!:\! A.B$ has type \coqdockw{Set} or \coqdockw{Prop} (depending on the type of $B$) then we have

\ptranslateIso{$∀ \coqdocvar{x}\!:\! A.B$} :=
\coqRefDefn{Top.paper}{mkIsoRel}
({$∀ \coqdocvar{x}\!:\! A.B$})
({$∀ \coqdocvarP{x}\!:\! \tprime{A}.\tprime{B}$})
(\ptranslateC{$∀ \coqdocvar{x}\!:\! A.B$}{\Pi})
$ptot$
$pone$

\coqdocnoindent
Here $ptot$ and $pone$ respectively are the proofs of the
\coqRefDefn{Top.paper}{Total} and \coqRefDefn{Top.paper}{OneToOne} properties,
whose construction was explained in \secref{sec:uniformProp:piProp} (if $B$:\coqdockw{Prop}) or
\secref{sec:uniformProp:funt} (otherwise): It is important to prefer the construction
in \secref{sec:uniformProp:piProp} (also see \lemref{lemma:unify}) because that uses fewer assumptions and thus increases
the potency of the $2^{nd}$ phase described in the next subsection.
More details can be found in \appref{appendix:correctness:typing:pi}.
\coqdocnoindent

Just like the case for $\Pi$ type, if an inductive type is in the 
\coqdockw{Set} or \coqdockw{Prop} universe, we bundle its relation
with the two proof terms produced as explained in \secref{sec:uniformProp:pnode} (for inductive propositions) 
or \secref{sec:uniformProp:indt} (otherwise).

The translation of the \coqdockw{match} and the \coqdockw{fix} constructs are
nearly the same as in the {\anyrel} translation. There was a small change needed
in the return types. Coq's kernel requires every pattern \coqdockw{match} to
include a return type (which is a function of the discriminee and its indices). 
The {\anyrel} translation of a 
\coqdockw{match} term (say $t$) of type $T$, is a \coqdockw{match} term
whose return type is \ptranslate{$T$} $t$ \tprime{$t$}.
In the {\isorel} translation, the return type is
(\projTyRel{T}{\ptranslateIso{T}}) $t$ \tprime{$t$}. A similar change was needed
in the translation of fixpoints.

\subsubsection{Correctness}
\label{sec:isorel:correctness}
As explained before,  
w.r.t. \ptranslate{}, the only changes in \ptranslateIso{} are:
1) The relations produced by \ptranslate{} of types/propositions in
\coqdockw{Set} or \coqdockw{Prop} are now paired with proofs of \coqRefDefn{Top.paper}{Total} and \coqRefDefn{Top.paper}{OneToOne} properties.
2) As a result, at some places, we project the relations out of the pairs.
In Sections~\ref{sec:uniformProp} and \ref{sec:uniformProp:type}, we explained in detail how to construct the proofs of 
\coqRefDefn{Top.paper}{Total} and \coqRefDefn{Top.paper}{OneToOne} properties.
Those constructions were originally done and proved correct in Coq.
Except for the construction of those proofs, 
the correctness argument for \ptranslateIso{} is almost identical to the correctness argument for 
\ptranslate{}: one proves that the translation preserves substitution, then reduction, and finally typehood~\cite[Lemma 2, Theorem 1]{Keller.Lasson2012}. 

In \appref{appendix:correctness}, we discuss a formal (but not machine checked) proof of correctness of \ptranslateIso{} for a
CoC-like core calculus. 
%
%
%
%
The formalized calculus excludes inductive types and associated constructs such as pattern matching and fixpoints.
However, we illustrate that \ptranslateIso{} correctly translates the W type (which can encode inductive types)
and its recursion principle (which can encode pattern matching and fixpoints).
We also show that \ptranslateIso{} preserves $\iota$ reduction of that recursion principle:
intuitively, axioms don't block preservation of reduction because we use axioms \emph{only} in proofs of the 
\coqRefDefn{Top.paper}{Total} and \coqRefDefn{Top.paper}{OneToOne} properties.
We have also tested \ptranslateIso{} on a large variety of inductives (e.g. multiple and dependent indices, multiple constructors, various shapes of arguments of constructors).
Recall that Coq typechecks the result of \ptranslateIso{} (after reflection): so soundness is not a concern.



\subsection{Eliminating Unused Hypotheses}
\label{sec:isorel:unused}
As mentioned before, \ptranslateStrongIso{} has a
post-processing stage where the user can ask the system to strengthen an
abstraction theorem generated by \ptranslateIso{}.
In \secref{sec:uniformProp}, we saw that our proofs of the desirable properties (\coqRefDefn{Top.paper}{IffProps}, \coqRefDefn{Top.paper}{CompleteRel}) of
propositions may not need one or both of the two properties
(\coqRefDefn{Top.paper}{Total}, \coqRefDefn{Top.paper}{OneToOne}) about the
relations of types mentioned in the propositions.
Similarly, the proof of the \coqRefDefn{Top.paper}{Total} or \coqRefDefn{Top.paper}{OneToOne} property for
relations of composite types may not need one or both 
of the two properties (\coqRefDefn{Top.paper}{Total}, \coqRefDefn{Top.paper}{OneToOne}) of
subcomponents (\secref{sec:uniformProp:type}).
Thus, we expect the proofs produced by \ptranslateIso{} to not mention some of
the hypotheses. We want to strengthen the statements of the theorems produced by
\ptranslateIso{} by pruning the unused hypotheses.  

There are many ways to define what it means for a variable \coqdocvar{x} (e.g. a hypothesis) to be unused
in a term (e.g. a proof) $p$. We say that a variable \coqdocvar{x} is \emph{definitionally unused} in $p$ if $\exists$ a term $p'$ such that $p'$ is \emph{definitionally} 
equal to $p$ and the free variables of $p'$ does not include \coqdocvar{x}.
It is easy to \emph{exactly} determine whether \coqdocvar{x} is definitionally unused in $p$: just strongly normalize $p$ and check if 
\coqdocvar{x} occurs in the free variables of the normal form.
However, for some realistic applications, strong normalization often ran for hours and then ran out of memory on our machines with 32GB RAM.
So, we use a publically available Coq plugin 
that avoids normalizing many subterms (e.g. whose free variables do not include \coqdocvar{x}), and 
is yet guaranteed to return the \emph{exact} answer.
This plugin runs within a few seconds in all our applications so far.
If it succeeds in eliminating \coqdocvar{x}, it also returns the term $p'$ where \coqdocvar{x} does not occur free.
$p'$ can be considered a proof of a stronger theorem which does not have the hypothesis \coqdocvar{x}.

As an example, consider a polymorphic proposition of the form $\lambda$
(\coqdocvar{T}:\coqdockw{Set}), $\theta$, where $\theta$ is some term.
\ptranslateIso{\ensuremath{\lambda} (\coqdocvar{T}:\coqdockw{Set}),
\ensuremath{\theta}} := $\lambda$ 
(\coqdocvar{T}:\coqdockw{Set})
(\coqdocvarP{T}:\coqdockw{Set})
(\coqdocvarR{T}: \coqRefDefn{Top.paper}{IsoRel} \coqdocvar{T} \coqdocvarP{T})
, \ptranslateIso{$\theta$}.
We $\eta$-expand \coqdocvarR{T}, say as variables \coqdocvar{R},
\coqdocvar{RTot}, \coqdocvar{ROne}, and then use the above-mentioned plugin, 
hoping one or both of \coqdocvar{RTot},
\coqdocvar{ROne} disappear in \ptranslateIso{$\theta$}. 

A more effective approach would be to aim for the following definition:
a variable \coqdocvar{x} is \emph{logically unused} in $p$ if 
$\exists$ a term $p'$ such that $p'$ is \emph{propositionally} 
equal to $p$ and the free variables of $p'$ does not include \coqdocvar{x}.
It is impossible to solve this variant exactly, but we believe there are heuristics that would
yield better results (stronger theorems) than the above approach in some applications.

\subsection{Limitations of the {\isorel} translation}
\label{sec:isorel:limitations}
\ptranslateIso{} fails for propositions where types of higher universes occur at certain places.
In universal quantification, the quantified type
must be in \coqdockw{Set} or \coqdockw{Prop}. 
In inductively defined propositions, the types of indices
and the types of arguments of constructors (except the parameters of the type)
must be in \coqdockw{Set} or \coqdockw{Prop}.
In \secref{sec:uniformProp}, we saw that the relations of the types at those positions may need
to have the 
\coqRefDefn{Top.paper}{Total} and/or the \coqRefDefn{Top.paper}{OneToOne} properties,
Unfortunately, it is not possible to systematically produce the proofs of those 
properties for types in higher universes:

Suppose we redefined, for $i>0$, \ptranslateIso{\coqdockw{Type}$_i$} to be just
like \ptranslateIso{\coqdockw{Set}}.
Then, the abstraction theorem for \coqdockw{Set}:\coqdockw{Type}$_1$ fails.
Now, \ptranslateIso{\coqdockw{Set}} needs to be augmented to also produce the
proofs of the \coqRefDefn{Top.paper}{Total} and the
\coqRefDefn{Top.paper}{OneToOne} property for the relation 
{\coqdocnotation{\ensuremath{\lambda}}} {\coqdocnotation{(}}\coqdocvar{A} \coqdocvar{\tprime{A}}: \coqdockw{Set}{\coqdocnotation{),}} 
\coqRefDefn{Top.paper}{IsoRel} \coqdocvar{A} \coqdocvar{\tprime{A}}.
The latter property is not provable:

\begin{lemma}
There is no axiom-free proof of \coqRefDefn{Top.paper}{OneToOne}
({\coqdocnotation{\ensuremath{\lambda}}} {\coqdocnotation{(}}\coqdocvar{A} \coqdocvar{\tprime{A}}: \coqdockw{Set}{\coqdocnotation{),}} 
\coqRefDefn{Top.paper}{IsoRel} \coqdocvar{A} \coqdocvar{\tprime{A}})
\end{lemma}
\begin{proof}
It is easy to produce a \coqRefDefn{Top.paper}{Total} and
\coqRefDefn{Top.paper}{OneToOne} relation between
the types
\coqRefInductive{Top.paper}{nat} and \coqRefInductive{Top.paper}{nat},
and between the types
\coqRefInductive{Top.paper}{nat} and {\CoqList} {\CoqTrue}.
Then, it is easy to see that 
\coqRefDefn{Top.paper}{OneToOne} 
({{\coqdocnotation{\ensuremath{\lambda}}} {\coqdocnotation{(}}\coqdocvar{A}
\coqdocvar{\tprime{A}}: \coqdockw{Set}{\coqdocnotation{),}}
\coqRefDefn{Top.paper}{IsoRel} \coqdocvar{A} \coqdocvar{\tprime{A}}})
implies \coqRefInductive{Top.paper}{nat} = {\CoqList} {\CoqTrue}, which is
\emph{unprovable} in Coq:
looking at the definition of = (\secref{sec:intro}), it is obvious that if $u=v$ and $u$ and $v$ are closed,
they must be \emph{definitionally} equal. 
\end{proof}
\coqRefInductive{Top.paper}{nat} = {\CoqList} {\CoqTrue} may be provable using the univalence
axiom~\cite[Sec.~2.10]{TheUnivalentFoundationsProgram2013}. 
However, that axiom refutes UIP (Unicity of Identity Proofs)
which is useful in many Coq developments. 
For example, the proof of the 
\coqref{Top.paper.inj pair2}{\coqdocdefinition{inj\_pair2}} 
lemma used in
\secref{sec:uniformProp:indt} uses the UIP axiom (proof irrelevance implies UIP).
Also, UIP is needed for the justification of erasing (equality) proofs
during the compilation of Coq programs~\cite{Letouzey2004}.

As an example of the above limitation, the {\isorel} translation fails on the following because the index type is in a higher universe.\coqdoceol
\coqdocnoindent
\coqdockw{Inductive}  \coqdef{Top.paper.isNat}{isNat}{\coqdocinductive{isNat}} : \coqdockw{\ensuremath{\forall}} (\coqdocvar{A}:\coqdockw{Set}), \coqdockw{Prop} := 
\coqdef{Top.paper.isnat}{isnat}{\coqdocconstructor{isnat}} : \coqref{Top.paper.isNat}{\coqdocinductive{isNat}} \coqref{Top.paper.nat}{\coqdocinductive{nat}}.\coqdoceol
\coqdocnoindent
\coqdockw{Set} \cancel{:} \coqdockw{Set}. Indeed,
the {\isorel} abstraction theorem for \coqRefInductive{Top.paper}{isNat}
is easily refutable.

Nevertheless, as explained in \secref{sec:intro}, \coqdockw{Set} and \coqdockw{Prop} suffice for many practical application domains, 
especially verification of computer and physical systems. 

As explained in \secref{sec:isorel:weak}, \ptranslateIso{$∀ \coqdocvar{x}\!:\! A.B$} works differently in the cases when $B$:\coqdockw{Prop} and $B$:\coqdockw{Set}.
Thus, \ptranslateIso{} may produce ill-typed results for terms whose typing derivations use the rule \coqdockw{Prop} :> \coqdockw{Set}
at that position of $∀$.
See \appref{appendix:correctness:subst:subtyping} for an example. Using \coqdockw{Prop} :> \coqdockw{Set} is not always problematic.
Also, instead of using \coqdockw{Prop} :> \coqdockw{Set}, one can duplicate definitions for \coqdockw{Prop} and \coqdockw{Set}:
for example, Coq has different pair constructions for propositions and types.

  
There are also some fixable limitations: \ptranslateIso{} currently does not handle nested inductive types or propositions. 
Also, it fails on terms whose typehood derivation uses the
property \coqdockw{Set} :> \coqdockw{Type}$_1$.
The relations for types in \coqdockw{Set} are bundled with proofs,
while the relations for types in \coqdockw{Type}$_i$ aren't ($i>0$).
If the reification mechanism marked the 
places where the subtyping property was
used, \ptranslateIso{} can insert projections to remove the proofs.

\fxnote{Mutual inductives with mixed sorts: Set/Type vs Prop: see examples/mutIndDiffSort.v
This is a problem with the {\anyrel} translation. Fix the title.}

\section{Applications}
\label{sec:app}

The parametricity translation presented and implemented in previous
work~\cite{Keller.Lasson2012} 
can already be used to obtain for free many Coq proofs that
Coq users often do manually, often spending several hours, if not days.
First, we illustrate this with a simplified version of an actual use case from
our ongoing compiler-verification project.
Then, we extend the example with a free theorem 
that our {\isorel} translation produces but the translations in previous work
were not designed to produce.

When using named variable bindings, we often have to prove that various
concepts, e.g. big-step operational semantics, respect $\alpha$ equality.
These proofs are tedious, especially if the language has several kinds of
reductions, such as $\beta$, $\zeta$ (let-bindings), $\iota$ (pattern-matching). 
However, all these proofs
mainly boil down to one fact: that substitution behaves uniformly, i.e.,
on related ($\alpha$ equal) inputs, it produces related ($\alpha$ equal)
outputs.
First, we show that by polymorphically defining the operational semantics over
an abstract interface, we can use parametricity ({\anyrel} translation) to obtain the
proof for free.\footnote{Our interface abstracts over both named and de Bruijn style 
variable bindings, and thus we were able to use parametricity ({\anyrel}
translation) to also obtain the proof that the big-step operational semantics is
preserved when changing the representation from de Bruijn indices to
named-variable representation.}
Then we show that using our {\isorel} translation, we also obtain for free that
a notion of observational equality, which is an undecidable relation, respects
$\alpha$ equality. The {\anyrel} translation and the translations in previous works
produce useless abstraction theorems for this polymorphic
proposition.The language in our example is the simply typed lambda calculus with natural
numbers.

\coqdocnoindent
\begin{minipage}[t]{0.4\textwidth}
\coqdocnoindent
\coqdockw{Variables} (\coqdef{Top.paper.Squiggle4.Tm}{Tm}{\coqdocvariable{Tm}} \coqdef{Top.paper.Squiggle4.BTm}{BTm}{\coqdocvariable{BTm}} : \coqdockw{Set}).\coqdoceol
\coqdocnoindent
\coqdockw{Variable} \coqdef{Top.paper.Squiggle4.applyBtm}{applyBtm}{\coqdocvariable{applyBtm}}: \coqdocvariable{BTm} \coqexternalref{:type scope:x '->' x}{http://coq.inria.fr/distrib/8.5pl3/stdlib/Coq.Init.Logic}{\coqdocnotation{\ensuremath{\rightarrow}}} \coqdocvariable{Tm} \coqexternalref{:type scope:x '->' x}{http://coq.inria.fr/distrib/8.5pl3/stdlib/Coq.Init.Logic}{\coqdocnotation{\ensuremath{\rightarrow}}} \coqdocvariable{Tm}.\coqdoceol
\coqdockw{Inductive} \coqdef{Top.paper.TmKind}{TmKind}{\coqdocinductive{TmKind}} :=\coqdoceol
\coqdocnoindent
\ensuremath{|} \coqdef{Top.paper.elam}{elam}{\coqdocconstructor{elam}} (\coqdocvar{bt}: \coqdocvariable{BTm}) \coqdoceol
\coqdocnoindent
\ensuremath{|} \coqdef{Top.paper.eapp}{eapp}{\coqdocconstructor{eapp}} (\coqdocvar{f} \coqdocvar{a}: \coqdocvariable{Tm}) \coqdoceol
\coqdocnoindent
\ensuremath{|} \coqdef{Top.paper.enum}{enum}{\coqdocconstructor{enum}} (\coqdocvar{n}: \coqref{Top.paper.nat}{\coqdocinductive{nat}})\coqdoceol
\coqdocnoindent
\ensuremath{|} \coqdef{Top.paper.evar}{evar}{\coqdocconstructor{evar}}.\coqdoceol
\coqdocnoindent
\coqdockw{Variable} \coqdef{Top.paper.Squiggle4.tmKind}{tmKind}{\coqdocvariable{tmKind}}:  \coqdocvariable{Tm} \coqexternalref{:type scope:x '->' x}{http://coq.inria.fr/distrib/8.5pl3/stdlib/Coq.Init.Logic}{\coqdocnotation{\ensuremath{\rightarrow}}} \coqref{Top.paper.TmKind}{\coqdocinductive{TmKind}}.\coqdoceol
\coqdocemptyline
\coqdocemptyline
\end{minipage}
\begin{minipage}[t]{0.6\textwidth}
The interface has two type variables: \coqRefVar{Top.paper.Squiggle4}{Tm} for
the type of terms and \coqRefVar{Top.paper.Squiggle4}{BTm} for the type of bound
terms~\citep[Sec. 2]{Howe1989}. 
In the $\lambda$ term $\lambda x. t$, $(x,t)$
can be considered a bound term. Bound terms only support the
\coqRefDefn{Top.paper.Squiggle4}{applyBtm} operation.
\coqRefDefn{Top.paper.Squiggle4}{applyBtm} $(x,t)$ $u$
represents $t[u/x]$. 
To define big step evaluation,
given a term
(\coqRefVar{Top.paper.Squiggle4}{Tm}),
we need to figure out what kind of a term it is: a $\lambda$, an
application, a number, or a variable.
The \coqRefDefn{Top.paper.Squiggle4}{tmKind} operation does just that.
It also allows limited access to subterms of a term.\coqdoceol
\coqdocemptyline
\coqdocemptyline
\end{minipage}
Note that the interface never allows
direct access to variables and can be instantiated even with de Bruijn terms and
de Bruijn substitution (for \coqRefDefn{Top.paper.Squiggle4}{applyBtm}).
Now, as shown in \figref{fig:app:evalObs}, we can polymorphically define not
only the big-step evaluation semantics
(\coqRefDefn{Top.paper.Squiggle4}{evaln}), but also a notion of observational
equivalence (\coqRefDefn{Top.paper.Squiggle4}{obseq}).

\ptranslate{\coqRefDefn{Top.paper.Squiggle4}{evaln}} is a proof that
on related instantiations of the above interface,
on related inputs \coqRefDefn{Top.paper.Squiggle4}{evaln} produces related
outputs. Given two concrete implementations of lambda terms and bound terms, say
\coqdocdefinition{LTm} and \coqdocdefinition{LBTm}, we instantiate 
\coqRefVar{Top.paper.Squiggle4}{Tm}, \tprime{\coqRefVar{Top.paper.Squiggle4}{Tm}}:= \coqdocdefinition{LTm};
\coqRefVar{Top.paper.Squiggle4}{BTm}, \tprime{\coqRefVar{Top.paper.Squiggle4}{BTm}}:= \coqdocdefinition{LBTm};
\trel{\coqRefVar{Top.paper.Squiggle4}{Tm}}, \trel{\coqRefVar{Top.paper.Squiggle4}{BTm}} := $\alpha$ equality.
This instantiation of \ptranslate{\coqRefDefn{Top.paper.Squiggle4}{evaln}} is a proof that on $\alpha$ equal
inputs, \coqRefDefn{Top.paper.Squiggle4}{evaln} produces $\alpha$ equal outputs.

In contrast, \ptranslate{\coqRefDefn{Top.paper.Squiggle4}{obseq}} is useless: it is a proof that on related
inputs, there is a relation between the output propositions
(\secref{sec:uniformProp}).
However, the {\isorel} translation also produces a proof that the relation has
the \coqRefDefn{Top.paper}{IffProps} (and the \coqRefDefn{Top.paper}{CompleteRel}) property, 
which means that the two
propositions are \emph{logically equivalent}. 
\appref{appendix:applications} shows the types of 
\ptranslate{\coqRefDefn{Top.paper.Squiggle4}{obseq}}, \ptranslateIso{\coqRefDefn{Top.paper.Squiggle4}{obseq}}, and \ptranslateStrongIso{\coqRefDefn{Top.paper.Squiggle4}{obseq}}. 

\ptranslateIso{\coqRefDefn{Top.paper.Squiggle4}{obsEq}} requires the relation 
\trel{\coqRefVar{Top.paper.Squiggle4}{Tm}}
to have the \coqRefDefn{Top.paper}{OneToOne} (and the
\coqRefDefn{Top.paper}{Total}) property.
This is problematic because the chosen relation, which is $\alpha$-equality,
is not \coqRefDefn{Top.paper}{OneToOne}: it is coarser than syntactic equality (=).
Fortunately, the 2$^{nd}$ stage
(\secref{sec:isorel:unused}) finds out that the
\coqRefDefn{Top.paper}{OneToOne} assumption is unused, and removes it.
Note that the definition of \coqRefDefn{Top.paper.Squiggle4}{obsEq} has
a universal quantification over the type \coqRefVar{Top.paper.Squiggle4}{Tm}.
Thus, \ptranslateIso{\coqRefDefn{Top.paper.Squiggle4}{obsEq}} (and \ptranslateIso{\coqRefDefn{Top.paper.Squiggle4}{obseq}}) 
\emph{does} use the assumption \coqRefDefn{Top.paper}{Total} \trel{\coqRefVar{Top.paper.Squiggle4}{Tm}} (\lemref{lemma:piIff}).
Fortunately, it is easy to prove that
$\alpha$ equality is a \coqRefDefn{Top.paper}{Total} relation, as are all
reflexive relations.

For closed $t$:$T$, unlike \ptranslate{$t$} and \ptranslateIso{$t$}, the type of
\ptranslateStrongIso{$t$} depends not only on $T$, but also on $t$: the potency
of the 2$^{nd}$ stage (\secref{sec:isorel:unused}) depends on how a proposition is expressed, not just its type.
For example, we could have expressed \coqRefDefn{Top.paper.Squiggle4}{obsEq} as an indexed-inductive proposition, where
the two arguments of type \coqRefVar{Top.paper.Squiggle4}{Tm} would be indices.
However, \ptranslateIso{\coqRefDefn{Top.paper.Squiggle4}{obsEq}} would then actually use 
\coqRefDefn{Top.paper}{OneToOne} \trel{\coqRefVar{Top.paper.Squiggle4}{Tm}}
because \coqRefVar{Top.paper.Squiggle4}{Tm} would be an index type (\corref{corr:iwp}).
Thus, the 2$^{nd}$ stage would then fail to remove it. 
The tables in Appendix~\ref{appendix:table}, which summarize the assumptions
of our uniformity lemmas, may be useful while trying to express a proposition in ways to get stronger abstraction theorems
from \ptranslateStrongIso{}.

\begin{figure}
\coqdocnoindent
\begin{minipage}[t]{0.49\textwidth}
\coqdockw{Fixpoint} \coqdef{Top.paper.evaln}{evaln}{\coqdocdefinition{evaln}} (\coqdocvar{n}:\coqref{Top.paper.nat}{\coqdocinductive{nat}}) (\coqdocvar{t}:\coqdocvariable{Tm}): \coqref{Top.paper.option}{\coqdocinductive{option}} \coqdocvariable{Tm} :=\coqdoceol
\coqdocnoindent
\coqdockw{match} \coqdocvariable{n} \coqdockw{with}\coqdoceol
\coqdocnoindent
\ensuremath{|} \coqref{Top.paper.O}{\coqdocconstructor{O}} \ensuremath{\Rightarrow} \coqref{Top.paper.None}{\coqdocconstructor{None}} \ensuremath{|} \coqref{Top.paper.S}{\coqdocconstructor{S}} \coqdocvar{n} \ensuremath{\Rightarrow} \coqdoceol
\coqdocindent{1.00em}
\coqdockw{match} (\coqdocvariable{tmKind} \coqdocvariable{t}) \coqdockw{with}\coqdoceol
\coqdocindent{1.00em}
\ensuremath{|} \coqref{Top.paper.evar}{\coqdocconstructor{evar}} \ensuremath{\Rightarrow} {\CoqNone}\coqdoceol
\coqdocindent{1.00em}
\ensuremath{|} \coqref{Top.paper.elam}{\coqdocconstructor{elam}} \coqdocvar{\_} \ensuremath{|} \coqref{Top.paper.enum}{\coqdocconstructor{enum}} \coqdocvar{\_} \ensuremath{\Rightarrow} \coqref{Top.paper.Some}{\coqdocconstructor{Some}} \coqdocvariable{t}\coqdoceol
\coqdocindent{1.00em}
\ensuremath{|} \coqref{Top.paper.eapp}{\coqdocconstructor{eapp}} \coqdocvar{f} \coqdocvar{a} \ensuremath{\Rightarrow}\coqdoceol
\coqdocindent{2.00em}
\coqdockw{match} \coqref{Top.paper.evaln}{\coqdocdefinition{evaln}} \coqdocvariable{n} \coqdocvar{f}, \coqref{Top.paper.evaln}{\coqdocdefinition{evaln}} \coqdocvariable{n} \coqdocvar{a} \coqdockw{with}\coqdoceol
\coqdocindent{2.00em}
\ensuremath{|} \coqref{Top.paper.Some}{\coqdocconstructor{Some}} \coqdocvar{f}, \coqref{Top.paper.Some}{\coqdocconstructor{Some}} \coqdocvar{a} \ensuremath{\Rightarrow}\coqdoceol
\coqdocindent{3.00em}
\coqdockw{match} (\coqdocvariable{tmKind} \coqdocvar{f}) \coqdockw{with}\coqdoceol
\coqdocindent{3.00em}
\ensuremath{|} \coqref{Top.paper.elam}{\coqdocconstructor{elam}} \coqdocvar{bt} \ensuremath{\Rightarrow} 
\coqref{Top.paper.evaln}{\coqdocdefinition{evaln}} \coqdocvariable{n} (\coqdocvariable{applyBtm} \coqdocvar{bt} \coqdocvar{a})\coqdoceol
\coqdocindent{3.00em}
\ensuremath{|} \coqdocvar{\_} \ensuremath{\Rightarrow} \coqref{Top.paper.None}{\coqdocconstructor{None}}\coqdoceol
\coqdocindent{3.00em}
\coqdockw{end}\coqdoceol
\coqdocindent{2.00em}
\ensuremath{|} \coqdocvar{\_},\coqdocvar{\_} \ensuremath{\Rightarrow} \coqref{Top.paper.None}{\coqdocconstructor{None}}\coqdoceol
\coqdocindent{2.00em}
\coqdockw{end}\coqdoceol
\coqdocindent{1.00em}
\coqdockw{end}\coqdoceol
\coqdocnoindent
\coqdockw{end}.\coqdoceol
\end{minipage}
\begin{minipage}[t]{0.49\textwidth}
\coqdockw{Definition} \coqdef{Top.paper.divergesIff}{divergesIff}{\coqdocdefinition{divergesIff}} (\coqdocvar{tl} \coqdocvar{tr}:\coqdocvariable{Tm}) : \coqdockw{Prop} :=\coqdoceol
\coqdocindent{0.0em}
\coqexternalref{:type scope:x '<->' x}{http://coq.inria.fr/distrib/8.5pl3/stdlib/Coq.Init.Logic}{\coqdocnotation{(}}\coqdockw{\ensuremath{\forall}} (\coqdocvar{nsteps}:\coqref{Top.paper.nat}{\coqdocinductive{nat}}), \coqref{Top.paper.isNone}{\coqdocdefinition{isNone}} (\coqRefDefn{Top.paper.Squiggle4}{evaln} \coqdocvariable{nsteps} \coqdocvariable{tl}) \coqref{Top.paper.:type scope:x '=' x}{\coqdocnotation{=}} \coqexternalref{true}{http://coq.inria.fr/distrib/8.5pl3/stdlib/Coq.Init.Datatypes}{\coqdocconstructor{true}}\coqexternalref{:type scope:x '<->' x}{http://coq.inria.fr/distrib/8.5pl3/stdlib/Coq.Init.Logic}{\coqdocnotation{)}} \coqdoceol
\coqdocindent{0.0em}
\coqexternalref{:type scope:x '<->' x}{http://coq.inria.fr/distrib/8.5pl3/stdlib/Coq.Init.Logic}{\coqdocnotation{\ensuremath{\leftrightarrow}}} 
\coqdockw{\ensuremath{\forall}} (\coqdocvar{nsteps}:\coqref{Top.paper.nat}{\coqdocinductive{nat}}), \coqref{Top.paper.isNone}{\coqdocdefinition{isNone}} (\coqRefDefn{Top.paper.Squiggle4}{evaln} \coqdocvariable{nsteps} \coqdocvariable{tr})\coqref{Top.paper.:type scope:x '=' x}{\coqdocnotation{=}} \coqexternalref{true}{http://coq.inria.fr/distrib/8.5pl3/stdlib/Coq.Init.Datatypes}{\coqdocconstructor{true}}.\coqdoceol
\coqdocemptyline
\coqdocnoindent
\coqdockw{Fixpoint} \coqdef{Top.paper.obsEq}{obsEq}{\coqdocdefinition{obsEq}} (\coqdocvar{k}:\coqref{Top.paper.nat}{\coqdocinductive{nat}}) 
(\coqdocvar{tl} \coqdocvar{tr}:\coqdocvariable{Tm}) : \coqdockw{Prop} :=\coqdoceol
\coqdocnoindent
\coqref{Top.paper.divergesIff}{\coqdocdefinition{divergesIff}} \coqdocvariable{tl} \coqdocvariable{tr} \coqexternalref{:type scope:x '/x5C' x}{http://coq.inria.fr/distrib/8.5pl3/stdlib/Coq.Init.Logic}{\coqdocnotation{\ensuremath{\land}}} \coqdockw{\ensuremath{\forall}} (\coqdocvar{nsteps}:\coqref{Top.paper.nat}{\coqdocinductive{nat}}), \coqdoceol
\coqdocnoindent
\coqdockw{match} \coqdocvariable{k} \coqdockw{with} \ensuremath{|} \coqref{Top.paper.O}{\coqdocconstructor{O}} \ensuremath{\Rightarrow} {\CoqTrue} \ensuremath{|} \coqref{Top.paper.S}{\coqdocconstructor{S}} \coqdocvar{k} \ensuremath{\Rightarrow}\coqdoceol
\coqdocindent{0.50em}
\coqdockw{match} \coqRefDefn{Top.paper.Squiggle4}{evaln} \coqdocvariable{nsteps} \coqdocvariable{tl}, \coqRefDefn{Top.paper.Squiggle4}{evaln} \coqdocvariable{nsteps} \coqdocvariable{tr} \coqdockw{with}\coqdoceol
\coqdocindent{0.50em}
\ensuremath{|}\coqref{Top.paper.Some}{\coqdocconstructor{Some}} \coqdocvar{vl}, \coqref{Top.paper.Some}{\coqdocconstructor{Some}} \coqdocvar{vr} \ensuremath{\Rightarrow} \coqdoceol
\coqdocindent{1.00em}
\coqdockw{match} \coqdocvariable{tmKind} \coqdocvar{vl}, \coqdocvariable{tmKind} \coqdocvar{vr} \coqdockw{with}\coqdoceol
\coqdocindent{1.00em}
\ensuremath{|} \coqref{Top.paper.enum}{\coqdocconstructor{enum}} \coqdocvar{nl} , \coqref{Top.paper.enum}{\coqdocconstructor{enum}} \coqdocvar{nr} \ensuremath{\Rightarrow} \coqdocvar{nl} \coqref{Top.paper.:type scope:x '=' x}{\coqdocnotation{=}} \coqdocvar{nr}\coqdoceol
\coqdocindent{1.00em}
\ensuremath{|} \coqref{Top.paper.elam}{\coqdocconstructor{elam}} \coqdocvar{btl} , \coqref{Top.paper.elam}{\coqdocconstructor{elam}} \coqdocvar{btr} \ensuremath{\Rightarrow} \coqdockw{\ensuremath{\forall}} (\coqdocvar{ta}: \coqdocvariable{Tm}), \coqdoceol
\coqdocindent{2.00em}
\coqref{Top.paper.obsEq}{\coqdocdefinition{obsEq}} \coqdocvariable{k} (\coqdocvariable{applyBtm} \coqdocvar{btl} \coqdocvariable{ta}) (\coqdocvariable{applyBtm} \coqdocvar{btr} \coqdocvariable{ta})\coqdoceol
\coqdocindent{1.00em}
\ensuremath{|} \coqdocvar{\_},\coqdocvar{\_} \ensuremath{\Rightarrow} {\CoqFalse}\coqdoceol
\coqdocindent{1.00em}
\coqdockw{end}\coqdoceol
\coqdocindent{0.50em}
\ensuremath{|}\coqdocvar{\_}, \coqdocvar{\_}  \ensuremath{\Rightarrow} {\CoqTrue}
\coqdoceol\coqdocindent{0.50em}
\coqdockw{end}\coqdoceol
\coqdocnoindent
\coqdockw{end}.\coqdoceol
\end{minipage}
\coqdocemptyline
\coqdocindent{17em}
\coqdockw{Definition} \coqdef{Top.paper.Squiggle4.obseq}{obseq}{\coqdocdefinition{obseq}} (\coqdocvar{tl} \coqdocvar{tr}:\coqdocvariable{Tm}) := 
\coqdockw{\ensuremath{\forall}} (\coqdocvar{k}:\coqref{Top.paper.nat}{\coqdocinductive{nat}}), \coqref{Top.paper.obsEq}{\coqdocdefinition{obsEq}} \coqdocvariable{k} \coqdocvariable{tl} \coqdocvariable{tr}.\coqdoceol
\caption{Left: big-step evaluation with fuel \coqdocvar{n}. Right: observational
equivalence.
}
\label{fig:app:evalObs}
\end{figure}

\section{Related Work and Conclusion}

The idea of globally enforcing that parametricity relations satisfy some
desirable properties was inspired by
\citet{Krishnaswami.Dreyer2013}: they globally enforce a \emph{zigzag-completeness} property.
However, that property is unrelated to our work which enforces
\coqRefDefn{Top.paper}{Total} and  \coqRefDefn{Top.paper}{OneToOne} properties.


For applications described in \secref{sec:app}, Coq developments typically
employ rewriting~\cite{Sozeau2010} and other proof-search mechanisms. For
example, \citet{Cohen.Denes.ea2013} use a library of proof search hints 
to semi-automatically refine algorithms (e.g. Strassen’s
matrix product) from simple data structures to complex but efficient ones.
Proof search mechanisms have less well-defined semantics and reliability
properties.
Our translation directly produces fully elaborated proof terms.
It is more automatic and preserves the
meaning of polymorphic propositions.


\citet{Zimmermann.Herbelin2015} built a Coq plugin to transfer
theorems across isomorphisms. Instead of using proof-search mechanisms,
they structurally recurse over the statement of the to-be-transferred theorem.
However, they consider a smaller class of propositions.
Inductively defined
propositions were not considered. Also, propositions produced by
pattern-matching (e.g. \coqRefDefn{Top.paper}{obsEq} in \secref{sec:app}) were not considered.

Transfer tools also exist for other proof assistants such as
Isabelle/HOL~\cite{Huffman.Kuncar2013}. However, our problem is more general
because HOL doesn't have dependent types.

Several works \cite{Krishnaswami.Dreyer2013, Atkey.Ghani.ea2014, Bernardy.Coquand.ea2015a} have constructed \emph{meta-theoretic} parametric models of variants of dependent type theory. Such models
may be useful in proving the consistency of such type theories and justify various useful extensions.
Our focus is not on the consistency of Coq or justifying extensions to Coq.
Like \citet{Keller.Lasson2012}, 
our translation produces proofs \emph{expressed in} Coq (Gallina) that are useful (\secref{sec:app}) without needing any extension to Coq.

There is one approach that is even more general than our \emph{weak} {\isorel} translation (\ptranslateIso{}):
Homotopy Type Theory (HoTT) \cite{TheUnivalentFoundationsProgram2013} is an area
of active research. 
It aims to serve as a foundation for full-fledged proof
assistants like Coq. The main advantage of HoTT is that it validates the
\emph{univalence principle} which says that, isomorphic types (more generally, equivalent types), even those in higher universes, are
\emph{equal}.
Also, as usual, every function, including the ones that return propositions, produces
equal outputs on equal inputs. Equal propositions are, of course, logically
equivalent!
Thus, HoTT may be able to provide some of the benefits that our weak {\isorel} translation
provides in Coq.
However, those benefits come at a cost.
\secref{sec:isorel:limitations} explained that univalence refutes UIP and why that can be problematic in Coq.

As mentioned before, our \emph{strong} {\isorel} translation (\ptranslateStrongIso{}) does not always require the two
instantiations to be isomorphic.
In \secref{sec:uniformProp}, we saw examples 
where one or \emph{both} of the \coqRefDefn{Top.paper}{Total} and 
\coqRefDefn{Top.paper}{OneToOne} assumptions are not needed.
In contrast, to use univalence to conclude that two types are equal, one needs
to always provide an isomorphism (more generally, an
equivalence~\cite[Sec.~4]{TheUnivalentFoundationsProgram2013}).
Thus, even in HoTT, a version of our \emph{strong} {\isorel} translation may be useful. Also, there it may be able to work for \emph{all} universes (\secref{sec:isorel:limitations}).

\paragraph{Conclusion}
We presented a new parametricity translation for a significant fragment of Coq.
Unlike the existing translations, it ensures that
parametrically related propositions are logically equivalent.
This allows us to obtain free proofs that polymorphic propositions behave uniformly.

Our goal was to develop a principled way to get free Coq proofs for our compiler-verification project.
We believe that our translation would be useful in many other application domains as well. 
%
Our implementation and test-suite are publically available on
Github:\\ \url{https://github.com/aa755/paramcoq-iff}.

\begin{acks}                            
  We thank Marc Lasson for help with understanding his and Chantal Keller's
  implementation of the paramcoq plugin, in particular the proof obligations that it generates.
  We also thank the anonymous ICFP 2017 reviewers for their detailed and constructive feedback. 

  This material is based upon work supported by the
  \grantsponsor{GS100000001}{National Science
    Foundation}{http://dx.doi.org/10.13039/100000001}. 
  Any opinions, findings, and
  conclusions or recommendations expressed in this material are those
  of the authors and do not necessarily reflect the views of the
  National Science Foundation.
\end{acks}

\bibliography{anand.bib}

\clearpage
\appendix
\section{Appendix}


\subsection{Universe-polymorphic inductive types are problematic for $\hat{\coqdockw{Set}}$ := \coqdockw{Prop}} 
\label{appendix:anyrel:univPoly}
Consider the following universe polymorphic inductive type:

\coqdocnoindent
\coqdockw{Inductive} \coqdef{Top.paper.list}{list@}{\coqdocinductive{list@}}\{\coqdocvar{i}\} (\coqdocvar{A} : \coqdocvar{Type@}\{\coqdocvar{i}\}) : \coqdocvar{Type@}\{\coqdocvar{i}\} :=  \coqdoceol
\coqdocnoindent
\ensuremath{|} \coqdef{Top.paper.nil}{nil}{\coqdocconstructor{nil}} : \coqref{Top.paper.list}{\coqdocinductive{list}} \coqdocvar{A} \coqdoceol
\coqdocnoindent
\ensuremath{|} \coqdef{Top.paper.cons}{cons}{\coqdocconstructor{cons}} : \coqdocvar{A} \coqexternalref{:type scope:x '->' x}{http://coq.inria.fr/distrib/8.5pl3/stdlib/Coq.Init.Logic}{\coqdocnotation{\ensuremath{\rightarrow}}} \coqref{Top.paper.list}{\coqdocinductive{list}} \coqdocvar{A} \coqexternalref{:type scope:x '->' x}{http://coq.inria.fr/distrib/8.5pl3/stdlib/Coq.Init.Logic}{\coqdocnotation{\ensuremath{\rightarrow}}} \coqref{Top.paper.list}{\coqdocinductive{list}} \coqdocvar{A}.\coqdoceol
\coqdocnoindent
The definition can be considered quantified over the universe \coqdocvar{i}.
The {\anyrel} parametricity translation of \coqdocinductive{list}, 
whether in the inductive style or in the deductive style, would also have to be polymorphic.
Recall from \secref{sec:anyrel} that we have
$\hat{{\CoqTypeSet}}$ := \coqdockw{Prop} and $\hat{s}$ := $s$ otherwise.
To the best of our knowledge, Coq's syntax for universe polymorphism is too restrictive to allow a definition like
the following (see the type of \coqdocvarR{A}):

\coqdocnoindent
\coqdockw{Fixpoint} \coqdef{Top.paper.list R}{list\_R@}{\coqdocdefinition{list\_R@}}\{\coqdocvar{i}\} (\coqdocvar{A}:
\coqdockw{Type@}\{\coqdocvar{i}\}) (\coqdocvar{A₂} : \coqdockw{Type@}\{\coqdocvar{i}\}) 
(\coqdocvarR{A} : \coqdocvariable{A} {\coqdocnotation{\ensuremath{\rightarrow}}} 
\coqdocvariable{A₂} {\coqdocnotation{\ensuremath{\rightarrow}}} 
if \coqdocvar{i} is 0 then \coqdockw{Prop} else \coqdockw{Type@}\{\coqdocvar{i}\})
$\hdots$ 
\coqdoceol
\coqdocnoindent
An inductive-style translation would also suffer from the same problem.
The problem doesn't arise if we choose $\hat{s}$~:=~$s$ for every universe.
Then, the type of \coqdocvarR{A} would simply be
\coqdocvariable{A} {\coqdocnotation{\ensuremath{\rightarrow}}} 
\coqdocvariable{A₂} {\coqdocnotation{\ensuremath{\rightarrow}}} 
\coqdockw{Type@}\{\coqdocvar{i}\}.


\subsection{Deductive-style {\anyrel} translation of inductive types: the
general case}
\label{appendix:anyrel:ded:ind}
Consider a general inductive type \coqdocinductive{$T$} of the form:

\coqdocnoindent
\coqdockw{Inductive} \coqdocinductive{$T$} (\coqdocvar{$\indexedVar{p}{1}$}: $\indexedVar{P}{1}$) $\hdots$
(\coqdocvar{$\indexedVar{p}{n}$}: $\indexedVar{P}{n}$) : 
$\forall$
(\coqdocvar{$\indexedVar{i}{1}$}: $\indexedVar{I}{1}$)
$\hdots$ 
(\coqdocvar{$\indexedVar{i}{k}$}: $\indexedVar{I}{k}$), $s$ 
:=\coqdoceol
\coqdocnoindent
$|$ \coqdocconstructor{$\indexedVar{c}{1}$} : $\indexedVar{C}{1}$
\coqdoceol\coqdocnoindent
$\vdots$
\coqdoceol\coqdocnoindent
$|$ \coqdocconstructor{$\indexedVar{c}{m}$} : $\indexedVar{C}{m}$.\coqdoceol

\coqdocnoindent
Recall from \secref{sec:anyrel:core} that $s$ denotes a universe
(\coqdockw{Prop} or \coqdockw{Type}$_i$).
Now we describe the deductive-style translation of \coqdocinductive{$T$}. 
First, we define the corresponding generalized equality proposition, as explained in \secref{sec:anyrel:ind}:\\

\coqdocnoindent
\coqdockw{Inductive} \coqdef{Top.indicesEq}{indicesEq}{\coqdocinductive{$T$\_indicesEq}}
\coqdoceol
(\coqdocvar{$\indexedVar{p}{1}$}: $\indexedVar{P}{1}$)
(\coqdocvarP{$\indexedVar{p}{1}$}: \tprime{$\indexedVar{P}{1}$})
(\coqdocvarR{$\indexedVar{p}{1}$}: \ptranslate{$\indexedVar{P}{1}$} \coqdocvar{$\indexedVar{p}{1}$} \coqdocvarP{$\indexedVar{p}{1}$})
$\hdots$
(\coqdocvar{$\indexedVar{p}{n}$}: $\indexedVar{P}{n}$)
(\coqdocvarP{$\indexedVar{p}{n}$}: \tprime{$\indexedVar{P}{n}$})
(\coqdocvarR{$\indexedVar{p}{n}$}: \ptranslate{$\indexedVar{P}{n}$} \coqdocvar{$\indexedVar{p}{n}$}
\coqdocvarP{$\indexedVar{p}{n}$})\coqdoceol
(\coqdocvar{$\indexedVar{i}{1}$}: $\indexedVar{I}{1}$) 
(\coqdocvarP{$\indexedVar{i}{1}$}: \tprime{$\indexedVar{I}{1}$}) 
(\coqdocvarR{$\indexedVar{i}{1}$}: \ptranslate{$\indexedVar{I}{1}$} \coqdocvar{$\indexedVar{i}{1}$} \coqdocvarP{$\indexedVar{i}{1}$}) 
$\hdots$ 
(\coqdocvar{$\indexedVar{i}{k}$}: $\indexedVar{I}{k}$) 
(\coqdocvarP{$\indexedVar{i}{k}$}: \tprime{$\indexedVar{I}{k}$}) 
(\coqdocvarR{$\indexedVar{i}{k}$}: \ptranslate{$\indexedVar{I}{k}$} \coqdocvar{$\indexedVar{i}{k}$} \coqdocvarP{$\indexedVar{i}{k}$})
:\coqdoceol
$\forall$
(\coqdocvarR{$\indexedVar{i}{1}$}\coqdocvar{$'$}: \ptranslate{$\indexedVar{I}{1}$} \coqdocvar{$\indexedVar{i}{1}$} \coqdocvarP{$\indexedVar{i}{1}$})
$\hdots$
(\coqdocvarR{$\indexedVar{i}{k}$}\coqdocvar{$'$}: \ptranslate{$\indexedVar{I}{k}$} \coqdocvar{$\indexedVar{i}{k}$} \coqdocvarP{$\indexedVar{i}{k}$})
: \coqdockw{Prop}
:=
\coqdoceol\coqdocnoindent
$|$ \coqdocconstructor{$T$\_refl} : \coqdocinductive{$T$\_indicesEq} 
\coqdocvar{$\indexedVar{p}{1}$}
\coqdocvarP{$\indexedVar{p}{1}$}
\coqdocvarR{$\indexedVar{p}{1}$}
$\hdots$
\coqdocvar{$\indexedVar{p}{n}$}
\coqdocvarP{$\indexedVar{p}{n}$}
\coqdocvarR{$\indexedVar{p}{n}$}
$\,$
\coqdocvar{$\indexedVar{i}{1}$}
\coqdocvarP{$\indexedVar{i}{1}$}
\coqdocvarR{$\indexedVar{i}{1}$}
$\hdots$
\coqdocvar{$\indexedVar{i}{k}$}
\coqdocvarP{$\indexedVar{i}{k}$}
\coqdocvarR{$\indexedVar{i}{k}$}
$\,$
\coqdocvarR{$\indexedVar{i}{1}$}\coqdocvar{$'$}
$\hdots$
\coqdocvarR{$\indexedVar{i}{k}$}\coqdocvar{$'$}.

\coqdocnoindent
In the future, instead of generating one such inductive proposition for each inductive type, 
we plan to have only one for each class of inductives that have the same number of indices (e.g. \coqdocinductive{$T$} has $k$ indices).

Let \coqdocvar{t} be a variable of the first class~(\secref{sec:anyrel:core})
such that \coqdocvar{t} is distinct from any variable in the above definitions.
Now, we can define the {\anyrel} relation for the above inductive type (\coqdocinductive{$T$}):

\begin{minipage}[t]{\textwidth}
\coqdocnoindent
\coqdockw{Fixpoint} \coqdocdefinition{\trel{$T$}}
(\coqdocvar{$\indexedVar{p}{1}$}: $\indexedVar{P}{1}$)
(\coqdocvarP{$\indexedVar{p}{1}$}: \tprime{$\indexedVar{P}{1}$})
(\coqdocvarR{$\indexedVar{p}{1}$}: \ptranslate{$\indexedVar{P}{1}$} \coqdocvar{$\indexedVar{p}{1}$} \coqdocvarP{$\indexedVar{p}{1}$})
$\hdots$
(\coqdocvar{$\indexedVar{p}{n}$}: $\indexedVar{P}{n}$)
(\coqdocvarP{$\indexedVar{p}{n}$}: \tprime{$\indexedVar{P}{n}$})
(\coqdocvarR{$\indexedVar{p}{n}$}: \ptranslate{$\indexedVar{P}{n}$} \coqdocvar{$\indexedVar{p}{n}$}
\coqdocvarP{$\indexedVar{p}{n}$})
\coqdoceol\coqdocindent{1em}
(\coqdocvar{$\indexedVar{i}{1}$}: $\indexedVar{I}{1}$) 
(\coqdocvarP{$\indexedVar{i}{1}$}: \tprime{$\indexedVar{I}{1}$}) 
(\coqdocvarR{$\indexedVar{i}{1}$}: \ptranslate{$\indexedVar{I}{1}$} \coqdocvar{$\indexedVar{i}{1}$} \coqdocvarP{$\indexedVar{i}{1}$}) 
$\hdots$ 
(\coqdocvar{$\indexedVar{i}{k}$}: $\indexedVar{I}{k}$) 
(\coqdocvarP{$\indexedVar{i}{k}$}: \tprime{$\indexedVar{I}{k}$}) 
(\coqdocvarR{$\indexedVar{i}{k}$}: \ptranslate{$\indexedVar{I}{k}$} \coqdocvar{$\indexedVar{i}{k}$} \coqdocvarP{$\indexedVar{i}{k}$}) 
\coqdoceol\coqdocindent{1em}
(\coqdocvar{t} : \coqdocinductive{$T$}
\coqdocvar{$\indexedVar{p}{1}$} $\hdots$ \coqdocvar{$\indexedVar{p}{n}$} $\,$
\coqdocvar{$\indexedVar{i}{1}$} $\hdots$ \coqdocvar{$\indexedVar{i}{k}$}) 
(\coqdocvarP{t} : \coqdocinductive{$T$}
\coqdocvarP{$\indexedVar{p}{1}$} $\hdots$ \coqdocvarP{$\indexedVar{p}{n}$} $\,$
\coqdocvarP{$\indexedVar{i}{1}$} $\hdots$ \coqdocvarP{$\indexedVar{i}{k}$})
\{\coqdockw{struct} \coqdocvar{t}\} 
: $\hat{s}$ :=\coqdoceol\coqdocnoindent
\coqdockw{match} \coqdocvar{t} \coqdockw{in}
\coqdocinductive{$T$}
\coqdocvar{\_} $\hdots$ \coqdocvar{\_} $\,$
\coqdocvar{$\indexedVar{i}{1}$} $\hdots$ \coqdocvar{$\indexedVar{i}{k}$}
\coqdockw{return} $Ret_{out}$ \coqdockw{with}
\coqdoceol\coqdocnoindent
$\vdots$
\coqdoceol\coqdocnoindent
$|$ \coqdocconstructor{$\indexedVar{c}{u}$} \coqdocvar{$\indexedVar{a}{1}$} $\hdots$ \coqdocvar{$\indexedVar{a}{l}$} 
$\Rightarrow$
\coqdoceol\coqdocindent{1.5em}
\coqdockw{match} \coqdocvarP{t} \coqdockw{in}
\coqdocinductive{$T$}
\coqdocvar{\_} $\hdots$ \coqdocvar{\_} $\,$
\coqdocvarP{$\indexedVar{i}{1}$} $\hdots$ \coqdocvarP{$\indexedVar{i}{k}$}
\coqdockw{return} $Ret_{in}$ \coqdockw{with}
\coqdoceol\coqdocindent{1.5em}
$\vdots$
\coqdoceol\coqdocindent{1.5em}
$|$ \coqdocconstructor{$\indexedVar{c}{u}$} \coqdocvarP{$\indexedVar{a}{1}$} $\hdots$ \coqdocvarP{$\indexedVar{a}{l}$} 
$\Rightarrow$ $\lambda$ \coqdocvarR{$\indexedVar{i}{1}$} $\hdots$ \coqdocvarR{$\indexedVar{i}{k}$},
\coqdoceol\coqdocindent{3em}
\{\coqdocvarR{$\indexedVar{a}{1}$} : \ptranslate{$\indexedVar{A}{1}$} \coqdocvar{$\indexedVar{a}{1}$} \coqdocvarP{$\indexedVar{a}{1}$} \& \{ $\hdots$
\&
\{
\coqdocvarR{$\indexedVar{a}{l}$} : \ptranslate{$\indexedVar{A}{l}$} \coqdocvar{$\indexedVar{a}{l}$} \coqdocvarP{$\indexedVar{a}{l}$}
\&
\coqdoceol\coqdocindent{4em}
\coqref{Top.indicesEq}{\coqdocinductive{$T$\_indicesEq}}
\coqdocvar{$\indexedVar{p}{1}$}
\coqdocvarP{$\indexedVar{p}{1}$}
\coqdocvarR{$\indexedVar{p}{1}$}
$\hdots$
\coqdocvar{$\indexedVar{p}{n}$}
\coqdocvarP{$\indexedVar{p}{n}$}
\coqdocvarR{$\indexedVar{p}{n}$}
$\,$
$CI1$
\tprime{$CI1$}
\ptranslate{$CI1$}
$\hdots$
$CIk$
\tprime{$CIk$}
\ptranslate{$CIk$}
$\,$
\coqdocvarR{$\indexedVar{i}{1}$}
$\hdots$
\coqdocvarR{$\indexedVar{i}{k}$}
\}
\coqdoceol\coqdocindent{1.5em}
$\vdots$
\coqdoceol\coqdocindent{1.5em}
$|$ \coqdocconstructor{$\indexedVar{c}{v}$} $\hdots$ 
$\Rightarrow$ $\lambda$ \coqdocvarR{$\indexedVar{i}{1}$} $\hdots$ \coqdocvarR{$\indexedVar{i}{k}$}, {\CoqFalse} 
\coqdoceol\coqdocindent{1.5em}
$\vdots$
\coqdoceol\coqdocindent{1.5em}
\coqdockw{end}
\coqdoceol\coqdocnoindent
$\vdots$
\coqdoceol\coqdocnoindent
\coqdockw{end} \coqdocvarR{$\indexedVar{i}{1}$} $\hdots$ \coqdocvarR{$\indexedVar{i}{k}$}.
\coqdoceol\coqdocnoindent
\end{minipage}

\coqdocnoindent
In the above, $1 \le {\ensuremath{u}} \le m$, $1 \le {\ensuremath{v}} \le m$, and 
${\ensuremath{u}}\neq {\ensuremath{v}}$. 
Also, $Cu$, the type declaration for \coqdocconstructor{cu} in the definition of \coqdocinductive{$T$} is:\\
$\forall$ 
(\coqdocvar{$\indexedVar{a}{1}$} : {$\indexedVar{A}{1}$})
$\hdots$
(\coqdocvar{$\indexedVar{a}{l}$} : {$\indexedVar{A}{l}$}),
\coqdocinductive{$T$} 
\coqdocvar{$\indexedVar{p}{1}$}
$\hdots$
\coqdocvar{$\indexedVar{p}{n}$}
$CI1$
$\hdots$
$CIk$.\\
$Ret_{out}$ is $\forall$ 
(\coqdocvarR{$\indexedVar{i}{1}$}: \ptranslate{$\indexedVar{I}{1}$} \coqdocvar{$\indexedVar{i}{1}$} \coqdocvarP{$\indexedVar{i}{1}$}) 
$\hdots$ 
(\coqdocvarR{$\indexedVar{i}{k}$}: \ptranslate{$\indexedVar{I}{k}$} \coqdocvar{$\indexedVar{i}{k}$} \coqdocvarP{$\indexedVar{i}{k}$}),
$\hat{s}$.\\
$Ret_{in}$ is a refined version of $Ret_{out}$, where the variables \coqdocvar{$\indexedVar{i}{1}$},$\hdots$,\coqdocvar{$\indexedVar{i}{k}$}
are respectively substituted with  
$CI1$,
$\hdots$,
$CIk$.\\

The translation of constructors (e.g. \coqdocconstructor{cu}) of \coqdocinductive{$T$} is straightforward.
Note that the abstraction theorem (\thmref{Abstraction}) already determines the type of the result of the translation.
The result of translating a constructor is a function
that packages some of its arguments into dependent pairs whose types were shown in the above definition.
The innermost member of such dependent pairs is always the canonical proof of the corresponding generalized equality proposition. 
For example, for the constructors of \coqdocinductive{$T$}, 
the innermost member is always of the form 
\coqref{Top.indicesEq}{\coqdocconstructor{$T$\_refl}} $\hdots$.

In the case of mutual inductive definitions, we produce mutually recursive functions.

\subsection{Deductive-style {\anyrel} translation of pattern matching on
inductive types:
the general case}
\label{appendix:anyrel:ded:match}
Now we will see how to translate a pattern match on 
a discriminee of the inductive type \coqdocinductive{$T$} defined above (\secref{appendix:anyrel:ded:ind}).
Consider a term\\
$m$ := \coqdoceol
\coqdocnoindent
\coqdockw{match} ($d$:$D$) \coqdockw{as} \coqdocvar{t} \coqdockw{in} \coqdocinductive{$T$}
\coqdocvar{\_} $\hdots$ \coqdocvar{\_} $\,$
\coqdocvar{$\indexedVar{i}{1}$} $\hdots$ \coqdocvar{$\indexedVar{i}{k}$}
\coqdockw{return} $R$ \coqdockw{with}
\coqdocnoindent
\coqdoceol\coqdocnoindent
$\vdots$
\coqdoceol\coqdocnoindent
$|$ \coqdocconstructor{$\indexedVar{c}{u}$} \coqdocvar{$\indexedVar{a}{1}$} $\hdots$ \coqdocvar{$\indexedVar{a}{l}$} $\Rightarrow$
$\indexedVar{b}{u}$ \coqdoceol\coqdocnoindent
$\vdots$
\coqdoceol\coqdocnoindent
\coqdockw{end}.
\coqdoceol

\coqdocnoindent
Note that the discriminee $d$ has type $D$. 
In Coq, $d$ need not be a variable.
In the representation of terms in Coq's kernel, the type of discriminee is not stored.
Our reifier computes that type and includes it in the reified terms.
Below, we will see that $D$ is needed in the translation.
Intuitively, the translation uses \ptranslate{$d$}, the translation of the discriminee.
Note that \ptranslate{$d$}:\ptranslate{$D$} $d$ \tprime{$d$}.
$D$ must be of the form \coqdocinductive{$T$} $\indexedVar{dP}{1}$ $\hdots$ $\indexedVar{dP}{n}$
$\indexedVar{dI}{1}$ $\hdots$ $\indexedVar{dI}{k}$.

Recall~\citep[Sec 8.2]{Chlipala2011} that in $m$, the return type $R$ can mention the variables
\coqdocvar{$\indexedVar{i}{1}$} $\hdots$ \coqdocvar{$\indexedVar{i}{k}$} and \coqdocvar{t}.
(Also, those variables are bound only in $R$.)
In other words, the return type of a \coqdockw{match} is a function of the 
discriminee and the indices of the (co-)inductive type of the discriminee.
While checking each branch, Coq substitutes those variables in $R$ to values corresponding to the constructor of the branch.
For example, $\indexedVar{b}{u}$ must be of type:\\
$R$ [ $\indexedVar{CI}{1}$ / \coqdocvar{$\indexedVar{i}{1}$}, $\hdots$,  
$\indexedVar{CI}{k}$ / \coqdocvar{$\indexedVar{i}{k}$}, 
(\coqdocconstructor{$\indexedVar{c}{u}$} 
$\indexedVar{dP}{1}$ $\hdots$ $\indexedVar{dP}{n}$
\coqdocvar{$\indexedVar{a}{1}$} $\hdots$ \coqdocvar{$\indexedVar{a}{l}$}
) / \coqdocvar{t}
]

\coqdocnoindent
The translation of the \coqdockw{match} term $m$ shown above is:

\coqdocnoindent
\begin{minipage}[t]{\textwidth}
\coqdocnoindent
\coqdockw{match} $d$ \coqdockw{as} \coqdocvar{t} \coqdockw{in} \coqdocinductive{$T$}
\coqdocvar{\_} $\hdots$ \coqdocvar{\_} $\,$
\coqdocvar{$\indexedVar{i}{1}$} $\hdots$ \coqdocvar{$\indexedVar{i}{k}$}
\coqdockw{return} $Ret_{out}$ \coqdockw{with}
\coqdocnoindent
\coqdoceol\coqdocnoindent
$\vdots$
\coqdoceol\coqdocnoindent
$|$ \coqdocconstructor{$\indexedVar{c}{u}$} \coqdocvar{$\indexedVar{a}{1}$} $\hdots$ \coqdocvar{$\indexedVar{a}{l}$} $\Rightarrow$
\coqdoceol\coqdocindent{1.5em}
\coqdockw{match} \tprime{$d$} \coqdockw{as} \coqdocvarP{t} \coqdockw{in} \coqdocinductive{$T$}
\coqdocvar{\_} $\hdots$ \coqdocvar{\_} $\,$
\coqdocvarP{$\indexedVar{i}{1}$} $\hdots$ \coqdocvarP{$\indexedVar{i}{k}$}
\coqdockw{return} $Ret_{in}$ \coqdockw{with}
\coqdoceol\coqdocindent{1.5em}
$\vdots$
\coqdoceol\coqdocindent{1.5em}
$|$ \coqdocconstructor{$\indexedVar{c}{u}$} \coqdocvarP{$\indexedVar{a}{1}$} $\hdots$ \coqdocvarP{$\indexedVar{a}{l}$} 
$\Rightarrow$ $\lambda$ \coqdocvarR{$\indexedVar{i}{1}$} $\hdots$ \coqdocvarR{$\indexedVar{i}{k}$} {\coqdocvarR{t}},
\coqdoceol\coqdocindent{3em}
\coqdockw{match} \coqdocvarR{t} \coqdockw{in} $\hdots$ \coqdockw{return} $\hdots$ \coqdockw{with} 
\coqdoceol\coqdocindent{3em}
$|$ {\CoqExistT} \coqdocvarR{\indexedVar{a}{1}} \coqdocvarR{t} $\Rightarrow$
\coqdoceol\coqdocindent{4em}
$\ddots$
\coqdoceol\coqdocindent{5em}
\coqdockw{match} \coqdocvarR{t} \coqdockw{in} $\hdots$ \coqdockw{return} $\hdots$ \coqdockw{with} 
\coqdoceol\coqdocindent{5em}
$|$ {\CoqExistT} \coqdocvarR{\indexedVar{a}{l}} \coqdocvar{pdeq} $\Rightarrow$
\coqdoceol\coqdocindent{6em}
\coqdockw{match} \coqdocvar{pdeq} \coqdockw{as} $\hdots$ \coqdockw{in} $\hdots$ \coqdockw{return} $\hdots$ \coqdockw{with} 
\coqdoceol\coqdocindent{6em}
$|$ \coqref{Top.indicesEq}{\coqdocconstructor{$T$\_refl}} $\Rightarrow$ \ptranslate{$\indexedVar{b}{u}$}
\coqdoceol\coqdocindent{6em}
\coqdockw{end}
\coqdoceol\coqdocindent{5em}
\coqdockw{end}
\coqdoceol\coqdocindent{4em}
$\iddots$
\coqdoceol\coqdocindent{3em}
\coqdockw{end}
\coqdoceol\coqdocindent{1.5em}
$\vdots$
\coqdoceol\coqdocindent{1.5em}
$|$ \coqdocconstructor{$\indexedVar{c}{v}$} $\hdots$ 
$\Rightarrow$ $\lambda$ \coqdocvarR{$\indexedVar{i}{1}$} $\hdots$ \coqdocvarR{$\indexedVar{i}{k}$}
\coqdocvarR{$t$}, {\CoqFalseRect} \_ $\;$ {\coqdocvarR{t}}
\coqdoceol\coqdocindent{1.5em}
$\vdots$
\coqdoceol\coqdocindent{1.5em}
\coqdockw{end}
\coqdoceol\coqdocnoindent
$\vdots$
\coqdoceol\coqdocnoindent
\coqdockw{end} \ptranslate{$\indexedVar{dI}{1}$} $\hdots$ \ptranslate{$\indexedVar{dI}{k}$} \ptranslate{d}.
\coqdoceol
\coqdocemptyline
\coqdocemptyline
\coqdocemptyline
\end{minipage}
Next, we describe the terms $Ret_{out}$ and $Ret_{in}$ mentioned in the above definition.
Given these, it should be easy to figure out the return types (of the inner \coqdockw{match} terms) that have been denoted by $\hdots$ for brevity.
Also, our implementation (as a Coq function) is publically available in a Github repository (\url{https://github.com/aa755/paramcoq-iff}).\\
First we define the term $ma$ which is obtained by replacing the discriminee $d$ in $m$ by the variable \coqdocvar{t}:

\coqdocnoindent
\begin{minipage}[t]{\textwidth}
$ma$ := \coqdoceol
\coqdocnoindent
\coqdockw{match} \coqdocvar{t} \coqdockw{as} \coqdocvar{t} \coqdockw{in} \coqdocinductive{$T$}
\coqdocvar{\_} $\hdots$ \coqdocvar{\_} $\,$
\coqdocvar{$\indexedVar{i}{1}$} $\hdots$ \coqdocvar{$\indexedVar{i}{k}$}
\coqdockw{return} $R$ \coqdockw{with}
\coqdocnoindent
\coqdoceol\coqdocnoindent
$\vdots$
\coqdoceol\coqdocnoindent
$|$ \coqdocconstructor{$\indexedVar{c}{u}$} \coqdocvar{$\indexedVar{a}{1}$} $\hdots$ \coqdocvar{$\indexedVar{a}{l}$} $\Rightarrow$
$\indexedVar{b}{u}$ \coqdoceol\coqdocnoindent
$\vdots$
\coqdoceol\coqdocnoindent
\coqdockw{end}.
\coqdoceol
\coqdocemptyline
\coqdocemptyline
\coqdocemptyline
\end{minipage}

\coqdocnoindent
Note that in the above definition of the term $ma$,
the occurrence of \coqdocvar{t} after \coqdockw{as} is a bound variable and not substitutable.
The occurrence at the position of discriminee is substitutable. 
\coqdocnoindent
$Ret_{out}$ and $Ret_{in}$ are obtained by performing substitutions in the following term:\\ 
$Ret$ := $\forall$ 
(\coqdocvarR{$\indexedVar{i}{1}$}: \ptranslate{$\indexedVar{I}{1}$} \coqdocvar{$\indexedVar{i}{1}$} \coqdocvarP{$\indexedVar{i}{1}$}) 
$\hdots$ 
(\coqdocvarR{$\indexedVar{i}{k}$}: \ptranslate{$\indexedVar{I}{k}$} \coqdocvar{$\indexedVar{i}{k}$} \coqdocvarP{$\indexedVar{i}{k}$})
(\coqdocvarR{$t$}: \ptranslate{$D$} \coqdocvar{$t$} \coqdocvarP{$t$}),
\ptranslate{R} $ma$ \tprime{$ma$} .\\
Now, we can define $Ret_{out}$ and $Ret_{in}$ as follows:\\
$Ret_{out}$ :=  $Ret$ [
\tprime{$\indexedVar{dI}{1}$} / \coqdocvarP{$\indexedVar{i}{1}$},
$\hdots$,
\tprime{$\indexedVar{dI}{k}$} / \coqdocvarP{$\indexedVar{i}{k}$},
\tprime{$d$} / \coqdocvarP{$t$}
]
\coqdoceol\coqdocnoindent
$Ret_{in}$ :=  $Ret$ [ $\indexedVar{CI}{1}$ / \coqdocvar{$\indexedVar{i}{1}$}, $\hdots$,  
$\indexedVar{CI}{k}$ / \coqdocvar{$\indexedVar{i}{k}$}, 
(\coqdocconstructor{$\indexedVar{c}{u}$} 
$\indexedVar{dP}{1}$ $\hdots$ $\indexedVar{dP}{n}$
\coqdocvar{$\indexedVar{a}{1}$} $\hdots$ \coqdocvar{$\indexedVar{a}{l}$}
) / \coqdocvar{t}
]

\subsection{{\anyrel} translation of fixpoints}
\label{appendix:anyrel:fix}

Our translation of \coqdockw{fix} (or \coqdockw{Fixpoint}) terms is largely 
as described by \citet{Keller.Lasson2012}. Roughly speaking,  
\ptranslate{\coqdockw{fix} $F$} is just \coqdockw{fix} \ptranslate{$F$}.
The translation of \coqdockw{fix} terms
depends a tiny bit on how the inductives are translated. 
Unlike in Agda, each \coqdockw{fix} term in Coq has a designated \coqdockw{struct} argument of an inductive type. 
Coq requires that any recursive call should be made on a structural subterm of the
\coqdockw{struct} argument. 
Coq can often infer the \coqdockw{struct} argument and in this paper, 
we have usually omitted the annotations stating the \coqdockw{struct} argument.
In \coqref{Top.paper.DedV.Vec R}{\coqdocdefinition{\trel{Vec}}}~(\secref{sec:anyrel:ind}), which is the
deductive-style translation of the type
\coqRefInductive{Top.paper}{Vec}, \coqdocvar{v} is the \coqdockw{struct}
argument.
Suppose we are translating \coqdockw{fix} $F$, where $F$ is
of the form $\lambda \hdots$ (\coqdocvar{v}:\coqdocinductive{$I$}) $\hdots$,
$\hdots$.
Suppose the \coqdockw{struct} argument is \coqdocvar{v}.
If \coqdocinductive{$I$} was translated in inductive style (e.g. when
\coqdocinductive{$I$} is a proposition), we must pick \coqdocvarR{v} as the
\coqdockw{struct} argument in the translation of \coqdockw{fix} $F$.
Coq guarantees that $F$ only makes recursive calls on subterms
of \coqdocvar{v}, which are obtained by pattern matching on \coqdocvar{v}.
In the inductive-style translation of $F$, those matches will be translated
to pattern matching on \coqdocvarR{v}.
In contrast, if 
\coqdocinductive{$I$} was translated in deductive style (e.g. when
\coqdocinductive{$I$} is a type), those matches will be translated into
matches on \coqdocvar{v} and \coqdocvar{\tprime{v}}. Thus we can choose either \coqdocvar{v} or
\coqdocvar{\tprime{v}} as the \coqdockw{struct} argument. We choose  \coqdocvar{v}.

A problem not mentioned in the literature, but partially addressed in the
implementation by \citet{Keller.Lasson2012}, is that the translation 
of \coqdockw{fix} $F$ needs to generate unfolding equations of the form
\coqdockw{fix} $F$ = $F$ (\coqdockw{fix} $F$).
For some pathological programs, these equations are
\emph{unprovable}.

Marc Lasson gave us the following example where the unfolding equation is unprovable:

\coqdocnoindent
\coqdockw{Fixpoint} \coqdef{Top.paper.zero}{zero}{\coqdocdefinition{zero}} (\coqdocvar{A} : \coqdockw{Type}) (\coqdocvar{x} : \coqdocvariable{A}) (\coqdocvar{p} : \coqdocvariable{x} \coqexternalref{:type scope:x '=' x}{http://coq.inria.fr/distrib/8.5pl3/stdlib/Coq.Init.Logic}{\coqdocnotation{=}} \coqdocvariable{x}) \{\coqdockw{struct} \coqdocvar{p}\}:= 0.\coqdoceol
\coqdocnoindent
To ensure strong normalization, a \coqdockw{fix} term only reduces (unfolds) when the \coqdockw{struct} argument is in head normal form.
In the definition above, it is impossible to prove that \coqdocvar{p} is equal to something in the head normal form~\cite{Hofmann.Streicher1998}.

\subsection{Locating proofs in the supplementary material}
\label{appendix:suppl}
For proofs in \secref{sec:uniformProp} and \secref{sec:uniformProp:type}, see the files Pi.v and IWTP.v in the supplementary material.

\subsubsection{The importance of the triviality property} 
\label{appendix:intro:triv}
See the admitted lemma in the file triviality.v in the supplementary material.

\subsection{1-phase strong {\isorel} translation}
\label{appendix:onephase}
As mentioned in \secref{sec:isorel}, it is natural to consider another design, where the main translation itself
determines the minimally needed assumptions on type variables (or variables
denoting type families) by, e.g., analysing the bodies of lambda terms, and
directly uses the appropriately minimal type for type variables.

Such a translation seems hard to implement for several reasons.
It would be non-compositional, while translating an application of 
some function $F$ to some type $T$,
we may need to prune the translation of $T$
depending on the translation of $F$.



Also, we are only interested in removing the top level arguments of an abstraction theorem.
It is not clear whether there is an advantage (disadvantage?) to removing the arguments
of $\lambda$ subterms that appear elsewhere.
%

\subsection{Abstraction theorems for {\protect\coqRefDefn{Top.paper.Squiggle4}{obseq}}}
\label{appendix:applications}

\subsubsection{\protect\ptranslate{{\coqdocdefinition{obseq}}}  ({\anyrel} translation)}
\coqdoceol\coqdocnoindent
\begin{minipage}[t]{\textwidth}
\coqdoceol
\coqdocnoindent
\coqdocnoindent
\coqdockw{\ensuremath{\forall}} (\coqdocvar{Tm} \coqdocvar{Tm₂} : \coqdockw{Set}) (\coqdocvar{Tmᵣ} : \coqdocvariable{Tm} \coqexternalref{:type scope:x '->' x}{http://coq.inria.fr/distrib/8.6/stdlib/Coq.Init.Logic}{\coqdocnotation{\ensuremath{\rightarrow}}} \coqdocvariable{Tm₂} \coqexternalref{:type scope:x '->' x}{http://coq.inria.fr/distrib/8.6/stdlib/Coq.Init.Logic}{\coqdocnotation{\ensuremath{\rightarrow}}} \coqdockw{Prop})
(\coqdocvar{BTm} \coqdocvar{BTm₂} : \coqdockw{Set}) (\coqdocvar{BTmᵣ} : \coqdocvariable{BTm} \coqexternalref{:type scope:x '->' x}{http://coq.inria.fr/distrib/8.6/stdlib/Coq.Init.Logic}{\coqdocnotation{\ensuremath{\rightarrow}}} \coqdocvariable{BTm₂} \coqexternalref{:type scope:x '->' x}{http://coq.inria.fr/distrib/8.6/stdlib/Coq.Init.Logic}{\coqdocnotation{\ensuremath{\rightarrow}}} \coqdockw{Prop})\coqdoceol
\coqdocindent{1.00em}
(\coqdocvar{applyBtm} : \coqdocvariable{BTm} \coqexternalref{:type scope:x '->' x}{http://coq.inria.fr/distrib/8.6/stdlib/Coq.Init.Logic}{\coqdocnotation{\ensuremath{\rightarrow}}} \coqdocvariable{Tm} \coqexternalref{:type scope:x '->' x}{http://coq.inria.fr/distrib/8.6/stdlib/Coq.Init.Logic}{\coqdocnotation{\ensuremath{\rightarrow}}} \coqdocvariable{Tm})
(\coqdocvar{applyBtm₂} : \coqdocvariable{BTm₂} \coqexternalref{:type scope:x '->' x}{http://coq.inria.fr/distrib/8.6/stdlib/Coq.Init.Logic}{\coqdocnotation{\ensuremath{\rightarrow}}} \coqdocvariable{Tm₂} \coqexternalref{:type scope:x '->' x}{http://coq.inria.fr/distrib/8.6/stdlib/Coq.Init.Logic}{\coqdocnotation{\ensuremath{\rightarrow}}} \coqdocvariable{Tm₂})\coqdoceol
\coqdocindent{1.00em}
(\coqdocvar{applyBtmᵣ} : \coqdockw{\ensuremath{\forall}} (\coqdocvar{b} : \coqdocvariable{BTm}) (\coqdocvar{b₂} : \coqdocvariable{BTm₂}) (\coqdocvar{bᵣ}: \coqdocvariable{BTm}\coqdocvariable{ᵣ} \coqdocvariable{b} \coqdocvariable{b₂})
(\coqdocvar{a} : \coqdocvariable{Tm}) (\coqdocvar{a₂} : \coqdocvariable{Tm₂}) (\coqdocvar{aᵣ} : \coqdocvariable{Tm}\coqdocvariable{ᵣ} \coqdocvariable{a} \coqdocvariable{a₂}),\coqdoceol
\coqdocindent{8.00em}
\coqdocvariable{Tm}\coqdocvariable{ᵣ} (\coqdocvariable{applyBtm} \coqdocvariable{b} \coqdocvariable{a}) (\coqdocvariable{applyBtm₂} \coqdocvariable{b₂} \coqdocvariable{a₂}))\coqdoceol
\coqdocindent{1.00em}
(\coqdocvar{tmKind} : \coqdocvariable{Tm} \coqexternalref{:type scope:x '->' x}{http://coq.inria.fr/distrib/8.6/stdlib/Coq.Init.Logic}{\coqdocnotation{\ensuremath{\rightarrow}}} \coqref{Top.paper.TmKind}{\coqdocinductive{TmKind}} \coqdocvariable{Tm} \coqdocvariable{BTm})
(\coqdocvar{tmKind₂} : \coqdocvariable{Tm₂} \coqexternalref{:type scope:x '->' x}{http://coq.inria.fr/distrib/8.6/stdlib/Coq.Init.Logic}{\coqdocnotation{\ensuremath{\rightarrow}}} \coqref{Top.paper.TmKind}{\coqdocinductive{TmKind}} \coqdocvariable{Tm₂} \coqdocvariable{BTm₂})\coqdoceol
\coqdocindent{1.00em}
(\coqdocvar{tmKindᵣ} : \coqdockw{\ensuremath{\forall}} (\coqdocvar{a} : \coqdocvariable{Tm}) (\coqdocvar{a₂} : \coqdocvariable{Tm₂}) (\coqdocvar{aᵣ}: \coqdocvariable{Tm}\coqdocvariable{ᵣ} \coqdocvariable{a} \coqdocvariable{a₂}),\coqdoceol
\coqdocindent{8.00em}
{\trel{\coqdocdefinition{TmKind}}} \coqdocvariable{Tm} \coqdocvariable{Tm₂} \coqdocvariable{Tm}\coqdocvariable{ᵣ} \coqdocvariable{BTm} \coqdocvariable{BTm₂} \coqdocvariable{BTm}\coqdocvariable{ᵣ}
(\coqdocvariable{tmKind} \coqdocvariable{a})
(\coqdocvariable{tmKind₂} \coqdocvariable{a₂}))\coqdoceol
\coqdocindent{1.00em}
(\coqdocvar{tl} : \coqdocvariable{Tm}) (\coqdocvar{tl₂} : \coqdocvariable{Tm₂}) (\coqdocvar{tlᵣ}: \coqdocvariable{Tm}\coqdocvariable{ᵣ} \coqdocvariable{tl} \coqdocvariable{tl₂})
(\coqdocvar{tr} : \coqdocvariable{Tm}) (\coqdocvar{tr₂} : \coqdocvariable{Tm₂}) (\coqdocvar{trᵣ}: \coqdocvariable{Tm}\coqdocvariable{ᵣ} \coqdocvariable{tr} \coqdocvariable{tr₂}),\coqdoceol
\coqdocindent{2.50em}
(\coqref{Top.paper.Squiggle4.obseq}{\coqdocdefinition{obseq}} \coqdocvariable{Tm} \coqdocvariable{BTm} \coqdocvariable{applyBtm} \coqdocvariable{tmKind} \coqdocvariable{tl} \coqdocvariable{tr})
\coqdoceol\coqdocindent{2.50em}
$\rightarrow$
(\coqref{Top.paper.Squiggle4.obseq}{\coqdocdefinition{obseq}} \coqdocvariable{Tm₂} \coqdocvariable{BTm₂} \coqdocvariable{applyBtm₂} \coqdocvariable{tmKind₂} \coqdocvariable{tl₂} \coqdocvariable{tr₂})
\coqdoceol\coqdocindent{2.50em}
$\rightarrow$\coqdockw{Prop}.\coqdoceol
\coqdocemptyline
\coqdocemptyline
\coqdocemptyline
\end{minipage}

\subsubsection{\protect\ptranslateIso{{\coqdocdefinition{obseq}}}  (weak {\isorel} translation)}
Recall that for any
relation \coqdocvar{R} between any two propositions \coqdocvar{A} and \coqdocvar{B}, 
\coqRefDefn{Top.paper}{Total} \coqdocvar{R} is logically equivalent to 
(\coqRefDefn{Top.paper}{IffProps} \coqdocvar{R} $\wedge$ \coqRefDefn{Top.paper}{CompleteRel} \coqdocvar{R}).

\coqdoceol\coqdocnoindent
\begin{minipage}[t]{\textwidth}
\coqdoceol
\coqdocnoindent
\coqdocnoindent
\coqdockw{\ensuremath{\forall}} (\coqdocvar{Tm} \coqdocvar{Tm₂} : \coqdockw{Set}) 
(\coqdocvar{Tmᵣ} : \coqRefDefn{Top.paper}{IsoRel} \coqdocvar{Tm} \coqdocvar{Tm₂})
(\coqdocvar{BTm} \coqdocvar{BTm₂} : \coqdockw{Set}) 
(\coqdocvar{BTmᵣ} : \coqRefDefn{Top.paper}{IsoRel} \coqdocvar{BTm} \coqdocvar{BTm₂})\coqdoceol
\coqdocindent{1.00em}
(\coqdocvar{applyBtm} : \coqdocvariable{BTm} \coqexternalref{:type scope:x '->' x}{http://coq.inria.fr/distrib/8.6/stdlib/Coq.Init.Logic}{\coqdocnotation{\ensuremath{\rightarrow}}} \coqdocvariable{Tm} \coqexternalref{:type scope:x '->' x}{http://coq.inria.fr/distrib/8.6/stdlib/Coq.Init.Logic}{\coqdocnotation{\ensuremath{\rightarrow}}} \coqdocvariable{Tm})
(\coqdocvar{applyBtm₂} : \coqdocvariable{BTm₂} \coqexternalref{:type scope:x '->' x}{http://coq.inria.fr/distrib/8.6/stdlib/Coq.Init.Logic}{\coqdocnotation{\ensuremath{\rightarrow}}} \coqdocvariable{Tm₂} \coqexternalref{:type scope:x '->' x}{http://coq.inria.fr/distrib/8.6/stdlib/Coq.Init.Logic}{\coqdocnotation{\ensuremath{\rightarrow}}} \coqdocvariable{Tm₂})\coqdoceol
\coqdocindent{1.00em}
(\coqdocvar{applyBtmᵣ} : \coqdockw{\ensuremath{\forall}} (\coqdocvar{b} : \coqdocvariable{BTm}) (\coqdocvar{b₂} : \coqdocvariable{BTm₂}) (\coqdocvar{bᵣ}: \coqdocvariable{BTm}\coqdocvariable{ᵣ} \coqdocvariable{b} \coqdocvariable{b₂})
(\coqdocvar{a} : \coqdocvariable{Tm}) (\coqdocvar{a₂} : \coqdocvariable{Tm₂}) (\coqdocvar{aᵣ} : {\CoqSigTProj} \coqdocvariable{Tm}\coqdocvariable{ᵣ} \coqdocvariable{a} \coqdocvariable{a₂}),\coqdoceol
\coqdocindent{8.00em}
{\CoqSigTProj} \coqdocvariable{Tm}\coqdocvariable{ᵣ} (\coqdocvariable{applyBtm} \coqdocvariable{b} \coqdocvariable{a}) (\coqdocvariable{applyBtm₂} \coqdocvariable{b₂} \coqdocvariable{a₂}))\coqdoceol
\coqdocindent{1.00em}
(\coqdocvar{tmKind} : \coqdocvariable{Tm} \coqexternalref{:type scope:x '->' x}{http://coq.inria.fr/distrib/8.6/stdlib/Coq.Init.Logic}{\coqdocnotation{\ensuremath{\rightarrow}}} \coqref{Top.paper.TmKind}{\coqdocinductive{TmKind}} \coqdocvariable{Tm} \coqdocvariable{BTm})
(\coqdocvar{tmKind₂} : \coqdocvariable{Tm₂} \coqexternalref{:type scope:x '->' x}{http://coq.inria.fr/distrib/8.6/stdlib/Coq.Init.Logic}{\coqdocnotation{\ensuremath{\rightarrow}}} \coqref{Top.paper.TmKind}{\coqdocinductive{TmKind}} \coqdocvariable{Tm₂} \coqdocvariable{BTm₂})\coqdoceol
\coqdocindent{1.00em}
(\coqdocvar{tmKindᵣ} : \coqdockw{\ensuremath{\forall}} (\coqdocvar{a} : \coqdocvariable{Tm}) (\coqdocvar{a₂} : \coqdocvariable{Tm₂}) (\coqdocvar{aᵣ}: {\CoqSigTProj} \coqdocvariable{Tm}\coqdocvariable{ᵣ} \coqdocvariable{a} \coqdocvariable{a₂}),\coqdoceol
\coqdocindent{8.00em}
{\trel{\coqdocdefinition{TmKind}}} \coqdocvariable{Tm} \coqdocvariable{Tm₂} ({\CoqSigTProj} \coqdocvariable{Tm}\coqdocvariable{ᵣ}) \coqdocvariable{BTm} \coqdocvariable{BTm₂} ({\CoqSigTProj} \coqdocvariable{BTm}\coqdocvariable{ᵣ})
(\coqdocvariable{tmKind} \coqdocvariable{a})
(\coqdocvariable{tmKind₂} \coqdocvariable{a₂}))\coqdoceol
\coqdocindent{1.00em}
(\coqdocvar{tl} : \coqdocvariable{Tm}) (\coqdocvar{tl₂} : \coqdocvariable{Tm₂}) (\coqdocvar{tlᵣ}: {\CoqSigTProj} \coqdocvariable{Tm}\coqdocvariable{ᵣ} \coqdocvariable{tl} \coqdocvariable{tl₂})
(\coqdocvar{tr} : \coqdocvariable{Tm}) (\coqdocvar{tr₂} : \coqdocvariable{Tm₂}) (\coqdocvar{trᵣ}: {\CoqSigTProj} \coqdocvariable{Tm}\coqdocvariable{ᵣ} \coqdocvariable{tr} \coqdocvariable{tr₂}),\coqdoceol
\coqdocindent{2.50em}
\coqref{Top.paper.Squiggle4.IsoRel}{\coqdocdefinition{IsoRel}} (\coqref{Top.paper.Squiggle4.obseq}{\coqdocdefinition{obseq}} \coqdocvariable{Tm} \coqdocvariable{BTm} \coqdocvariable{applyBtm} \coqdocvariable{tmKind} \coqdocvariable{tl} \coqdocvariable{tr})
(\coqref{Top.paper.Squiggle4.obseq}{\coqdocdefinition{obseq}} \coqdocvariable{Tm₂} \coqdocvariable{BTm₂} \coqdocvariable{applyBtm₂} \coqdocvariable{tmKind₂} \coqdocvariable{tl₂} \coqdocvariable{tr₂}).\coqdoceol
\coqdocemptyline
\coqdocemptyline
\coqdocemptyline
\end{minipage}

\subsubsection{\protect\ptranslateStrongIso{{\coqdocdefinition{obseq}}}  (strong {\isorel} translation)}
The conclusion is the same as before (\ptranslateIso{}) but 3 assumptions (\coqRefDefn{Top.paper}{OneToOne} \coqdocvarR{Tm}, \coqRefDefn{Top.paper}{Total} \coqdocvarR{BTm}, \coqRefDefn{Top.paper}{OneToOne} \coqdocvarR{BTm}) were removed by the second phase~(\secref{sec:isorel:unused}).
\coqdoceol\coqdocnoindent
\begin{minipage}[t]{\textwidth}
\coqdoceol
\coqdocnoindent
\coqdocnoindent
\coqdockw{\ensuremath{\forall}} (\coqdocvar{Tm} \coqdocvar{Tm₂} : \coqdockw{Set}) (\coqdocvar{Tmᵣ} : \coqdocvariable{Tm} \coqexternalref{:type scope:x '->' x}{http://coq.inria.fr/distrib/8.6/stdlib/Coq.Init.Logic}{\coqdocnotation{\ensuremath{\rightarrow}}} \coqdocvariable{Tm₂} \coqexternalref{:type scope:x '->' x}{http://coq.inria.fr/distrib/8.6/stdlib/Coq.Init.Logic}{\coqdocnotation{\ensuremath{\rightarrow}}} \coqdockw{Prop})
(\coqdocvar{Tmᵣ}\coqdocvar{tot}: \coqref{Top.paper.Total}{\coqdocdefinition{Total}} \coqdocvariable{Tm}\coqdocvariable{ᵣ})\coqdoceol
\coqdocindent{1.00em}
(\coqdocvar{BTm} \coqdocvar{BTm₂} : \coqdockw{Set}) (\coqdocvar{BTmᵣ} : \coqdocvariable{BTm} \coqexternalref{:type scope:x '->' x}{http://coq.inria.fr/distrib/8.6/stdlib/Coq.Init.Logic}{\coqdocnotation{\ensuremath{\rightarrow}}} \coqdocvariable{BTm₂} \coqexternalref{:type scope:x '->' x}{http://coq.inria.fr/distrib/8.6/stdlib/Coq.Init.Logic}{\coqdocnotation{\ensuremath{\rightarrow}}} \coqdockw{Prop})\coqdoceol
\coqdocindent{1.00em}
(\coqdocvar{applyBtm} : \coqdocvariable{BTm} \coqexternalref{:type scope:x '->' x}{http://coq.inria.fr/distrib/8.6/stdlib/Coq.Init.Logic}{\coqdocnotation{\ensuremath{\rightarrow}}} \coqdocvariable{Tm} \coqexternalref{:type scope:x '->' x}{http://coq.inria.fr/distrib/8.6/stdlib/Coq.Init.Logic}{\coqdocnotation{\ensuremath{\rightarrow}}} \coqdocvariable{Tm})
(\coqdocvar{applyBtm₂} : \coqdocvariable{BTm₂} \coqexternalref{:type scope:x '->' x}{http://coq.inria.fr/distrib/8.6/stdlib/Coq.Init.Logic}{\coqdocnotation{\ensuremath{\rightarrow}}} \coqdocvariable{Tm₂} \coqexternalref{:type scope:x '->' x}{http://coq.inria.fr/distrib/8.6/stdlib/Coq.Init.Logic}{\coqdocnotation{\ensuremath{\rightarrow}}} \coqdocvariable{Tm₂})\coqdoceol
\coqdocindent{1.00em}
(\coqdocvar{applyBtmᵣ} : \coqdockw{\ensuremath{\forall}} (\coqdocvar{b} : \coqdocvariable{BTm}) (\coqdocvar{b₂} : \coqdocvariable{BTm₂}) (\coqdocvar{bᵣ}: \coqdocvariable{BTm}\coqdocvariable{ᵣ} \coqdocvariable{b} \coqdocvariable{b₂})
(\coqdocvar{a} : \coqdocvariable{Tm}) (\coqdocvar{a₂} : \coqdocvariable{Tm₂}) (\coqdocvar{aᵣ} : \coqdocvariable{Tm}\coqdocvariable{ᵣ} \coqdocvariable{a} \coqdocvariable{a₂}),\coqdoceol
\coqdocindent{8.00em}
\coqdocvariable{Tm}\coqdocvariable{ᵣ} (\coqdocvariable{applyBtm} \coqdocvariable{b} \coqdocvariable{a}) (\coqdocvariable{applyBtm₂} \coqdocvariable{b₂} \coqdocvariable{a₂}))\coqdoceol
\coqdocindent{1.00em}
(\coqdocvar{tmKind} : \coqdocvariable{Tm} \coqexternalref{:type scope:x '->' x}{http://coq.inria.fr/distrib/8.6/stdlib/Coq.Init.Logic}{\coqdocnotation{\ensuremath{\rightarrow}}} \coqref{Top.paper.TmKind}{\coqdocinductive{TmKind}} \coqdocvariable{Tm} \coqdocvariable{BTm})
(\coqdocvar{tmKind₂} : \coqdocvariable{Tm₂} \coqexternalref{:type scope:x '->' x}{http://coq.inria.fr/distrib/8.6/stdlib/Coq.Init.Logic}{\coqdocnotation{\ensuremath{\rightarrow}}} \coqref{Top.paper.TmKind}{\coqdocinductive{TmKind}} \coqdocvariable{Tm₂} \coqdocvariable{BTm₂})\coqdoceol
\coqdocindent{1.00em}
(\coqdocvar{tmKindᵣ} : \coqdockw{\ensuremath{\forall}} (\coqdocvar{a} : \coqdocvariable{Tm}) (\coqdocvar{a₂} : \coqdocvariable{Tm₂}) (\coqdocvar{aᵣ}: \coqdocvariable{Tm}\coqdocvariable{ᵣ} \coqdocvariable{a} \coqdocvariable{a₂}),\coqdoceol
\coqdocindent{8.00em}
{\trel{\coqdocdefinition{TmKind}}} \coqdocvariable{Tm} \coqdocvariable{Tm₂} \coqdocvariable{Tm}\coqdocvariable{ᵣ} \coqdocvariable{BTm} \coqdocvariable{BTm₂} \coqdocvariable{BTm}\coqdocvariable{ᵣ}
(\coqdocvariable{tmKind} \coqdocvariable{a})
(\coqdocvariable{tmKind₂} \coqdocvariable{a₂}))\coqdoceol
\coqdocindent{1.00em}
(\coqdocvar{tl} : \coqdocvariable{Tm}) (\coqdocvar{tl₂} : \coqdocvariable{Tm₂}) (\coqdocvar{tlᵣ}: \coqdocvariable{Tm}\coqdocvariable{ᵣ} \coqdocvariable{tl} \coqdocvariable{tl₂})
(\coqdocvar{tr} : \coqdocvariable{Tm}) (\coqdocvar{tr₂} : \coqdocvariable{Tm₂}) (\coqdocvar{trᵣ}: \coqdocvariable{Tm}\coqdocvariable{ᵣ} \coqdocvariable{tr} \coqdocvariable{tr₂}),\coqdoceol
\coqdocindent{2.50em}
\coqref{Top.paper.Squiggle4.IsoRel}{\coqdocdefinition{IsoRel}} (\coqref{Top.paper.Squiggle4.obseq}{\coqdocdefinition{obseq}} \coqdocvariable{Tm} \coqdocvariable{BTm} \coqdocvariable{applyBtm} \coqdocvariable{tmKind} \coqdocvariable{tl} \coqdocvariable{tr})
(\coqref{Top.paper.Squiggle4.obseq}{\coqdocdefinition{obseq}} \coqdocvariable{Tm₂} \coqdocvariable{BTm₂} \coqdocvariable{applyBtm₂} \coqdocvariable{tmKind₂} \coqdocvariable{tl₂} \coqdocvariable{tr₂}).\coqdoceol
\coqdocemptyline
\coqdocemptyline
\coqdocemptyline
\end{minipage}

\subsection{Tabulation of assumptions in lemmas in Section \ref{sec:uniformProp} and \ref{sec:uniformProp:type}}
\label{appendix:table}
\subsubsection{Canonical Propositions}
Recall that for any
relation \coqdocvar{R} between any two propositions \coqdocvar{A} and \coqdocvar{B}, 
\coqRefDefn{Top.paper}{Total} \coqdocvar{R} is logically equivalent to 
(\coqRefDefn{Top.paper}{IffProps} \coqdocvar{R} $\wedge$ \coqRefDefn{Top.paper}{CompleteRel} \coqdocvar{R}).
Also, \coqRefDefn{Top.paper}{OneToOne} \coqdocvar{R} is a trivial consequence of  
{\coqexternalref{proof irrelevance}{https://coq.inria.fr/library/Coq.Logic.ProofIrrelevance}{proof irrelevance}}.
\paragraph{Universal Quantification}
($\forall$\coqdocvar{x}:$A$,$B$):\coqdockw{Prop}\coqdoceol
\begin{tabular}{| l | l | l  | l |}
\hline
proof of & assumptions on \ptranslate{A} & axioms & lemma\\
\hline
\coqRefDefn{Top.paper}{IffProps} \ptranslate{$\forall$\coqdocvar{x}:$A$,$B$} & \coqRefDefn{Top.paper}{Total} & & \ref{lemma:piIff}\\
\hline
\coqRefDefn{Top.paper}{CompleteRel} \ptranslate{$\forall$\coqdocvar{x}:$A$,$B$} &  & 
& \ref{lemma:piComplete}\\
\hline
\end{tabular}

\paragraph{Inductive propositions}
(\coqRefDefn{Top.paper}{IWP} $I$ $A$ $B$ $AI$ $BI$ $i$):\coqdockw{Prop}\coqdoceol
\begin{tabular}{| l | l | l | l | l | l | l |}
\hline
proof of & \multicolumn{3}{|c|}{assumptions on} & axioms & lemma\\
 & \ptranslate{I} & \ptranslate{A} & \ptranslate{B} &  & \\
\hline
\coqRefDefn{Top.paper}{IffProps} & \coqRefDefn{Top.paper}{OneToOne}& \coqRefDefn{Top.paper}{Total}& \coqRefDefn{Top.paper}{Total} & & \ref{corr:iwp}\\
\hline
\coqRefDefn{Top.paper}{CompleteRel} & \coqRefDefn{Top.paper}{OneToOne}& \coqRefDefn{Top.paper}{Total}& \coqRefDefn{Top.paper}{Total}  & 
{\coqexternalref{proof irrelevance}{https://coq.inria.fr/library/Coq.Logic.ProofIrrelevance}{proof irrelevance}} & \ref{corr:iwp}\\
\hline
\end{tabular}
\coqdocemptyline
\coqdoceol
\coqdocemptyline
\coqdocnoindent
For general inductive propositions, index types behave like (regarding the use of assumptions) $I$,
types of non-recursive arguments (\secref{sec:uniformProp:indt}) to constructors (except parameters) behave like $A$,
and the domain types in recursive arguments behave like $B$. 

\subsubsection{Canonical Types}
\paragraph{dependent function types}
($\forall$\coqdocvar{x}:$A$,$B$):\coqdockw{Set}\coqdoceol
\begin{tabular}{| l | l | l | l | l | l |}
\hline
proof of & \multicolumn{2}{|c|}{assumptions on} & axioms & lemma\\
  & \ptranslate{A} & \ptranslate{B} &  & \\
\hline
\coqRefDefn{Top.paper}{Total} & \coqRefDefn{Top.paper}{Total}, \coqRefDefn{Top.paper}{OneToOne} & \coqRefDefn{Top.paper}{Total} & 
{\coqexternalref{proof irrelevance}{https://coq.inria.fr/library/Coq.Logic.ProofIrrelevance}{proof irrelevance}}
 & \ref{lemma:piTot}\\
\hline
\coqRefDefn{Top.paper}{OneToOne} & \coqRefDefn{Top.paper}{Total}& \coqRefDefn{Top.paper}{OneToOne}  & 
{\coqexternalref{functional_extensionality_dep}{https://coq.inria.fr/library/Coq.Logic.FunctionalExtensionality}{function extensionality}} & \ref{lemma:piOne}\\
\hline
\end{tabular}
\coqdocemptyline
\coqdoceol
\coqdocemptyline
\coqdocnoindent

\paragraph{Inductive types}
(\coqRefDefn{Top.paper}{IWT} $I$ $A$ $B$ $AI$ $BI$ $i$):\coqdockw{Set}\coqdoceol
\begin{tabular}{| l | l | l | l | l | l | l |}
\hline
proof of & \multicolumn{3}{|c|}{assumptions on} & axioms & lemma\\
 & \ptranslate{I} & \ptranslate{A} & \ptranslate{B} &  & \\
\hline
\coqRefDefn{Top.paper}{Total} & \coqRefDefn{Top.paper}{OneToOne}& \coqRefDefn{Top.paper}{Total}& \coqRefDefn{Top.paper}{Total}, \coqRefDefn{Top.paper}{OneToOne} & 
{\coqexternalref{proof irrelevance}{https://coq.inria.fr/library/Coq.Logic.ProofIrrelevance}{proof irrelevance}} & \ref{lemma:iwtTot}\\
\hline
\coqRefDefn{Top.paper}{OneToOne} & & \coqRefDefn{Top.paper}{OneToOne}& \coqRefDefn{Top.paper}{Total}  & 
{\coqexternalref{proof irrelevance}{https://coq.inria.fr/library/Coq.Logic.ProofIrrelevance}{proof irrelevance}},
 & \ref{lemma:iwtOne}\\
& & & & {\coqexternalref{functional_extensionality_dep}{https://coq.inria.fr/library/Coq.Logic.FunctionalExtensionality}{function extensionality}} & \\
\hline
\end{tabular}
\coqdocemptyline
\coqdoceol
\coqdocemptyline
\coqdocnoindent

\section{Correctness of the weak {\isorel} translation}
\label{appendix:correctness}
In this section, we discuss a formal proof of correctness of \ptranslateIso{} for {\cocminus}, a
CoC-like core calculus.

\figref{fig:subtyping} (adapted from \cite{Keller.Lasson2012}) shows the subtyping rules of {\cocminus}.
W.r.t, CoC, the omissions are
 \coqdockw{Type}$_0$ :> \coqdockw{Type}$_1$ and 
 \coqdockw{Prop} :> \coqdockw{Set}.
Recall that \coqdockw{Type}$_0$ is written as \coqdockw{Set} in Coq.
We can add back the former rule if Coq gave us terms that make explicit all uses of that subtyping rule~(\secref{sec:isorel:limitations}).
For example, it would be sensible for a future version of Coq's typechecker to furnish a typing derivation for terms that it deems well-typed. 
The problem with the latter rule is explained below in \appref{appendix:correctness:subst:subtyping}.

\figref{fig:typing} (adapted from \cite{Keller.Lasson2012}) shows the typing rules of {\cocminus}.
$\equiv_{\Gamma}$ is essentially Coq's $\beta$-equivalence, except that it maintains an invariant (\appref{appendix:correctness:capture}) that prevents capture
during the translation.
The only omission (highlighted) is that when constructing a proposition using universal quantification, one can only quantify 
over types in \coqdockw{Set} or 
\coqdockw{Prop}.  
Our proof of the uniformity of universal quantification (\lemref{lemma:piIff}) needs the \coqRefDefn{Top.paper}{Total} property
for the relation of the quantified type. We were unable to systematically build that property for types in higher universes~(\secref{sec:isorel:limitations}, \secref{sec:uniformProp:type}). 

Recall that \ptranslateIso{} is implemented as a \emph{structurally recursive} function in Coq (Gallina). Its input is obtained by a reifier that translates the OCaml representation of Coq terms to 
a Coq datatype. We use the inverse operation (reflection) to declare the output \ptranslateIso{} in Coq's environment, but only after Coq typechecks the output
of reflection. (We use a monad to automate these steps.)

The grammar of {\cocminus} is essentially the grammar of CoC presented in \secref{sec:anyrel:core}, except that we make explicit some implementation details:
Recall (\secref{sec:isorel:weak}) that \ptranslateIso{} needs to make different choices depending on whether a type is in the universe \coqdockw{Set}, \coqdockw{Prop}
or \coqdockw{Type}$_i$ ($i > 0$). For example, 
it needs to pair the relations of types/propositions in
\coqdockw{Set} or \coqdockw{Prop} with proofs of \coqRefDefn{Top.paper}{Total} and \coqRefDefn{Top.paper}{OneToOne} properties.
As a result, at some places (e.g. \ptranslateIso{$\lambda \hdots$}), it needs to project the relations out of such pairs.
To ensure the simplicity of \ptranslateIso{}, we push the task of determining the universe of types to the reifier, which has access to Coq's typechecker.
The terms produced by the reifier has flags indicating the universe information wherever needed~(\secref{sec:isorel:weak}).
 
We make these flags explicit in the grammar of {\cocminus}: In ($\forall$ (\coqdocvar{x}:$A$), $B$) , we use a 2-letter subscript respectively 
denoting the universes of $A$ and $B$. The letters are: $S$ for \coqdockw{Set}, $P$ for \coqdockw{Prop}, and $T$ for \coqdockw{Type}$_i$ ($i > 0$).
For example, the syntax ($\forall_{SP}$ (\coqdocvar{x}:$A$), $B$) implies $A$:\coqdockw{Set} and $B$:\coqdockw{Prop}.
Similarly, ($\lambda_{S}$ (\coqdocvar{x}:$A$), $b$) implies $A$:\coqdockw{Set}.
We will omit the subscripts in contexts where they do not matter.

Unlike \ptranslate{}, even for terms in {\cocminus}, \ptranslateIso{} produces terms that are \emph{not} in {\cocminus} (not even in CoC).
For example, 
\ptranslateIso{\coqdockw{Set}}:=
{\coqdocnotation{\ensuremath{\lambda}}} {\coqdocnotation{(}}\coqdocvar{A} \coqdocvar{\tprime{A}}: \coqdockw{Set}
{\coqdocnotation{),}} 
\coqRefDefn{Top.paper}{IsoRel} \coqdocvar{A} \coqdocvar{\tprime{A}}. 
\coqRefDefn{Top.paper}{IsoRel} is defined using
$\Sigma$ types, which are missing in {\cocminus}.
Instead of defining an extended core calculus for interpreting the \emph{output} of \ptranslateIso{}, we take the luxury of interpreting 
it in Coq (CiC). 
Also, our translation invokes (transparent) lemmas proved in Coq. In the proofs in this section, we assume that the proof terms
corresponding to those lemmas indeed have the types proven in Coq.

\begin{figure*}
\begin{prooftree}

 \AXC{}
 \LeftLabel{$\highlight{0 < i} < j$}
 \RightLabel{\textsc{(Sub$_{T}$)}}
 \UIC{$\Type_i <: \Type_j$}

 \AXC{$A <: B$}
 \RightLabel{\textsc{(Sub$_\pi$)}}
 \UIC{$∀x:C.A <: ∀x:C.B$}
 
 \noLine
 \BIC{}
\end{prooftree}
\caption{Subtyping rules of {\cocminus}}
\label{fig:subtyping}
\end{figure*}

\begin{figure*}
\begin{prooftree}
 \AXC{}
 \RightLabel{\textsc{(Ax$_1$)}}
 \UIC{$⊢ \Prop : \Type_1$}

 \AXC{}
 \RightLabel{\textsc{(Ax$_2$)}}
 \UIC{$⊢ \Type_i : \Type_{i+1}$}

 \noLine
 \BIC{}
\end{prooftree}

\begin{prooftree}
 \AXC{$Γ ⊢ A:s$}
 \LeftLabel{$x \not\in Γ, s \in \mathcal{S}$}
 \RightLabel{\textsc{(Var)}}
 \UIC{$Γ,x:A ⊢ x:A$}

 \AXC{$Γ ⊢ B : C$}
 \AXC{$Γ ⊢ A:s$}
 \LeftLabel{$x \not\in Γ, s \in \mathcal{S}$}
 \RightLabel{\textsc{(Weak)}}
 \BIC{$Γ,x:A ⊢ B : C$}

 \noLine
 \BIC{}
\end{prooftree}

\begin{prooftree}
 \AXC{$Γ ⊢  A : C$}
 \AXC{$Γ ⊢  B : s$}
 \LeftLabel{$B ≡_Γ C, s \in \mathcal{S}$}
 \RightLabel{\textsc{(Conv)}}
 \BIC{$Γ ⊢  A : B$}

 \AXC{$Γ ⊢ A : B$}
 \LeftLabel{$B <: C$}
 \RightLabel{\textsc{(Cum)}}
 \UIC{$Γ ⊢ A : C$}

 \noLine
 \BIC{}
\end{prooftree}

\begin{prooftree}
  \AXC{$Γ ⊢  A : \Type_i$}
  \AXC{$Γ, x : A ⊢  B : \Type_i$}
  \RightLabel{\textsc{($∀_1$)}}
  \BIC{$Γ ⊢ ∀x:A.B : \Type_i$}
\end{prooftree}

\begin{prooftree}
  \AXC{$Γ ⊢  A : s $}
  \AXC{$Γ, x : A ⊢  B : \Prop$}
 \LeftLabel{$s \in \{\highlight{\Type_0, \coqdockw{Prop}}\}$}
  \RightLabel{\textsc{($∀_2$)}}
  \BIC{$Γ ⊢ ∀x:A.B : \Prop$}
\end{prooftree}

\begin{prooftree}
  \AXC{$Γ ⊢ M : ∀x:A.B$}
  \AXC{$Γ ⊢ N : A$}
  \RightLabel{\textsc{(App)}}
  \BIC{$Γ ⊢ M\,N : B[N/x]$}

  \AXC{$Γ, x : A ⊢  B : C$}
  \RightLabel{\textsc{(Abs)}}
  \UIC{$Γ ⊢ λx :A.B : ∀x:A.C$}
 \noLine
 \BIC{}
\end{prooftree}
\caption{typing rules of {\cocminus}. the set $\mathcal{S}$ contains all universes (\coqdockw{Prop} and \coqdockw{Type}$_i$ for all $i$)}
\label{fig:typing}
\end{figure*}

\subsection{Avoiding variable capture in parametricity translations}
\label{appendix:correctness:capture}
As mentioned before~(\secref{sec:isorel:correctness}), 
except for the construction of proofs of the \coqRefDefn{Top.paper}{Total} and \coqRefDefn{Top.paper}{OneToOne} properties, 
the correctness argument for \ptranslateIso{} is almost identical to the correctness argument for 
\ptranslate{}: one proves that the translation preserves substitution, then reduction, and finally typehood~\cite[Lemma 2, Theorem 1]{Keller.Lasson2012}. 
However, we needed to make some assumptions of those theorems explicit. (We have done some parts of the proof in Coq, just to increase confidence in our paper proof.) 
For example, Theorem 1 of \citet{Keller.Lasson2012} doesn't hold for input terms that have shadowed bound variables:
\ptranslate{$\lambda$ (\coqdocvar{x}: {\CoqNat}) (\coqdocvar{x}: \coqRefInductive{Top.paper}{Vec} {\CoqNat} \coqdocvar{x}), \coqdocvar{x}}
:= 
\\
$\lambda$ (\coqdocvar{x}: {\CoqNat}) (\coqdocvarP{x}: {\CoqNat}) (\coqdocvarR{x}: \coqref{Top.paper.Ded.nat R}{\coqdocdefinition{\trel{nat}}} 
\coqdocvar{x} \coqdocvarP{x}) 
(\coqdocvar{x}: \coqRefInductive{Top.paper}{Vec} {\CoqNat} \coqdocvar{x})
(\coqdocvarP{x}: \coqRefInductive{Top.paper}{Vec} {\CoqNat} \coqdocvarP{x})
(\coqdocvarR{x}: \coqref{Top.paper.DedV.Vec R}{\trel{\coqdocdefinition{Vec}}} {\CoqNat} {\CoqNat} \coqref{Top.paper.Ded.nat R}{\coqdocdefinition{\trel{nat}}}
 \highlight{\coqdocvar{x}} \highlight{\coqdocvarP{x}} \coqdocvarR{x} \coqdocvar{x} \coqdocvarP{x}), \coqdocvarR{x}\\
, which is ill-typed: the highlighted arguments to
 \coqref{Top.paper.DedV.Vec R}{\trel{\coqdocdefinition{Vec}}} have incorrect type.
 The problem is easily rectified by $\alpha$ renaming the input:
\ptranslate{$\lambda$ (\coqdocvar{y}: {\CoqNat}) (\coqdocvar{x}: \coqRefInductive{Top.paper}{Vec} {\CoqNat} \coqdocvar{y}), \coqdocvar{x}}
:= 
\\
$\lambda$ (\coqdocvar{y}: {\CoqNat}) (\coqdocvarP{y}: {\CoqNat}) (\coqdocvarR{y}: \coqref{Top.paper.Ded.nat R}{\coqdocdefinition{\trel{nat}}} 
\coqdocvar{y} \coqdocvarP{y}) 
(\coqdocvar{x}: \coqRefInductive{Top.paper}{Vec} {\CoqNat} \coqdocvar{y})
(\coqdocvarP{x}: \coqRefInductive{Top.paper}{Vec} {\CoqNat} \coqdocvarP{y})
(\coqdocvarR{x}: \coqref{Top.paper.DedV.Vec R}{\trel{\coqdocdefinition{Vec}}} {\CoqNat} {\CoqNat} \coqref{Top.paper.Ded.nat R}{\coqdocdefinition{\trel{nat}}}
 \highlight{\coqdocvar{y}} \highlight{\coqdocvarP{y}} \coqdocvarR{y} \coqdocvar{x} \coqdocvarP{x}), \coqdocvarR{x}

Because \ptranslateIso{} uses \ptranslate{} at its core, it suffers from the same problem.
In general, a natural way to fix the problem is to $\alpha$-rename the input to ensure that there are no repeated bound variables. 
However, a weaker condition suffices: the input must be in Barendregt's convention. Formally, a closed term should have no shadowed bound variables 
(nested bound variables with same name).  Open terms in a typing context, say $\Gamma$, must satisfy an additional property: their bound variables should
be distinct from variables in  $\Gamma$. 
We believe that using the weaker condition simplified some of our proofs in the next two subsections.

Recall that we have 5 disjoint classes of variables (\secref{sec:anyrel:core}). The input must only have variables of the first
class: to avoid capture the other classes are reserved for use by the translation. Also, the variable \coqdocvar{c} must
not occur in the input because it is reserved for translating universes.
We believe that having separate classes of variables resulted in simpler proofs.
Similar techniques have been used before in mechanized proofs about CPS translation~\cite{Dargaye.Leroy2007}. 

\BCGamma{l}{$t$} denotes a conjunction of such capture-safety conditions on the input $t$ in the context that binds variables $l$: 
bound variables of $t$ are disjoint from the
variables $l$, there is no shadowing of bound variables in $t$, \fvars{$t$} $\subseteq$ $l$, all variables
in $t$ are of the first class, and the variable \coqdocvar{c} does not occur in $t$. 
In $\lambda$\coqdocvar{x}:$A$.$B$ and $\forall$\coqdocvar{x}:$A$.$B$, the ``no shadowing'' condition ensures that the variable 
\coqdocvar{x} does not occur in the 
bound variables of $B$.
 \BCGamma{l}{} additionally requires that \coqdocvar{x} does not occur in the bound variables of $A$. We believe this additional
condition is not necessary, but our current proof of \lemref{lemma:alphaeqTranslate} uses it. 

\BCGamma{l}{$t$} is sufficient for \ptranslateIso{} to be well-defined upto $=_\alpha$, thus eliminating the possibility of capture:
\begin{lemma}
\label{lemma:alphaeqTranslate}
 \BCGamma{l}{$t_1$} $\rightarrow$ \BCGamma{l}{$t_2$} $\rightarrow$ \alphaeq{$t_1$}{$t_2$ } 
 $\rightarrow$ \alphaeq{\ptranslateIso{$t_1$}}{\ptranslateIso{$t_2$}} 
\end{lemma}
\subsection{Preservation of substitution}
\label{appendix:correctness:subst}
Because the typing rules of {\cocminus} mention $\beta$ equivalence~(\figref{fig:typing}), in this and the next subsection, we prove that \ptranslateIso{} preserves 
$\beta$ equivalence.

In a context that binds variables $l$, consider the term (($\lambda$ (\coqdocvar{x}:$A$), $b$) $t$). In {\cocminus}, this term will $\beta$ reduce to
$b$ [ $t$ / \coqdocvar{x}]. Thus, we need to characterize
\ptranslateIso{$b$ [ $t$ / \coqdocvar{x}]}.
As explained in the previous subsection, we require that the input to \ptranslateIso{} satisfies the \BCGamma{l}{} property . 
Thus, in {\cocminus}, we use a substitution operation that preserves it.
Let \bcsubst{l}{$b$}{\coqdocvar{x}}{$t$} denote the substitution of $t$ for \coqdocvar{x} in $b$, performed in the following way:
First $b$ is $\alpha$ renamed to $b'$, such that its bound variables are disjoint from \emph{all} the variables of $t$, 
and \BCGamma{x::l}{$b'$}. Finally, we perform a naive structurally recursive substitution, say \coqdocdefinition{unsafeSubst}, of $t$ for \coqdocvar{x} in $b'$, without
doing any further $\alpha$ renaming. It is easy to prove that 
\BCGamma{l}{\ensuremath{((\lambda} (\coqdocvar{x}:$A$), $b$) $t$)} implies 
\BCGamma{l}{(\bcsubst{l}{$b$}{\coqdocvar{x}}{$t$})}

To understand \ptranslateIso{\bcsubst{l}{$b$}{\coqdocvar{x}}{$t$}}, it is helpful to understand the free variables of \ptranslateIso{$b$}.
Let \trel{\tprime{\coqdocdefinition{lv}}} denote a function from lists of variables to lists of variables, such that
\trel{\tprime{\coqdocdefinition{lv}}} $l$ = $l$ {\CoqAppend} {\CoqMap} ($\lambda$ \coqdocvar{x}, \coqdocvarP{x}) $l$ {\CoqAppend} 
{\CoqMap} ($\lambda$ \coqdocvar{x}, \coqdocvarR{x}) $l$.
Intuitively, for every variable \coqdocvar{x} in $l$,  the list (\trel{\tprime{\coqdocdefinition{lv}}} $l$) contains not only
\coqdocvar{x} but also \coqdocvarP{x} and \coqdocvarR{x}. 
\begin{lemma}
\label{lemma:freeVars}
 \BCGamma{l}{$t$} 
 $\rightarrow$ \fvars{\ptranslateIso{$t$}} $\subseteq$ \trel{\tprime{\coqdocdefinition{lv}}} (\fvars{$t$})
\end{lemma}
\coqdocnoindent
The proof is by structural induction on $t$.

Thus, in \ptranslateIso{\bcsubst{l}{$b$}{\coqdocvar{x}}{$t$}}, if we perform the substitution \emph{after} the translation of $b$, we will need to
substitute for not only \coqdocvar{x} but also \coqdocvarP{x} and \coqdocvarR{x}:
\begin{lemma}
\label{lemma:subst}
 \BCGamma{\coqdocvar{x}::l}{$b$} $\rightarrow$ \BCGamma{l}{$t$} $\rightarrow$ 
 \alphaeq{\ptranslateIso{\bcsubst{l}{$b$}{\coqdocvar{x}}{$t$}}}{\ptranslateIso{$b$} [$t$/\coqdocvar{x}] [\tprime{$t$}/\coqdocvarP{x}] [\ptranslateIso{$t$}/\coqdocvarR{x}]} 
 \end{lemma}

Note that the RHS of the equation uses the regular capture-avoiding substitution. We only need the \emph{input} of \ptranslateIso{} to be safe.
The proof is tedious but straightforward. We begin by rewriting with $\alpha$ equality to replace the substitution operations
on both sides with \coqdocdefinition{unsafeSubst}, which is structurally recursive because it does not have to do $\alpha$ renaming
before recursing under binders. Then the proof proceeds by structural recursion on $b'$. 
For rewriting, we use \lemref{lemma:alphaeqTranslate} and the following lemma about bound variables of translations: 
\begin{lemma}
\label{lemma:boundVars}
 \BCGamma{l}{$t$} 
 $\rightarrow$ \bvars{\ptranslateIso{$t$}} $\subseteq$ \coqdocvar{c}::\coqdocvarP{c}::(\trel{\tprime{\coqdocdefinition{lv}}}$_{45}$ (\bvars{$t$}))
\end{lemma}

\subsubsection{\coqdockw{Prop} \cancel{:>} \coqdockw{Set}}
\label{appendix:correctness:subst:subtyping}
Our proof of \lemref{lemma:subst} crucially depends on the fact that substitution does not change the universe flags in $\forall$.
Thus, \ptranslateIso{} makes the same decision before and after the substitution;
\appref{appendix:correctness:typing:pi} presents \ptranslateIso{$\forall\,\hdots$} in much more detail than \secref{sec:isorel:weak}.

For the correctness of our implementation, it is also important to ensure that \emph{on well-typed inputs}, the reifier produce the same flags before and after the substitution.
This is why allowing the rule \coqdockw{Prop} :> \coqdockw{Set} in the input may be problematic.
If we had \coqdockw{Prop} :> \coqdockw{Set},
it would be legal to substitute a proposition, say {\CoqFalse}, for a variable \coqdocvar{X}:\coqdockw{Set}.
For example, the term (($\lambda$ (\coqdocvar{X}:\coqdockw{Set}), $\forall$ (\coqdocvar{x}:{\CoqNat}), \coqdocvar{X}) {\CoqFalse}) would
be well typed.
Our reifier reifies
(($\lambda$ (\coqdocvar{X}:\coqdockw{Set}), $\forall$ (\coqdocvar{x}:{\CoqNat}), \coqdocvar{X}) {\CoqFalse})
as
(($\lambda$ (\coqdocvar{X}:\coqdockw{Set}), $\forall_{S\highlight{S}}$ (\coqdocvar{x}:{\CoqNat}), \coqdocvar{X}) {\CoqFalse}),
but reifies the $\beta$ redex
($\forall$ (\coqdocvar{x}:{\CoqNat}), {\CoqFalse}) as 
($\forall_{S\highlight{P}}$ (\coqdocvar{x}:{\CoqNat}), {\CoqFalse}).
\ptranslateIso{} will thus make different decisions 
(different combinators for the \coqRefDefn{Top.paper}{Total} proof) because of the difference in flags. 
Thus the end-to-end translation (\ptranslateIso{} composed with the reifier and reflector) would not preserve this $\beta$ reduction.

Preservation of definitional equality is necessary, at least in the presence of inductive types. 
If closed terms $u$ and $v$ are definitionally equal, then 
\coqref{Coq.Init.Logic.eq refl}{\coqdocconstructor{eq\_refl}}:$u = v$.
The corresponding abstraction theorem holds iff the end-to-end translations of $u$ and $v$ are definitionally equal.

Using \coqdockw{Prop} :> \coqdockw{Set} is not always a problem: many other parts of \ptranslateIso{} do not differentiate between the two.
For example, the reduction of (($\lambda$ (\coqdocvar{X}:\coqdockw{Set}), $\forall$ (\coqdocvar{x}:\coqdocvar{X}), {\CoqNat}) {\CoqFalse}) is preserved.

\subsection{Preservation of $\beta$ equivalence}
\label{appendix:correctness:beta}
$\beta$ equivalence ($\equiv_\Gamma$), which is used in the typing rules in \figref{fig:typing}, is the conditionally reflexive, symmetric, transitive closure of the $\beta$-reduction
explained in the previous subsection. Reflexivity only holds for safe terms: $t$ $\equiv_\Gamma$ $t$ iff \BCGamma{\coqdocdefinition{vars}\,\Gamma}{$t$}.
$\coqdocdefinition{vars}\,\Gamma$ denotes the variables of the typing context $\Gamma$. 
For example, \coqdocdefinition{vars} [\coqdocvar{x}: {\CoqNat}, \coqdocvar{y}: {\CoqBool}]) = [\coqdocvar{x}; \coqdocvar{y}]. 
Overloading notation, below, \BCGamma{\Gamma}{$t$} will denote \BCGamma{\coqdocdefinition{vars}\,\Gamma}{$t$}.

In $\equiv_\Gamma$, the $\beta$ reductions steps may occur even in subterms, even under binders: when recursing under a binder, we add the variable to the context.

In a context that binds the variables $l$, the term (($\lambda$ (\coqdocvar{x}:$A$), $b$) $t$) $\beta$ reduces in {\cocminus} to
\bcsubst{l}{$b$}{\coqdocvar{x}}{$t$}.\\
\ptranslateIso{($\lambda$ (\coqdocvar{x}:$A$), $b$) $t$} :=
(($\lambda$ (\coqdocvar{x}:$A$) (\coqdocvarP{x}:\tprime{$A$}) (\coqdocvarR{x}:$\hdots$), \ptranslateIso{$b$}) $t$ \tprime{$t$} \ptranslateIso{$t$}),
which is definitionally equivalent in Coq to
\ptranslateIso{$b$} [$t$/\coqdocvar{x}] [\tprime{$t$}/\coqdocvarP{x}] [\ptranslateIso{$t$}/\coqdocvarR{x}],
which is exactly the RHS of \lemref{lemma:subst}. 
 
Using \lemref{lemma:subst}, it is easy to prove the following:
\begin{lemma}
\label{lemma:betaEquiv}
$u \equiv_\Gamma v$ $\rightarrow$ \ptranslateIso{u} $\equiv$ \ptranslateIso{v}
\end{lemma}
On the RHS, we have Coq's definitional equivalence ($\equiv$), which is unconditionally reflexive.
As mentioned before, only the input to \ptranslateIso{} needs to be in Barandregt's convention.

\subsection{Preservation of subtyping}
\label{appendix:correctness:subtyping}
The typing rules of {\cocminus}~(\figref{fig:typing}) mention the subtyping relation~(\figref{fig:subtyping}).
Thus, we prove that \ptranslateIso{} preserves the subtyping relation.
The predicate \BCGammaC{} lifts the \BCGamma{}{} property to contexts, ensuring that all types in the context are safe inputs to \ptranslateIso{}.
:>$_{\coqsubscript}$ and \vdashq are respectively the subtyping relations of Coq (CIC), not {\cocminus}.
\begin{lemma}
\label{lemma:subtyping}
\BCGamma{\Gamma}{$U$} 
$\rightarrow$ \BCGamma{\Gamma}{$V$} 
$\rightarrow$ \BCGammaC{$\Gamma$}
$\rightarrow 
 \Gamma \vdash$ $U$ :> $V$ $\rightarrow$
$\ptranslateIso{\Gamma} \vdashq$ $u$:$U$ $\rightarrow$
$\ptranslateIso{\Gamma} \vdashq$ $u'$:\tprime{$U$} $\rightarrow$
$\ptranslateIso{\Gamma} \vdashq$ $v$:$V$ $\rightarrow$
$\ptranslateIso{\Gamma} \vdashq$ $v'$:\tprime{$V$} $\rightarrow$
$\ptranslateIso{\Gamma} \vdashq$ (\projTyRel{U}{\ptranslateIso{$U$}} $u$ $u'$)  :>$_{\coqsubscript}$ (\projTyRel{V}{\ptranslateIso{$V$}} $v$ $v'$)
\end{lemma}
The proof is straightforward, by induction on the derivation of $\Gamma \vdash$ $U$ :> $V$.

\subsection{Preservation of typehood}
\label{appendix:correctness:typehood}

\begin{thm}[\label{IsoAbstraction}Abstraction Theorem]
\BCGamma{\Gamma}{$a$} 
$\rightarrow$ \BCGamma{\Gamma}{$B$} 
$\rightarrow$ \BCGammaC{$\Gamma$}
$\rightarrow$ 
$Γ ⊢ a : B$ \coqdoceol
$\rightarrow$ $\ptranslateIso{Γ} \vdashq a : B$ $\wedge$ $\ptranslateIso{Γ} \vdashq \tprime{a} :
\tprime{B}$ $\wedge$ $\ptranslateIso{Γ} \vdashq \ptranslateIso{a} : \projTyRel{B}{\ptranslateIso{B}}\, a\,
\tprime{a}$
\end{thm}
The proof is by induction on the derivation of $Γ ⊢ a : B$. In the next three subsubsections, we will look at the three cases that are most 
different between \ptranslateIso{} and \ptranslate{}.

\subsubsection{universes}
The interesting cases are $\Gamma \vdash$ \coqdockw{Set} : \coqdockw{Type}$_1$ and $\Gamma \vdash$ \coqdockw{Prop} : \coqdockw{Type}$_1$.
For $i >0$, \ptranslateIso{$\Type_i$} = \ptranslate{$\Type_i$}, so the proofs for the cases $\Gamma \vdash$ \coqdockw{Type}$_{i}$ : \coqdockw{Type}$_{i+1}$
for \ptranslateIso{} are the same as the proofs for \ptranslate{}.

Because \coqdockw{Prop} and \coqdockw{Set} are closed terms, it suffices to consider the empty context.
Below, we consider \coqdockw{Prop} : \coqdockw{Type}$_1$. The other case is similar.

We need to prove \ptranslateIso{\coqdockw{Prop}} : \projTyRel{\coqdockw{Type}_1}{\ptranslateIso{\coqdockw{Type}$_1$}} \coqdockw{Prop} \coqdockw{Prop},
which is (on unfolding definitions)  
{\coqdocnotation{\ensuremath{\lambda}}} {\coqdocnotation{(}}\coqdocvar{c} \coqdocvar{\tprime{c}}: \coqdockw{Prop}
{\coqdocnotation{),}} 
\coqRefDefn{Top.paper}{IsoRel} \coqdocvar{c} \coqdocvar{\tprime{c}}:
\coqdockw{Prop} $\rightarrow$ \coqdockw{Prop} $\rightarrow$ \coqdockw{Type}$_1$,
which boils down to
\coqdocvar{c}: \coqdockw{Prop}, \coqdocvarP{c}: \coqdockw{Prop} $\vdashq$  \coqRefDefn{Top.paper}{IsoRel} \coqdocvar{c} \coqdocvarP{c} : \coqdockw{Type}$_1$,
which (using Coq's typing rules for Inductives) mainly boils down to
\coqdocvar{c}: \coqdockw{Prop}, \coqdocvarP{c}: \coqdockw{Prop} $\vdashq$ \coqdocvar{c} $\rightarrow$ \coqdocvarP{c} $\rightarrow$ \coqdockw{Prop} : \coqdockw{Type}$_1$,
which holds because \coqdockw{Prop}: \coqdockw{Type}$_1$, \coqdockw{Prop} :>$_{\coqsubscript}$ \coqdockw{Set}, and \coqdockw{Set} :>$_{\coqsubscript}$ \coqdockw{Type}$_1$.
Note that we interpret the result of \ptranslateIso{} in Coq, not {\cocminus}. Thus, we were able to use the subtyping rules omitted in {\cocminus}.

\subsubsection{$\forall$ (rules $∀_1$ and $∀_2$ in \figref{fig:typing})}
\label{appendix:correctness:typing:pi}
W.r.t. \ptranslate{}, in \ptranslateIso{}, the interesting cases are when \ensuremath{\Gamma \vdash (∀\coqdocvar{x}\!:\! A.B)} : \coqdockw{Set} and
\ensuremath{\Gamma \vdash (∀\coqdocvar{x}\!:\! A.B)} : \coqdockw{Prop}.
We first explain how \ptranslateIso{} works in these cases in more detail (this presentation is slightly different, but equivalent to the one in \secref{sec:isorel:weak}). Correctness would then be obvious.
In these cases, our implementation merely invokes one of the following definitions (combinators) that have been already accepted (deemed well-typed) by Coq:

\coqdoceol
\coqdocnoindent
\begin{minipage}[t]{\textwidth}
\coqdoceol
\coqdocnoindent
\coqdockw{Definition} \coqdef{Top.paper.piProp}{piProp}{\coqdocdefinition{piProp}} (\coqdocvar{A1} \coqdocvar{A2} :\coqdockw{Set}) (\coqdocvarR{A}: \coqref{Top.paper.IsoRel}{\coqdocdefinition{IsoRel}} \coqdocvariable{A1} \coqdocvariable{A2}) 
(\coqdocvar{B1}: \coqdocvariable{A1} \coqexternalref{:type scope:x '->' x}{http://coq.inria.fr/distrib/8.6/stdlib/Coq.Init.Logic}{\coqdocnotation{\ensuremath{\rightarrow}}} \coqdockw{Prop}) 
(\coqdocvar{B2}: \coqdocvariable{A2} \coqexternalref{:type scope:x '->' x}{http://coq.inria.fr/distrib/8.6/stdlib/Coq.Init.Logic}{\coqdocnotation{\ensuremath{\rightarrow}}} \coqdockw{Prop}) \coqdoceol
\coqdocindent{1.00em}
(\coqdocvarR{B}: \coqdockw{\ensuremath{\forall}} \coqdocvar{a1} \coqdocvar{a2}, {\CoqSigTProj} \coqdocvarR{A} \coqdocvariable{a1} \coqdocvariable{a2} \coqexternalref{:type scope:x '->' x}{http://coq.inria.fr/distrib/8.6/stdlib/Coq.Init.Logic}{\coqdocnotation{\ensuremath{\rightarrow}}} \coqref{Top.paper.IsoRel}{\coqdocdefinition{IsoRel}} (\coqdocvariable{B1} \coqdocvariable{a1}) (\coqdocvariable{B2} \coqdocvariable{a2}))
: \coqref{Top.paper.IsoRel}{\coqdocdefinition{IsoRel}} (\coqdockw{\ensuremath{\forall}} \coqdocvar{a} : \coqdocvariable{A1}, \coqdocvariable{B1} \coqdocvariable{a}) (\coqdockw{\ensuremath{\forall}} \coqdocvar{a} : \coqdocvariable{A2}, \coqdocvariable{B2} \coqdocvariable{a}) :=
\coqdoceol\coqdocindent{2em}
{\CoqExistT}
($\lambda$ (\coqdocvar{f1} : \coqdockw{\ensuremath{\forall}} \coqdocvar{a} : \coqdocvar{A1}, \coqdocvar{B1} \coqdocvariable{a}) (\coqdocvar{f2} : \coqdockw{\ensuremath{\forall}} \coqdocvar{a} : \coqdocvar{A2}, \coqdocvar{B2} \coqdocvariable{a}) \ensuremath{\Rightarrow}
\coqdoceol\coqdocindent{7em}
\coqdockw{\ensuremath{\forall}} (\coqdocvar{a1} : \coqdocvar{A1}) (\coqdocvar{a2} : \coqdocvar{A2}) (\coqdocvar{p} : {\CoqSigTProj} \coqdocvarR{A} \coqdocvariable{a1} \coqdocvariable{a2}), {\CoqSigTProj} (\coqdocvarR{B} \coqdocvariable{a1} \coqdocvariable{a2} \coqdocvariable{p}) (\coqdocvariable{f1} \coqdocvariable{a1}) (\coqdocvariable{f2} \coqdocvariable{a2}))
\coqdoceol\coqdocindent{4.5em} ($\hdots$).
\end{minipage}

\coqdoceol
\coqdocnoindent
\begin{minipage}[t]{\textwidth}
\coqdoceol\coqdocemptyline
\coqdocnoindent
\coqdockw{Definition} \coqdef{Top.paper.piSet}{piSet}{\coqdocdefinition{piSet}} (\coqdocvar{A1} \coqdocvar{A2} :\coqdockw{Set}) (\coqdocvarR{A}: \coqref{Top.paper.IsoRel}{\coqdocdefinition{IsoRel}} \coqdocvariable{A1} \coqdocvariable{A2}) 
(\coqdocvar{B1}: \coqdocvariable{A1} \coqexternalref{:type scope:x '->' x}{http://coq.inria.fr/distrib/8.6/stdlib/Coq.Init.Logic}{\coqdocnotation{\ensuremath{\rightarrow}}} \coqdockw{Set}) 
(\coqdocvar{B2}: \coqdocvariable{A2} \coqexternalref{:type scope:x '->' x}{http://coq.inria.fr/distrib/8.6/stdlib/Coq.Init.Logic}{\coqdocnotation{\ensuremath{\rightarrow}}} \coqdockw{Set}) \coqdoceol
\coqdocindent{1.00em}
(\coqdocvarR{B}: \coqdockw{\ensuremath{\forall}} \coqdocvar{a1} \coqdocvar{a2}, {\CoqSigTProj} \coqdocvarR{A} \coqdocvariable{a1} \coqdocvariable{a2} \coqexternalref{:type scope:x '->' x}{http://coq.inria.fr/distrib/8.6/stdlib/Coq.Init.Logic}{\coqdocnotation{\ensuremath{\rightarrow}}} \coqref{Top.paper.IsoRel}{\coqdocdefinition{IsoRel}} (\coqdocvariable{B1} \coqdocvariable{a1}) (\coqdocvariable{B2} \coqdocvariable{a2}))
: \coqref{Top.paper.IsoRel}{\coqdocdefinition{IsoRel}} (\coqdockw{\ensuremath{\forall}} \coqdocvar{a} : \coqdocvariable{A1}, \coqdocvariable{B1} \coqdocvariable{a}) (\coqdockw{\ensuremath{\forall}} \coqdocvar{a} : \coqdocvariable{A2}, \coqdocvariable{B2} \coqdocvariable{a}) := 
\coqdoceol\coqdocindent{2em}
{\CoqExistT}
($\lambda$ (\coqdocvar{f1} : \coqdockw{\ensuremath{\forall}} \coqdocvar{a} : \coqdocvar{A1}, \coqdocvar{B1} \coqdocvariable{a}) (\coqdocvar{f2} : \coqdockw{\ensuremath{\forall}} \coqdocvar{a} : \coqdocvar{A2}, \coqdocvar{B2} \coqdocvariable{a}) \ensuremath{\Rightarrow}
\coqdoceol\coqdocindent{7em}
\coqdockw{\ensuremath{\forall}} (\coqdocvar{a1} : \coqdocvar{A1}) (\coqdocvar{a2} : \coqdocvar{A2}) (\coqdocvar{p} : {\CoqSigTProj} \coqdocvarR{A} \coqdocvariable{a1} \coqdocvariable{a2}), {\CoqSigTProj} (\coqdocvarR{B} \coqdocvariable{a1} \coqdocvariable{a2} \coqdocvariable{p}) (\coqdocvariable{f1} \coqdocvariable{a1}) (\coqdocvariable{f2} \coqdocvariable{a2}))
\coqdoceol\coqdocindent{4.5em} ($\hdots$).
\coqdocemptyline
\coqdocemptyline
\coqdocemptyline
\end{minipage}

\coqdocnoindent
The bodies of these definitions are dependent pairs whose first components are essentially the {\anyrel} translations of $\Pi$ types.
The second components are huge and thus shown as $\hdots$: 
they are proofs of the \coqRefDefn{Top.paper}{Total} and \coqRefDefn{Top.paper}{OneToOne} properties, which 
were already explained respectively in \secref{sec:uniformProp:piProp} (also \lemref{lemma:unify}) and \secref{sec:uniformProp:funt}.
\ptranslateIso{} merely refers to one of these two constants by name (the string ``piSet'' or ``piProp'') and then applies the six arguments.
The reflector turns those strings to references to the above definitions.

Using these definitions (instead of constructing their bodies by hand in \ptranslateIso{}) greatly simplified our implementation and proofs. 
Many proofs, e.g. \lemref{appendix:correctness:subst}, did not have to reason about the horrendously complex bodies of those definitions. 
In \lemref{appendix:correctness:subst}, we only had to perform the substitution on the arguments
to the constant, which are relatively very simple, as we will show soon. 
Even in this subsection, we don't need to reason about the correctness of the proof parts shown above as $\hdots$, 
because Coq has already checked them for us! Below, we will merely argue that the arguments to the above definitions are of correct types.

When \ensuremath{\Gamma \vdash (∀\coqdocvar{x}\!:\! A.B)} : \coqdockw{Prop}, as indicated by the flags ($\forall_{SP}$ or $\forall_{PP}$),
\ptranslateIso{} invokes the lemma \coqRefDefn{Top.paper}{piProp}.
When \ensuremath{\Gamma \vdash (∀\coqdocvar{x}\!:\! A.B)} : \coqdockw{Set}, as indicated by the flags ($\forall_{SS}$ or $\forall_{PS}$),
\ptranslateIso{} invokes the lemma \coqRefDefn{Top.paper}{piSet}. 
The lemma \coqRefDefn{Top.paper}{piProp} uses fewer assumptions about the arguments \coqdocvarR{A} and \coqdocvarR{B} because it exploits proof irrelevance.
Thus, it is important to prefer the lemma \coqRefDefn{Top.paper}{piProp}.
\emph{Note that \coqdockw{Prop} :>$_{\coqsubscript}$ \coqdockw{Set}, even though \coqdockw{Prop} \cancel{:>} \coqdockw{Set}}.
In both cases, the 6 arguments are the same:
(\coqdocvar{A1} := $A$),
(\coqdocvar{A2} := \tprime{$A$}),
(\coqdocvarR{A} := \ptranslateIso{$A$}),
(\coqdocvar{B1} := $\lambda$ (\coqdocvar{x}:$A$), $B$),
(\coqdocvar{B2} := $\lambda$ (\coqdocvarP{x}:\tprime{$A$}), \tprime{$B$}), and
(\coqdocvarR{B} := $\lambda$ (\coqdocvar{x}:$A$) (\coqdocvarP{x}:\tprime{$A$}) (\coqdocvarR{x}: {{\CoqSigTProj}} {\ptranslateIso{A}} \coqdocvar{x} \coqdocvarP{x}), \ptranslateIso{$B$}).
These arguments are in the context \ptranslateIso{Γ}.

We will consider the case \ensuremath{\Gamma \vdash (∀\coqdocvar{x}\!:\! A.B)} : \coqdockw{Prop} (rule $\forall_2$ in \figref{fig:typing}, with $A$:\coqdockw{Set}). 
The other cases are similar.
It is easy to check that for the above instantiation, the return type is correct (exactly what the abstraction theorem needs).
The correctness of the types of the arguments follows from the induction 
hypotheses (abstraction theorems for the two premises of the rule in \figref{fig:typing}).
The two induction hypotheses (after unfolding definitions) are:
$\ptranslateIso{Γ} \vdashq A : \coqdockw{Set}$ $\wedge$ $\ptranslateIso{Γ} \vdashq \tprime{A} :
\coqdockw{Set}$ $\wedge$ $\ptranslateIso{Γ} \vdashq \ptranslateIso{A} : \coqRefDefn{Top.paper}{IsoRel}\, A\,
\tprime{A}$
and \\
$\ptranslateIso{Γ},\,
\coqdocvar{x}: A,\,
\coqdocvarP{x}: \tprime{A},\,
\coqdocvarR{x}: {{\CoqSigTProj}} \ptranslateIso{A} \coqdocvar{x}\,\coqdocvarP{x},\, 
\vdashq B : \coqdockw{Prop}$ 
$\wedge$ \\
$\ptranslateIso{Γ},\,
\coqdocvar{x}: A,\,
\coqdocvarP{x}: \tprime{A},\,
\coqdocvarR{x}: {{\CoqSigTProj}} \ptranslateIso{A} \coqdocvar{x}\,\coqdocvarP{x},\, 
\vdashq \tprime{B} :
\coqdockw{Prop}$ $\wedge$ \\
$\ptranslateIso{Γ},\,
\coqdocvar{x}: A,\,
\coqdocvarP{x}: \tprime{A},\,
\coqdocvarR{x}: {{\CoqSigTProj}} \ptranslateIso{A} \coqdocvar{x}\,\coqdocvarP{x},\, 
\vdashq \ptranslateIso{B}:
\coqRefDefn{Top.paper}{IsoRel}
\, B\,
\tprime{B}$

\noindent
$B$ does not mention the variables \coqdocvarP{x} and \coqdocvarR{x}.
\tprime{$B$} does not mention the variables \coqdocvar{x} and \coqdocvarR{x}.
Thus, the instantiations of \coqdocvar{B1} and \coqdocvar{B2} are well-defined (and correctly typed) in the context \ptranslateIso{Γ}.

\begin{lemma}
\label{lemma:freeVarsPrime}
 \fvars{\tprime{$t$}} = {\CoqMap} ($\lambda$ \coqdocvar{x}, \coqdocvarP{x}) (\fvars{{$t$}}) 
\end{lemma}


\subsubsection{$\lambda$ (rule $ABS$ in \figref{fig:typing})}
\label{appendix:correctness:typing:lam}
Now we consider the case $Γ ⊢ (λ\coqdocvar{x} :A.B) : (∀\coqdocvar{x}:A.C$)

W.r.t. \ptranslate{}, in \ptranslateIso{}, the interesting case is when the type $∀\coqdocvar{x}:A.C$ is in the universe \coqdockw{Set} or \coqdockw{Prop}.
Consider the case when ($∀\coqdocvar{x}:A.C$):\coqdockw{Prop}.
We need to prove that in the typing context $\ptranslateIso{Γ}$,  $\ptranslateIso{λ\coqdocvar{x} :A.B}$ has type
 $\projTyRel{(∀\coqdocvar{x}:A.C)}{\ptranslateIso{∀\coqdocvar{x}:A.C}} (λ\coqdocvar{x} :A.B)
 (λ\coqdocvarP{x} :\tprime{A}.\tprime{B})$, which (as explained in the previous subsubsection) is 
({{\CoqSigTProj}} (\coqRefDefn{Top.paper}{piProp} $A$ $\tprime{A}$ $\ptranslateIso{A}$ $\hdots$)) $(λ\coqdocvar{x} :A.B)
 (λ\coqdocvarP{x} :\tprime{A}.\tprime{B})$.
 Coq's definitional equality includes $\delta$ and $\iota$ reductions (definition unfolding and pattern matching). After 
 unfolding the definition of \coqRefDefn{Top.paper}{piProp}, we get a dependent pair.
 {{\CoqSigTProj}} then acts on the pair ($\iota$ reduction) to produce the first component, which is essentially the {\anyrel} translation
 of $∀\coqdocvar{x}:A.C$. The second component (proofs of \coqRefDefn{Top.paper}{Total} and \coqRefDefn{Top.paper}{OneToOne} properties) 
 get thrown away by {{\CoqSigTProj}}. The rest of this proof is essentially the same as that for
 the {\anyrel} translation (\ptranslate{}).




\subsection{Translation of the W type and its induction principle}
\label{appendix:correctness:W}
\coqRefDefn{Top.paper}{IWTind} is a general induction (recursion) principle for the type \coqRefDefn{Top.paper}{IWT} in \secref{sec:uniformProp:indt}.
\coqdoceol\coqdocnoindent
\coqdockw{Definition} \coqdef{Top.paper.IWTind}{IWTind}{\coqdocdefinition{IWTind}} :=
\coqdoceol\coqdocnoindent
$\lambda$ (\coqdocvar{I} \coqdocvar{A} : \coqdockw{Set}) (\coqdocvar{B} : \coqdocvariable{A} \coqexternalref{:type scope:x '->' x}{http://coq.inria.fr/distrib/8.6/stdlib/Coq.Init.Logic}{\coqdocnotation{\ensuremath{\rightarrow}}} \coqdockw{Set}) (\coqdocvar{AI} : \coqdocvariable{A} \coqexternalref{:type scope:x '->' x}{http://coq.inria.fr/distrib/8.6/stdlib/Coq.Init.Logic}{\coqdocnotation{\ensuremath{\rightarrow}}} \coqdocvariable{I}) (\coqdocvar{BI} : \coqdockw{\ensuremath{\forall}} \coqdocvar{a} : \coqdocvariable{A}, \coqdocvariable{B} \coqdocvariable{a} \coqexternalref{:type scope:x '->' x}{http://coq.inria.fr/distrib/8.6/stdlib/Coq.Init.Logic}{\coqdocnotation{\ensuremath{\rightarrow}}} \coqdocvariable{I})\coqdoceol
\coqdocindent{1.00em}
(\coqdocvar{P} : \coqdockw{\ensuremath{\forall}} \coqdocvar{i} : \coqdocvariable{I}, \coqref{Top.paper.IWT}{\coqdocinductive{IWT}} \coqdocvariable{I} \coqdocvariable{A} \coqdocvariable{B} \coqdocvariable{AI} \coqdocvariable{BI} \coqdocvariable{i} \coqexternalref{:type scope:x '->' x}{http://coq.inria.fr/distrib/8.6/stdlib/Coq.Init.Logic}{\coqdocnotation{\ensuremath{\rightarrow}}} \coqdockw{Set})\coqdoceol
\coqdocindent{1.00em}
(\coqdocvar{f} : \coqdockw{\ensuremath{\forall}} (\coqdocvar{a} : \coqdocvariable{A}) (\coqdocvar{lim} : \coqdockw{\ensuremath{\forall}} \coqdocvar{b} : \coqdocvariable{B} \coqdocvariable{a}, \coqref{Top.paper.IWT}{\coqdocinductive{IWT}} \coqdocvariable{I} \coqdocvariable{A} \coqdocvariable{B} \coqdocvariable{AI} \coqdocvariable{BI} (\coqdocvariable{BI} \coqdocvariable{a} \coqdocvariable{b})),\coqdoceol
\coqdocindent{3.50em}
\coqexternalref{:type scope:x '->' x}{http://coq.inria.fr/distrib/8.6/stdlib/Coq.Init.Logic}{\coqdocnotation{(}}\coqdockw{\ensuremath{\forall}} \coqdocvar{b} : \coqdocvariable{B} \coqdocvariable{a}, \coqdocvariable{P} (\coqdocvariable{BI} \coqdocvariable{a} \coqdocvariable{b}) (\coqdocvariable{lim} \coqdocvariable{b})\coqexternalref{:type scope:x '->' x}{http://coq.inria.fr/distrib/8.6/stdlib/Coq.Init.Logic}{\coqdocnotation{)}} \coqexternalref{:type scope:x '->' x}{http://coq.inria.fr/distrib/8.6/stdlib/Coq.Init.Logic}{\coqdocnotation{\ensuremath{\rightarrow}}} \coqdocvariable{P} (\coqdocvariable{AI} \coqdocvariable{a}) (\coqref{Top.paper.iwt}{\coqdocconstructor{iwt}} \coqdocvariable{I} \coqdocvariable{A} \coqdocvariable{B} \coqdocvariable{AI} \coqdocvariable{BI} \coqdocvariable{a} \coqdocvariable{lim})), \coqdoceol
\coqdocnoindent
\coqdockw{fix} \coqdocvar{F} (\coqdocvar{i} : \coqdocvariable{I}) (\coqdocvar{i0} : \coqref{Top.paper.IWT}{\coqdocinductive{IWT}} \coqdocvariable{I} \coqdocvariable{A} \coqdocvariable{B} \coqdocvariable{AI} \coqdocvariable{BI} \coqdocvariable{i}) \{\coqdockw{struct} \coqdocvar{i0}\} : \coqdocvariable{P} \coqdocvariable{i} \coqdocvariable{i0} :=\coqdoceol
\coqdocindent{1.00em}
\coqdockw{match} \coqdocvariable{i0} \coqdockw{as} \coqdocvar{i2} \coqdoctac{in} (\coqref{Top.paper.IWT}{\coqdocinductive{IWT}} \coqdocvar{\_} \coqdocvar{\_} \coqdocvar{\_} \coqdocvar{\_} \coqdocvar{\_} \coqdocvar{i1}) \coqdockw{return} (\coqdocvariable{P} \coqdocvariable{i1} \coqdocvariable{i2}) \coqdockw{with}\coqdoceol
\coqdocindent{1.00em}
\ensuremath{|} \coqref{Top.paper.iwt}{\coqdocconstructor{iwt}} \coqdocvar{\_} \coqdocvar{\_} \coqdocvar{\_} \coqdocvar{\_} \coqdocvar{\_} \coqdocvar{a} \coqdocvar{lim} \ensuremath{\Rightarrow} \coqdocvariable{f} \coqdocvar{a} \coqdocvar{lim} ($\lambda$ (\coqdocvar{b} : \coqdocvariable{B} \coqdocvar{a}), \coqdocvariable{F} (\coqdocvariable{BI} \coqdocvar{a} \coqdocvariable{b}) (\coqdocvar{lim} \coqdocvariable{b}))\coqdoceol
\coqdocindent{1.00em}
\coqdockw{end}.\coqdoceol

Just as inductive types can be encoded as instantiations of \coqRefDefn{Top.paper}{IWT}, Coq's pattern matching and fixpoints (recursive functions)
can be encoded as instantiations of \coqRefDefn{Top.paper}{IWTind}.
Thus, we checked that \ptranslateIso{\coqRefInductive{Top.paper}{IWT}} and \ptranslateIso{\coqRefDefn{Top.paper}{IWTind}} succeed and are of correct type.
We also checked that in the most general context, \ptranslateIso{} preserves the 
reduction (unfolding \coqdockw{fix} and $\iota$ reduction of pattern matching) of \coqRefDefn{Top.paper}{IWTind}.
(As explained in the above subsections, preservation of reduction is a step in proving that \ptranslateIso{} preserves typing.)

In Coq, reductions can happen even under binders.
Thus, below we pick terms \coqRefDefn{Top.paper}{LHS} and \coqRefDefn{Top.paper}{RHS} which observe the reduction of \coqRefDefn{Top.paper}{IWTind}
in the most general context. \coqRefDefn{Top.paper}{LHS} and \coqRefDefn{Top.paper}{RHS} are the same except the highlighted part.
\coqRefDefn{Top.paper}{LHS} reduces to \coqRefDefn{Top.paper}{RHS}.

\coqdoceol\coqdocnoindent
\begin{minipage}[t]{\textwidth}
\coqdoceol\coqdocnoindent
\coqdockw{Definition} \coqdef{Top.paper.LHS}{LHS}{\coqdocdefinition{LHS}} :=
\coqdoceol\coqdocnoindent
$\lambda$ (\coqdocvar{I} \coqdocvar{A} : \coqdockw{Set}) (\coqdocvar{B} : \coqdocvariable{A} \coqexternalref{:type scope:x '->' x}{http://coq.inria.fr/distrib/8.6/stdlib/Coq.Init.Logic}{\coqdocnotation{\ensuremath{\rightarrow}}} \coqdockw{Set}) (\coqdocvar{AI} : \coqdocvariable{A} \coqexternalref{:type scope:x '->' x}{http://coq.inria.fr/distrib/8.6/stdlib/Coq.Init.Logic}{\coqdocnotation{\ensuremath{\rightarrow}}} \coqdocvariable{I})  (\coqdocvar{BI} : \coqdockw{\ensuremath{\forall}} (\coqdocvar{a} : \coqdocvariable{A}), \coqdocvariable{B} \coqdocvariable{a} \coqexternalref{:type scope:x '->' x}{http://coq.inria.fr/distrib/8.6/stdlib/Coq.Init.Logic}{\coqdocnotation{\ensuremath{\rightarrow}}} \coqdocvariable{I})\coqdoceol
\coqdocindent{1.00em}
(\coqdocvar{P} : \coqdockw{\ensuremath{\forall}} \coqdocvar{i} : \coqdocvariable{I}, \coqref{Top.paper.IWT}{\coqdocinductive{IWT}} \coqdocvariable{I} \coqdocvariable{A} \coqdocvariable{B} \coqdocvariable{AI} \coqdocvariable{BI} \coqdocvariable{i} \coqexternalref{:type scope:x '->' x}{http://coq.inria.fr/distrib/8.6/stdlib/Coq.Init.Logic}{\coqdocnotation{\ensuremath{\rightarrow}}} \coqdockw{Set})\coqdoceol
\coqdocindent{1.00em}
(\coqdocvar{f} : \coqdockw{\ensuremath{\forall}} (\coqdocvar{a} : \coqdocvariable{A}) (\coqdocvar{lim} : \coqdockw{\ensuremath{\forall}} \coqdocvar{b} : \coqdocvariable{B} \coqdocvariable{a}, \coqref{Top.paper.IWT}{\coqdocinductive{IWT}} \coqdocvariable{I} \coqdocvariable{A} \coqdocvariable{B} \coqdocvariable{AI} \coqdocvariable{BI} (\coqdocvariable{BI} \coqdocvariable{a} \coqdocvariable{b})),\coqdoceol
\coqdocindent{3.50em}
\coqexternalref{:type scope:x '->' x}{http://coq.inria.fr/distrib/8.6/stdlib/Coq.Init.Logic}{\coqdocnotation{(}}\coqdockw{\ensuremath{\forall}} \coqdocvar{b} : \coqdocvariable{B} \coqdocvariable{a}, \coqdocvariable{P} (\coqdocvariable{BI} \coqdocvariable{a} \coqdocvariable{b}) (\coqdocvariable{lim} \coqdocvariable{b})\coqexternalref{:type scope:x '->' x}{http://coq.inria.fr/distrib/8.6/stdlib/Coq.Init.Logic}{\coqdocnotation{)}} \coqexternalref{:type scope:x '->' x}{http://coq.inria.fr/distrib/8.6/stdlib/Coq.Init.Logic}{\coqdocnotation{\ensuremath{\rightarrow}}} \coqdocvariable{P} (\coqdocvariable{AI} \coqdocvariable{a}) (\coqref{Top.paper.iwt}{\coqdocconstructor{iwt}} \coqdocvariable{I} \coqdocvariable{A} \coqdocvariable{B} \coqdocvariable{AI} \coqdocvariable{BI} \coqdocvariable{a} \coqdocvariable{lim}))\coqdoceol
\coqdocindent{1.00em}
(\coqdocvar{a}:\coqdocvariable{A}) \coqdoceol
\coqdocindent{1.00em}
(\coqdocvar{lim}: \coqdockw{\ensuremath{\forall}} \coqdocvar{b} : \coqdocvariable{B} \coqdocvariable{a}, \coqref{Top.paper.IWT}{\coqdocinductive{IWT}} \coqdocvariable{I} \coqdocvariable{A} \coqdocvariable{B} \coqdocvariable{AI} \coqdocvariable{BI} (\coqdocvariable{BI} \coqdocvariable{a} \coqdocvariable{b})), \coqdoceol
\coqdocindent{1.00em}
\highlight{\coqref{Top.paper.IWTind}{\coqdocdefinition{IWTind}} \coqdocvar{\_} \coqdocvar{\_} \coqdocvar{\_} \coqdocvar{\_} \coqdocvar{\_} \coqdocvar{\_} \coqdocvariable{f} \coqdocvar{\_} (\coqref{Top.paper.iwt}{\coqdocconstructor{iwt}} \coqdocvar{\_} \coqdocvar{\_} \coqdocvar{\_} \coqdocvar{\_} \coqdocvar{\_} \coqdocvariable{a} \coqdocvariable{lim})}.\coqdoceol
\coqdocemptyline
\coqdocnoindent
\coqdockw{Definition} \coqdef{Top.paper.RHS}{RHS}{\coqdocdefinition{RHS}} :=\coqdoceol
\coqdocnoindent
$\lambda$ (\coqdocvar{I} \coqdocvar{A} : \coqdockw{Set}) (\coqdocvar{B} : \coqdocvariable{A} \coqexternalref{:type scope:x '->' x}{http://coq.inria.fr/distrib/8.6/stdlib/Coq.Init.Logic}{\coqdocnotation{\ensuremath{\rightarrow}}} \coqdockw{Set}) (\coqdocvar{AI} : \coqdocvariable{A} \coqexternalref{:type scope:x '->' x}{http://coq.inria.fr/distrib/8.6/stdlib/Coq.Init.Logic}{\coqdocnotation{\ensuremath{\rightarrow}}} \coqdocvariable{I})  (\coqdocvar{BI} : \coqdockw{\ensuremath{\forall}} (\coqdocvar{a} : \coqdocvariable{A}), \coqdocvariable{B} \coqdocvariable{a} \coqexternalref{:type scope:x '->' x}{http://coq.inria.fr/distrib/8.6/stdlib/Coq.Init.Logic}{\coqdocnotation{\ensuremath{\rightarrow}}} \coqdocvariable{I})\coqdoceol
\coqdocindent{1.00em}
(\coqdocvar{P} : \coqdockw{\ensuremath{\forall}} \coqdocvar{i} : \coqdocvariable{I}, \coqref{Top.paper.IWT}{\coqdocinductive{IWT}} \coqdocvariable{I} \coqdocvariable{A} \coqdocvariable{B} \coqdocvariable{AI} \coqdocvariable{BI} \coqdocvariable{i} \coqexternalref{:type scope:x '->' x}{http://coq.inria.fr/distrib/8.6/stdlib/Coq.Init.Logic}{\coqdocnotation{\ensuremath{\rightarrow}}} \coqdockw{Set})\coqdoceol
\coqdocindent{1.00em}
(\coqdocvar{f} : \coqdockw{\ensuremath{\forall}} (\coqdocvar{a} : \coqdocvariable{A}) (\coqdocvar{lim} : \coqdockw{\ensuremath{\forall}} \coqdocvar{b} : \coqdocvariable{B} \coqdocvariable{a}, \coqref{Top.paper.IWT}{\coqdocinductive{IWT}} \coqdocvariable{I} \coqdocvariable{A} \coqdocvariable{B} \coqdocvariable{AI} \coqdocvariable{BI} (\coqdocvariable{BI} \coqdocvariable{a} \coqdocvariable{b})),\coqdoceol
\coqdocindent{3.50em}
\coqexternalref{:type scope:x '->' x}{http://coq.inria.fr/distrib/8.6/stdlib/Coq.Init.Logic}{\coqdocnotation{(}}\coqdockw{\ensuremath{\forall}} \coqdocvar{b} : \coqdocvariable{B} \coqdocvariable{a}, \coqdocvariable{P} (\coqdocvariable{BI} \coqdocvariable{a} \coqdocvariable{b}) (\coqdocvariable{lim} \coqdocvariable{b})\coqexternalref{:type scope:x '->' x}{http://coq.inria.fr/distrib/8.6/stdlib/Coq.Init.Logic}{\coqdocnotation{)}} \coqexternalref{:type scope:x '->' x}{http://coq.inria.fr/distrib/8.6/stdlib/Coq.Init.Logic}{\coqdocnotation{\ensuremath{\rightarrow}}} \coqdocvariable{P} (\coqdocvariable{AI} \coqdocvariable{a}) (\coqref{Top.paper.iwt}{\coqdocconstructor{iwt}} \coqdocvariable{I} \coqdocvariable{A} \coqdocvariable{B} \coqdocvariable{AI} \coqdocvariable{BI} \coqdocvariable{a} \coqdocvariable{lim}))\coqdoceol
\coqdocindent{1.00em}
(\coqdocvar{a}:\coqdocvariable{A}) \coqdoceol
\coqdocindent{1.00em}
(\coqdocvar{lim}: \coqdockw{\ensuremath{\forall}} \coqdocvar{b} : \coqdocvariable{B} \coqdocvariable{a}, \coqref{Top.paper.IWT}{\coqdocinductive{IWT}} \coqdocvariable{I} \coqdocvariable{A} \coqdocvariable{B} \coqdocvariable{AI} \coqdocvariable{BI} (\coqdocvariable{BI} \coqdocvariable{a} \coqdocvariable{b})), \coqdoceol
\coqdocindent{1.00em}
\highlight{\coqdocvariable{f} \coqdocvariable{a} \coqdocvariable{lim} ($\lambda$ (\coqdocvar{b} : \coqdocvariable{B} \coqdocvariable{a}), \coqref{Top.paper.IWTind}{\coqdocdefinition{IWTind}} \coqdocvar{\_} \coqdocvar{\_} \coqdocvar{\_} \coqdocvar{\_} \coqdocvar{\_} \coqdocvar{\_} \coqdocvariable{f} (\coqdocvariable{BI} \coqdocvariable{a} \coqdocvariable{b}) (\coqdocvariable{lim} \coqdocvariable{b}))}.\coqdoceol
\coqdocemptyline
\coqdocemptyline
\coqdocemptyline
\end{minipage}
 
We observed that \ptranslateIso{\coqRefDefn{Top.paper}{LHS}} and \ptranslateIso{\coqRefDefn{Top.paper}{RHS}} succeed and
that \ptranslateIso{\coqRefDefn{Top.paper}{LHS}} is definitionally equal to \ptranslateIso{\coqRefDefn{Top.paper}{RHS}}.
One way to check that terms $u$ and $v$ are definitionally equal is to ask Coq to check 
(\coqref{Coq.Init.Logic.eq refl}{\coqdocconstructor{eq\_refl}} :  ($u$ = $v$)). We used this method.

Although we did only one reduction experiment for \coqRefDefn{Top.paper}{IWTind}, because Coq's reductions are preserved under substitutions (and how \ptranslateIso{} translates $\lambda$ and application terms),
we have hereby proved that reductions of \coqRefDefn{Top.paper}{IWTind} in \emph{all} well-typed instantiations are preserved.

In our implementation repository (\url{https://github.com/aa755/paramcoq-iff}), the experiments in this subsection can be found in
the file {\tt test-suite/iso/IWTS.v}

\end{document}